\documentclass[letterpaper,11pt]{article}

\newif\ifSODAsubmission
\SODAsubmissionfalse

\usepackage{hyperref}

\usepackage{commands-tam}
\usepackage{hyperref}
\usepackage{subcaption}% \captionsetup{compatibility=false}
\usepackage{mathrsfs}
\usepackage{wrapfig}
\usepackage{enumitem}

\usepackage{amsmath}
\usepackage{amssymb}
\usepackage{amsthm}
\usepackage{thmtools, thm-restate}

\usepackage{algorithm}
\usepackage{algpseudocode}

\declaretheorem[]{lemma}
\declaretheorem[]{corollary}

\usepackage{multirow}
\usepackage[capitalize, noabbrev]{cleveref}
\usepackage[dvipsnames]{xcolor}
\usepackage{tikz}

\usepackage{cancel}
\usepackage{wasysym}
\usepackage{readarray}

\usepackage{listings}
\usepackage{multicol}

\usepackage{commands-tam}
\usepackage[utf8]{inputenc}
\usepackage{subcaption}% \captionsetup{compatibility=false}
\usepackage{blkarray}

\usepackage{wrapfig}
\usepackage{booktabs, multirow} % for borders and merged ranges
\usepackage{soul}% for underlines
\usepackage{xcolor,colortbl} % for cell colors
\usepackage{changepage,threeparttable} % for wide tables
\usepackage{cleveref}

\crefname{@theorem}{theorem}{theorems}
\Crefname{@theorem}{Theorem}{Theorems}

\usepackage{tikz}
\usetikzlibrary{shapes.geometric, arrows, decorations.pathmorphing, calc, decorations.pathreplacing, math, matrix, decorations.markings}
\usetikzlibrary{intersections, fit, shapes.symbols, automata, positioning, arrows.meta}
\usetikzlibrary{tikzmark}
\tikzstyle{startstop} = [rectangle, rounded corners, 
minimum width=3cm, 
minimum height=1cm,
text centered, 
draw=black, 
fill=red!30]

\tikzstyle{io} = [trapezium, 
trapezium stretches=true, % A later addition
trapezium left angle=70, 
trapezium right angle=110, 
minimum width=3cm, 
minimum height=1cm, text centered, 
draw=black, fill=blue!30]

\tikzstyle{process} = [rectangle, 
minimum width=3cm, 
minimum height=1cm, 
text centered, 
text width=3cm, 
draw=black, 
fill=orange!30]

\tikzstyle{decision} = [diamond, 
minimum width=3cm, 
minimum height=1cm, 
text centered, 
draw=black, 
fill=green!30]
\tikzstyle{arrow} = [thick,->,>=stealth]
\newcommand{\shapeasm}[1]{\mathcal{S}_{\Box}[\mathcal{#1}]}

\newcommand{\calS}{\ensuremath{\mathcal{S}}}
\newcommand{\calU}{\ensuremath{\mathcal{U}}}
\newcommand{\Z}{\mathbb{Z}}
\usepackage{stmaryrd}
\newcommand{\intinterval}[1]{\llbracket #1 \rrbracket}

\newtheorem{definition}{Definition}[section]

\ifSODAsubmission
    \newcommand{\todo}[1]{}
\else
    \usepackage[textsize=scriptsize]{todonotes}
    
\fi

\usepackage[utf8]{inputenc}
\usepackage{subcaption}% \captionsetup{compatibility=false}

\usepackage[letterpaper,left=1in, right=1in, top=1in, bottom=1in]{geometry}

\usepackage{lineno} % Load the line numbering package
\usepackage[backend=biber,maxbibnames=99]{biblatex}
\begin{document}

\title{The Rhythm of the Tiles: Narrow and Low Temperature Computation and Simulation in Synchronous Tile Self-Assembly}

\author{
 Florent Becker
  \thanks{Université d'Orléans, INSA CVL, LIFO, UR 4022, Orléans, France,
  \protect\url{florent.becker@univ-orleans.fr}
    Research supported in part by CNRS project CalcADN
  }
\and
 Phillip Drake%
  \thanks{Electrical Engineering and Computer Science Department, University of Arkansas,
    \protect\url{padrake@uark.edu}
    Supported in part by National Science Foundation Grant FET-2329908.}
\and
 Benjamin Durbin%
  \thanks{Electrical Engineering and Computer Science Department, University of Arkansas,
    \protect\url{bd041@uark.edu}
    Supported in part by National Science Foundation Grant FET-2329908.}
\and
 Nicolas Ollinger%
  \thanks{Université d'Orléans, INSA CVL, LIFO, UR 4022, Orléans, France,
    \protect\url{nicolas.ollinger@univ-orleans.fr}}
\and
 Brent Orlina%
  \thanks{Electrical Engineering and Computer Science Department, University of Arkansas,
    \protect\url{borlina@uark.edu}
    Supported in part by National Science Foundation Grant FET-2329908.}
\and
  Matthew J. Patitz%
    \thanks{Electrical Engineering and Computer Science Department, University of Arkansas,
      \protect\url{patitz@uark.edu}
      Supported in part by National Science Foundation Grant FET-2329908.}
}

\date{}

\maketitle

\begin{abstract}
The majority of research into mathematical models of self-assembling systems has involved asynchronous models such as the abstract Tile Assembly Model (aTAM) in which tiles attach one at a time to a growing assembly in a nondeterministic and asynchronous manner. Such systems have great computational power because they can be designed so that the growth of assemblies follows pre-designed algorithms, allowing them to simulate arbitrary Turing machines as well as to be guided to form complex shapes. However, these powers have been shown to be limited to systems which allow so-called ``cooperative’’ attachments, in which the attachment of one tile can only occur if it forms bonds with at least two other tiles. In non-cooperative systems, in which tiles attach to only a single other tile, the computational powers are reduced to essentially those of regular languages. Recently, however, the Synchronous Tile Assembly Model (syncTAM) was introduced in which, at every step of assembly, tiles simultaneously bind into every location in which tiles can bind. It was shown that even non-cooperative systems are computationally universal in the syncTAM, and also that the aTAM and syncTAM have mutually exclusive powers, in that each has systems that cannot be mimicked by the other.

In this paper, we prove that the synchronicity provided by the syncTAM can be leveraged to perform two additional feats impossible in the asynchronous aTAM. The aTAM has long been known to be equivalent to finite automata in the strictly one-dimensional case. For our first set of results, we prove that this still holds when one dimension is restricted to a finite, fixed amount of space. Yet, surprisingly, in the synchronous setting, a small, fixed, width is enough to simulate a Turing machine. This holds even in the non-cooperative setting.

For our second set of main results we prove that there exists a single set of syncTAM tiles that can be used to simulate any syncTAM system simply by using those tiles to create a scaled version of the seed structure of the original system and then letting the tiles from that set self-assemble from the scaled seed structure. While it has been shown (FOCS 2012) that such a tile set also exists in the aTAM for simulating all aTAM systems, that ``intrinsically universal’’ (IU) tile set requires cooperativity (and is impossible without cooperativity as shown in SODA 2014). We prove there exists an IU tile set in the syncTAM that works without cooperativity. Furthermore, there exists a tile set that is IU for the class of directed systems (i.e. systems which each have a single terminal assembly) while remaining directed during simulation. This was shown to be impossible in the aTAM (FOCS 2016).

\end{abstract}

\pagenumbering{arabic}

\ifSODAsubmission

\else
    \tableofcontents
\fi

\section{Introduction and Outline}

The study of algorithmic self-assembly began with the introduction of the abstract Tile Assembly Model (aTAM) in \cite{Winf98}. In algorithmic self-assembly, sets of tiles are designed so that, as they attach to grow assemblies via matching glues on their sides, the rules of binding and careful design of the tiles enforce that pre-defined algorithms are inherently followed. Such systems are capable of Turing universal computation \cite{Winf98} and the generation of complex shapes and patterns \cite{SolWin07,RotWin00,jCCSA,IUSA,FractalsSelfAssembleSODA,jSSADST}. In an attempt to more fully explore and understand the dynamics of self-assembling systems, many generalizations of the aTAM that allow for deviations from its tile properties and dynamics have been developed \cite{AGKS05g,Polygons,Polyominoes,jSignals,jDuples,SingleNegative2,syncTAM,DDDIU,GeoTiles,jRTAM,BlockTAM}. This has generated a rich landscape of models and results showing their powers and limitations, and has often prompted comparisons across models via notions of simulation. %\todo{PD: Slightly changed the wording of the last sentence}

A recent example of a derivative model of the aTAM is the Synchronous Tile Assembly Model (syncTAM) \cite{syncTAM,LSyncTAM}. In the aTAM, at each step of assembly a single tile attaches, irrespective of the number of locations into which a valid tile binding can occur. The particular location is chosen nondeterministically from the set of available locations, and therefore the aTAM is fully asynchronous and nondeterministic in tile attachments. In contrast, in the syncTAM, at every assembly step \emph{every} location which can receive a tile does so, thus making tile attachments synchronous. While the syncTAM was shown to have different powers than the aTAM, the results of \cite{syncTAM} and \cite{LSyncTAM} entail that simulation is not the right tool for comparing the two models.

The present work proposes a new approach to assess the power of the syncTAM: addressing the question whether the ``sync'' in syncTAM can be used for computing purposes: can this synchronicity be substituted for the resources needed for algorithmic work in the aTAM? For this, let us turn to the self-assembly literature in search of the different shapes computation takes in the aTAM setting, as well as the conditions it needs. In his seminal work~\cite{Winf98}, Winfree used the simulation of Cellular Automata to argue that self-assembling systems indeed compute. Follow-up efforts such as~\cite{RotWin00} have used simulations of Turing machines within the self-assembly process in order to enable its algorithmic control. Both of those constructions feature two characteristics which are shared by later occurrences of computation in the aTAM: they use cooperative attachments of tiles (both constructions are at temperature 2, in which tile attachments can require them to form bonds with two other tiles), and the computation takes place in rectangular regions. These two characteristics have been deemed necessary for computation in the aTAM, leading to the investigation of two conjectures:
\begin{enumerate}
    \item the temperature 1 conjecture, stating that cooperation is necessary for any computation, and
    \item the strict fractal conjecture, stating that discrete fractal shapes cannot be obtained in the aTAM, for lack of room within them to compute.
\end{enumerate}

Their investigation, from the formalization of their statements to their eventual resolutions, has led to a refinement of the understanding of computation within self-assembly. One essential step in this process is the transplantation of the concept of \emph{intrinsic universality} from cellular automata~\cite{ollinger2003intrinsic,AlbertCulik87,Goles-etal-2011,Ollinger-CSP08} to the aTAM~\cite{USA,IUSA,DirectedNotIU}. An intrinsically universal model is one where computation can be reified so that there is a universal system within that model can read the description of any system and simulate it. In the aTAM, this feat is possible at temperature 2~\cite{IUSA}, but not at temperature 1~\cite{temp1notIU}, showing that cooperation cannot be done entirely without. In other words, the cooperative aTAM is a well-structured model of a computation, while the non-cooperative aTAM much less so. Subsequent progress in the non-cooperative aTAM~\cite{directedTemp1Decidable,Temp1PathsSTOC,temp1notIU,PierreTemp1LongPaths} have successfully gone back to whether these system can escape periodic behaviors at all, showing the limitations of computation in the asynchronous and non-cooperative setting.

As for the necessity for room in which to compute and the associated fractal conjecture, despite some early results showing that the aTAM in one dimension is but a finite state machine in disguise \cite{Staged1D_NACO} and that the discrete Sierpinski triangle indeed cannot be self-assembled~\cite{jSSADST}, the resolution of the fractal conjecture is recent~\cite{FractalsSelfAssembleSODA}. This may have been due to the conjecture actually not holding! Indeed, by leveraging the possibility of self-simulation in an intrinsically universal model such as the cooperative aTAM, it is possible to obtain some, and actually most discrete fractal shapes. Why not the discrete Sierpinski triangle? The characterization of fractal shapes which can be strictly self-assembled in~\cite{FractalsSelfAssembleSODA} reveals a crucial parameter in terms of space: the connected tree-width of the assembled shape~\cite{DBLP:journals/combinatorica/DiestelM18}. When this parameter is bounded for unbounded assemblies, some periodicity cannot be avoided, thwarting attempts to run arbitrary computations. When it is not, the assembly encircles arbitrarily large squares of the grid.

Equipped with a new tool, synchronous dynamics, it is natural to ask whether it can be used instead of either cooperation, infinite connected treewidth or both. Our first main result shows that the need for room in which to compute can be minimized in a synchronous setting: the syncTAM can perform universal computation even when restricted to a 19-tile wide band. As a consequence of this bound on the literal width of the productions, the connected tree-width is bounded, as well as more stringent measures of width such as path-width.

The second main result of our investigation into the power of synchronicity provides an intrinsically universal system at temperature 1 for the syncTAM: not only is the syncTAM able to decode the description of a system and then simulate it, it is able to do so without even using cooperation. Since it is known that no such non-cooperative universal simulator exists in the aTAM \cite{temp1notIU}, this is a distinction between the two models. Furthermore, it is also known that there is no universal simulator in the aTAM which can always remain directed (i.e. it can create only a single terminal assembly) while simulating a directed system \cite{DirectedNotIU}. However, we show another distinction between models, and power of synchronicity, by proving that such a universal simulator does exist in the syncTAM, and it also works at temperature 1. As an additional corollary, up to a constant rescaling, anything the aTAM can do in a \emph{directed} manner, the syncTAM can do at temperature 1. Morally, \emph{a global clock can supersede cooperation}.

From a more general perspective, mathematically speaking the syncTAM is interesting to study since it’s on the furthest end of the synchronization spectrum (opposite the aTAM), allowing rigorous theoretical study of the power of fully-synchronized self-assembly. It also allows for studying the comparison between tile assembly and Cellular Automata. There are three main distinctions between the temperature-1 aTAM and Cellular Automata: (1) irreversibility: add this to Cellular Automata to get Freezing Cellular Automata, (2) synchronization, and (3) decision of a location from one glue versus an arbitrary neighborhood and local transition function. Our results show that the meaningful distinction is Synchronization, as even alone, it allows non-cooperative tile assembly to achieve computational and intrinsic universality.

The layout of this paper is as follows. In \cref{sec:prelims} we provide definitions of the models used and several additional components used throughout the paper. In \cref{sec:tour} we provide a high-level overview of the full set of results. Then, in \cref{simulation-def} we provide rigorous definitions of simulation and intrinsic universality. In \cref{sec:bounded} we provide the full details and proofs of the results related to bounded-width computation. In \cref{sec:adj-proof} - \cref{sec:univ-temp1} we present the series of results (and their proofs) that result in the universal simulator for the syncTAM at temperature 1. Finally, \cref{sec:directed} contains the proof of the directed universal simulator.

% Keywords: Synchronous self-assembly, intrinsic universality, abstract Tile Assembly Model

% Question: is it possible to simulate any 1-change, Von Neumann FCA in the syncTAM?

\section{Preliminary Definitions and Models}
\label{sec:prelims}

\subsection{The abstract Tile-Assembly Model}

%\subsubsection{Definition}

We work within the abstract Tile-Assembly Model\cite{Winf98} in 2 dimensions, abbreviated \emph{aTAM}. These definitions are borrowed from \cite{DDDIU} and we note that \cite{RotWin00} and \cite{jSSADST} are good introductions to the model for unfamiliar readers. 

Let $\mathbb{N}$ be the set of nonnegative integers,  for $l,u \in \Z$, let $\intinterval{l, u} = \{ k \in \Z \mid l \leq k < u\}$. For $n \in \mathbb{N}$, let $[n] = \intinterval{0,n} = \{0, 1, ..., n-2, n-1\}$.

Fix $\Sigma$ to be some alphabet with $\Sigma^*$ its finite strings. A \emph{glue} $g\in\Sigma^*\times\mathbb{N}$ consists of a finite string \emph{label} and non-negative integer \emph{strength}. There is a single glue of strength $0$, referred to as the \emph{null} glue. A \emph{tile type} is a tuple $t\in(\Sigma^*\times\mathbb{N})^{4}$, thought of as a unit square with a glue on each side. A \emph{tile set} is a finite set of tile types. We always assume a finite set of tile types, but allow an infinite number of copies of each tile type to occupy locations in the $\mathbb{Z}^2$ lattice, each called a \emph{tile}. A \textit{glue set} $G_T$ is the set of all glues associated with a tile set $T$. Given a direction $d \in \{N, E, S, W\} = \calD$ and a tile set $T$, $G_{T, d}$ is the set of all glues on the edges of tiles in direction $D$ and $g_{t, d}$ is the glue $g$ on tile $t$ in direction $D$; $\operatorname{str}(g)$ is the strength of $g$. 

Given a tile set $T$, a \emph{configuration} is an arrangement (possibly empty) of tiles in the lattice $\mathbb{Z}^2$, i.e.\ a partial function $\alpha:\mathbb{Z}^2\dashrightarrow T$. Two adjacent tiles in a configuration \emph{interact}, or are \emph{bound} or \emph{attached}, if the glues on their abutting sides are equal (in both label and strength) and have positive strength. Each configuration $\alpha$ induces a \emph{binding graph} $B_\alpha$ whose vertices are those points occupied by tiles, with an edge of weight $s$ between two vertices if the corresponding tiles interact with strength $s$. An \emph{assembly} is a configuration whose domain (as a graph) is connected and non-empty. The \emph{shape} of $X \subset \mathbb{Z}^2$ is $\{ X + (x, y) | x, y \in \Z \}$, i.e. the set of all its translations. The shape of an assembly $\alpha$ is the shape $S_\alpha$ of its domain: two assemblies have the same shape if they have the same domain modulo translation. An assembly $\alpha$ contains a shape $S$ if there is $s \in S$ such that $s \subset \dom \alpha$.  For some $\tau\in\mathbb{Z}^+$, an assembly $\alpha$ is \emph{$\tau$-stable} if every cut of $B_\alpha$ has weight at least $\tau$, i.e.\ a $\tau$-stable assembly cannot be split into two pieces without separating bound tiles whose shared glues have cumulative strength $\tau$. Given two assemblies $\alpha,\beta$, we say $\alpha$ is a \emph{subassembly} of $\beta$ (denoted $\alpha \sqsubseteq \beta$) if $\dom \alpha \subseteq \dom \beta$ and for all $p\in \dom \alpha$, $\alpha(p)=\beta(p)$ (i.e., they have tiles of the same types in all locations of $\dom \alpha$). 

A \emph{tile-assembly system} (TAS) is a triple $\calT=(T,\sigma,\tau)$, where $T$ is a tile set, $\sigma$ is a finite $\tau$-stable assembly called the \emph{seed assembly}, and $\tau\in\mathbb{Z}^+$ is called the \emph{binding threshold} (a.k.a. \emph{temperature}).
%If the seed $\sigma$ consists of a single tile, i.e. $|\sigma| = 1$, we say $\calT$ is \emph{singly-seeded}.
Given a TAS $\calT=(T,\sigma,\tau)$ and two $\tau$-stable assemblies $\alpha$ and $\beta$, we say that $\alpha$ \emph{$\calT$-produces} $\beta$ \emph{in one step} (written $\alpha \to^{\calT}_1 \beta$) if $\alpha \sqsubseteq \beta$ and $|S_\beta \setminus S_\alpha| = 1$.
That is, $\alpha \to^{\calT}_1 \beta$ if $\beta$ differs from $\alpha$ by the addition of a single tile.
The \emph{$\calT$-frontier} is the set $\partial^{\calT}\alpha = \bigcup_{\alpha \to^{\calT}_1 \beta} S_\beta \setminus S_\alpha$ of locations in which a tile could $\tau$-stably attach to $\alpha$.
When $\calT$ is clear from context we simply refer to these as the \emph{frontier} locations.

We use $\mathcal{A}^T$ to denote the set of all assemblies of tiles in tile set $T$. Given a TAS $\calT=(T, \sigma, \tau)$, a sequence of $k\in\mathbb{Z}^+ \cup \{\infty\}$ assemblies $\alpha_0, \alpha_1, \ldots$ over $\mathcal{A}^T$ is called a \emph{$\calT$-assembly sequence} if, for all $1\le i < k$, $\alpha_{i-1} \to^{\calT}_1 \alpha_i$. The \emph{result} $\lim \alpha$ of an assembly sequence $\alpha$ is the unique limiting assembly of the sequence. For finite assembly sequences, this is the final assembly; whereas for infinite assembly sequences, this is the assembly consisting of all tiles from any assembly in the sequence. We say that \emph{$\alpha$ $\calT$-produces $\beta$} (denoted $\alpha\to^{\calT} \beta$) if there is a $\calT$-assembly sequence starting with $\alpha$ whose result is $\beta$. We say $\alpha$ is \emph{$\calT$-producible} if $\sigma\to^{\calT}\alpha$ and write $\prodasm{\calT}$ to denote the set of $\calT$-producible assemblies. We say $\alpha$ is \emph{$\calT$-terminal} if $\alpha$ is $\tau$-stable and there exists no assembly that is $\calT$-producible from $\alpha$. We denote the set of $\calT$-producible and $\calT$-terminal assemblies by $\termasm{\calT}$. If $|\termasm{\calT}| = 1$, i.e., there is exactly one terminal assembly, we say that $\calT$ is \emph{directed}. When $\calT$ is clear from context, we may omit $\calT$ from notation. 
For any side of a location $d \in \{N,E,S,W\}$, let $\texttt{OPP}(d)$ be a function which returns the side opposite to $d$ (i.e. $\texttt{OPP}(N) = S$). Let $l \in \mathbb{Z}^2$ be a location in $\calT$ that is initially empty (i.e. not in the seed $\sigma$). Let $\texttt{IN}(l,i) \subseteq \{N,E,S,W\}$ be the set of directions such that for each $d \in \texttt{IN}(l,i)$, at step $i$ of $\vec{\alpha}$ a tile is placed in the location adjacent to $l$, in direction $d$, which has a non-null glue on side $\texttt{OPP}(d)$. That is, the set of directions $\texttt{IN}(l,i)$ represents the sides of $l$ adjacent to neighboring locations that received tiles in step $i$ and have glues facing $l$.

Given TAS $\calT$, we define $\shapeasm{T} = \{S_\alpha \mid \alpha \in \termasm{T}\}$ as the set of all shapes of terminal assemblies in $\calT$. Clearly, if $\calT$ is directed $|\shapeasm{T}| = 1$.% However, even if $\calT$ is not directed it may be the case that $|\shapeasm{T}| = 1$ (see \cite{BryChiDotKarSekTOC} for an example). In such a case we call the system \emph{shape directed}.

\subsection{The Synchronous Tile Assembly Model}

The Synchronous Tile Assembly Model (syncTAM) is the variant of the aTAM where at every step of assembly, rather than a single location of the frontier being non-deterministically selected to receive a tile, \emph{all} frontier locations simultaneously receive a tile each. As the frontier of a system can grow arbitrarily large, the number of tiles added per assembly step may also grow arbitrarily large. As in the aTAM, if any single frontier location allows $\tau$-strength binding of tiles of multiple types, one of those types is non-deterministically chosen for each such location.

We formally define the syncTAM in the same way as the aTAM with the exception of the definition of a single step in the assembly sequence. Given a syncTAS $\calS = (T, \sigma, \tau)$ and two $\tau$-stable assemblies $\alpha$ and $\beta$, we say that $\alpha$ \emph{$\calS$-produces $\beta$ in one step}, written $\alpha \to^{\calS}_1 \beta$, if $\alpha \sqsubseteq \beta$ and $\dom(\beta) \setminus \dom(\alpha) = \partial^{\calS}(\alpha)$. Note that like in the aTAM, we define $\partial^{\calS}(\alpha)$ to be the set of locations where a tile can $\tau$-stably attach to $\alpha$.  The syncTAM inherits all other definitions and notation from the aTAM.

To distinguish between a tile assembly system (a.k.a. TAS) in the aTAM versus one in the syncTAM, we will use the term \emph{syncTAS} to refer to a TAS in the syncTAM. The class $\mathfrak{T}$ will be the set of all syncTASs.

\subsection{Additional definitions}

This section includes some additional helper functions, terminology related to simulation, etc.

Let $D = \{N,E,S,W\}$ be the set of directions. The function $\texttt{OPP}: D \rightarrow D$, given direction $d \in D$, returns the opposite direction (e.g. $\texttt{OPP}(N) = S$). We say that two sides are \emph{adjacent} if their directions are one of the pairs in the set $\{\{N,E\},\{E,S\},\{S,W\},\{W,N\}\}$, and \emph{across-the-gap} if their directions are one of the pairs in the set $\{\{d, \texttt{OPP}(d)\} | d \in D\} = \{\{N,S\},\{E,W\}\}$.

Given a tile type $t$ and direction $d \in D$, the function $G_d(t)$ returns the glue on side $d$ of $t$ as an ordered pair $(l,s)$ where $l$ is the glue's label and $s$ is the glue's strength. If $t$ does not have a glue on side $d$, $G_d(t)$ instead returns \emph{null}.

Given $S$, a connected set of points in $\mathbb{Z}^2$, we define a version of $S$ \emph{scaled by factor $c$}, and denote it by $c \cdot S$, as $c \cdot S = \{(x,y) \in \mathbb{Z}^2 \mid (\lfloor \frac{x}{c} \rfloor, \lfloor \frac{y}{c} \rfloor) \in S\}$. Intuitively, $c \cdot S$ is a version of $S$ expanded by replacing each point of $S$ by a $c \times c$ square of points. Likewise, for a set $A$ of shapes, $c \cdot A$ is defined as $\{ c\cdot S \mid S \in A \}$.

%We define a \emph{path} $p$ as an ordered list of distinct locations in $\mathbb{Z}^2$, and refer to the $i$th location in $p$ as $p[i]$, so that each location $p[i]$, for $0 \le i < |p|-1$, is adjacent to $p[i+1]$.

%\begin{definition}[Dependence path]\label{def:dep-path}
%Given an aTAM system $\calT$, a producible assembly $\alpha \in \prodasm{\calT}$, assembly sequence $\vec{\alpha}$ which produces $\alpha$, and locations $l_1,l_2 \in \dom{\alpha}$, we say that there is a \emph{dependence path from $l_1$ to $l_2$} if and only if there exists a path $p$ from $1_1$ to $l_2$ in $\dom{\alpha}$ such that, in $\vec{\alpha}$, for each location $p[i]$, for $0 < i < |p|-1$, (1) the tile at location $p[i]$ was placed before the tile in $p[i+1]$, and (2) the tile at location $p[i+1]$ was required to form a bond with the tile at location $p[i]$ in order to attach.
%\end{definition}

%Intuitively, a dependence path from $l_1$ to $l_2$ is a path formed by sequential tile attachments that leads from $l_1$ to $l_2$. Note that, by definition, a dependence path is directional and therefore a dependence path from $l_1$ to $l_2$ cannot be the same path as a dependence path from $l_2$ to $l_1$.

In the spirit of locally consistent systems \cite{USA}, some so-called \emph{locally adjacent-only} systems exhibit remarkably simple behavior, making them an invaluable intermediary step in the simulation of arbitrary systems. Similarly for \emph{staggered} syncTAS.

%\begin{definition}
%\label{def:adj-only}
%    A system $\calS$ in either the aTAM or the syncTAM is \emph{adjacent-only} if, for each producible assembly $\alpha$ of $\calS$, for each frontier location in $\alpha$, any tile type which can $\tau$-stably bind in that location can do so by either:
%    \begin{itemize}
%        \item binding to exactly one neighbor in $\alpha$
%        \item binding to exactly two neighbors in $\alpha$ on adjacent sides.
%    \end{itemize}
%\end{definition}

%\begin{definition}
%\label{def:locally-adj-only}
%A system $\calS$ in either the aTAM or the syncTAM is \emph{locally adjacent-only} if $\calS$ is %adjacent-only and, for each producible assembly $\alpha$ of $\calS$, for each frontier location in $%\alpha$, if a tile type can $\tau$-stably bind to two neighbors in $\alpha$ on adjacent sides then %there exists no other set of two adjacent sides onto which a tile type can $\tau$-stably bind.
%\end{definition}\todo{FlB: should we fold adjacent-only and locally adjacent-only into the same definition? I don’t think we ever use ``merely adjacent-only''}

\begin{definition}[Locally adjacent-only]
    \label{def:locally-adj-only}
    A system $\calS$ in either the aTAM or the syncTAM, with tile set $T$, is \emph{locally adjacent-only} iff, for each producible assembly $\alpha$ of $\calS$, for each frontier location $l \in \partial^\calS(\alpha)$, if a tile of type $t \in T$ can $\tau$-stably bind to two neighbors at location $l$, then these neighbors are on adjacent sides of $l$ and also, there exists no different set of two sides adjacent to $l$ onto which a tile of any type $t' \in T$, with $t \ne t'$, can $\tau$-stably bind.
\end{definition}%\todo{FlB: I made the executive decision to fold the two definition into one, as we barely used ``merely adjacent-only'' at all. MJP: Ok, I was a bit confused by what was left but am trying to make things consistent. Since you left ``locally adjacent-only'' as the name, and ``adjacent-only'' has been used before, let's just use locally adjacent-only throughout.}

\begin{definition}[Staggered]
\label{def:staggered}
    A system $\calS$ in the syncTAM is \emph{staggered} if for every assembly sequence $\vec{\alpha}$ of $\calS$, in every step of $\vec{\alpha}$, for each empty position $z \in \mathbb{Z}^2$, at most one position neighboring $z$ receives a tile which contains a non-null glue abutting $z$. 
\end{definition}

\section{Overview of the Results}\label{sec:tour}

In this section, we provide high-level overviews of our main results, with full details and proofs deferred to later sections.

%%%%%%%%%%%%%%%%%%%%%%%%%%%%%%%%%%%%%%%%%%%%%%%%%%%%%%%%%%
%%%%%%%%%%%%%%%%%%%%%%%%%%%%%%%%%%%%%%%%%%%%%%%%%%%%%%%%%%
\subsection{Bounded-Width Computation in the syncTAM}

One of the earliest tools for exploring the limits of the aTAM, and arguably one of the most fruitful is the Window Movie Lemma~\cite{TempOneNotIU}, which bounds the amount of information that can be shared across a narrow cut of the assembly. When the productions of a given aTAS $\calS$ all fit in a $k$-wide band, this restriction applies everywhere along that band. \Cref{thm:aTAM-buchi} uses this limited information passing to prove that the set of productions of $\calS$ admits  A result of this kind cannot hold in a synchronous setting. The Window Movie bounds the amount of information which can cross a narrow cut of an assembly. In contrast, the syncTAM exhibits undecidable behavior even for assemblies with no attachments outside $\{-9, \dots, 9\} \times \mathbb{Z}$.

\begin{restatable}[$k$-Band Boundedness]{definition}{kBandProblem}
    A TAS $\calS$ in the aTAM or the syncTAM is \emph{$k$-band bounded} if each producible assembly $\alpha$ of $\calS$ satisfies $\operatorname{dom} \alpha \subset (\{-k, \dots, k\} \times \mathbb{Z})$.
\end{restatable}

\begin{restatable}{lemma}{aTAMkbounded}\label{lem:aTAM-k-bounded}
    For all $k$, it is decidable given a TAS whether it is $k$-Band Bounded. Actually, that problem is in Co-NEXPTIME.
\end{restatable}

While proving \cref{lem:aTAM-k-bounded} is a relatively simple application of the Window Movie Lemma to be found in \cref{sec:bounded}, such a result cannot hold in the syncTAM, which is shown in the next section.

The comparision between the aTAM and the syncTAM in the $k$-band-bounded regime is especially striking. The productions and dynamics of the aTAM are described by a Finite Automaton, using a Büchi acceptance condition for infinite productions. The details, including some extra definitions are in \ref{sec:buchi}.

\begin{restatable}{theorem}{aTAMbuchi}\label{thm:aTAM-buchi}
    Let $\calT$ be a $k$-band-bounded aTAS on tile-set $T$, and $\mathcal{L}_\prec(T)$ be the set of representations of decorated (final) productions of $\calT$. Then $\mathcal{L}_\prec(\calT)$ can be recognized by a deterministic finite automaton with a Büchi acceptance condition.
\end{restatable}

Given such a notable weakness in the computational power of the aTAM when assemblies are confined in one dimension, the following result reveals significant power derived from the synchronicity of the syncTAM. It shows that even when one dimension is bounded by a constant, resulting essentially in one-dimensional, linear space, syncTAM systems are capable of Turing universal computation. We prove this by construction, showing that an arbitrary two-counter machine can be simulated using assemblies of fixed width (at temperature 2 and also at temperature 1). A two-counter machine has two \emph{counters}, or \emph{registers}, each of which is capable of storing an arbitrary natural number. The key technique of the construction is to encode the current value in each register as the amount of space between the tops of two growing subassemblies. These subassemblies grow in zig-zag patterns, with the ability of the ``program'' portion of the assembly to cause one of them to speed up or slow down relative to the other (by changing the lengths of the zig-zags). A two-counter machine only needs to be able to increment or decrement the value in a register, or to tell when it is zero. The first two operations are performed by the slowing or speeding mentioned, and testing for zero is done by growth that can detect when the subassemblies for a counter are at the same height. This is all easily done within a fixed width that is bounded by the sums of the maximum zig and zig lengths. A simple conversion of two-counter machine instructions into sets of tile types, along with the generation of a seed with well-defined distances between the portions that set the initial values of the registers, allows a syncTAM system to simulate a two-counter machine. Since two-counter machines can simulate Turing machines, syncTAM systems bounded in width are thus Turing universal.

\begin{restatable}{theorem}{syncTAMkbounded}\label{thm:syncTAM-k-bounded}
\label{thm:Turing_band}
  It is undecidable, given a SyncTAS whether it is $9$-band bounded.  
\end{restatable}

We prove \cref{thm:syncTAM-k-bounded} by providing a construction that allows an arbitrary two-counter machine to be simulated by a syncTAS without placing any tiles with $x$-coordinates that are $<-9$ or $>9$. However, systems can be designed to place tiles outside of those bounds if and only if a given two-counter machine halts. Since two-counter machines are Turing universal, their halting problem is undecidable and thus the syncTAM-$9$-Band Boundedness problem is undecidable. Full details can be found in \cref{sec:bounded}.

The construction is made of a number of gadgets, each with its function and invariants. All productions fit in the  $\{-9, \dots, 9\} \times \mathbb{Z}$ band, except a few tiles which are only added if $M$ halts. Tiles in the columns between $-4$ and $4$ are from the ``transition gadgets'', columns $-9$ to $-4$ are from the ``left counter gadgets'', and columns $4$ to $9$ are from the ``right counter gadgets''. The trick is that while altitude in the assembly of $\calS_M$ represents increasing time in the run of $M$ on input $(0, 0)$, the three gadgets around altitude $6i$ which represent the $i$-th transition of the run are not assembled at the same time. In other words, \emph{the value of the counters are represented by the delay between either counter gadget and the transition gadget}. Their nullity can be observed as it is the only case where the counter gadgets arrive at $z_l(i)$ or $z_r(i)$ before the transition gadget. They can be increased or decreased by having the counter gadgets rise faster or slower than the rate of the transition gadget.

\cref{fig:band-example-sim} is an illustration of this construction. An interactive version is available as a JavaScript simulator for two-counter machines, along with functionality that can convert an arbitrary two-counter machine into a syncTAM system that uses bounded-width and can be run in an embedded syncTAM simulator. They are available (along with examples)  online here: \url{http://self-assembly.net/software/two_counter_machines/two_counter_machines-latest/} \cite{TwoCounterSimulator}.

As a consequence of the \cref{lem:adj-only} and \cref{lem:temp1simulation-adj}, that problem is undecidable even at temperature 1, but we also provide an explicit construction.

\begin{restatable}{corollary}{syncTAMToneBounded}\label{cor:syncTAMone-k-bounded}
    There is a $k$ such that given a temperature 1 syncTAS, it is undecidable whether it is $k$-band bounded.
\end{restatable}

%%%%%%%%%%%%%%%%%%%%%%%%%%%%%%%%%%%%%%%%%%%%%%%%%%%%%%%%%%
%%%%%%%%%%%%%%%%%%%%%%%%%%%%%%%%%%%%%%%%%%%%%%%%%%%%%%%%%%
\subsection{Intrinsic Universality and the syncTAM}

Our second main result is that the syncTAM is \emph{Intrinsically Universal}, and moreover, the \emph{universal simulator} works at temperature 1. The formal definitions of the concepts of simulation and intrinsic universality for both the syncTAM and the aTAM are given in \cref{simulation-def}. For this overview, it suffices to know that a TAS $\calS$ can be \emph{simulated} by a TAS $\calU$ if there is a $k \in \mathbb{N}$ such that the evolution of the $k \times k$ square at $(kx, ky)$ (a \emph{macrotile}) in $\calU$  corresponds to the evolution of $(x, y)$: gathering input glues from neighboring macrotiles until an attachment eventually becomes possible and is simulated. The TAS $\calU$ shall then transmit the glues of the newly attached tile to the neighboring macrotiles.

%and a \emph{representation function} $R$ mapping assemblies $\alpha$ of $\calU$ with $\dom \alpha \subset \llbracket 0, k \rrbracket^2$ to tile types such that each 
%The main result of this subsection is \cref{thm:T1-IU}.

\begin{restatable}{theorem}{ToneIU}
\label{thm:T1-IU}

%Let $\mathfrak{T}$ represent the full class of all syncTAM systems. There is a tile set $\mathcal{U}_1$, a representation function generator $\mathcal{R}: \mathfrak{T} \to \REPL$ and a seed generator function $S: \mathfrak{T} \to \mathcal{A}^{U_1}_{< \infty}$ such that $\mathcal{U}_1$ is intrinsically universal for $\mathfrak{T}$ at temperature 1.\todo{this does not state that $\mathcal{R}(T)$ does not depend on the seed of $T$}
Let $\mathfrak{T}$ represent the full class of all syncTAM systems. There exists a tile set $\mathcal{U}_1$ that is a universal simulator for $\mathfrak{T}$ such that, for each $\calT \in \mathfrak{T}$, $\mathcal{U}_1$ can simulate $\calT$ at temperature 1. 
\end{restatable}

To prove \cref{thm:T1-IU}, we chain together a series of constructions that are meant to each account for distinct aspects of the dynamics of universal simulation. (A schematic depiction of the chain is shown in \cref{fig:crazy_diagram}.)

%%%%%%%%%%%%%%%%%%%%%%%%%%%%%%%%%%%%%%%%%%%%%%%%%%%%%%%%%%
\subsubsection{Reducing dynamics}

The first step, and arguably the most technically challenging, is to simulate any arbitrary syncTAM system by one which is locally adjacent-only, staggered, and at temperature 2. This part is also very specific to the syncTAM, as the concept of being staggered does not apply at all in an asynchronous setting. Furthermore, the simulation is \emph{seed-independent} (see \cref{def:seed-ind}) which means that, if the seed of the system to be simulated is changed, then only the seed of the simulating system needs to be changed (i.e. its tile set, temperature are unchanged).

\begin{restatable}{lemma}{adjonly}\label{lem:adj-only}
Let $\calS = (T,\sigma,\tau)$ be an arbitrary %temperature-2
syncTAM system. Then, there exists $S' = (T',\sigma', 2)$ such that $S'$ synchronously simulates $\calS$ at scale factor $O(|T|)$, and $S'$ is locally adjacent-only and staggered.
 Moreover, this simulation scheme is seed-independent, and $\calS'$ can be obtained from $\calS$ computably.
\end{restatable}

\cref{lem:adj-only} is proven by construction in \cref{sec:adj-proof}. Here, we provide a brief sketch of several main aspects of the construction.

Since the system which is to be simulated, $\calS$, is an arbitrary syncTAM system, its temperature parameter may be any positive integer and it may not be locally adjacent-only or staggered. Therefore, the simulating system, $\calS'$, must account for tile attachments in which: (1) any subset of 1, 2, 3, or 4 tile sides may be the \emph{input} sides that are used for initial binding, (2) any subset of the neighboring tiles providing input sides may arrive in any order and/or simultaneously, (3) input glues may be adjacent and/or across-the-gap, and (4) there may be nondeterminism that allows tiles of multiple type to attach in the same location (possibly using the same or different sets of sides as input sides). $\calS'$ simulates each tile of $\calS$ using a square macrotile (an example macrotile can be seen in \cref{fig:adj-only-macrotile_tour}), and if $n$ tiles are added during a step in the assembly of $\calS$, $\calS'$ ensures that the $n$ macrotiles that grow to represent those tiles all complete at the exact same step in $\calS'$, regardless of which of any of the previously mentioned scenarios occur for each of the $n$ tiles in $\calS$. In this way, the growth of macrotiles remains synchronized throughout the system.

\begin{figure}
    \centering
    \begin{tikzpicture}
    \begin{scope}[scale=1.8]
        
        %\fill[gray!25!white] (0,0) rectangle (1, 1);

        \draw[brown, ultra thick] (0, 1) -- (1, 1);
        \draw[magenta, ultra thick] (1, 1) -- (1, 0) 
            (.95, .95) -- (.95, .05)
            (.9, .9) -- (.9, .1);
        \draw[green, ultra thick] (1, 0) -- (0, 0);
        \draw[yellow, ultra thick] (0, 0) -- (0, 1) (.075, .075) -- (.075, .925);
        \draw[gray!50!white] (0, 0) -- (.2, .2) -- (.2, .8) -- (.8, .8) -- (.8, .2) -- (.2, .2);
        \draw[gray!50!white] (1, 0) -- (.8, .2)
            (1, 1) -- (.8, .8)
            (0, 1) -- (.2, .8);

        \draw (-2, 1) -- (-2, 0);
        \draw [cyan, ultra thick] (-2, 0) -- (-3, 0);
        \draw [yellow, ultra thick] (-3, 0) -- ++(0, -1) (-3.075, 0) -- (-3.075, -1);
        \draw [green, ultra thick] (-3, -1) -- (-2, -1);
        \draw (-2, -1) -- (-2, -2);
        \draw[gray!50!white] (-3, 0) -- (-3.075, 0) (-3, -1) -- (-3.075, -1);

        \draw[->] (.5, .5) to[out=180, in=0] (-2.5, -.5);
        \node[above] at (-2.5, 1) {$\tau = 3$};
    \end{scope}
    
        \node[right] (macro) at (3,0) {    
            \includegraphics[width=3.0in]{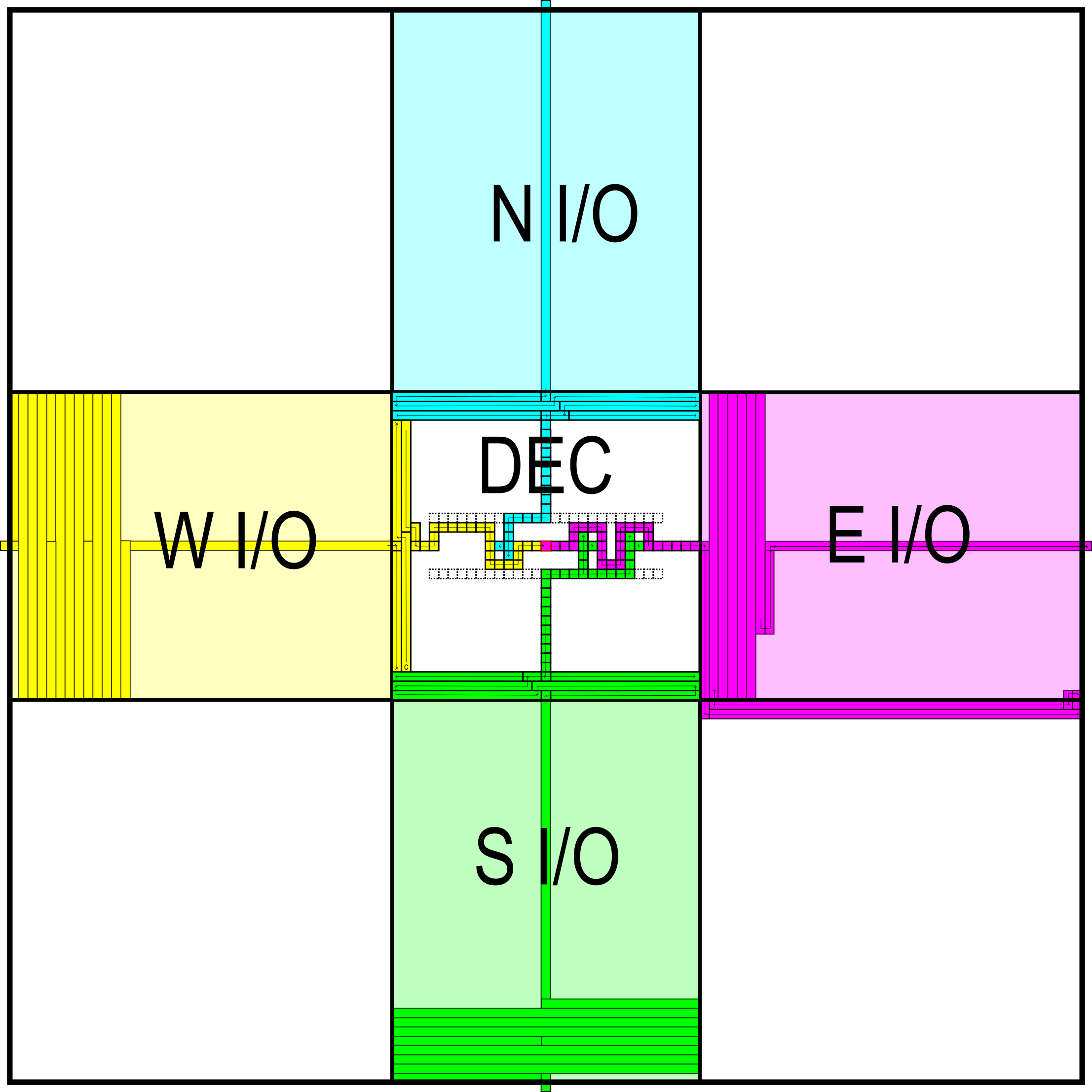}
        };

        %\draw[<-] (img.west) -- +(-1, 0);
        %\draw[<-] (img.north) -- +(0, 1);
        %\draw[<-] (img.south) -- +(0, -1);
        
    \end{tikzpicture}
    \caption{An example macrotile for the construction of the proof of \cref{lem:adj-only}, with (most) submacrotile names shown. In this example, the N, W, and S were input sides and the E an output side. The N had input delay 0 so arrived at the Decision submacrotile first, then the S arrived, and finally the W, which provided enough input glue strength to simulate a tile attachment by placing a tile in the central point of competition (in the Decision region labeled DEC, outlined in red), and then cause output of the one remaining glue, on the E side.\label{fig:adj-only-macrotile_tour}}
\end{figure}
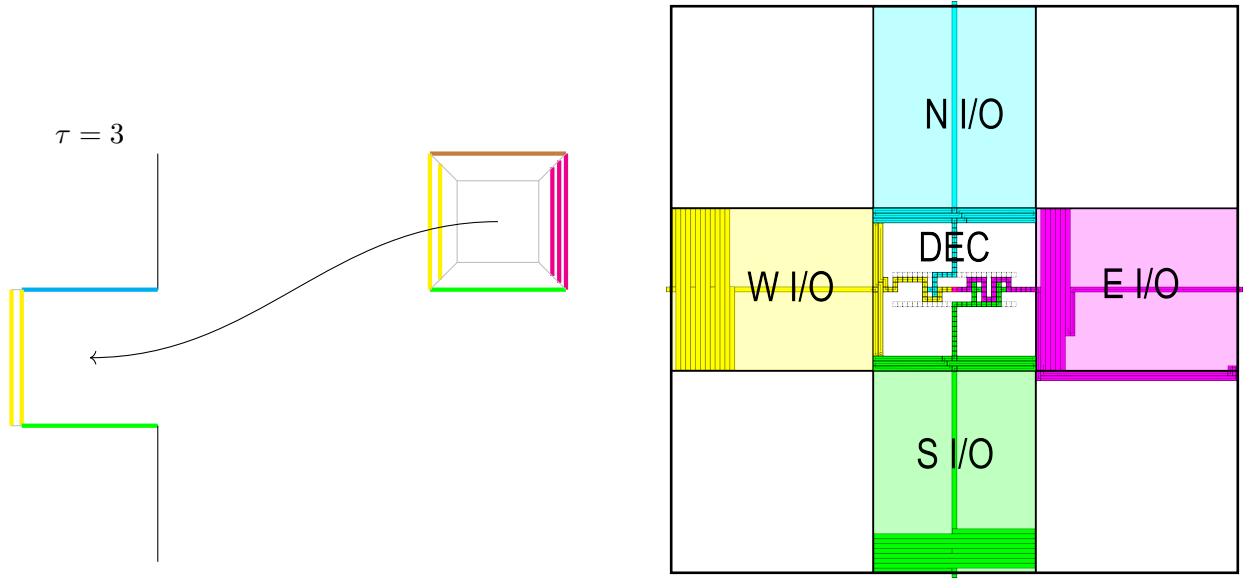

The key technique for ensuring such synchronization is for each side of a growing macrotile that represents an input glue to select a random amount of time to delay its growth toward the center (which it does by growing a longer zig-zag path, rather than a straight path). The set of random delays from which one can be selected by each side are specific for each side and designed so that all possible orderings of arrivals of input side growth to the center are fairly represented and possible. This is done following the protocol of ``Go-First Dice'' \cite{gofirst4player}, and allows the computation at the center of a macrotile to be done in a way that allows each side to contribute its input and compute whether or not a tile attachment can be simulated one at a time. Doing so allows all paths of tiles to be carefully designed so that no empty location simultaneously receives multiple neighboring tiles (with incident glues), and thus the system is staggered. It also allows for paths to be designed so that, whenever cooperation is needed for a tile placement, that cooperation is done using pairs of adjacent tile sides that are unique, keeping $\calS'$ locally adjacent-only. The randomness is needed in order to simulate systems which are not directed (i.e. some locations may receive tiles of different types in different assembly sequences). The randomness allows the possibility of each pathway being taken (and is therefore not needed when simulating a directed system as all pathways lead to the same outcome). While fairness (with respect to the balanced probability of selecting each option) is not technically necessary, it comes for free with the use of Go-First Dice. 

Finally, the paths and points of cooperation are designed so that, once a macrotile receives sufficient input to grow into the representation of a tile in $\calS$, it is then able to ensure that paths that represent output glues grow to the edge of the macrotile and all arrive at exactly the same step (across all outputs of the macrotile, and all macrotiles that grow into tile representations during the same simulated step of $\calS$).

%%%%%%%%%%%%%%%%%%%%%%%%%%%%%%%%%%%%%%%%%%%%%%%%%%%%%%%%%%
\subsubsection{Removing cooperativity}

Next the temperature of the simulating system needs to be lowered to 1. This transformation needs to be done before feeding the system to be simulated to the universal simulator. %, as well as on the output of \cref{thm:T2-IU-for-T1}. 
Hence it needs to be algorithmic and seed-independent.

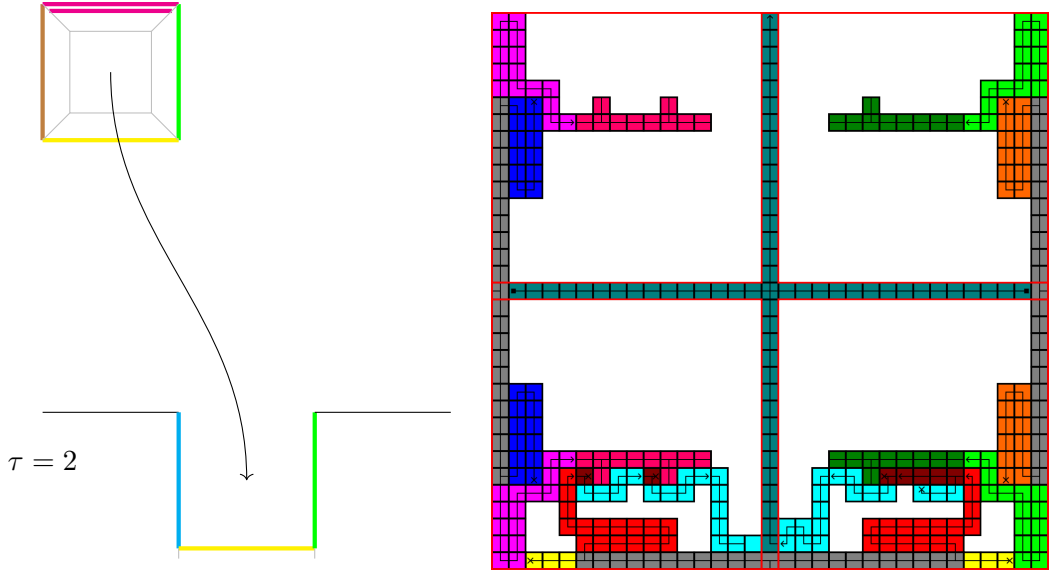
\begin{figure}
    \centering
    \begin{tikzpicture}
        \begin{scope}[scale=1.8, rotate=90]
        
        %\fill[gray!25!white] (0,0) rectangle (1, 1);

        \draw[brown, ultra thick] (0, 1) -- (1, 1);
        \draw[magenta, ultra thick] (1, 1) -- (1, 0) 
            (.95, .95) -- (.95, .05);
        \draw[green, ultra thick] (1, 0) -- (0, 0);
        \draw[yellow, ultra thick] (0, 0) -- (0, 1);
        \draw[gray!50!white] (0, 0) -- (.2, .2) -- (.2, .8) -- (.8, .8) -- (.8, .2) -- (.2, .2);
        \draw[gray!50!white] (1, 0) -- (.8, .2)
            (1, 1) -- (.8, .8)
            (0, 1) -- (.2, .8);

        \draw (-2, 1) -- (-2, 0);
        \draw [cyan, ultra thick] (-2, 0) -- (-3, 0);
        \draw [yellow, ultra thick] (-3, 0) -- ++(0, -1);
        \draw [green, ultra thick] (-3, -1) -- (-2, -1);
        \draw (-2, -1) -- (-2, -2);
        \draw[gray!50!white] (-3, 0) -- (-3.075, 0) (-3, -1) -- (-3.075, -1);

        \draw[->] (.5, .5) to[out=180, in=0] (-2.5, -.5);
        \node[above] at (-2.5, 1) {$\tau = 2$};
    \end{scope}
    
        \node[right] at (4, -2) {
    \includegraphics[width=2.9in]{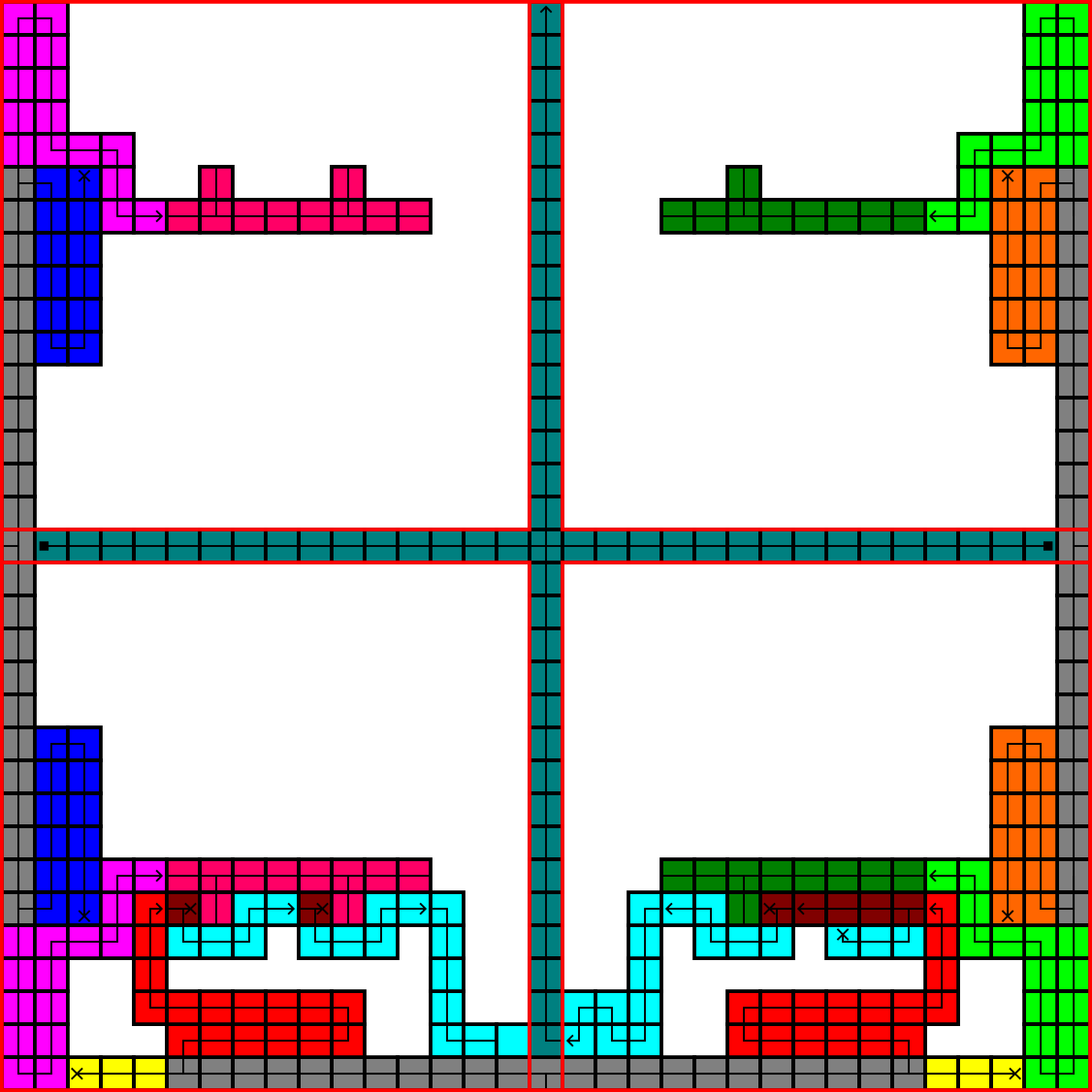}};
    \end{tikzpicture}
    \caption{An example macrotile for the construction of the proof of \cref{lem:temp1simulation-adj} which receives strength-1 input from the S, E, and W adjacent macrotiles. In this instance, the S and E glues end up cooperating and start the output. S and W fail to cooperate for lack of a suitable tile type in $\calT$. The East and West inputs arrived before the S input and each claimed the two corners next to them (in fuchsia and green respectively), which allows the corner-claiming gadgets of S (yellow) to start simulating cooperation (red tiles). Each macrotile region is enclosed with a red border. }
    \label{fig:temp1simulator-example-macrotile-tour}
\end{figure}

\begin{restatable}{lemma}{tempOneSimulationAdj}
        \label{lem:temp1simulation-adj}
    For any temperature 2 TAS $\calS$ which is locally adjacent-only and staggered in the syncTAM, there is a staggered temperature 1 TAS $\calS_1$ which simulates $\calS$ in the syncTAM.

    Moreover, this simulation scheme is seed-independent, and $\calS_1$ can be obtained from $\calS$ computably.
\end{restatable}

This construction relies on two main gadgets which show how a global clock can substitute for cooperation. (See \cref{fig:temp1simulator-example-macrotile-tour} for an example macrotile.) Because $\calT$ is locally adjacent-only, it suffices to simulate cooperation in a corner of each macrotile. The construction makes use of gadgets known as \emph{bit-readers} and \emph{bit-writers} that can be used in non-cooperative systems to read and write information necessary for computation (first used in the syncTAM in \cite{syncTAM}, but widely used in other models and constructions \cite{CookFuSch11,Polygons,Polyominoes,SingleNegative,jDuples}). The bit-reader and bit writer gadgets combine data about the glues from two input sides in a corner. The corner-claiming gadget breaks the symmetry between the first and second input arriving in the same corner, so that the earlier one can wait and the later one can collect both inputs and proceed. Because $\calT$ is staggered, there can be no ties in their arrival times, and the time scale of the simulation is set to be large enough that the first corner gadget is done growing when the second input arrives. Carefully placed delays ensure this. Then, whenever one or two inputs suffice to allow simulation of a tile placement, paths are able to grow to place a tile in the center that is unique to the tile type being simulated and then to the output sides to represent output glues. The paths grow with sufficient delays to account for possible variance in the lengths of paths needed to read different values, keeping the completion of all macrotiles synchronized.

%%%%%%%%%%%%%%%%%%%%%%%%%%%%%%%%%%%%%%%%%%%%%%%%%%%%%%%%%%
\subsubsection{A universal simulator}

%The rest of this section consists in the proof of \cref{thm:T1-IU}.
The next step, and the heart of the universal simulator at temperature 1, is a specialized temperature 2 system for simulating temperature 1 systems from their descriptions, i.e. a universal simulator using a single tile set that encodes all necessary information about any system to be simulated within the seed-representing assembly.

\begin{restatable}{theorem}{TtwoIUforTone}
    \label{thm:T2-IU-for-T1}
    There is a temperature 2 tile set $\calU$, a representation function generator function $\mathfrak{R}$ and a seed generator function $\mathfrak{Y}$ such that, using $\mathfrak{R}$ and $\mathfrak{Y}$, $\calU$ is a universal simulator for $\mathfrak{S}_1$, the class of staggered temperature 1 syncTAM systems. Moreover, for any $\calT \in \mathfrak{S}_1$, the TAS $(\calU,\mathfrak{Y}(\calT),2 )$ that synchronously simulates $\calT$ is staggered and is locally adjacent-only.
\end{restatable}

\Cref{thm:T2-IU-for-T1} is proved by construction of a universal simulator that works by (1) creating a seed consisting of macrotiles that represent the tile(s) of the simulated seed and which also encode a description of the full system $\calT$ to be simulated, and (2) building macrotiles that decode $\calT$, simulate tile attachments when possible, and then pass around the description of $\calT$ and of the newly simulated tile. Here, the main challenge, beyond decoding and simulating, is to ensure that for a given $\calT$, the delay between input and output is the same for all input and output pairs of all macrotiles. The amount of cases to consider here is greatly reduced by the assumption that $\calT$ is staggered and at temperature 1: each attachment is done through one neighbor at once.

Decoding and simulating is done by a Turing machine embedded near the input of each macrotile. 
Here, the fact that the simulator can work at temperature 2 lets the design of the various gadgets rely on known techniques for computing in self-assembly.
Care has to be taken so that the space occupied by its space-time diagram only depends on $\calT$ and not on anything specific to the macrotile at hand. Its embedding into the macrotile also needs to ensure that the time it takes the assembly process to traverse its space-time diagram is also constant given $\calT$. The routing of all information can then be synchronized and all delays accounted for using simple gadgets that rotate and reflect information through growing assemblies.

\begin{figure}
    \centering
    \includegraphics[width=0.4\linewidth]{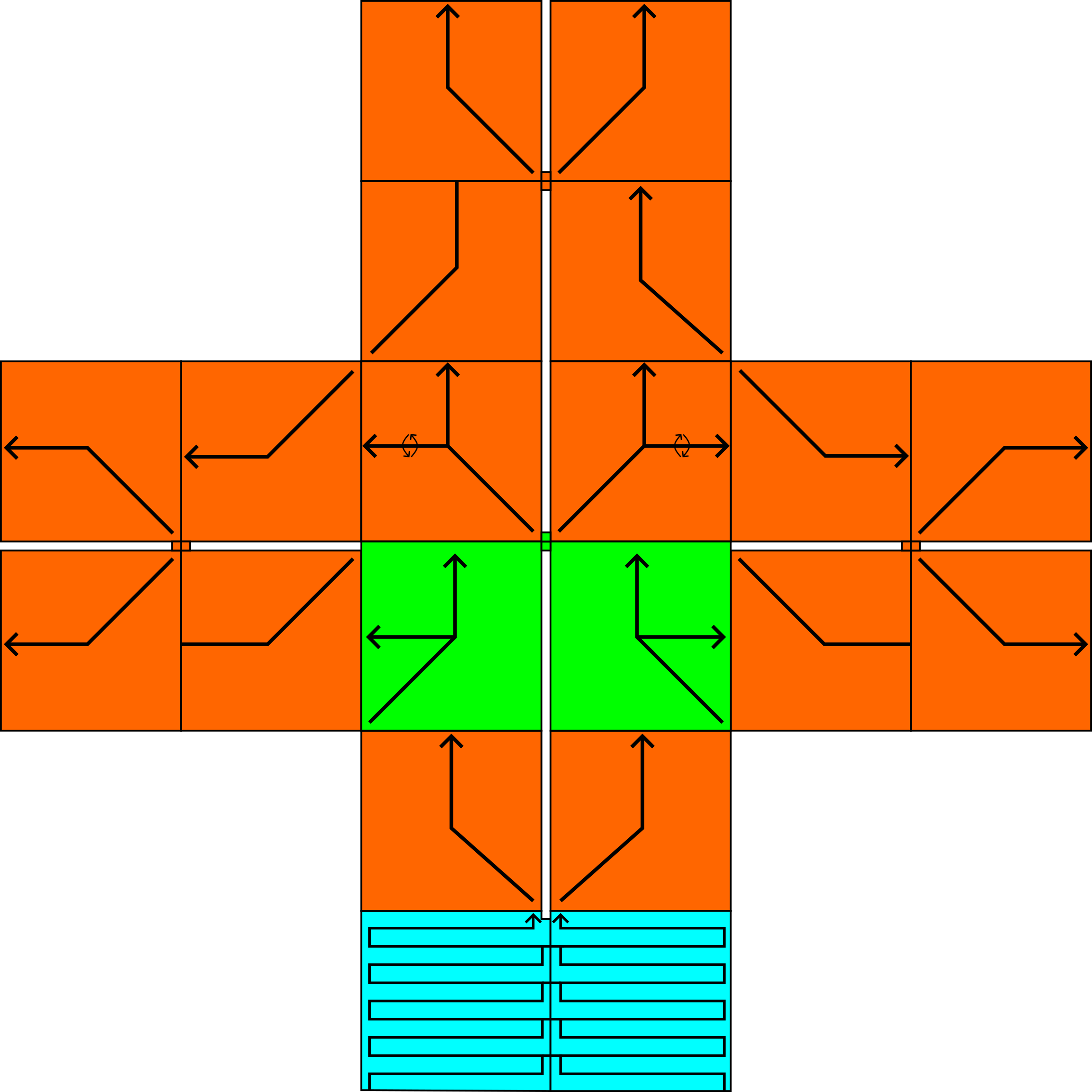}
    \caption{A (very) high level depiction of a macrotile for the construction of the proof of \cref{thm:T2-IU-for-T1} that simulates the attachment of a tile using a south input glue. Squares are colored corresponding to gadgets which are fully described in \cref{sec:temp2iufortemp1}. } 
    \label{fig:temp2IUforTemp1-macrotile-layout}
\end{figure}

In this way, a tile set is created that is universal for all syncTAM systems that are temperature 1 and staggered, in such a way that the simulator using these tiles to simulate any such system $\calT$ itself remains staggered and, although it is now at temperature 2, it is locally adjacent-only.

%%%%%%%%%%%%%%%%%%%%%%%%%%%%%%%%%%%%%%%%%%%%%%%%%%%%%%%%%%
\subsubsection{Universal simulation at temperature 1}

The final link in the chain of constructions is to create a tile set that can be used at temperature 1 to simulate any system created for the universal simulator from the construction of \cref{thm:T2-IU-for-T1}. By careful design, those systems are temperature 2, staggered, and locally adjacent-only, and we already have a construction that takes any such system as input and outputs a temperature 1 tile set to simulate it: the construction of \cref{lem:temp1simulation-adj}. Furthermore, that construction yields a simulation scheme that is seed-independent. This means that, for the entire class of systems using the universal simulator for temperature 2, since they all use the same tile set and temperature and thus vary only by seed structure, the construction of \cref{lem:temp1simulation-adj} yields a single universal tile set that works at temperature 1 to simulate any system that can be simulated by the systems of the construction for \cref{thm:T2-IU-for-T1}. We will call this tile set $\mathcal{U}_1$.

Any arbitrary syncTAM system $\calS$ can be simulated by a system $\calS'$ that is temperature 2, locally adjacent-only, and staggered (\cref{lem:adj-only}). That system $\calS'$ can then be simulated by a temperature 1 system $\calS_1$ that is staggered (\cref{lem:temp1simulation-adj}). A seed over the tile set $\mathcal{U}$, which is universal for temperature 1 systems that are staggered, can then be designed so that the resulting system, $\mathcal{S}_{\mathcal{U}}$, using tile set $\mathcal{U}$ at temperature 2 simulates $\calS_1$. Finally, a seed over the tiles of $\mathcal{U}_1$ can be created so that the resulting system, $\mathcal{S}_{\mathcal{U}_1}$, using tile set $\mathcal{U}_1$ at temperature 1 simulates $\mathcal{S}_{\mathcal{U}_2}$.  Thus, $\mathcal{S}_{\mathcal{U}_1}$ simulates the original syncTAM system $\calS$ using tile set $\mathcal{U}_1$ at temperature 1, proving \cref{thm:T1-IU}.

An overview of the full chain can be seen in \cref{fig:crazy_diagram}.

% The simulator of \cref{thm:T2-IU-for-T1} is still far removed from the temperature 1 intrinsically universal system \cref{thm:T1-IU} claims to exhibit. Two gaps need to be bridged:
% \begin{itemize}
%     \item the input system needs to be simulated by a temperature 1 staggered system
%     \item the output system is staggered, but it needs to be simulated by a temperature 1 system; that simulation scheme needs to be seed-independent so that \cref{lem:seed-ind-IU} applies and intrinsic universality carries through.
% \end{itemize}

% These three constructions together, put together as on \cref{fig:crazy_diagram} produce the intrinsically universal system of \cref{thm:T1-IU}. \Cref{lem:seed-ind-IU} makes it possible to use \cref{lem:temp1simulation-adj} on the universal simulator of \cref{thm:T2-IU-for-T1}.

\begin{figure}
    \centering
    \input{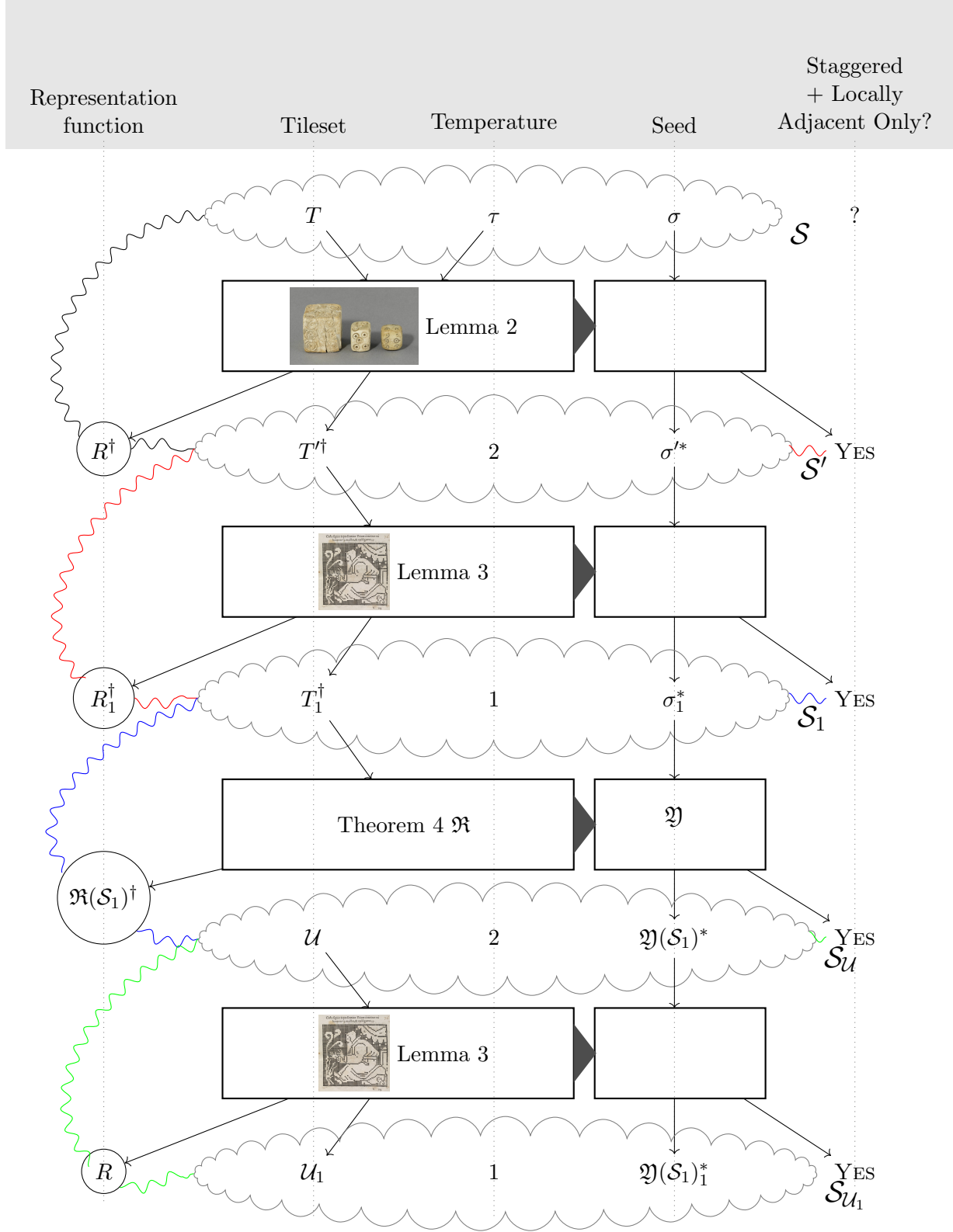}
    \caption{Depiction of the whole machinery entailed by \cref{thm:T1-IU}. The box for each lemma or theorem is cleft to separate the outputs which are independent of the seed of the input system from those which do depend on it. Elements marked with a dagger depend on the input tileset, while those marked with a star depend on the whole input system.}
    \label{fig:crazy_diagram}
\end{figure}

%%%%%%%%%%%%%%%%%%%%%%%%%%%%%%%%%%%%%%%%%%%%%%%%%%%%%%%%%%
\subsubsection{Maintaining directedness}

Next, we show that the previous constructions can be (very slightly) modified to make a universal simulator with the property that, when simulating any directed system in the syncTAM, the simulator works at temperature 1 and also remains directed. This was proven to be impossible in the aTAM \cite{DirectedNotIU}, at any temperature.

\begin{restatable}{theorem}{directedIU}\label{thm:directed}
    Let $\mathfrak{D}$ be the class of all tile assembly systems within the syncTAM which are directed. $\mathfrak{D}$ is intrinsically universal through an intrinsically universal system at temperature 1.
\end{restatable}

In order to prove \cref{thm:directed}, modifications need only be made to the constructions in \cref{lem:adj-only} and \cref{thm:T2-IU-for-T1} to remove the intentionally added nondeterminism from each construction. Specifically, these modifications consist of hard-coding the ``Go-First Dice'' used in \cref{lem:adj-only} to a single option for each direction, and hard-coding the binary choices made by the Turing machine simulations in \cref{thm:T2-IU-for-T1}. With these modifications, the full chain of simulation remains completely deterministic.

Note that it is undecidable to determine if a given syncTAM system is directed, and thus whether it can be accurately simulated by the directed version of the simulator, or instead requires the undirected simulator.

As a consequence of the previous results, anything that can be assembled in the aTAM by a directed system can be assembled in the temperature 1 syncTAM, up to a rescaling.

\begin{corollary}
\label{cor:directed-shapes-T1}
    Let $S \subset \mathbb{Z}^2$ be a (finite or infinite) shape which can be self-assembled in the aTAM by a directed TAS $\calT$. Then there is an integer $k$ and a TAS $\calT'$ in the syncTAM such that $\calT'$ assembles $k \cdot S$ and is directed. 
\end{corollary}

\begin{proof}
    Because $\calT$ is directed in the aTAM, it has the same terminal production in the syncTAM (Lemma 2 of \cite{syncTAM}). By \cref{lem:adj-only} and \cref{lem:temp1simulation-adj}, that system can be simulated by a temperature 1 syncTAM system, meaning that system produces the same terminal assembly modulo scale factor. 
\end{proof}

Having completed the high-level tour of our results, the following sections contain the full details and proofs of each.
\section{Formalization of Simulations and Universality}
\label{simulation-def}

\subsection{Synchronous and Asynchronous intrinsic simulation}

%The definition of simulation for synchronous systems simulating synchronous systems is shoplifted from \cite{DBLP:conf/cie/BeckerMOT18}.
Intrinsic universality via intrinsic simulation has been studied in a variety of tile-assembly models (e.g. \cite{IUSA,2HAMIU,DirectedNotIU,DDDIU,TempOneNotIU,DDDxModelSimICALP}). Essentially, intrinsic simulation deals with one system simulating the behavior of another modulo scale factor, with squares of tiles in the simulating system representing individual tiles of the simulated system. Intrinsic universality (IU) for a model (or class of systems within a model) occurs when there is a single tile set (along with associated functions, to be discussed) that can be used to create a seed structure from its tiles for any of the systems in the model (or class) and then intrinsically simulate them. Now we provide formal definitions which are largely taken from \cite{DDDxModelSimICALP} but are modified to incorporate the dynamic of synchronicity of the syncTAM.

%Intuitively, simulation of a system $\mathcal{T}$ by another system $\mathcal{S}$ is done by utilizing some scale factor $m \in \mathbb{Z}^+$ such that $m \times m \times m$ cubes of tiles in $\mathcal{S}$ represent individual tiles in $\mathcal{T}$, and there is a ``representation function'' which is able to interpret the assemblies of $\mathcal{S}$ as assemblies in $\mathcal{T}$.

From this point on, let $T$ be a tile set and let the scale factor be $m\in\Z^+$.
An \emph{$m$-block macrotile} over $T$ is a partial function $\alpha : \Z_m^2 \dashrightarrow T$, where $\Z_m = \{0,1,\ldots,m-1\}$.
Let $B^T_m$ be the set of all $m$-block macrotiles over $T$.
The $m$-block with no domain is said to be $\emph{empty}$.
For a general assembly $\alpha:\Z^2 \dashrightarrow T$ and $(x',y')\in\Z^2$, define $\alpha^m_{(x',y')}$ to be the $m$-block macrotile defined by $\alpha^m_{(x',y')}(i_x,i_y) = \alpha(mx'+i_x,my'+i_y)$ for $0 \leq i_x,i_y< m$.
For some tile set $S$, a partial function $R: B^{S}_m \dashrightarrow T$ is said to be a \emph{valid $m$-block macrotile representation} from $S$ to $T$ if for any $\alpha,\beta \in B^{S}_m$ such that $\alpha \sqsubseteq \beta$ and $\alpha \in \dom R$, then $R(\alpha) = R(\beta)$.

For a given valid $m$-block macrotile representation function $R$ from tile set~$S$ to tile set $T$, define the \emph{assembly representation function}\footnote{Note that $R^*$ is a total function since every assembly of $S$ represents \emph{some} assembly of~$T$; the functions $R$ and $\alpha$ are partial to allow undefined points to represent empty space.}  $R^*: \mathcal{A}^{S} \rightarrow \mathcal{A}^T$ such that $R^*(\alpha') = \alpha$ if and only if $\alpha(x,y) = R\left(\alpha'^m_{(x,y)}\right)$ for all $(x,y) \in \Z^2$.
For an assembly $\alpha' \in \mathcal{A}^{S}$ such that $R^*(\alpha') = \alpha$, $\alpha'$ is said to map \emph{cleanly} to $\alpha \in \mathcal{A}^T$ under $R^*$ if for all non empty blocks $\alpha'^m_{(x,y)}$, $(x,y)+(u_x,u_y) \in \dom(\alpha)$ for some $(u_x,u_y) \in \{(0,\pm 1), (\pm 1, 0)\}$, or if $\alpha'$ has at most one non-empty $m$-block $\alpha^m_{0,0}$.  In other words, $\alpha'$ may have tiles in macrotile blocks representing empty space in $\alpha$, but only if that position is adjacent to a macrotile that represents a tile in $\alpha$.  We call such growth ``around the edges'' of $\alpha'$ \emph{fuzz} and thus restrict it to be adjacent to only valid macrotiles, but not diagonally adjacent (i.e.\ we do not permit \emph{diagonal fuzz}).

In the following definitions, let $\mathcal{T} = \left(T, \sigma_T, \tau_T\right)$  and $\mathcal{S} = \left(S, \sigma_S, \tau_S\right)$ each be a syncTAS, and let $R$ be an $m$-block representation function $R:B^S_m \rightarrow T$. We will define how $\calS$ \emph{synchronously simulates} $\calT$ at \emph{spatial scale factor} $m$ and \emph{temporal scale factor} $s \in \mathbb{Z}^+$.

\begin{definition}\label{def:equiv-assemblies}
    Let $\calS = (T, \sigma, \tau)$ be a syncTAS, and let $\alpha$ and $\beta$ be two $\tau$-stable assemblies in $\calS$.  We say that $\beta$ is \emph{equivalent} to $\alpha$, and write $\beta \approx \alpha$, if and only if $\exists \gamma$ such that $\gamma$ is a $\tau$-stable assembly in $\calS$ and the following conditions hold:
    \begin{enumerate}
        \item $\alpha \rightarrow^\calS_1 \gamma$

        \item $\alpha \sqsubseteq \beta \sqsubseteq \gamma$

        %\item $\exists i$, where $0 \le i < l$, such that $|\dom(\alpha)| + i = |\dom(\beta)| = |\dom(\gamma)| - l + i$

        \item $\dom(\beta) \setminus \dom(\alpha) \subseteq \partial^\calS(\alpha)$
    \end{enumerate}
\end{definition}

Essentially, if $\beta \approx \alpha$, $\beta$ can be thought of as an ``intermediate'' assembly between $\alpha$ and $\gamma$ (where $\alpha$ produces $\gamma$ in one step) where only a subset of the locations in the frontier of $\alpha$ have been added. While such an assembly $\beta$ is not producible in $\calS$ if  $\dom(\beta) \setminus \dom(\alpha) \neq \emptyset$, it will be useful to define such an assembly in the context of a simulator that may require many tile additions to simulate each individual tile addition of $\calS$. 

\begin{definition}
\label{def-equiv-prod} We say that $\mathcal{S}$ and $\mathcal{T}$ have \emph{equivalent productions} (under $R$), and we write $\mathcal{S} \Leftrightarrow \mathcal{T}$ if the following conditions hold:
\begin{enumerate}
        %\item $\left\{R^*(\alpha') | \alpha' \in \prodasm{\mathcal{S}}\right\} = \prodasm{\mathcal{T}}$.
        
        \item $\forall \alpha' \in \prodasm{S}, \exists \alpha \in \prodasm{T}$ such that $R^*(\alpha') \approx \alpha$.

        \item $\forall \alpha \in \prodasm{T}, \exists \alpha' \in \prodasm{S}$ such that $R^*(\alpha') = \alpha$.
        
        \item $\left\{R^*(\alpha') | \alpha' \in \termasm{\mathcal{S}}\right\} = \termasm{\mathcal{T}}$.
%        \item For all $\alpha' \in \termasm{\mathcal{S}}$, $R^*(\alpha') \in \termasm{T}$.
        \item For all $\alpha'\in \prodasm{\mathcal{S}}$, $\alpha'$ maps cleanly to $R^*(\alpha')$.
\end{enumerate}
\end{definition}

\begin{definition}%\todo{FlB: Check (and note) that \emph{follows} is the definition from the literature. MJP: It is actually slightly dfferent due to the ``equivalence modulo $l$''}

\label{def:t-follows-s} We say that $\mathcal{T}$ \emph{follows} $\mathcal{S}$ (under $R$), and we write $\mathcal{T} \dashv_R \mathcal{S}$, if $\alpha' \rightarrow^\mathcal{S} \beta'$ for some $\alpha',\beta' \in \prodasm{\mathcal{S}}$ implies that $\exists \alpha'', \beta'' \in \prodasm{\calS}$ such that $R^*(\alpha') \approx R^*(\alpha'')$, $R^*(\beta') \approx R^*(\beta'')$ and $R^*(\alpha'') \to^\mathcal{T} R^*(\beta'')$.
\end{definition}

\begin{definition}\label{def:t-sync-follows-s}%\todo{FlB: can we allow some temporal slack, i.e. $s(i-1) \leq i' < s \cdot i$ instead of $i' = s \cdot i$ and still have our negative results? If so, I think we should. MJP: Yes, good point. Our results will still hold and that would be better.}
We say that $\calT$ \emph{synchronously follows} $\calS$ (under $R$ with temporal scale factor $s$), and we write  $\mathcal{T} \dashv^s_R \mathcal{S}$, if $\mathcal{T} \dashv_R \mathcal{S}$ and, for every assembly sequence $\vec{\beta} = \beta_0, \beta_1, \beta_2, \ldots$ in $\calS$, there exists assembly sequence $\vec{\alpha} = \alpha_0, \alpha_1, \alpha_2, \ldots$ in $\calT$ such that: 
    \begin{enumerate}
        \item $\alpha_0 = \sigma$ and $R^*(\beta_0) = \sigma$
        
        \item For every $0 \le i < |\vec{\alpha}|$, there is $i' \in [i \cdot (s - 1), i \cdot (s+1)]$ such that $R^*(\beta_{i'}) = \alpha_i$

        %\item For every $0 \le i < |\vec{\alpha}|$, and $i' = i*s$, for every $\alpha'_i \in \texttt{PROD}_1(\alpha_i)$, $\exists j \in \mathbb{Z}^+$ where $i' \le j < i'+s$ and $\beta_{j}$ can produce $\beta'$ where $R(\beta') = \alpha'_i$.\todo{finish this...}
    \end{enumerate}
\end{definition}

Note in particular that this definition is lax enough that it \emph{is} possible for a staggered system to synchronously follow a non-staggered system.

The next definition essentially specifies that every time $\mathcal{S}$ simulates an assembly $\alpha \in \prodasm{\mathcal{T}}$, there must be at least one valid growth path in $\mathcal{S}$ for each of the possible next steps that $\mathcal{T}$ could make from $\alpha$ which results in an assembly in $\mathcal{S}$ that maps to that next step. While this definition is unfortunately dense, it accommodates subtle situations such as where $\calS$ must ``commit to'' a subset of possible representations in $\mathcal{T}$ before being explicitly mapped, under $R$, to any one in particular.

\begin{definition}\label{def:s-models-t}
We say that $\mathcal{S}$ \emph{models} $\mathcal{T}$ (under $R$), and we write $\mathcal{S} \models_R \mathcal{T}$, if for every $\alpha \in \prodasm{\mathcal{T}}$, there exists $\Pi \subset \prodasm{\mathcal{S}}$ where $\Pi \neq \emptyset$ and  $R^*(\alpha') = \alpha$ for all $\alpha' \in \Pi$, such that, for every $\beta \in \prodasm{\mathcal{T}}$ where $\alpha \rightarrow^\mathcal{T} \beta$, (1) for every $\alpha' \in \Pi$ there exists $\beta' \in \prodasm{\mathcal{S}}$ where $R^*(\beta') = \beta$ and $\alpha' \rightarrow^\mathcal{S} \beta'$, and (2) for every $\alpha'' \in \prodasm{\mathcal{S}}$ where $\alpha'' \rightarrow^\mathcal{S} \beta'$, $\beta' \in \prodasm{\mathcal{S}}$, $R^*(\alpha'') \approx \alpha$, and $R^*(\beta') = \beta$, there exists $\alpha' \in \Pi$ such that $\alpha' \rightarrow^\mathcal{S} \alpha''$.
\end{definition}

\begin{definition}\label{def:s-simulates-t}
We say that $\mathcal{S}$ \emph{synchronously simulates} $\mathcal{T}$ (under $R$ and temporal scale factor $s$) if $\mathcal{S} \Leftrightarrow_R \mathcal{T}$ (equivalent productions), $\mathcal{T} \dashv^s_R \mathcal{S}$ and $\mathcal{S} \models_R \mathcal{T}$ (equivalent dynamics).
\end{definition}

In \cref{def:s-simulates-t}, every aspect of $\mathcal{S}$ is allowed to depend on the seed of $\mathcal{T}$. Yet, in most cases, only the \emph{seed} of $\mathcal{S}$ depends on that of $\mathcal{T}$. This situation is called \emph{seed-independent simulation}.

\begin{definition}[Seed-Independent Simulation]\label{def:seed-ind}
\label{def:si-simulation}
    A pair $(U, \tau)$ of a tile set and a temperature is said to \emph{seed-independently simulate} another $(T, \tau')$ (under representation function $R$ and scale factor $s$) if there is a computable function $S$ such that for each $\sigma$ over $T$, $(U, S(T, \sigma, \tau'), \tau)$ simulates $(T, \sigma, \tau')$.
\end{definition}

\subsection{Intrinsic Universality}
\label{sec:iu_def}

\newcommand{\REPL}{\mathsf{REPR}}
\newcommand{\frakC}{\mathfrak{C}}

Now that we have a formal definition of what it means for one tile system to simulate another, we can proceed to formally define the concept of intrinsic universality, i.e., when there is one general-purpose tile set that can be appropriately programmed to simulate any other tile system from a specified class of tile systems.

Let $\REPL$ denote the set of all macrotile representation functions (i.e., $m$-block macrotile representation functions for some $m\in\Z^+$).  Define $\frakC$ to be a class of tile assembly systems, and let $U$ be a tile set.  Note that each element of $\frakC$, $\REPL$, and $\mathcal{A}^U_{< \infty}$ is a finite object, hence encoding and decoding of simulated and simulator assemblies can be defined to be computable via standard models such as Turing machines and Boolean circuits. 
% representation functions the encodings we use are such that they can be be represented in a suitable format for computation in some formal system such as Turing machines. 
% `` simple'' use are Turing machine computable (or better). % 

\begin{definition}\label{def:iu-specific-temp}
We say $U$ is a \emph{universal simulator} for $\frakC$ \emph{at temperature} $\tau' \in \Z^+$ %\marginpar{SMS: added this definition for intrinsic universality at a specific temperature to help with wording of theorem statements}
if there are computable functions $\mathcal{R}:\frakC \to \REPL$ and $S:\frakC \to \mathcal{A}^U_{< \infty}$ such that, for each $\mathcal{T} = (T,\sigma,\tau) \in \frakC$, there is a constant $m\in\Z^+$ such that, letting $R = \mathcal{R}(\mathcal{T})$, $\sigma_\mathcal{T}=S(\mathcal{T})$, and $\mathcal{U}_\mathcal{T} = (U,\sigma_\mathcal{T},\tau')$, $\mathcal{U}_\mathcal{T}$ simulates $\mathcal{T}$ at scale $m$ and using macrotile representation function~$R$.

If additionally, for all $\calT$, $(U, \sigma_{\calT}, \tau') \in \frakC$, $U$ is then said to be \emph{intrinsically universal} for $\frakC$ at tempetature $\tau'$.
\end{definition}

That is, $\mathcal{R}(\mathcal{T})$, the \emph{representation function generator function}, outputs a representation function that interprets assemblies of $\mathcal{U}_\mathcal{T}$ as assemblies of $\mathcal{T}$, and $S(\mathcal{T})$, the \emph{seed generator function}, outputs the seed assembly used to program tiles from $U$ to represent the seed assembly of $\mathcal{T}$. 

\begin{lemma}
\label{lem:seed-ind-IU}
    If $U$ is a universal simulator for some $\frakC$ at temperature $\tau$, and $(U', \tau')$ seed-independently simulates $U$, then $U'$ is a universal simulator for $\frakC$ at temperature $\tau'$.
\end{lemma}

\begin{proof}
    Assume $(U', \tau')$ seed-independently simulates $(U, \tau)$ under representation function $R_0$, thanks to the seed function $S_0$. Assume moreover that $U$ is a universal simulator through the functions $\mathcal{R}$ and $S$.

    Consider an arbitrary system $\calT$. By universal simulation, $\calT' = (U, \tau, S(\calT))$ simulates $\calT$ with representation function $\mathcal{R}(\calT)$. But then $(U', \tau', S_0(\calT'))$ simulates $\calT'$ and consequently $\calT$.

    Thus, $U'$ is a universal simulator at temperature $\tau'$ through the functions $\mathcal{R}': \calT \mapsto R_0 \circ \mathcal{R}(\calT)$ and $\mathcal S': \calT \mapsto S_0(U, \tau, S(\calT))$.
\end{proof}

\begin{definition}
\label{def:iu-general}
We say that~$U$ is \emph{intrinsically universal} for $\frakC$ if it is intrinsically universal for $\frakC$ at some temperature $\tau'\in \Z^+$.
\end{definition} %\marginpar{SMS: this is the normal definition of IU. DW: agreed! Nice.}

\begin{definition}
We say that $\frakC$ is intrinsically universal if there exists some $U$ such that $U$ is intrinsically universal for $\frakC$.
\end{definition}

\subsection{Local Predecessor Relation}

\begin{definition}[local predecessor relation]%\todo{this must exist somewhere under some name}
    Let $\calT$ be an aTAS, and $\alpha$ an assembly sequence of $\calT$. The local predecessor relation $\prec_\alpha$ of $\alpha$ is the relation defined on $\dom \lim \alpha$ by $z \prec_\alpha z'$ if $z$ and $z'$ are neighbor positions and $z$ is attached before $z'$ in $\alpha$.
\end{definition}

The terminal productions of $\calT$ and the possible local predecessor relations of each one are enough to fully describe the dynamics of $\calT$.

\begin{lemma}
\label{lem:local-prec}
    Let $\calT$ be a TAS with seed $\sigma$ and temperature $\tau$, $A$ an assembly of $\calT$ and $\prec$ a relation on $\dom A$ such that:
    \begin{itemize}
        \item $\prec$ is an orientation of the edges of $\dom A$ with one extremity at least outside $\dom \sigma$,
        %\item $\prec$ induces no infinite path,
        %\item $\prec$ is a spanning dag of $\dom A$, with $\dom \sigma$ as its set of minima
        \item $\dom \sigma$ is the set of minima of $\prec$, and
        \item for all $z \notin \dom \sigma$, the binding strength between $\{z\}$ and $\{ z' | z' \prec z\}$ in $A$ is at least $\tau$.
    \end{itemize}

    Then any sequence $\alpha$ of sub-assemblies of $A$ starting at $\sigma$ such that for all $i$, $\dom \alpha_{i}$ is closed under $\prec$ and $|\dom \alpha_{i + 1} \setminus \dom \alpha_i| = 1$ is an assembly sequence of $\calT$ producing $A$ with $\prec_\alpha = \prec$.
\end{lemma}

\begin{proof}
    Let $A$ and $\prec$ be as in the lemma. Proceed by induction on the sub-assemblies of $A$ which happen to be closed downwards for $\prec$ and contain $\sigma$.

    The only sequence of sub-assemblies of $A$ starting at $\sigma$ and ending at $\sigma$ is the trivial one, and it is indeed an assembly sequence of $\calT$ producing $\sigma$.

    Let $A'$ be such a sub-assembly of $A$, and $\alpha'$ a sequence of sub-assemblies of $A'$ of length $n$ satisfying the hypothesis of the lemma. If $\sigma \subsetneq A'$, $n > 0$ must hold and let $z$ be the last attached position in $\alpha'$. Let $\alpha''$ be the prefix of length $n - 1$ of $\alpha'$ and $A'' = \lim \alpha''$. By induction, $\alpha''$ is an assembly sequence producing $A''$. The tile addition at position $z$ at the end of $\alpha'$ is a legal tile addition because the binding strength there is at least $\tau$. Thus, the induction hypothesis holds for $A'$ too.
\end{proof}

It is very easy to characterize when the sequence $\alpha$ of \cref{lem:local-prec} exists.

\begin{lemma}
    \label{lem:local-prec-charac}
    Let $\calT$ be a TAS with seed $\sigma$ and temperature $\tau$, $A$ an assembly of $\calT$ and $\prec$ a relation on $\dom A$ such that:
    \begin{itemize}
        \item $\prec$ is an orientation of the edges of $\dom A$ with one extremity at least outside $\dom \sigma$,
        %\item $\prec$ induces no infinite path,
        %\item $\prec$ is a spanning dag of $\dom A$, with $\dom \sigma$ as its set of minima
        \item $\dom \sigma$ is the set of minima of $\prec$, and
        \item for all $z \notin \dom \sigma$, the binding strength between $z$ and $\{ z' | z' \prec z\}$ in $A$ is at least $\tau$.
    \end{itemize}

    There is an assembly sequence $\alpha$ with $\prec_\alpha = \prec$ if and only if the transitive closure $\prec^*$ of $\prec$ has no cycles and no infinite paths.
\end{lemma}

\begin{proof}
    First consider a production $A = \lim \alpha$ of $\calT$ and $\prec_\alpha$ the local predecessor relation of $\alpha$. If $\alpha$ had a cycle or infinite path, then the tiles in that part of the assembly could never be attached by $\alpha$.

    Then consider an assembly $A$ with an orientation $\prec$ of its edges as in the lemma. The transitive closure $\prec^*$ of $\prec$ forms a spanning directed acyclic graph of $\dom A$ rooted at the tiles of the seed. Since it has no infinite path, a breadth-first search $(z_i)_{i < |A|}$ of $\prec^*$ enumerates all the positions of $A$. Let $\alpha_i$ be the sub-assembly of $A$ with $\dom \alpha_i = \{z_k | k < i\}$. By \cref{lem:local-prec}, the sequence $\alpha$ is the desired assembly sequence.
\end{proof}

By this lemma, the set of local predecessor relations of all productions of a TAS $\calT$ describes its whole dynamics, as the assembly sequences are the $\prec$-non-decreasing sequences of positions. If, instead, one considers the total orders extending each $\prec$, the result is the set of all \emph{ordinal} assembly sequences as in the Tree Pump Theorem of \cite{FractalsSelfAssembleSODA}.

\subsection{Window Movie Lemma}\label{sec:wml}
The Window Movie Lemma is a powerful tool often used in impossibility results for tile-based self-assembly models.  Originally introduced in \cite{temp1notIU}, it allows surgery on an assembly sequence to obtain a new producible assembly (assuming that some conditions are met).  In its simplest form, the lemma states that if there are two enclosed areas of the plane that have the same ``shape'' and the same glues appear in the same order at the same position relative to these two enclosed shapes (that is, the same window movie), then the subassemblies produced in the two enclosed areas are interchangeable.  See Figure~\ref{fig:window-movie-lemma} for a graphical representation of this simple case of the Window Movie Lemma.

\begin{figure}[ht]
    \centering
    \includegraphics[width=\textwidth]{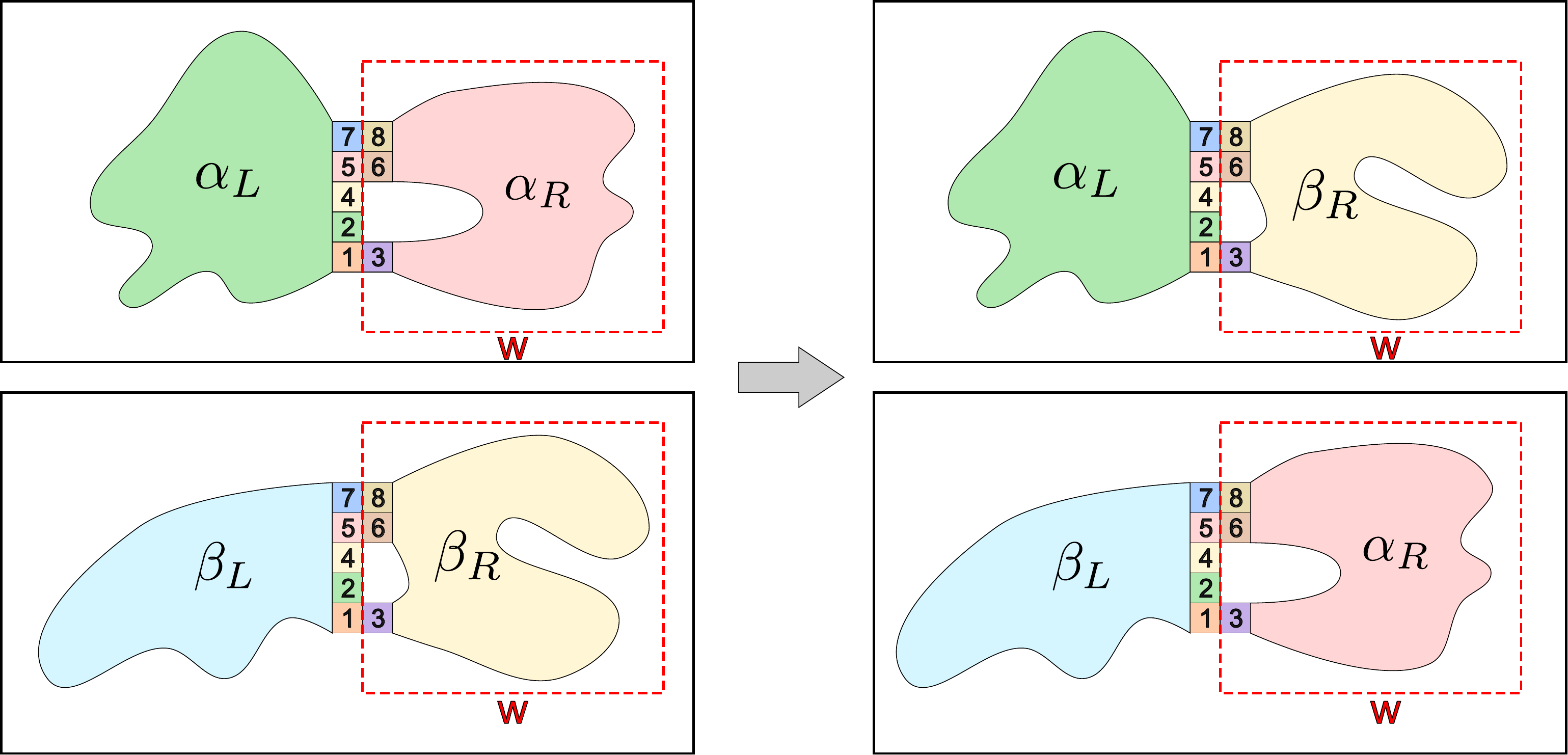}
    \caption{An illustration of the window movie lemma.  The squares represent tiles and the color of the tile is an indication of the glue type the tile exposes along the window $w$ (represented by a red dotted box).  The numbers on the tile represent the relative order in which the glues contained on the tile appear on the window as the assembly sequence is played forward.  The Window Movie Lemma tells us that since $\alpha_L \cup \alpha_R$ is producible (shown top left) and $\beta_L \cup \beta_R$ is producible (show bottom left) and because the assembly sequence produces the same window movie sequence along $w$, the assemblies shown on the right are also producible.}
    \label{fig:window-movie-lemma}
\end{figure}

The above surgical procedure can be extended in order to pump assemblies.  Indeed, if there is a fixed width subassembly of sufficient length, then there is some window movie which repeats.  We can then pump the subassembly between these two window movies indefinitely.  See Figure~\ref{fig:window-movie-lemma-pumping} for a graphical representation showing how assemblies can be pumped via the Window Movie Lemma.

The following definitions are taken from \cite{temp1notIU}, and replicated here for completeness.  A \emph{window} is an edge cut which partitions the lattice graph ($\mathbb{Z}^2$ in 2D or $\mathbb{Z}^3$ in 3D) into two regions. Given some window $w$ and some assembly sequence $\vec{\alpha}$ in a TAS $\calT$, a \emph{window movie} $M$ is defined to be the ordered sequence of glues presented along $w$ by tiles in $\calT$ during the assembly sequence $\vec{\alpha}$. Informally, the window $w$ is a cut dividing two regions of tile locations and $M$ is constructed by recording the glues which appear on the cut and their relative order as the assembly sequence $\vec{\alpha}$ is played forward. More formally, a \emph{window movie} is the sequence $M^{\vec{\alpha}}_w = \{(v_i, g_i)\}$ of pairs of grid graph vertices $v_i$ and glues $g_i$, given by order of appearance of the glues along window $w$ during $\vec{\alpha}$. Furthermore, if $k$ glues appear along $w$ during the same assembly step in $\vec{\alpha}$, then these glues appear contiguously and are listed in lexicographical order of the unit vectors describing their orientation in $M^{\vec{\alpha}}_w$.

\begin{figure}
    \centering
    \includegraphics[width=\textwidth]{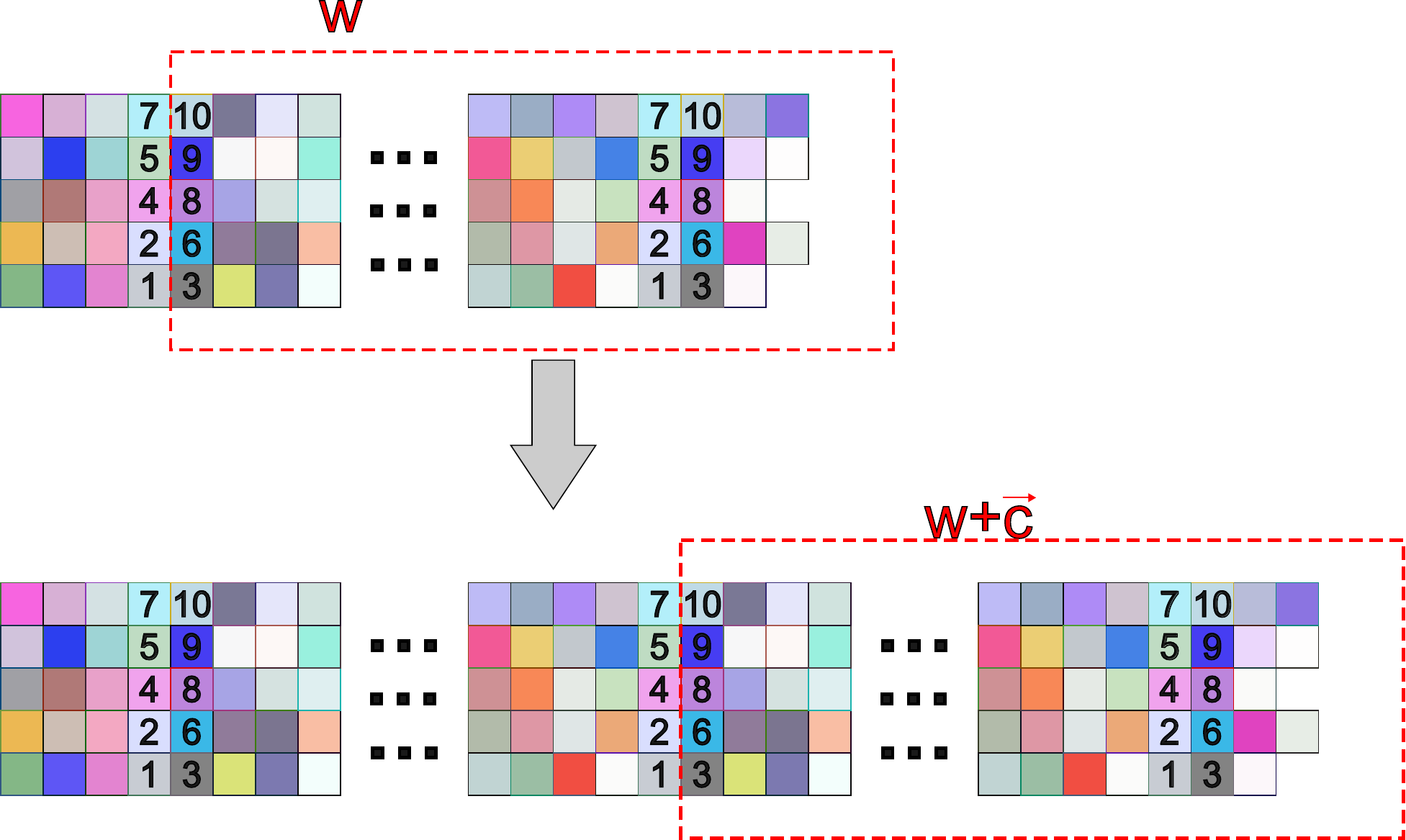}
    \caption{An illustration of how the Window Movie Lemma can be leveraged to pump assemblies.  The squares represent tiles and the color of the tile is an indication of the glue type the tile exposes along the window $w$ (represented by a red dotted box).  The numbers on the tile represent the relative order in which the glues contained on the tile appear on the window as the assembly sequence is played forward.  The top assembly shows a fixed width ribbon grown sufficiently long. By the pigeon hole principle, there must exist a window that has the same window movie when translated to two distinct positions along the ribbon (as indicated by the same colored tiles with the same numbers on them occurring twice in the ribbon).  Then, by the Window Movie Lemma, the assembly between these two translations of the window can be repeated indefinitely.  This is illustrated in the bottom portion of the figure which shows that the portion of the ribbon between repeating window movies can be also be assembled beginning where the last repeating window movie is located.}
    \label{fig:window-movie-lemma-pumping}
\end{figure}

\begin{lemma}[Window Movie Lemma]
    Let $\vec{\alpha} = \{\alpha_i\}$ and $\vec{\beta}=\{\beta_i\}$ be assembly sequences in TAS $\calT$ and let $\alpha, \beta$ be the result assemblies of each respectively. Let $w$ be a window that partitions $\alpha$ into two configurations $\alpha_L$ and $\alpha_R$ and let $w'=w+\vec{c}$ be a translation of $w$ that partitions $\beta$ into two configurations $\beta_L$ and $\beta_R$ (with $\alpha_L$ and $\beta_L$ being the configurations containing their respective seed tiles). Furthermore define $M^{\vec{\alpha}}_w$ and $M^{\vec{\beta}}_w$ to be the window movies for $\vec{\alpha},w$ and $\vec{\beta},w'$ respectively. Then if $M^{\vec{\alpha}}_w = M^{\vec{\beta}}_w$, the assemblies $\alpha_L\cup\beta'_R$ and $\beta'_L\cup\alpha_R$ (where $\beta'_L = \beta_L - \vec{c}$ and $\beta'_R = \beta_R - \vec{c}$) are also producible.
\end{lemma}

This classical result admits a concise modern proof as a corollary of \cref{lem:local-prec-charac}.

\begin{proof}
    By symmetry, it is enough to show that $\gamma = \alpha_L \cup \beta'_R$ is producible.
    Let $\prec_\alpha$ and $\prec_\beta$ be the local predecessor relations of $\vec{\alpha}$ and $\vec{\beta}$ respectively. Let $\prec_\gamma$ be the relation on $\gamma$ which coincides with $\prec_\alpha$ on $\alpha_L$, with $\prec_\beta$ on $\beta'_R$ (modulo translation by $\vec{c}$), and with the order of the window indices for the edges across $w$.

    It suffices to prove that $\prec_\gamma$ satisfies the hypothesis of \cref{lem:local-prec-charac}. For all the conditions of \cref{lem:local-prec-charac} which are checked edge-wise (the bullet points), since $\prec_\alpha$ and $\prec_\beta$ agree on any edge through $w$ and satisfy them, so does $\prec_\gamma$. Because $\prec^*_\alpha$ and $\prec^*_\beta$ do not admit cycles or infinite paths, any cycle of $\prec_\gamma^*$ would have to go through $w$, and any infinite path of $\prec^*_\gamma$ would have to go infinitely through it. In turn, such a cycle or infinite path through $w$ would beget a cycle or infinite downward path in the window indices.

    Thus, \cref{lem:local-prec-charac} makes $\gamma$ an assembling sequence.
\end{proof}

\section{Connection Between k-Band-Bounded aTAM Systems and Büchi Automata}\label{sec:buchi}

\subsection{The $k$-Band-Bounded Case}

In the $k$-band-bounded case, the set of local predecessor relations is recognizable by a deterministic Büchi Automaton\footnote{A deterministic Büchi Automaton is a deterministic Finite Automaton which, instead of receiving as input finite strings, receives infinite strings. Its acceptance condition is then that, on a given input, it must pass through one of its accepting states infinitely many times.}. Hence, for a tile set $T$, each $k$-band bounded assembly decorated with a local predecessor relation needs to be represented as an infinite word on some finite alphabet. For this, let $\Sigma_{T} = T \times \mathcal{P}(\{N, E, S, W\})$ be the possible values of a tile decorated with its set of predecessors according to some local predecessor relation. An assembly $\alpha: \mathbb{Z} \times \{-k, \ldots, k\} \dashrightarrow T$ together with a local predecessor relation $\prec$ on $\dom \alpha$ is represented as a word $w_\alpha$ on $(\Sigma_T \cup \#)^{2 \cdot (2k + 1)}$, where the $i$-th component of $w_\alpha(x)$ is $\alpha(x,i-k-1)$ with its set of $\prec$-predecessors if $i \leq 2k+1$, and $\alpha(-x,i-3k-2)$ otherwise.

\begin{figure}
    \centering
    \begin{tikzpicture}
        \input{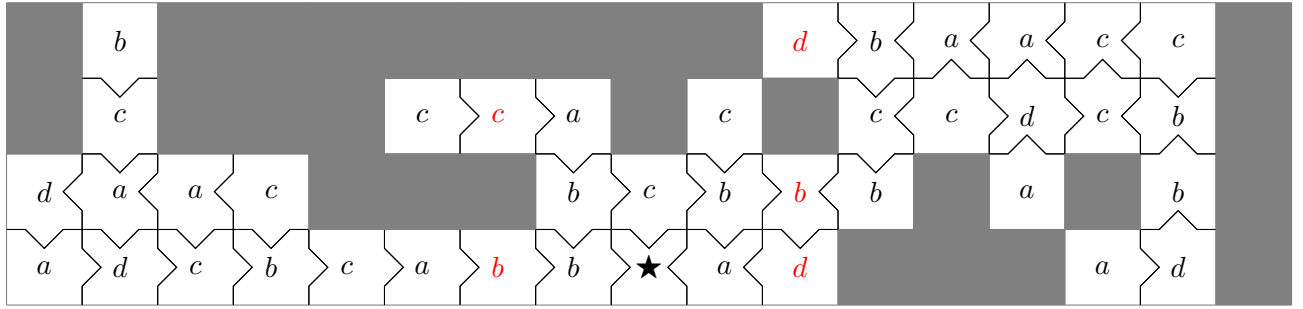}

        \begin{scope}[yshift=-5cm]
        \input{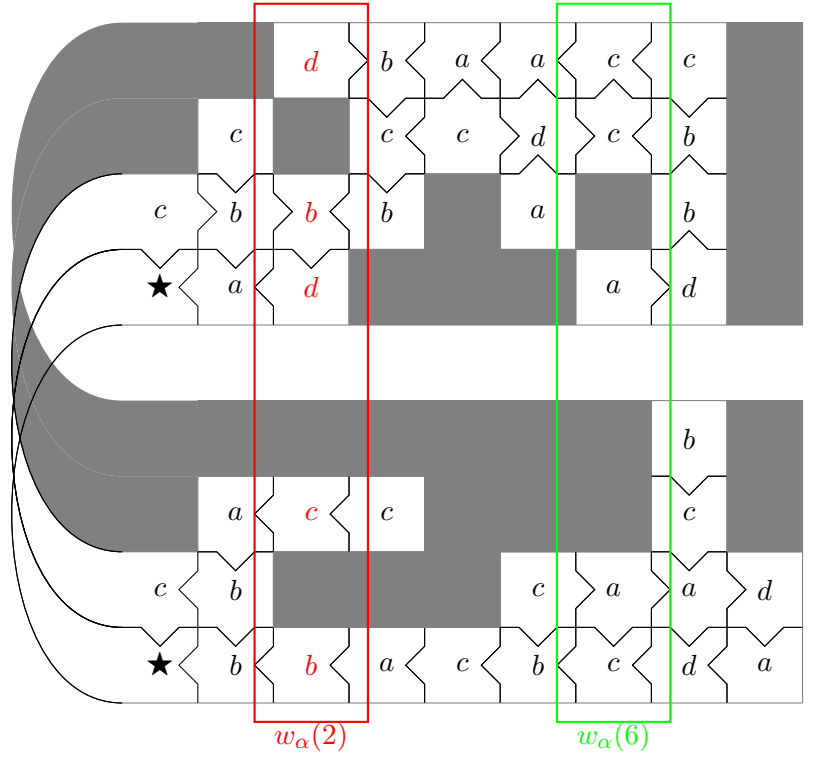}
        \draw[thick, red] (2, 4) +(-.25, .25) rectangle +(1.25, -9.25);
\node[below=.3em, red] at (2.5, -5) {$w_\alpha(2)$};
\draw[thick, green] (6, 4) +(-.25, .25) rectangle +(1.25, -9.25);
\node[below=.3em, green] at (6.5, -5) {$w_\alpha(6)$};
        \end{scope}
    \end{tikzpicture}
    \caption{Above: a production $\alpha$ with bounded height $k = 4$ together with a predecessor relation $\prec_\alpha$. There is a notch from each tile to each of its predecessors according to $\prec_\alpha$. Below: its encoding as a (potentially infinite) word $w_\alpha$. Each letter of $w_\alpha$ is a column of height $2k$ with the labels of the tiles as well as their notches, i.e. the direction of their predecessors.}
    \label{fig:buechi_encoding}
\end{figure}

\aTAMbuchi*

\begin{proof}
    The language $\mathcal{L}_\prec(\calT)$ can be defind as the intersection of four languages, each recognized by a DFA with a Büchi acceptance condition.

    The first language $S$ asserts that $\emptyset$ appears as a set of predecessors exactly on the positions of $\dom \sigma$, and that for every edge outside $\dom \sigma$, one of its extremity is the predecessor of the other. The second language $B$ asserts that the binding strength between each tile and its predecessors is at least $\tau$. Both conditions are entirely local and easily checked by a DFA with a Büchi acceptance condition.

    The third language $C$ asserts that the transitive closure $\prec^*$ of $\prec$ has no cycle. The associated automaton keeps as its state the relations between position in the leftmost and rightmost visited columns by $\prec^*$, and updates it as it reads each new pair of columns.

    \begin{figure}
        \centering
        \begin{tikzpicture}[scale=.75, mid arrow/.style={postaction={decorate,decoration={
        markings,
        mark=at position .5 with {\arrow[#1]{stealth}}
      }}}]
      \begin{scope}[opacity=.2]
        \input{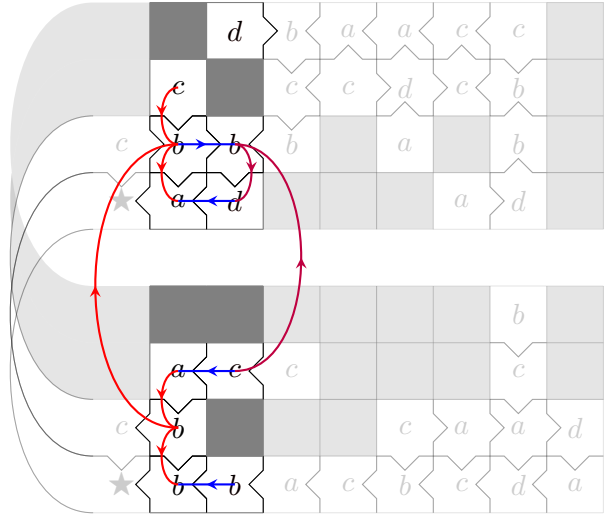}
      \end{scope}
            \begin{scope}[shift={(1,0)}]
\node at (.5, .5) { $a$ };
\draw[rotate around={(0:(.5, .5))}] (0, 1) -- (.25, 1) -- (.5, .75) -- (.75, 1) -- (1, 1);
\draw[rotate around={(270:(.5, .5))}] (0, 1) -- (.25, 1) -- (.5, .75) -- (.75, 1) -- (1, 1);
\draw[gray, rotate around={(180:(.5, .5))}] (0, 1) -- (1, 1);
\draw[rotate around={(90:(.5, .5))}] (0, 1) -- (.25, 1) -- (.5, 1.25) -- (.75, 1) -- (1, 1);
\end{scope}
\begin{scope}[shift={(1,1)}]
\node at (.5, .5) { $b$ };
\draw[rotate around={(0:(.5, .5))}] (0, 1) -- (.25, 1) -- (.5, .75) -- (.75, 1) -- (1, 1);
\draw[rotate around={(270:(.5, .5))}] (0, 1) -- (.25, 1) -- (.5, 1.25) -- (.75, 1) -- (1, 1);
\draw[rotate around={(180:(.5, .5))}] (0, 1) -- (.25, 1) -- (.5, 1.25) -- (.75, 1) -- (1, 1);
\draw[rotate around={(90:(.5, .5))}] (0, 1) -- (.25, 1) -- (.5, .75) -- (.75, 1) -- (1, 1);
\end{scope}
\begin{scope}[shift={(1,2)}]
\node at (.5, .5) { $c$ };
\draw[gray, rotate around={(0:(.5, .5))}] (0, 1) -- (1, 1);
\draw[gray, rotate around={(270:(.5, .5))}] (0, 1) -- (1, 1);
\draw[rotate around={(180:(.5, .5))}] (0, 1) -- (.25, 1) -- (.5, 1.25) -- (.75, 1) -- (1, 1);
\draw[gray, rotate around={(90:(.5, .5))}] (0, 1) -- (1, 1);
\end{scope}
\begin{scope}[shift={(1,3)}]
\fill[gray] (0, 0) rectangle +(1, 1);
\draw[gray, rotate around={(0:(.5, .5))}] (0, 1) -- (1, 1);
\draw[gray, rotate around={(270:(.5, .5))}] (0, 1) -- (1, 1);
\draw[gray, rotate around={(180:(.5, .5))}] (0, 1) -- (1, 1);
\draw[gray, rotate around={(90:(.5, .5))}] (0, 1) -- (1, 1);
\end{scope}
\begin{scope}[shift={(1, -5)}, rotate={180}, yscale=-1]
\begin{scope}[shift={(-1,0)}]
\node at (.5, .5) { $b$ };
\draw[rotate around={(0:(.5, .5))}] (0, 1) -- (.25, 1) -- (.5, .75) -- (.75, 1) -- (1, 1);
\draw[rotate around={(270:(.5, .5))}] (0, 1) -- (.25, 1) -- (.5, 1.25) -- (.75, 1) -- (1, 1);
\draw[gray, rotate around={(180:(.5, .5))}] (0, 1) -- (1, 1);
\draw[rotate around={(90:(.5, .5))}] (0, 1) -- (.25, 1) -- (.5, .75) -- (.75, 1) -- (1, 1);
\end{scope}
\begin{scope}[shift={(-1,1)}]
\node at (.5, .5) { $b$ };
\draw[rotate around={(0:(.5, .5))}] (0, 1) -- (.25, 1) -- (.5, .75) -- (.75, 1) -- (1, 1);
\draw[rotate around={(270:(.5, .5))}] (0, 1) -- (.25, 1) -- (.5, 1.25) -- (.75, 1) -- (1, 1);
\draw[rotate around={(180:(.5, .5))}] (0, 1) -- (.25, 1) -- (.5, 1.25) -- (.75, 1) -- (1, 1);
\draw[gray, rotate around={(90:(.5, .5))}] (0, 1) -- (1, 1);
\end{scope}
\begin{scope}[shift={(-1,2)}]
\node at (.5, .5) { $a$ };
\draw[gray, rotate around={(0:(.5, .5))}] (0, 1) -- (1, 1);
\draw[gray, rotate around={(270:(.5, .5))}] (0, 1) -- (1, 1);
\draw[rotate around={(180:(.5, .5))}] (0, 1) -- (.25, 1) -- (.5, 1.25) -- (.75, 1) -- (1, 1);
\draw[rotate around={(90:(.5, .5))}] (0, 1) -- (.25, 1) -- (.5, .75) -- (.75, 1) -- (1, 1);
\end{scope}
\begin{scope}[shift={(-1,3)}]
\fill[gray] (0, 0) rectangle +(1, 1);
\draw[gray, rotate around={(0:(.5, .5))}] (0, 1) -- (1, 1);
\draw[gray, rotate around={(270:(.5, .5))}] (0, 1) -- (1, 1);
\draw[gray, rotate around={(180:(.5, .5))}] (0, 1) -- (1, 1);
\draw[gray, rotate around={(90:(.5, .5))}] (0, 1) -- (1, 1);
\end{scope}
\end{scope}
\draw[shift={(.5, .5)}, thick, red, mid arrow] (1,-4) to[out=180, in=180] (1, -5);
\draw[shift={(.5, .5)}, thick, red, mid arrow] (1,-4) to[out=180, in=180] (1, 1);
\draw[shift={(.5, .5)}, thick, red, mid arrow] (1,-3) to[out=180, in=180] (1, -4);
\draw[shift={(.5, .5)}, thick, red, mid arrow] (1,1) to[out=180, in=180] (1, 0);
\draw[shift={(.5, .5)}, thick, red, mid arrow] (1,2) to[out=180, in=180] (1, 1);
\begin{scope}[shift={(2,0)}]
\node at (.5, .5) { $d$ };
\draw[rotate around={(0:(.5, .5))}] (0, 1) -- (.25, 1) -- (.5, .75) -- (.75, 1) -- (1, 1);
\draw[gray, rotate around={(270:(.5, .5))}] (0, 1) -- (1, 1);
\draw[gray, rotate around={(180:(.5, .5))}] (0, 1) -- (1, 1);
\draw[rotate around={(90:(.5, .5))}] (0, 1) -- (.25, 1) -- (.5, 1.25) -- (.75, 1) -- (1, 1);
\end{scope}
\begin{scope}[shift={(2,1)}]
\node at (.5, .5) { $b$ };
\draw[gray, rotate around={(0:(.5, .5))}] (0, 1) -- (1, 1);
\draw[rotate around={(270:(.5, .5))}] (0, 1) -- (.25, 1) -- (.5, .75) -- (.75, 1) -- (1, 1);
\draw[rotate around={(180:(.5, .5))}] (0, 1) -- (.25, 1) -- (.5, 1.25) -- (.75, 1) -- (1, 1);
\draw[rotate around={(90:(.5, .5))}] (0, 1) -- (.25, 1) -- (.5, .75) -- (.75, 1) -- (1, 1);
\end{scope}
\begin{scope}[shift={(2,2)}]
\fill[gray] (0, 0) rectangle +(1, 1);
\draw[gray, rotate around={(0:(.5, .5))}] (0, 1) -- (1, 1);
\draw[gray, rotate around={(270:(.5, .5))}] (0, 1) -- (1, 1);
\draw[gray, rotate around={(180:(.5, .5))}] (0, 1) -- (1, 1);
\draw[gray, rotate around={(90:(.5, .5))}] (0, 1) -- (1, 1);
\end{scope}
\begin{scope}[shift={(2,3)}]
\node at (.5, .5) { $d$ };
\draw[gray, rotate around={(0:(.5, .5))}] (0, 1) -- (1, 1);
\draw[rotate around={(270:(.5, .5))}] (0, 1) -- (.25, 1) -- (.5, 1.25) -- (.75, 1) -- (1, 1);
\draw[gray, rotate around={(180:(.5, .5))}] (0, 1) -- (1, 1);
\draw[gray, rotate around={(90:(.5, .5))}] (0, 1) -- (1, 1);
\end{scope}
\begin{scope}[shift={(1, -5)}, rotate={180}, yscale=-1]
\begin{scope}[shift={(-2,0)}]
\node at (.5, .5) { $b$ };
\draw[gray, rotate around={(0:(.5, .5))}] (0, 1) -- (1, 1);
\draw[rotate around={(270:(.5, .5))}] (0, 1) -- (.25, 1) -- (.5, 1.25) -- (.75, 1) -- (1, 1);
\draw[gray, rotate around={(180:(.5, .5))}] (0, 1) -- (1, 1);
\draw[rotate around={(90:(.5, .5))}] (0, 1) -- (.25, 1) -- (.5, .75) -- (.75, 1) -- (1, 1);
\end{scope}
\begin{scope}[shift={(-2,1)}]
\fill[gray] (0, 0) rectangle +(1, 1);
\draw[gray, rotate around={(0:(.5, .5))}] (0, 1) -- (1, 1);
\draw[gray, rotate around={(270:(.5, .5))}] (0, 1) -- (1, 1);
\draw[gray, rotate around={(180:(.5, .5))}] (0, 1) -- (1, 1);
\draw[gray, rotate around={(90:(.5, .5))}] (0, 1) -- (1, 1);
\end{scope}
\begin{scope}[shift={(-2,2)}]
\node at (.5, .5) { $c$ };
\draw[gray, rotate around={(0:(.5, .5))}] (0, 1) -- (1, 1);
\draw[rotate around={(270:(.5, .5))}] (0, 1) -- (.25, 1) -- (.5, 1.25) -- (.75, 1) -- (1, 1);
\draw[gray, rotate around={(180:(.5, .5))}] (0, 1) -- (1, 1);
\draw[rotate around={(90:(.5, .5))}] (0, 1) -- (.25, 1) -- (.5, .75) -- (.75, 1) -- (1, 1);
\end{scope}
\begin{scope}[shift={(-2,3)}]
\fill[gray] (0, 0) rectangle +(1, 1);
\draw[gray, rotate around={(0:(.5, .5))}] (0, 1) -- (1, 1);
\draw[gray, rotate around={(270:(.5, .5))}] (0, 1) -- (1, 1);
\draw[gray, rotate around={(180:(.5, .5))}] (0, 1) -- (1, 1);
\draw[gray, rotate around={(90:(.5, .5))}] (0, 1) -- (1, 1);
\end{scope}
\end{scope}
\draw[shift={(.5, .5)}, thick, purple, mid arrow] (2,-3) to[out=0, in=0] (2, 1);
\draw[shift={(.5, .5)}, thick, purple, mid arrow] (2,1) to[out=0, in=0] (2, 0);
\draw[shift={(.5, .5)}, thick, blue, mid arrow] (2,0) -- (1, 0);
\draw[shift={(.5, .5)}, thick, blue, mid arrow] (1,1) -- (2, 1);
\draw[shift={(.5, .5)}, thick, blue, mid arrow] (2,-5) -- (1, -5);
\draw[shift={(.5, .5)}, thick, blue, mid arrow] (2,-3) -- (1, -3);
        \end{tikzpicture}
        \caption{Principle of the Büchi automaton for the sublanguage $C$ asserting that there are no dependency cycles, demonstrated on the transition from $w_\alpha(1)$ to $w_\alpha(2)$. The transitive closure of the predecessor relations seen so far by the automaton is in red; the automaton only stores the relations between elements in the pair of column it has seen last, here columns $\pm 1$. Then it uses the predecessor relations of the tiles in the current pair of columns (the notches) to compute the transitive relation between these (in purple). From there, it can deduce the relations between the tiles in columns $\pm 2$ (in purple). }
        \label{fig:placeholder}
    \end{figure}

    Last, the language $F$ asserts that $\prec^*$ has no infinite path. It has a Büchi acceptance condition on transition rather than states, which is equivalent. It does so by keeping in its state, a list of edges going outwards from the leftmost and rightmost column, ordered by the smallest $|i|$ such that there is a path from a position in column $i$ to that outgoing edge (ties broken arbitrarily). This is only a finite amount of information, as there are at most $2 \cdot (2k + 1)$ such edges, and only their \emph{order} is of interest. Each new pair of columns contains all the information that is needed to update this information. Accepting transitions are those where the first edge of the old state has no path to the first edge of the new state, or both states have empty lists of outgoing states. If there are infinite paths in $\prec^*$, one of them eventually has its edges at the first position of the state. From that point on, the automaton cannot enter an accepting transition again. Conversely, if the automaton stops running into accepting transitions, it must be because some fixed position keeps depending on ever further edges, which implies an infinite path in $\prec^*$.

    \begin{figure}
        \centering
        \begin{tikzpicture}[scale=.75]
            \draw[dashed, red] (-2, -5.25) rectangle (6, 4.25);
            \draw[dashed, blue] (-2.1, -5.35) rectangle (7, 4.35);
            \begin{scope}[opacity=.2]
            \input{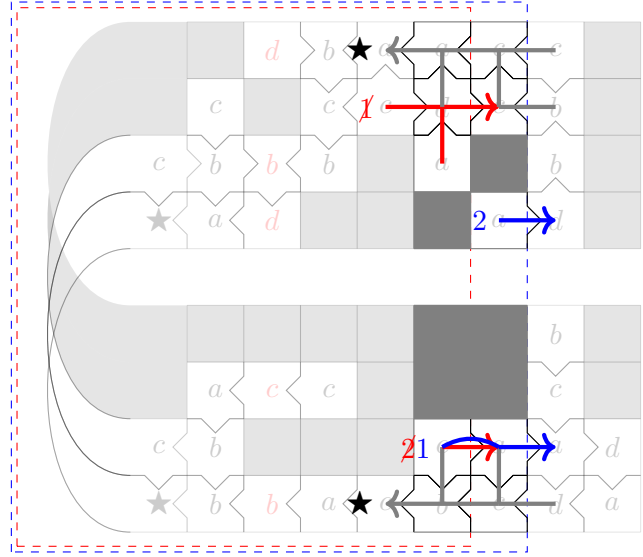}
            \end{scope}
            \begin{scope}[shift={(6,0)}]
\node at (.5, .5) {  };
\draw[gray, rotate around={(0:(.5, .5))}] (0, 1) -- (1, 1);
\draw[rotate around={(270:(.5, .5))}] (0, 1) -- (.25, 1) -- (.5, 1.25) -- (.75, 1) -- (1, 1);
\draw[gray, rotate around={(180:(.5, .5))}] (0, 1) -- (1, 1);
\draw[gray, rotate around={(90:(.5, .5))}] (0, 1) -- (1, 1);
\end{scope}
\begin{scope}[shift={(6,1)}]
\fill[gray] (0, 0) rectangle +(1, 1);
\draw[gray, rotate around={(0:(.5, .5))}] (0, 1) -- (1, 1);
\draw[gray, rotate around={(270:(.5, .5))}] (0, 1) -- (1, 1);
\draw[gray, rotate around={(180:(.5, .5))}] (0, 1) -- (1, 1);
\draw[gray, rotate around={(90:(.5, .5))}] (0, 1) -- (1, 1);
\end{scope}
\begin{scope}[shift={(6,2)}]
\node at (.5, .5) {  };
\draw[rotate around={(0:(.5, .5))}] (0, 1) -- (.25, 1) -- (.5, 1.25) -- (.75, 1) -- (1, 1);
\draw[rotate around={(270:(.5, .5))}] (0, 1) -- (.25, 1) -- (.5, .75) -- (.75, 1) -- (1, 1);
\draw[gray, rotate around={(180:(.5, .5))}] (0, 1) -- (1, 1);
\draw[rotate around={(90:(.5, .5))}] (0, 1) -- (.25, 1) -- (.5, .75) -- (.75, 1) -- (1, 1);
\end{scope}
\begin{scope}[shift={(6,3)}]
\node at (.5, .5) {  };
\draw[gray, rotate around={(0:(.5, .5))}] (0, 1) -- (1, 1);
\draw[rotate around={(270:(.5, .5))}] (0, 1) -- (.25, 1) -- (.5, .75) -- (.75, 1) -- (1, 1);
\draw[rotate around={(180:(.5, .5))}] (0, 1) -- (.25, 1) -- (.5, .75) -- (.75, 1) -- (1, 1);
\draw[rotate around={(90:(.5, .5))}] (0, 1) -- (.25, 1) -- (.5, 1.25) -- (.75, 1) -- (1, 1);
\end{scope}
\begin{scope}[shift={(1, -5)}, rotate={180}, yscale=-1]
\begin{scope}[shift={(-6,0)}]
\node at (.5, .5) {  };
\draw[rotate around={(0:(.5, .5))}] (0, 1) -- (.25, 1) -- (.5, .75) -- (.75, 1) -- (1, 1);
\draw[rotate around={(270:(.5, .5))}] (0, 1) -- (.25, 1) -- (.5, 1.25) -- (.75, 1) -- (1, 1);
\draw[gray, rotate around={(180:(.5, .5))}] (0, 1) -- (1, 1);
\draw[rotate around={(90:(.5, .5))}] (0, 1) -- (.25, 1) -- (.5, .75) -- (.75, 1) -- (1, 1);
\end{scope}
\begin{scope}[shift={(-6,1)}]
\node at (.5, .5) {  };
\draw[gray, rotate around={(0:(.5, .5))}] (0, 1) -- (1, 1);
\draw[rotate around={(270:(.5, .5))}] (0, 1) -- (.25, 1) -- (.5, .75) -- (.75, 1) -- (1, 1);
\draw[rotate around={(180:(.5, .5))}] (0, 1) -- (.25, 1) -- (.5, 1.25) -- (.75, 1) -- (1, 1);
\draw[rotate around={(90:(.5, .5))}] (0, 1) -- (.25, 1) -- (.5, 1.25) -- (.75, 1) -- (1, 1);
\end{scope}
\begin{scope}[shift={(-6,2)}]
\fill[gray] (0, 0) rectangle +(1, 1);
\draw[gray, rotate around={(0:(.5, .5))}] (0, 1) -- (1, 1);
\draw[gray, rotate around={(270:(.5, .5))}] (0, 1) -- (1, 1);
\draw[gray, rotate around={(180:(.5, .5))}] (0, 1) -- (1, 1);
\draw[gray, rotate around={(90:(.5, .5))}] (0, 1) -- (1, 1);
\end{scope}
\begin{scope}[shift={(-6,3)}]
\fill[gray] (0, 0) rectangle +(1, 1);
\draw[gray, rotate around={(0:(.5, .5))}] (0, 1) -- (1, 1);
\draw[gray, rotate around={(270:(.5, .5))}] (0, 1) -- (1, 1);
\draw[gray, rotate around={(180:(.5, .5))}] (0, 1) -- (1, 1);
\draw[gray, rotate around={(90:(.5, .5))}] (0, 1) -- (1, 1);
\end{scope}
\end{scope}
\begin{scope}[shift={(5,0)}]
\fill[gray] (0, 0) rectangle +(1, 1);
\draw[gray, rotate around={(0:(.5, .5))}] (0, 1) -- (1, 1);
\draw[gray, rotate around={(270:(.5, .5))}] (0, 1) -- (1, 1);
\draw[gray, rotate around={(180:(.5, .5))}] (0, 1) -- (1, 1);
\draw[gray, rotate around={(90:(.5, .5))}] (0, 1) -- (1, 1);
\end{scope}
\begin{scope}[shift={(5,1)}]
\node at (.5, .5) {  };
\draw[rotate around={(0:(.5, .5))}] (0, 1) -- (.25, 1) -- (.5, 1.25) -- (.75, 1) -- (1, 1);
\draw[gray, rotate around={(270:(.5, .5))}] (0, 1) -- (1, 1);
\draw[gray, rotate around={(180:(.5, .5))}] (0, 1) -- (1, 1);
\draw[gray, rotate around={(90:(.5, .5))}] (0, 1) -- (1, 1);
\end{scope}
\begin{scope}[shift={(5,2)}]
\node at (.5, .5) {  };
\draw[rotate around={(0:(.5, .5))}] (0, 1) -- (.25, 1) -- (.5, 1.25) -- (.75, 1) -- (1, 1);
\draw[rotate around={(270:(.5, .5))}] (0, 1) -- (.25, 1) -- (.5, 1.25) -- (.75, 1) -- (1, 1);
\draw[rotate around={(180:(.5, .5))}] (0, 1) -- (.25, 1) -- (.5, .75) -- (.75, 1) -- (1, 1);
\draw[rotate around={(90:(.5, .5))}] (0, 1) -- (.25, 1) -- (.5, .75) -- (.75, 1) -- (1, 1);
\end{scope}
\begin{scope}[shift={(5,3)}]
\node at (.5, .5) {  };
\draw[gray, rotate around={(0:(.5, .5))}] (0, 1) -- (1, 1);
\draw[rotate around={(270:(.5, .5))}] (0, 1) -- (.25, 1) -- (.5, .75) -- (.75, 1) -- (1, 1);
\draw[rotate around={(180:(.5, .5))}] (0, 1) -- (.25, 1) -- (.5, .75) -- (.75, 1) -- (1, 1);
\draw[rotate around={(90:(.5, .5))}] (0, 1) -- (.25, 1) -- (.5, 1.25) -- (.75, 1) -- (1, 1);
\end{scope}
\begin{scope}[shift={(1, -5)}, rotate={180}, yscale=-1]
\begin{scope}[shift={(-5,0)}]
\node at (.5, .5) {  };
\draw[rotate around={(0:(.5, .5))}] (0, 1) -- (.25, 1) -- (.5, .75) -- (.75, 1) -- (1, 1);
\draw[rotate around={(270:(.5, .5))}] (0, 1) -- (.25, 1) -- (.5, 1.25) -- (.75, 1) -- (1, 1);
\draw[gray, rotate around={(180:(.5, .5))}] (0, 1) -- (1, 1);
\draw[rotate around={(90:(.5, .5))}] (0, 1) -- (.25, 1) -- (.5, .75) -- (.75, 1) -- (1, 1);
\end{scope}
\begin{scope}[shift={(-5,1)}]
\node at (.5, .5) {  };
\draw[gray, rotate around={(0:(.5, .5))}] (0, 1) -- (1, 1);
\draw[gray, rotate around={(270:(.5, .5))}] (0, 1) -- (1, 1);
\draw[rotate around={(180:(.5, .5))}] (0, 1) -- (.25, 1) -- (.5, 1.25) -- (.75, 1) -- (1, 1);
\draw[rotate around={(90:(.5, .5))}] (0, 1) -- (.25, 1) -- (.5, 1.25) -- (.75, 1) -- (1, 1);
\end{scope}
\begin{scope}[shift={(-5,2)}]
\fill[gray] (0, 0) rectangle +(1, 1);
\draw[gray, rotate around={(0:(.5, .5))}] (0, 1) -- (1, 1);
\draw[gray, rotate around={(270:(.5, .5))}] (0, 1) -- (1, 1);
\draw[gray, rotate around={(180:(.5, .5))}] (0, 1) -- (1, 1);
\draw[gray, rotate around={(90:(.5, .5))}] (0, 1) -- (1, 1);
\end{scope}
\begin{scope}[shift={(-5,3)}]
\fill[gray] (0, 0) rectangle +(1, 1);
\draw[gray, rotate around={(0:(.5, .5))}] (0, 1) -- (1, 1);
\draw[gray, rotate around={(270:(.5, .5))}] (0, 1) -- (1, 1);
\draw[gray, rotate around={(180:(.5, .5))}] (0, 1) -- (1, 1);
\draw[gray, rotate around={(90:(.5, .5))}] (0, 1) -- (1, 1);
\end{scope}
\end{scope}
            \draw[shift={(.5, .5)}, ultra thick, red, ->] (4,2) node[left] {\cancel{1}} -- (5,2) -- (6,2);
\draw[shift={(.5, .5)}, ultra thick, red] (5,1) -- (5,2);
%\draw[shift={(.5, .5)}, ultra thick, red, ->] (5,-4) -- +(-1, 0) node[left] {\cancel{2}\textcolor{blue}{1}};
\draw[shift={(.5, .5)}, ultra thick, gray, ->] (6,2) -- (6,3) -- (5,3) -- +(-1, 0) node[left, black] {$\bigstar$};
\draw[shift={(.5, .5)}, ultra thick, gray] (5,2) -- (5,3) (5,-4) -- (5,-5) (6,-4) -- (6,-5);
\draw[shift={(.5, .5)}, ultra thick, gray] (6,2) -- (7,2) (6,-5) -- (7,-5) (7,3) -- (6,3);
\draw[shift={(.5, .5)}, ultra thick, gray, ->] (6, -5) -- (5,-5) -- +(-1, 0) node[left, black] {$\bigstar$};

\draw[shift={(.5, .5)}, ultra thick, red, ->] (5,-4) node[left] {\cancel{2}\textcolor{blue}{1}} to (6,-4);
\draw[shift={(.5, .5)}, ultra thick, blue] (5, -4) to[bend left] (6, -4);
\draw[shift={(.5, .5)}, ultra thick, blue, ->] (6,-4) -- (7,-4);
\draw[shift={(.5, .5)}, ultra thick, blue, ->] (6,0) node[left] {2} -- (7,0);
        \end{tikzpicture}
        \caption{Principle of the Büchi automaton for the language $F$ asserting that there are no infinite paths, demonstrated on the transtion from $w_\alpha(5)$ to $w_\alpha(6)$. The red edges are the one from which one can exit the 5 central columns following arcs of $\prec_\alpha$; the red number at the left of each of their connected component orders them according to how close they come to column $0$, \emph{id est} how early the automaton encountered them. Likewise for the blue edges and numbers regarding the 6 central columns. All paths starting from the gray edges reach back to the seed. The red edges marked \textcolor{red}{1} disappear as the right column is read: the automaton just discovered their missing edges back to the seed. The red \textcolor{red}{2} becomes a blue \textcolor{blue}{1}, as this sub-DAG is now the blue one reaching closest to column $0$. One new tile has its only predecessor in column 7, which creates a new family of hanging edges, which takes the lowest available number: \textcolor{blue}{2}. Since the old \textcolor{red}{1} has been replaced, this transition is accepting.}
        \label{fig:placeholder}
    \end{figure}

    For the language of terminal productions, a fifth condition is added, which checks that no empty position is attachable. Once again, this is a local safety condition which can be checked by a DFA with a Büchi acceptance condition.
\end{proof}

As a consequence, the set of (final) productions of a $k$-band-bounded TAS is $\omega$-regular, as it is the projection of $\mathcal{L}_\prec(\calT)$ which forgets the $\prec$ decorations.
\section{Bounded-Width Computation}\label{sec:bounded}

In this section, we prove that it is decidable to determine if an arbitrary aTAM system grows outside of a fixed-width region of the plane, but that the same question is undecidable for an arbitrary syncTAM system.

\kBandProblem*

\subsection{Decidability in the aTAM}

\aTAMkbounded*

\begin{proof}
Let $\calS = (\calT, \tau, \sigma)$ be a TAS in the aTAM and let $G$ be its set of glues. Assume without loss of generality that $\sigma$ is included in $\{-k, \dots, k\} \times \{-h,\dots,h\}$ for some $h$.

Let $\mathfrak{C}_k = (G \cup \{\#\})^{(\{-k, \dots, k\})}$ be the set of configurations of glues of $\calT$ on a horizontal segment around the $y$ axis. Let $\mathfrak{M}_k \subset \mathfrak{C}_k^*$ be the set of potential glue movies on such a segment, i.e. the set of sequences $s$ where for all $k$, $s_k$ and $s_{k+1}$ only differ in one position, and that position is empty in $s_k$.  Pose $G = |\mathfrak{M}_k|$.

Assume $\calT$ has a production $P$ which has exactly one tile outside the vertical band of radius $k$ around the vertical axis, and that this tile is at position $(k+1, y)$ with $y > h + G$. Then there are two values $h < w_1 < w_2 < y$ such that the horizontal lines
at height $w_1$ and $w_2$ define windows with the same glue movies. Thus, by the Window Movie Lemma, there is a production of $\calT$ where the first attachment outside the central vertical band is located closer to the seed than in $P$. Finally, $\calT$ satisfies $k$-Band Boundedness if and only if all of its production $\alpha$ with $\dom \alpha \subset \{-h -G, \dots, h + G\} \times \mathbb{Z}$ satisfy $\dom \alpha \subset \mathbb{Z} \times \{-k, \dots, k\}$.

Checking this last condition can be done by searching through all assembly sequences which do not reach too far up or down for one going horizontally more than $k$ away from the axis. The bound on the distance is given by the number of glue movies of width $k$: $d_{\max} = k! (|G|+1)^k$; an assembly sequence of $\calS$ can take at most $2 (k + |\sigma|)$ within that region, so ensuring none of them are ready to stray outside is in Co-NEXPTIME.

\end{proof}

\subsection{Undecidability in the syncTAM}\label{sec:two-counter-result}

We now show that the same problem is undecidable in the syncTAM by providing a construction capable of simulating an arbitrary two-counter machine in the syncTAM without placing any tiles outside of the region bounded by $x$-coordinates $-9$ and $9$. Since two-counter machines are Turing universal \cite{Minsky1967}, this results in the undecidability of the $k$-Band Boundedness problem.

A \emph{two-counter machine} (a.k.a. \emph{two-register machine} or \emph{Minsky machine}) is a type of finite state automaton that has as its memory two \emph{counters} (or \emph{registers}), each of which can hold any non-negative integer value. For each transition, a two-counter machine is able to take as input the status of each counter in terms of either \emph{zero} or \emph{not zero} (i.e. whether or not its current value is $0$), along with the machine's current state, with the result of the transition being the next state for the machine and the possibility of altering each counter by either decrementing by 1, incrementing by 1, or leaving it unchanged. An important caveat is that, if the current value of a counter is $0$ and the transition attempts to decrement that counter, it remains unchanged.  (See \cref{sec:two-counter-example} for an example two-counter machine and \url{https://self-assembly.net/software/two_counter_machines/two_counter_machines-latest/} for a two-counter machine simulator and a function to create a simulating syncTAM system, along with examples.)

\begin{definition}[Two-Counter Machine]
    A \emph{two-counter machine} $M$ is a tuple $(Q, q_0, H, \delta)$ where:
    \begin{itemize}
        \item $Q$ is a finite set of states
        \item $q_0 \in Q$ is the initial state
        \item $H \subset Q$ is the set of halting states
        \item \(\delta: Q \times \{0, 1\}^2 \to Q \times \{-1, 0, 1\}^2\) is the transition function, which must satisfy $\delta(q, v_l, v_r) = (q', d_l, d_r) \Rightarrow (v_l + d_l \geq 0 \wedge v_r + d_r \geq 0)$
    \end{itemize}

    A \emph{configuration} of $M$ is a triplet $(q, l, r) \in Q \times \mathbb{N}^2$; the successor $\Delta(q, l, r)$ of a configuration $(q, l, r)$ is $(q', l + c_l, r + c_r)$ where $(q', c_l, c_r) = \delta(q, v_l, v_r)$, $v_l = \max(l, 1)$ and $v_r = \max(r, 1)$.
    
    The \emph{run} of $M$ on input $(l_0, r_0) \in \mathbb{N}^2$ is the sequence of configurations defined by $c_0 = (q_0, l_0, r_0)$ and $c_{i+1} = \Delta(c_i)$. The machine $M$ \emph{halts} on input $(l_0, r_0)$ if its run enters $H \times \mathbb{N}^2$.
\end{definition}\syncTAMkbounded*

\subsubsection{Example two-counter machine}\label{sec:two-counter-example}

An example of a simple two-counter machine that calculates the Collatz function can be seen in \cref{fig:two-counter-collatz}. A trace of a run when its counter values are initialized to $3$ and $0$ is shown in Listing \ref{lst:trace}.

%Two-counter machines are known to be computationally universal \cite{Minsky1967}.

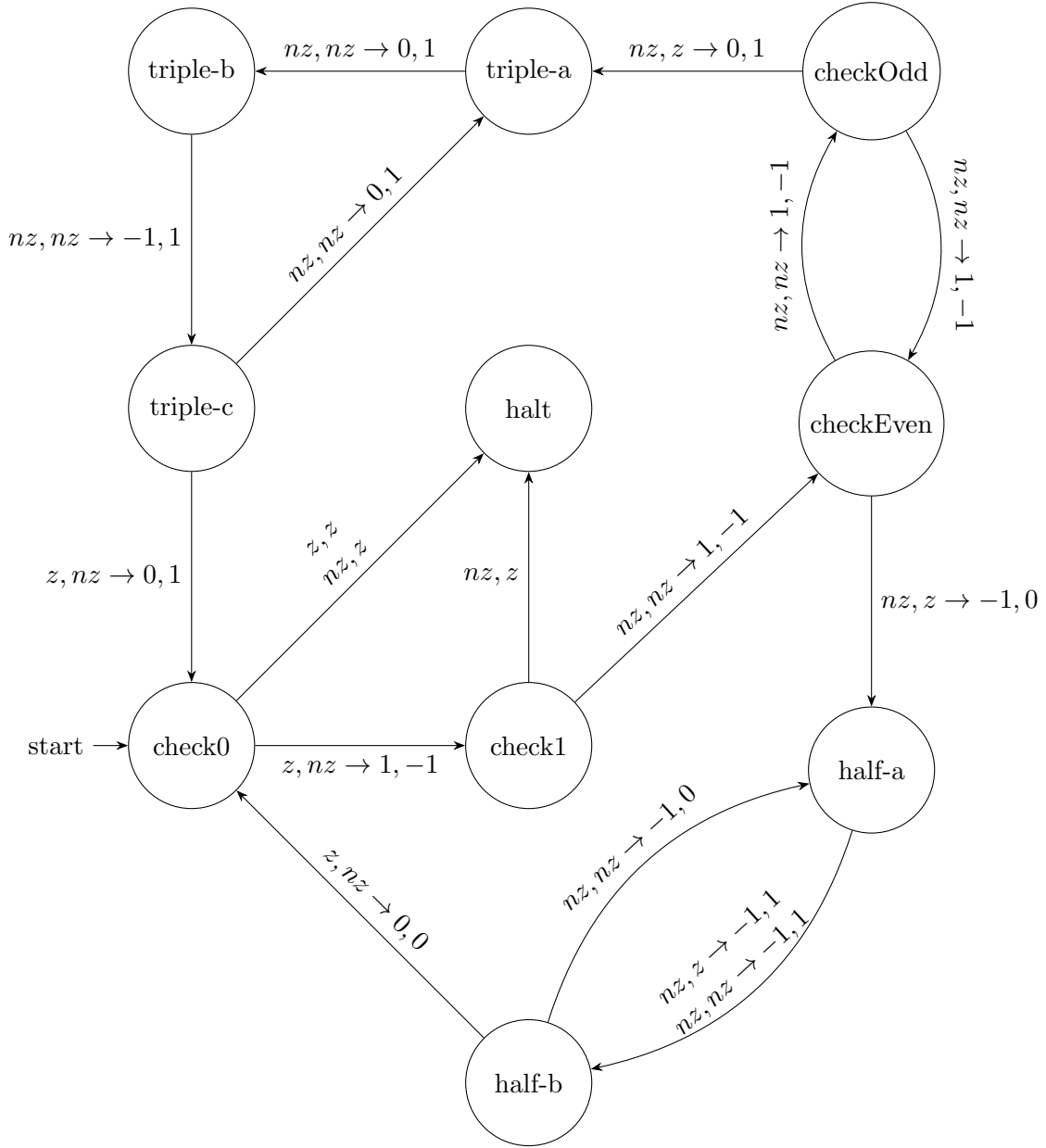
\begin{figure}
    \centering
\begin{tikzpicture}[->, >=Stealth, auto, node distance=3.0cm, every state/.append style={minimum size=1.8cm}]
        \node[state, initial] (check0) {check0};
        \node[state, right=of check0] (check1) {check1};
        \node[state, above=of check1] (halt) {halt};
        \node[state, above=of check0] (triple-c) {triple-c};
        \node[state, above=of triple-c] (triple-b) {triple-b};
        \node[state, right=of triple-b] (triple-a) {triple-a};
        \node[state, right=of triple-a] (odd) {checkOdd};
        \node[state, below=of odd] (even) {checkEven};
        \node[state, below=of even] (half-a) {half-a};
        \node[state, below=of check1] (half-b) {half-b};
        \path
            (check0) edge node[sloped, align=center] {$z,z$ \\ $nz,z$} (halt)
            (check0) edge node[swap] {$z,nz \rightarrow 1, -1$} (check1)
            (check1) edge node {$nz,z$} (halt)
            (check1) edge node[sloped] {$nz,nz \rightarrow 1, -1$} (even)
            (even) edge[bend left, sloped] node {$nz,nz \rightarrow 1, -1$} (odd)
            (odd) edge[bend left, sloped] node {$nz,nz \rightarrow 1, -1$} (even)
            (even) edge node {$nz,z \rightarrow -1, 0$} (half-a)
            (half-a) edge[bend left, sloped] node[align=center]
                {$nz,z \rightarrow -1, 1$ \\ $nz,nz \rightarrow -1, 1$} (half-b)
            (half-b) edge[bend left, sloped] node {$nz,nz \rightarrow -1, 0$} (half-a)
            (half-b) edge node[sloped] {$z,nz \rightarrow 0, 0$} (check0)
            (odd) edge node[swap] {$nz,z \rightarrow 0, 1$} (triple-a)
            (triple-a) edge node[swap] {$nz,nz \rightarrow 0, 1$} (triple-b)
            (triple-b) edge node[swap] {$nz,nz \rightarrow -1, 1$} (triple-c)
            (triple-c) edge node[sloped] {$nz,nz \rightarrow 0, 1$} (triple-a)
            (triple-c) edge node[swap] {$z,nz \rightarrow 0, 1$} (check0);
    \end{tikzpicture}
    \caption{Example two-counter machine that computes the Collatz function}
    \label{fig:two-counter-collatz}
\end{figure}

\begin{lstlisting}[basicstyle=\ttfamily\small, frame=single,
                   caption={Beginning of the trace of the two-counter machine from \cref{fig:two-counter-collatz} when its counters are initialized to 0 and 3},
                   label={lst:trace}]
    # Check if the value in r2 is 0, 1, even, or odd, moving it to r1:
    check_zero 0 3 -> check_one 1 2
    check_one 1 2 -> check_even 2 1
    check_even 2 1 -> check_odd 3 0
    
    # The value in r1 is odd, multiply by 3 and add 1, moving it to r2:
    check_odd 3 0 -> triplea 3 1
    triplea 3 1 -> tripleb 3 2
    tripleb 3 2 -> triplec 2 3
    triplec 2 3 -> triplea 2 4
    triplea 2 4 -> tripleb 2 5
    tripleb 2 5 -> triplec 1 6
    triplec 1 6 -> triplea 1 7
    triplea 1 7 -> tripleb 1 8
    tripleb 1 8 -> triplec 0 9

    # Check if the value in r2 is 0, 1, even, or odd, moving it to r1:
    triplec 0 9 -> check_zero 0 10
    check_zero 0 10 -> check_one 1 9
    check_one 1 9 -> check_even 2 8
    check_even 2 8 -> check_odd 3 7
    check_odd 3 7 -> check_even 4 6
    check_even 4 6 -> check_odd 5 5
    check_odd 5 5 -> check_even 6 4
    check_even 6 4 -> check_odd 7 3
    check_odd 7 3 -> check_even 8 2
    check_even 8 2 -> check_odd 9 1
    check_odd 9 1 -> check_even 10 0

    # The value in r1 is even, half it, and move it to r2:
    check_even 10 0 -> halfa 9 0
    halfa 9 0 -> halfb 8 1
    halfb 8 1 -> halfa 7 1
    halfa 7 1 -> halfb 6 2
    halfb 6 2 -> halfa 5 2
    halfa 5 2 -> halfb 4 3
    halfb 4 3 -> halfa 3 3
    halfa 3 3 -> halfb 2 4
    halfb 2 4 -> halfa 1 4
    halfa 1 4 -> halfb 0 5

    # Check if the value in r2 is 0, 1, even, or odd, moving it to r1:
    halfb 0 5 -> check_zero 0 5
    check_zero 0 5 -> check_one 1 4
    check_one 1 4 -> check_even 2 3
    check_even 2 3 -> check_odd 3 2
    check_odd 3 2 -> check_even 4 1
    check_even 4 1 -> check_odd 5 0
\end{lstlisting}

\subsubsection{Constructing a syncTAS to simulate a two-counter machine}

The rest of this section describes a scheme for compiling a two-counter machine $M$ into a TAS $\calS_M$ which places a tile outside $\{-9, \dots, 9\} \times \mathbb{Z}$ exactly when the counter machine halts. \Cref{thm:Turing_band} then follows from this construction.

% \begin{figure}
%     \centering
%     \input{tikz/band/example-machine}
    
%     \caption{An example two-counter machine $C$ computing the Collatz sequence starting from 0 in the left counter and an initial value in the right counter. The symbols on each transition are $v_l / v_r / c_l / c_r$; a dash for $v_l$ or $v_l$ is a wildcard.}
%     \label{fig:band-example-machine}
% \end{figure}

% \begin{table}
% \centering
%     \begin{tabular}{l|l|l|l|l|l}
%     0 3 $e$ & 3 0 $o$ & 0 10 $e$ & 10 0 $e$ & 0 5 $e$ &  \\
%     1 2 $o$ & 3 1 $*3$ & 1 9 $o$ & 9 0 $/2$ & 1 4 $o$ \\
%     2 1 $e$ & 3 2 $*3'$ & 2 8 $e$ & 8 1 $/2'$ & 2 3 $e$\\
%     3 0 $o$ & 2 3 $*3''$ & 3 7 $o$ & 7 1 $/2$ & 3 2 $o$\\
%           & 2 4 $*3$ & ... & 6 2 $/2'$ & 4 1 $e$\\
%           & 2 5 $*3'$ & 10 0 $e$ & 5 2 $/2$ & 5 0 $o$\\
%           & 1 6 $*3''$ & & 4 3 $/2'$ & \\
%           & 1 7 $*3$ & &3 3 $/2$ & \\
%           & 1 8 $*3'$ & & 2 4 $/2'$ & \\
%           & 0 9 $*3''$ & & 1 4 $/2$ & \\
%           & 0 10 $e$ & & 0 5 $/2'$ & \\
%           &          & & 0 5 $e$ & & ...\\
%           \hline
%     3 is odd & $3*3 + 1 = 10$ & 10 is even & $10/2 = 5$ & 5 is odd & ...\\
%     \end{tabular}
%     \caption{Beginning of the run of $C$ on input 3 (see \cref{fig:band-example-machine}).}
%     \label{tab:band-example-execution}
% \end{table}

The construction is made of a number of gadgets, each with its function and invariants. All productions fit in the  $\{-9, \dots, 9\} \times \mathbb{Z}$ band, except a few tiles which are only added if $M$ halts. Tiles in the columns between $-4$ and $4$ are from the ``transition gadgets'', columns $-9$ to $-4$ are from the ``left counter gadgets'', and columns $4$ to $9$ are the ``right counter gadgets''. The trick is that while altitude in the assembly of $\calS_M$ represents increasing time in the run of $M$ on input $(0, 0)$, the three gadgets around altitude $6i$ which represent the $i$th transition of the run are not assembled at the same time. In particular, let $z_l(i)$ be the position $(-4, 6i + 4)$ and $(q_i, l_i, r_i)$ be the $i$th configuration of the run of $M$ on input $(0, 0)$: the transition gadget reaches $z_l(i)$ at time $20i + 8$, while the left counter gadget reaches it at time $8(i + l_i) + 4$. Likewise, $z_r(i) = (4, 6i + 4)$ is reached at the same time as $z_l(i)$ by the transition gadget, but at time $20i + 8r_i + 4$ by the right transition gadget. In other words, \emph{the value of the counters are represented by the delay between either counter gadget and the transition gadget}. Their nullity can be observed as it is the only case where the counter gadgets arrives at $z_l(i)$ or $z_r(i)$ before the transition gadget. They can be increased or decreased by having the counter gadgets rise slower (28 ticks per gadget) or faster (12 ticks per gadget) than the rate of the transition gadget (20 ticks per gadget). These operations suffice to simulate $M$, their details follow. \cref{fig:band-example-sim} shows a time-lapse simulation of the described construction of the two-counter machine that calculates the Collatz function from \cref{fig:two-counter-collatz}.

\paragraph*{Transition gadget}

Let $i = (q, z_l, z_r) \to o = (q', c_l, c_r)$ be a transition, where:
\begin{itemize}
    \item $q, q'$ are states,
    \item $z_l, z_r$ are boolean values: whether each counter's value is 0
    \item $c_l, c_r \in \{-1, 0,1\}$ are the values to be added to the counters
\end{itemize}

The tiles for the transition gadget are depicted in \cref{fig:band-transition-tiles}. Starting from a tile representing $i = (q, z_l, z_r)$ at $(0, 0)$ at time 0, it places a tile representing $c_l$ at $(-3, 1)$ and a tile representing $c_r$ at $(3, 1)$, both at time 3. It then places a tile representing $i' = (q', z_l', z_r')$ at time 20 at position $(0, 6)$, where $z_l'$ is true if the position $(-4, 6)$ has been occupied by a \texttt{zl} tile by time 16, and false if that position is free at that time. Likewise for $z_r$, with respect to a \texttt{zr} tile at position $(4, 6)$. This implements a cycle of the two-counter machine, assuming that the placement of the tiles at $(-3, 1)$ and $(3, 1)$ trigger the appropriate update in the forthcoming counter gadgets, and that the attachment or absence of the \texttt{zl} and \texttt{zr} tiles corresponds to the nullity of the counter values after that update. 

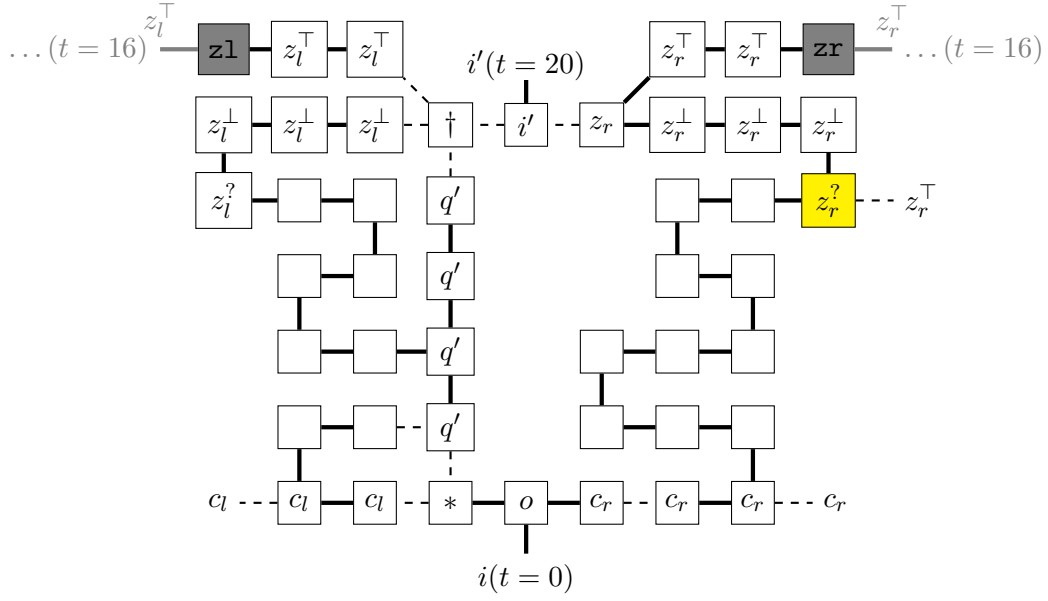
\begin{figure}
    \centering
    \begin{tikzpicture}[%
    tile/.style={draw, regular polygon, regular polygon sides=4, minimum size=.8cm, inner sep=1pt},%
    s1/.style={thick,dashed},%
    s2/.style={ultra thick}%
    ]
      \node (start) at (0,0) {$i (t=0)$};
      \node[tile] (next) at (0,6) {$i'$};
      \node[tile] (dispatch) at (0,1) {$o$};
      \node[tile, inner sep=0] (0l) at (-1,1) {$*$};
      \node[tile] (0r) at (1,1) {$c_r$};
      \node[tile] (1l) at (-2,1) {$c_l$};
      \node[tile] (1r) at (2,1) {$c_r$};
      \node[tile] (2l) at (-3,1) {$c_l$};
      \node[tile] (2r) at (3,1) {$c_r$};
      \node[tile] (3l) at (-3,2) {};
      \node[tile] (3r) at (3,2) {};
      \node[tile] (4l) at (-2,2) {};
      \node[tile] (4r) at (2,2) {};
      \node[tile] (5l) at (-1,2) {$q'$};
      \node[tile] (5r) at (1,2) {};
      \node[tile] (6l) at (-1,3) {$q'$};
      \node[tile] (6r) at (1,3) {};
      \node[tile] (7l) at (-2,3) {};
      \node[tile] (7r) at (2,3) {};
      \node[tile] (8l) at (-3,3) {};
      \node[tile] (8r) at (3,3) {};
      \node[tile] (9l) at (-3,4) {};
      \node[tile] (9r) at (3,4) {};
      \node[tile] (10l) at (-2,4) {};
      \node[tile] (10r) at (2,4) {};
      \node[tile] (11l) at (-2,5) {};
      \node[tile] (11r) at (2,5) {};
      \node[tile] (12l) at (-3,5) {};
      \node[tile] (12r) at (3,5) {};
      \node[tile] (checkl) at (-4,5) {$z_l^?$};
      \node[tile, fill=yellow] (checkr) at (4,5) {$z_r^?$};
      \node[tile] (win0l) at (-4,6) {$z_l^\bot$};
      \node[tile] (win0r) at (4,6) {$z_r^\bot$};
      \node[tile] (win1l) at (-3,6) {$z_l^\bot$};
      \node[tile] (win1r) at (3,6) {$z_r^\bot$};
      \node[tile] (win2l) at (-2,6) {$z_l^\bot$};
      \node[tile] (win2r) at (2,6) {$z_r^\bot$};
      \node[tile, fill=gray] (lose0l) at (-4,7) {\texttt{zl}};
      \node[tile, fill=gray] (lose0r) at (4,7) {\texttt{zr}};
      \node[tile] (lose1l) at (-3,7) {$z_l^\top$};
      \node[tile] (lose1r) at (3,7) {$z_r^\top$};
      \node[tile] (lose2l) at (-2,7) {$z_l^\top$};
      \node[tile] (lose2r) at (2,7) {$z_r^\top$};
      \node[tile] (state0) at (-1,4) {$q'$};
      \node[tile] (state1) at (-1,5) {$q'$};
      \node[tile] (state2) at (-1,6) {$\dagger$};
      \node[tile] (out) at (1,6) {$z_r$};
            
      \draw[ultra thick] (start) -- (dispatch);
      \foreach \side in {l, r} {
        \draw[s2] (dispatch) -- (0\side);
        \draw[s1] (0\side) -- (1\side);
        \draw[s2] (1\side) -- (2\side) -- (3\side) -- (4\side) (5\side) -- (6\side)
            -- (7\side) -- (8\side) -- (9\side) -- (10\side) -- (11\side) -- (12\side)
            -- (check\side) -- (win0\side);
        \draw[s2]
            (win0\side) -- (win1\side)
            (lose1\side) -- (lose0\side);
        \draw[s2] (win1\side) -- (win2\side) (lose1\side) -- (lose2\side);
      }
      \draw[s2] (4r) -- (5r);
      \draw[s1] (4l) -- (5l) -- (0l);
      \draw[s2] (6l) -- (state0) -- (state1);
      \draw[s1] (state1) -- (state2) -- (win2l);
      \draw[s1] (state2) -- (lose2l);
      \draw[s2] (win2r) -- (out) -- (lose2r);
      \draw[s1] (out) -- (next) -- (state2);
      \draw[s2] (next.north) -- +(0, .5) node[fill=white, inner sep=0] {$i' (t=20)$};
      \draw[s1] (2l.west) -- +(-.5, 0) node [left] {$c_l$};
      \draw[s1] (2r.east) -- +(.5, 0) node [right] {$c_r$};
      \draw[gray, s2] (lose0l.west) -- +(-.5, 0) node[above] {$z^\top_l$}  node [left] {$\dots (t=16)$};
      \draw[gray, s2] (lose0r.east) -- +(.5, 0)  node[above] {$z^\top_r$} node [right] {$\dots (t=16)$};
      \draw[s1] (checkr.east) -- +(.5, 0) node [right] {$z_r^\top$};
\end{tikzpicture}
    \caption{Transition gadget. The tile marked $*$ corresponds to the pair $(q', c_l)$ and the tile marked $\dagger$ to $(q', z_l)$.
    Starting from a strength $2$ glue representing the inputs $i$ of a transition, the tile $o$ representing its output can attach.
    The new state $q'$ as well as the changes $c_l, c_r$ to the counters are routed upwards, left and right respectively.
    The tiles $z_l^?$ and $z_r^?$ become attachable 14 time steps after $o$ if that position is then free.
    %The $z^\bot$ series of tiles can attach from there.
    If one of the $z^?$ positions is occupied, a tile with a $z^\top$ glue will attach at time step 16, followed by a series of tiles so that by time 20, a tile $i'$ with all the information for the next transition attaches in the central column, restarting the process for the next transition.
    The $z^\bot$ series of tiles corresponds to a non-zero counter value, and the $z^\top$ to a counter at 0.
    }
    \label{fig:band-transition-tiles}
\end{figure}

\paragraph*{Counter gadgets}
There are 5 counter update gadgets, shown in \cref{fig:band-decl,fig:band-keepnz,fig:band-keepz,fig:band-incrnz,fig:band-incrz}: three increments ($+1, 0, -1$) each in a zero and a non-zero version, except for $-1$, which has no zero version. Each of these gadgets has a right and a symmetrical left orientation. The description of the right counter gadgets are as follows:

All gadgets start at $(4, -1)$ at time 0. The non-zero versions rendezvous with the instruction from the transition gadget at $(3, 1)$ at time 7. The decrease gadget then reaches $(4, 4)$ at time 10. If the tile $z_r^?$ is already present there, a non-zero gadget can start at position $(4, 5)$; otherwise a zero gadget can start at this position at time 14. For a keep gadget, the delay before the next iteration is 20, and that next iteration is always a non-zero. For an increase gadget, the next iteration is always non-zero, and starts at time $28$.

The zero gadgets are ready to gather their instruction (keep or increase) at time 2. The keep gadget starts a new, non-zero iteration 10 time steps after it has cooperated with the instruction from the transition gadget. The increase gadget does the same, but with a delay of 18 after it has received its instruction.

\paragraph*{Seed configuration} The seed configuration consists of:
\begin{itemize}
    \item the starting tile of the transition gadgets corresponding to the starting state of $M$ and both counters at 0, placed at $(0,0)$
    \item the starting tile of the left and right counters in their zero variant, instruction yet undetermined, at $(0, -5)$ and $(0, 5)$ respectively.
\end{itemize}

\begin{proof}[Proof of \cref{thm:Turing_band}]

Starting from the seed configuration, a standard induction establishes the following. Columns $-9$ to $-4$ contain left counter gadgets, columns $-4$ to 4 contain the transition gadgets, and columns 4 to 9 contain right counter gadgets. For each horizontal line at height $6t - 1$, the state represented on the transition gadget is the state of the two-counter machine at time $t$, and the delay between the first attachment of the transition gadget and that of the left (resp. right) counter gadget is $8l - 4$ if the value of the first (resp. second) counter is $l$ at time $t$.

For any halting state of the two counter machine, a few special tiles are added which stop the growth of both counters and reach columns -10 and 10.
\end{proof}

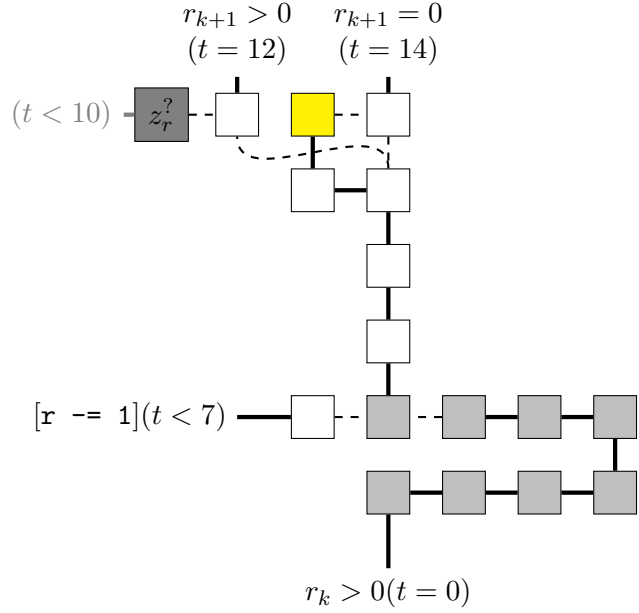
\begin{figure}
    \centering
    \begin{tikzpicture}[%
    tile/.style={draw, regular polygon, regular polygon sides=4, minimum size=.8cm, inner sep=1pt},%
    s1/.style={thick,dashed},%
    s2/.style={ultra thick}%
    ]
    \node[left] (instr0) at (-2, 1) {$[\texttt{r -= 1}] (t < 7)$};
    \begin{scope}[every node/.style={fill=gray!50!white}]
    \node[tile] (dec0) at (0, 1) {};
    \node[tile] (dec1) at (1, 1) {};
    \node[tile] (start) at (0, 0) {};
    \node[tile] (t4) at (3, 1) {};
    \node[tile] (t5) at (2, 1) {};
   \foreach \n in {1,...,3} {
        \node[tile] (t\n) at (\n, 0) {};
    }        
    \end{scope}
    \node[tile] (instr) at (-1, 1) {};
    \node[below] (go) at (0, -1) {$r_k > 0 (t = 0)$};
     \foreach \n/\y in {6/2,7/3,8/4} {
        \node[tile] (t\n) at (0, \y) {};
    }
    \node[tile] (challenge) at (-1, 4) {};
    \node[tile, fill=yellow] (winning) at (-1, 5) {};
    \node[tile] (resultnz) at (0, 5) {};
    \node[tile, fill=gray] (badnews) at (-3, 5) {$z_r^?$};
    \node[tile] (resultz) at (-2, 5) {};

    \draw[s2] (instr0) -- (instr);
    \draw[s1] (instr) -- (dec0);
    \draw[s2] (go) -- (start) -- (t1) -- (t2) -- (t3) -- (t4) -- (t5) -- (dec1);
    \draw[s1] (dec1) -- (dec0) (winning) -- (resultnz) -- (t8); 
    \draw[s2] (dec0) -- (t6) -- (t7) -- (t8) -- (challenge) -- (winning);
    \draw[s1] (badnews) -- (resultz);
    \draw[s1] (t8.north) to[in=270, out=90] (resultz.south);
    \draw[s2, gray] (badnews) -- +(-.5, 0) node[left] {$(t < 10)$};
    \draw[s2] (resultnz) -- +(0, .5) node[above, align=center] {$r_{k+1} = 0$\\$(t = 14)$};
    \draw[s2] (resultz) -- +(0, .5) node[above, align=center] {$r_{k+1} > 0$\\$(t = 12)$};
\end{tikzpicture}
    \caption{Right \emph{Decrease} gadget. The left one is symmetrical. The light gray tile types at the bottom are shared with the gadgets for non-zero input of all instructions.}
    \label{fig:band-decl}
\end{figure}

\begin{figure}
    \centering
    \begin{tikzpicture}[%
    tile/.style={draw, regular polygon, regular polygon sides=4, minimum size=.8cm, inner sep=1pt},%
    s1/.style={thick,dashed},%
    s2/.style={ultra thick}%
    ]
    \node[left] (instr0) at (-2, 1) {$\texttt{[keep]} (t < 7)$};
    \node[tile] (instr) at (-1, 1) {};
    \node[below] (go) at (0, -1) {$r_k > 0 (t = 0)$};
    \begin{scope}[every node/.style={fill=gray!50!white}]
        \node[tile] (dec0) at (0, 1) {};
        \node[tile] (dec1) at (1, 1) {};
        \node[tile] (start) at (0, 0) {};
        \node[tile] (t4) at (3, 1) {};
        \node[tile] (t5) at (2, 1) {};
        \foreach \n in {1,...,3} {
            \node[tile] (t\n) at (\n, 0) {};
        }        
    \end{scope}
    \foreach \n/\x/\y in {6/0/2,7/1/2,8/2/2,9/3/2,10/4/2,%
                        11/4/3,12/3/3,13/2/3,14/1/3,15/0/3,16/0/4} {
        \node[tile] (t\n) at (\x, \y) {};
    }
    %\node[tile] (t4) at (3, 1) {};
    %\node[tile] (t5) at (2, 1) {};
    %\foreach \n/\y in {6/2,7/3,8/4} {
    %    \node[tile, red] (t\n) at (0, \y) {};
    %}
    \begin{scope}[gray]
    \node[tile] (challenge) at (-1, 4) {};
    \node[tile, fill=yellow] (winning) at (-1, 5) {};
    \node[tile] (resultnz) at (0, 5) {};        
    \end{scope}

    \node[tile, fill=gray] (badnews) at (-3, 5) {$z_r^?$};
    \node[tile] (resultz) at (-2, 5) {};

    \draw[s2] (instr0) -- (instr);
    \draw[s1] (instr) -- (dec0);
    \draw[s2] (go) -- (start) -- (t1) -- (t2) -- (t3) -- (t4) -- (t5) -- (dec1);
    \draw[s1] (dec1) -- (dec0);
    \draw[s1, gray](winning) -- (resultnz) -- (t16); 
    \draw[s2] (dec0) -- (t6) -- (t7) -- (t8) -- (t9) -- (t10) -- (t11) -- (t12) --
    (t13) -- (t14) -- (t15) -- (t16);
    \draw[s2, gray] (t16) --  (challenge) -- (winning);
    \draw[s1] (badnews) -- (resultz);
    \draw[s1] (t16.north) to[in=270, out=90] (resultz.south);
    \draw[s2] (badnews) -- +(-.5, 0) node[left] {$(t < 10)$};
    \draw[s2, gray] (resultnz) -- +(0, .5) node[above, align=center] {$r_{k+1} = 0$\\$(\textrm{never})$};
    \draw[s2] (resultz) -- +(0, .5) node[above, align=center] {$r_{k+1} > 0$\\$(t = 20)$};
\end{tikzpicture}
    \caption{Right \emph{Keep} gadget for a non-zero value. The left one is symmetrical. The two tiles at the top-right would attach to signal a zero value, which cannot happen starting from a non-zero value.}
    \label{fig:band-keepnz}
\end{figure}
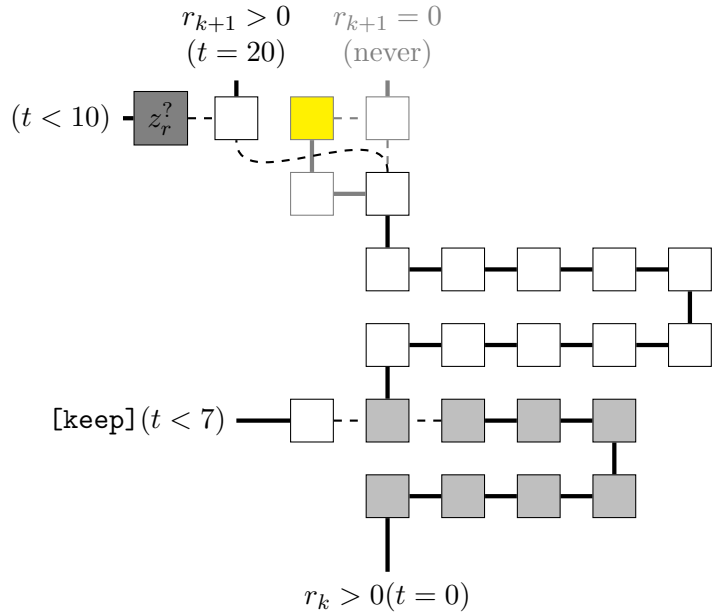

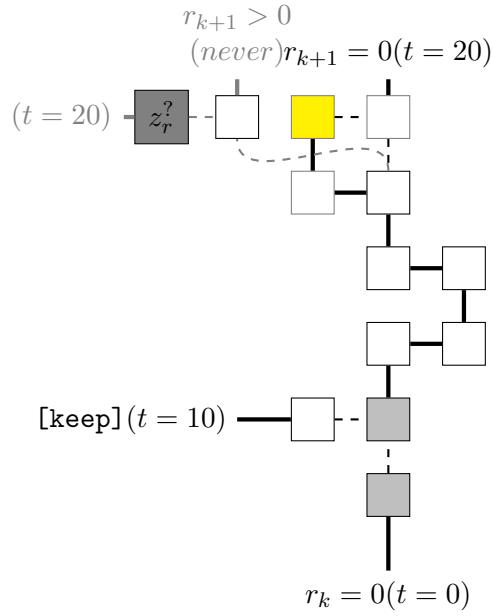
\begin{figure}
    \centering
    \begin{tikzpicture}[%
    tile/.style={draw, regular polygon, regular polygon sides=4, minimum size=.8cm, inner sep=1pt},%
    s1/.style={thick,dashed},%
    s2/.style={ultra thick}%
    ]
    \node[left] (instr0) at (-2, 1) {$\texttt{[keep]} (t =10)$};
    \node[tile] (instr) at (-1, 1) {};
    \node[below] (go) at (0, -1) {$r_k = 0 (t = 0)$};
    \begin{scope}[every node/.style={fill=gray!50!white}]
        \node[tile] (dec0) at (0, 1) {};
        \node[tile] (start) at (0, 0) {};        
    \end{scope}
    \foreach \n/\x/\y in {6/0/2,7/1/2,%
                          14/1/3,15/0/3,16/0/4} {
        \node[tile] (t\n) at (\x, \y) {};
    }
    %\node[tile] (t4) at (3, 1) {};
    %\node[tile] (t5) at (2, 1) {};
    %\foreach \n/\y in {6/2,7/3,8/4} {
    %    \node[tile, red] (t\n) at (0, \y) {};
    %}
    \begin{scope}[gray]
    \node[tile] (challenge) at (-1, 4) {};
    \node[tile, fill=yellow] (winning) at (-1, 5) {};
    \node[tile] (resultnz) at (0, 5) {};        
    \end{scope}

    \node[tile, fill=gray] (badnews) at (-3, 5) {$z_r^?$};
    \node[tile] (resultz) at (-2, 5) {};

    \draw[s2] (instr0) -- (instr);
    \draw[s1] (instr) -- (dec0);
    \draw[s1] (start) -- (dec0);
    \draw[s1](winning) -- (resultnz) -- (t16); 
    \draw[s2] (go) -- (start) (dec0) -- (t6) -- (t7) -- (t14) -- (t15) -- (t16);
    \draw[s2] (t16) --  (challenge) -- (winning);
    \draw[s1, gray] (badnews) -- (resultz);
    \draw[s1, gray] (t16.north) to[in=270, out=90] (resultz.south);
    \draw[s2, gray] (badnews) -- +(-.5, 0) node[left] {$(t = 20)$};
    \draw[s2] (resultnz) -- +(0, .5) node[above, align=center] {$r_{k+1} = 0 (t = 20)$};
    \draw[s2, gray] (resultz) -- +(0, .5) node[above, align=center] {$r_{k+1} > 0$\\$(never)$};
\end{tikzpicture}
    \caption{Right \emph{Keep} gadget when the counter value is zero. The left one is symmetrical.}
    \label{fig:band-keepz}
\end{figure}

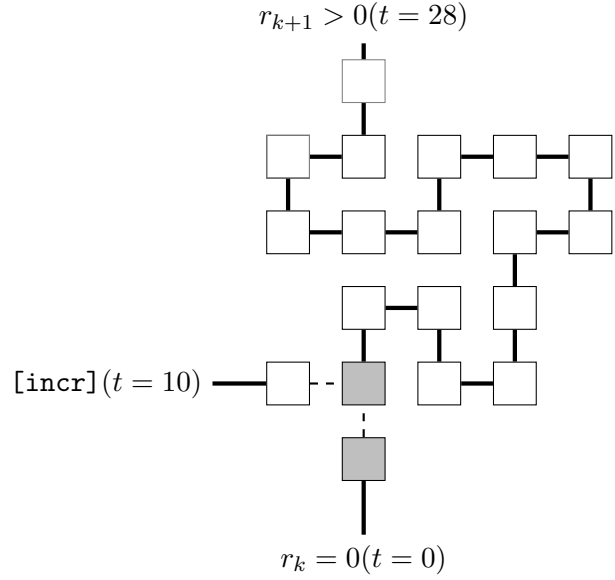
\begin{figure}
    \centering
    \begin{tikzpicture}[%
    tile/.style={draw, regular polygon, regular polygon sides=4, minimum size=.8cm, inner sep=1pt},%
    s1/.style={thick,dashed},%
    s2/.style={ultra thick}%
    ]
    \node[left] (instr0) at (-2, 1) {$\texttt{[incr]} (t =10)$};
    \node[tile] (instr) at (-1, 1) {};
    \node[below] (go) at (0, -1) {$r_k = 0 (t = 0)$};
    \begin{scope}[every node/.style={fill=gray!50!white}]
        \node[tile] (dec0) at (0, 1) {};
        \node[tile] (start) at (0, 0) {};        
    \end{scope}
    \foreach \n/\x/\y in {6/0/2,7/1/2,8/1/1,9/2/1,10/2/2,11/2/3,12/3/3,
                            13/3/4, 14/2/4,15/1/4,16/1/3, 17/0/3, 18/-1/3,%
                            19/-1/4, 20/0/4
                            } {
        \node[tile] (t\n) at (\x, \y) {};
    }
    %\node[tile] (t4) at (3, 1) {};
    %\node[tile] (t5) at (2, 1) {};
    %\foreach \n/\y in {6/2,7/3,8/4} {
    %    \node[tile, red] (t\n) at (0, \y) {};
    %}
    \begin{scope}[gray]
    \node[tile] (challenge) at (-1, 4) {};
    \node[tile] (resultnz) at (0, 5) {};        
    \end{scope}

    \draw[s2] (go) -- (start) (instr0) -- (instr);
    \draw[s1] (instr) -- (dec0);
    \draw[s1] (start) -- (dec0);
    \draw[s2] (dec0) -- (t6) -- (t7) -- (t8) -- (t9) -- (t10) -- (t11)  -- (t12) -- (t13) -- (t14) -- (t15) -- (t16) -- (t17) -- (t18) -- (t19) -- (t20);
    \draw[s2] (t20) -- (resultnz);
    \draw[s2] (resultnz) -- +(0, .5) node[above, align=center] {$r_{k+1} > 0 (t = 28)$};
\end{tikzpicture}
    \caption{Right \emph{Increase} gadget when the counter value is zero. The left one is symmetrical.}
    \label{fig:band-incrz}
\end{figure}

\begin{figure}
    \centering
    \begin{tikzpicture}[%
    tile/.style={draw, regular polygon, regular polygon sides=4, minimum size=.8cm, inner sep=1pt},%
    s1/.style={thick,dashed},%
    s2/.style={ultra thick}%
    ]
    \node[left] (instr0) at (-2, 1) {$\texttt{[incr]} (t < 7)$};
    \node[tile] (instr) at (-1, 1) {};
    \node[below] (go) at (0, -1) {$r_k > 0 (t = 0)$};
    \begin{scope}[every node/.style={fill=gray!50!white}]
        \node[tile] (dec0) at (0, 1) {};
        \node[tile] (dec1) at (1, 1) {};
        \node[tile] (start) at (0, 0) {};
        \node[tile] (t4) at (3, 1) {};
        \node[tile] (t5) at (2, 1) {};
        \foreach \n in {1,...,3} {
            \node[tile] (t\n) at (\n, 0) {};
        }        
    \end{scope}
    \foreach \n/\x/\y in {6/0/2,7/1/2,8/2/2,9/3/2,10/4/2,%
                        11/4/3, 12/4/4, 13/4/5, 14/3/5, 15/3/4,
                        16/3/3, 17/2/3, 18/2/4, 19/2/5, 20/1/5,
                        21/1/4, 22/1/3, 23/0/3, 24/0/4}{
                        %12/3/3,13/2/3,14/1/3,15/0/3,16/0/4} {
        \node[tile] (t\n) at (\x, \y) {};
    }

    \node[tile, fill=gray] (badnews) at (-1, 5) {$z_r^?$};
    \node[tile] (resultz) at (0, 5) {};

    \draw[s2] (instr0) -- (instr);
    \draw[s1] (instr) -- (dec0);
    \draw[s2] (go) -- (start) -- (t1) -- (t2) -- (t3) -- (t4) -- (t5) -- (dec1) (dec0) -- (t6);
    \draw[s1] (dec1) -- (dec0);
    \foreach \n [remember=\n as \lastn (initially 6)] in {7,...,24} {
        \draw[s2] (t\lastn) -- (t\n);
    }
    \draw[s1] (badnews) -- (resultz);
    \draw[s1] (t24.north) to[in=270, out=90] (resultz.south);
    \draw[s2] (badnews) -- +(-.5, 0) node[left] {$(t < 10)$};
    \draw[s2] (resultz) -- +(0, .5) node[above, align=center] {$r_{k+1} > 0$\\$(t = 28)$};
\end{tikzpicture}
    \caption{Right \emph{Increase} gadget when the counter value is non-zero. The left one is symmetrical.}
    \label{fig:band-incrnz}
\end{figure}

\begin{figure}
    \centering
    \input{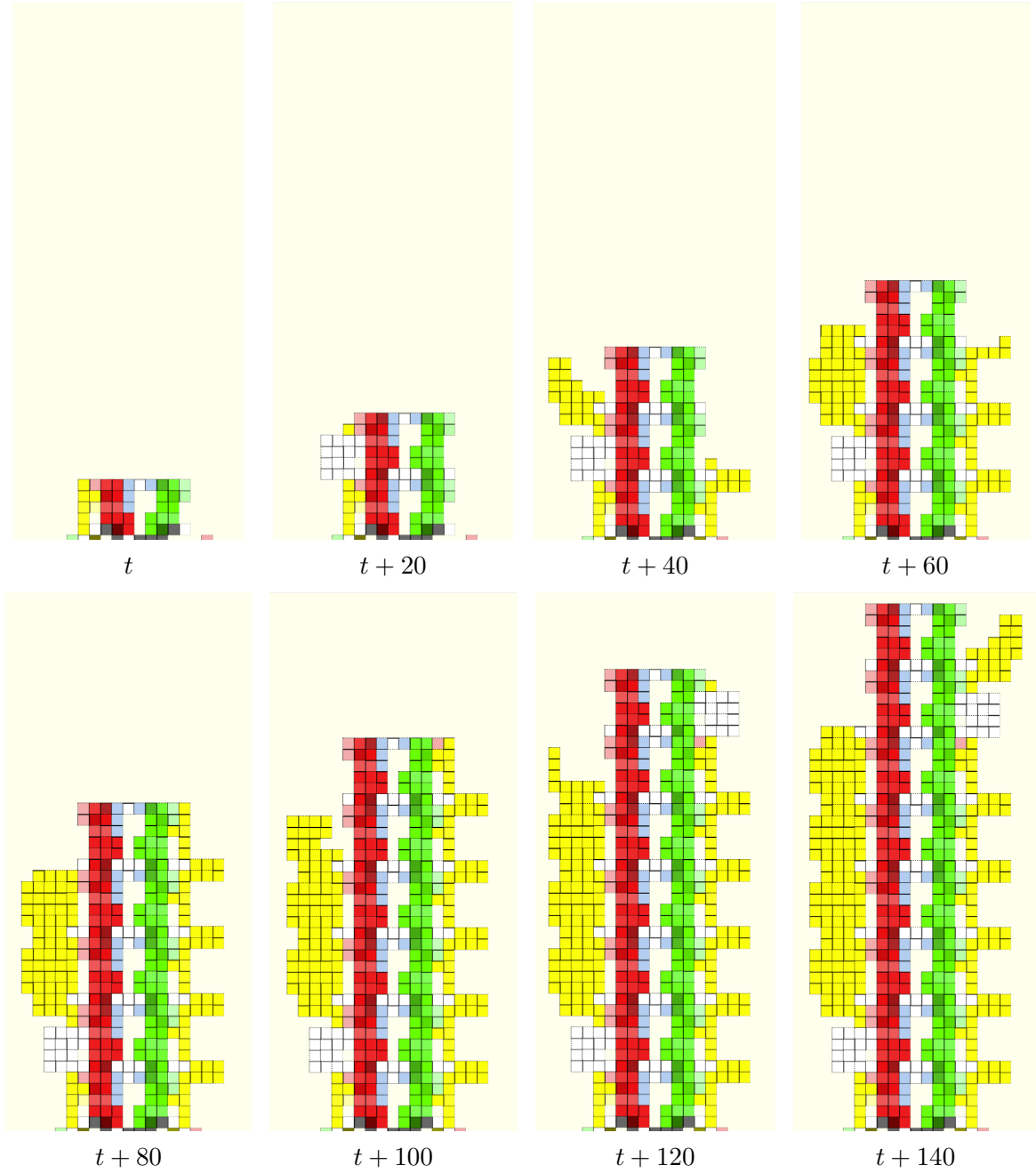}
    \caption{Simulation of the two-counter machine of \cref{fig:two-counter-collatz}: a time-lapse of the determination that 5 is odd, done while copying the value 5 from the right counter to the left one. Images are taken every 20 ticks; the left counter starts at value 0 and wins the race against the decision gadget (pink tile in the -4 column). Meanwhile the right counter, which holds the value 5 only enters the picture $8 * 5 = 40$ ticks later. The right counter gadget climbs fast as the value of the counter decreases, reaching zero five computation steps later, i.e. 30 rows higher. This happens at time step $t + 100$, when the right counter gadget wins the race against the transition gadget. The left counter gadget increases its value, thus rising at a leisurely trot. Having by then accumulated $5$ counter value worth of delay, in only reaches the 30th row at time step $t + 140$.}
    \label{fig:band-example-sim}
\end{figure}

\syncTAMToneBounded*

\paragraph*{Non-cooperative construction of a two-counter machine} 
\cref{cor:syncTAMone-k-bounded} follows directly from \cref{thm:Turing_band,thm:T1-IU} by simulating the cooperative construction described above in temperature 1. However, a non-cooperative construction can be made directly to show that given a temperature 1 syncTAS, it is also undecidable whether it is 9-band bounded. The non-cooperative construction works similarly to the cooperative construction with only minor differences in the timings. Let $z_l(i)$ and $z_r(i)$ be the positions $(-4, 12i+3)$ and $(4, 12i+3)$, respectively, and $(q_i, l_i, r_i)$ as the $i$th configuration of the run of $M$. The transition gadget reaches $z_l(i)$ and $z_r(i)$ at time $34i + 10$ while the left and right counter gadget reaches $z_l(i)$ and $z_r(i)$ at time $34i + 10l_i + 8$ and $34i + 10r_i + 8$, respectively. The left and right counters have a value of zero if they arrive at $z_l(i)$ and $z_r(i)$ before the transition gadget. The counters can be increased or decreased by having the counter gadgets rise slower (44 ticks per gadget) or faster (24 ticks per gadget) than the rate of the transition gadget (34 ticks per gadget). \cref{fig:band-example-sim-temp1} shows a time-lapse simulation of the non-cooperative construction of the two counter machine that calculates the Collatz function from \cref{fig:two-counter-collatz}.

The non-cooperative transition gadget, shown in \cref{fig:band-transition_nz_nz-temp1,fig:band-transition_z_z-temp1}, starts at $(-3, 0)$ and $(3, 0)$ at time 0 and carries the current state $q$ throughout the left side of the gadget. At time 10, the left register is determined to be non-zero, $z^\bot$, or zero $z^\top$ by the absence or presence of a \texttt{zl} tile at position $(-4, 3)$. The absence of a \texttt{zl} tile allows for the light gray tiles representing $z_l^\bot$ to grow, preventing further growth of the dark gray tiles, representing $z_l^\top$ as shown in \cref{fig:band-transition_nz_nz-temp1}. The presence of a \texttt{zl} tile prevents the growth of the light gray tiles entirely, allowing the dark gray tiles to grow as shown in \cref{fig:band-transition_z_z-temp1}. Likewise for the right register, with respect to a \texttt{zr} tile at position $(4, 3)$. The two sides of the gadget communicate using a similar mechanism of preventing and allowing growth of certain tiles. It then grows the magenta tiles which represent the instruction to be done to the counter using a ternary bit. The next transition gadget starts with the new state $q'$ at time 34.

Similar to the cooperative version, there are 5 non-cooperative counter update gadgets, shown in \cref{fig:band-decr-temp1,fig:band-keep-temp1,fig:band-incr-temp1}: three increments (+1, 0, -1) each in a zero and a non-zero version, except for -1 which has no zero version. Each gadget has a left and symmetrical right orientation. Let time 0 be the time that the previous right counter gadget start a new non-zero counter gadget, and let position $(7, 3)$ be the start of the new non-zero counter gadget. If the current value of the counter is zero, the gadget will grow a \texttt{zr} tile at position $(4, 3)$ which grows a tile at $(4, 4)$ to prevent the growth of the non-zero counter gadget. A non-zero counter gadget then grows from the transition gadget at $(4, 5)$ at time 18. If the current value of the counter is non-zero, then a $z_r^\bot$ tile from the transition gadget blocks a \texttt{zl} tile from appearing, which allows the non-zero counter gadget to keep growing. The counter gadgets then travel to and perform ternary bit-reading on the magenta tiles to determine which increment should be done, i.e., which counter gadget to use.

\begin{figure}
    \centering
    \includegraphics[width=1.0\textwidth]{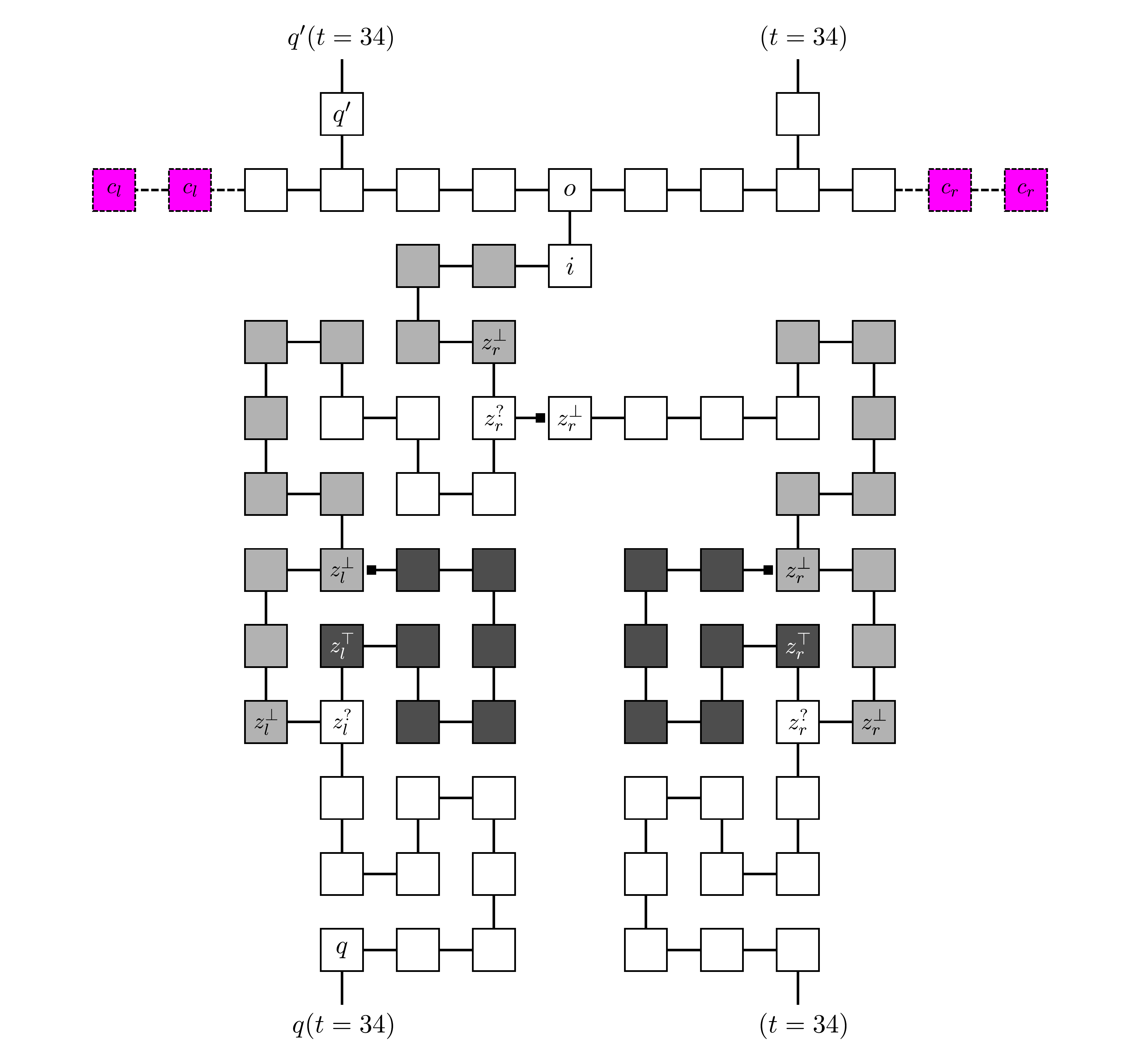}
    \caption{Non-cooperative transition gadget where both counter values are non-zero. It carries the current state of the $M$ on the left side of the gadget, signified by the tile marked with $q$. The tiles marked with $z^?$ creates a series of $z^\top$ and $z^\bot$ tiles which represent an input of zero and non-zero for each counter for the current state, respectively. The series of tiles indicated by $z^\top$ are unable to complete their growth, which signifies that the input for the current state is non-zero. The tiles marked with $c$ represent a ternary bit for the counters to read which indicates whether to decrease, increase, or stay the same. The tile marked with $q'$ represents the state of the next configuration of $M$.}
    \label{fig:band-transition_nz_nz-temp1}
\end{figure}

\begin{figure}
    \centering
    \includegraphics[width=1.0\textwidth]{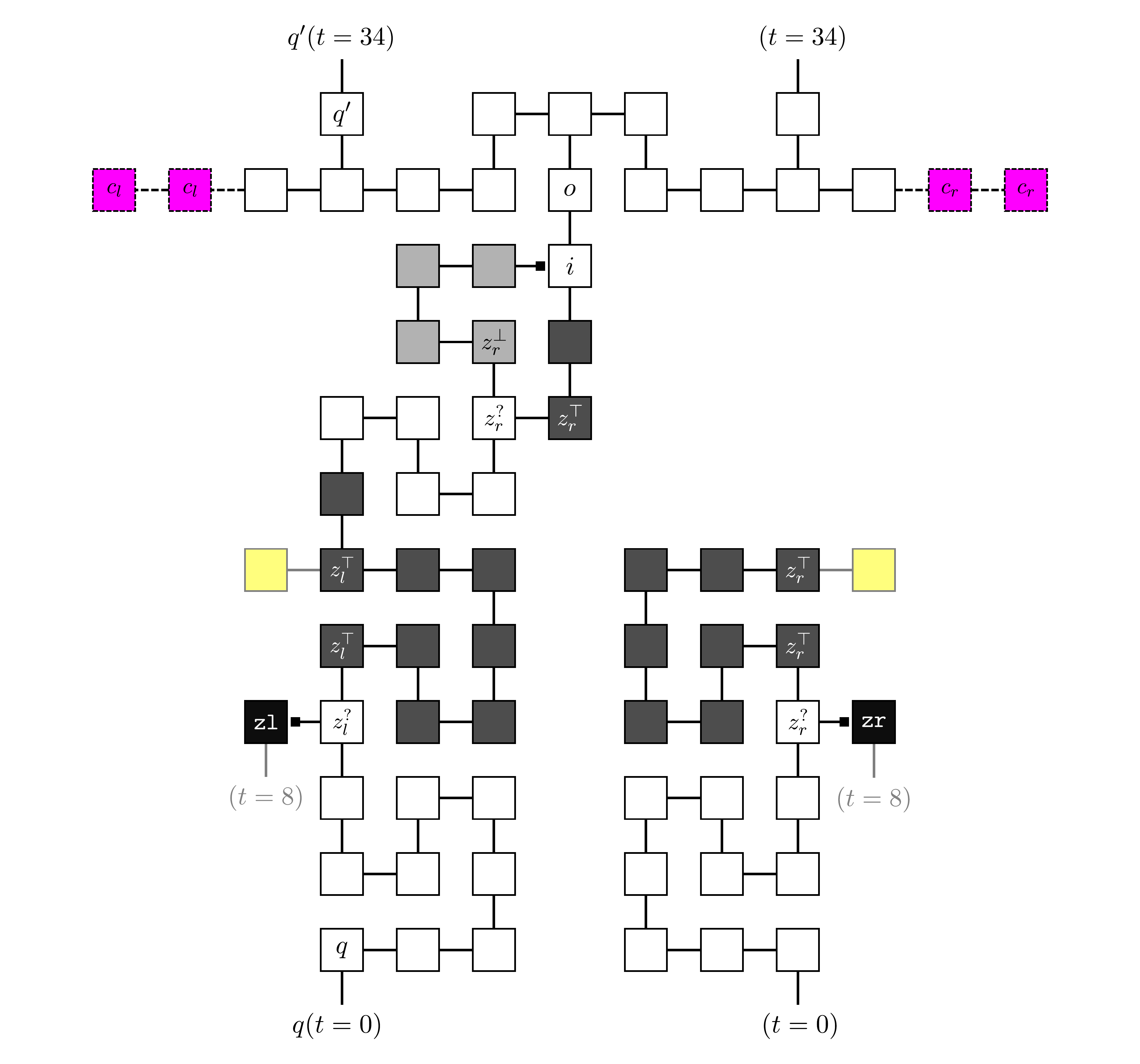}
    \caption{Non-cooperative transition gadget when both counter values are zero, indicated by the presence of \texttt{zl} and \texttt{zr} tiles, which are placed by the counter gadgets. Since the series of $z^\bot$ tiles are absent, the series of $z^\top$ tiles are able to grow to signify the input for this current state is zero for each counter. The series of $z^\top$ tiles also grow new counter gadgets, represented by the yellow tiles, since the \texttt{zl} and \texttt{zr} tiles stop the growth of the counter gadgets. This is done to ensure that the counter gadgets do not grow past the transition gadget.}
    \label{fig:band-transition_z_z-temp1}
\end{figure}

\begin{figure}
    \centering
    \includegraphics[width=.57\textwidth]{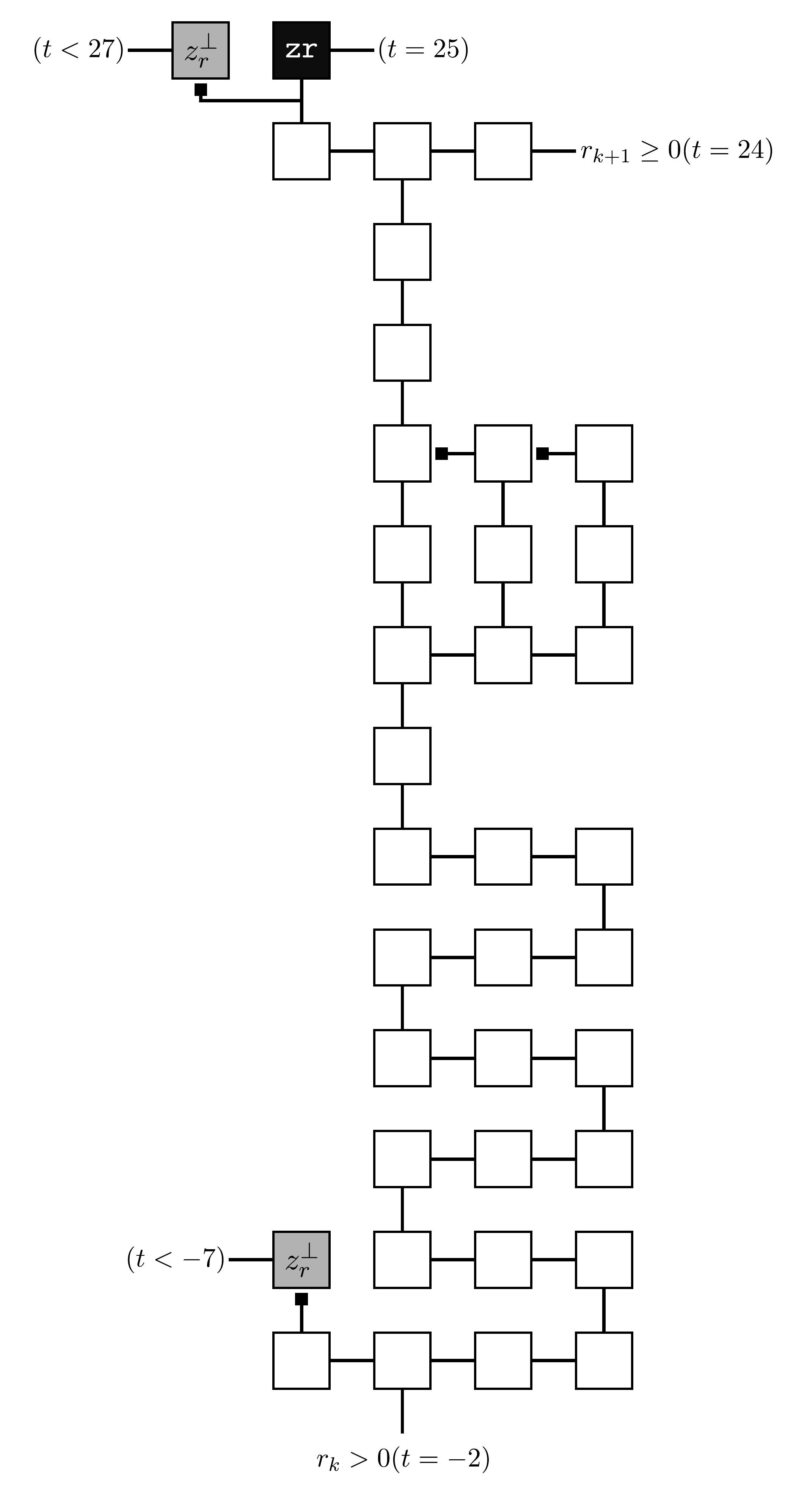}
    \caption{Non-cooperative Right \emph{Decrease} gadget. The left one is symmetrical. The counter value can only be non-zero. The ternary bit read is 0.}
    \label{fig:band-decr-temp1}
\end{figure}

\begin{figure}
    \centering
    \begin{subfigure}[t]{0.5\textwidth}
        \centering
        \includegraphics[width=\linewidth]{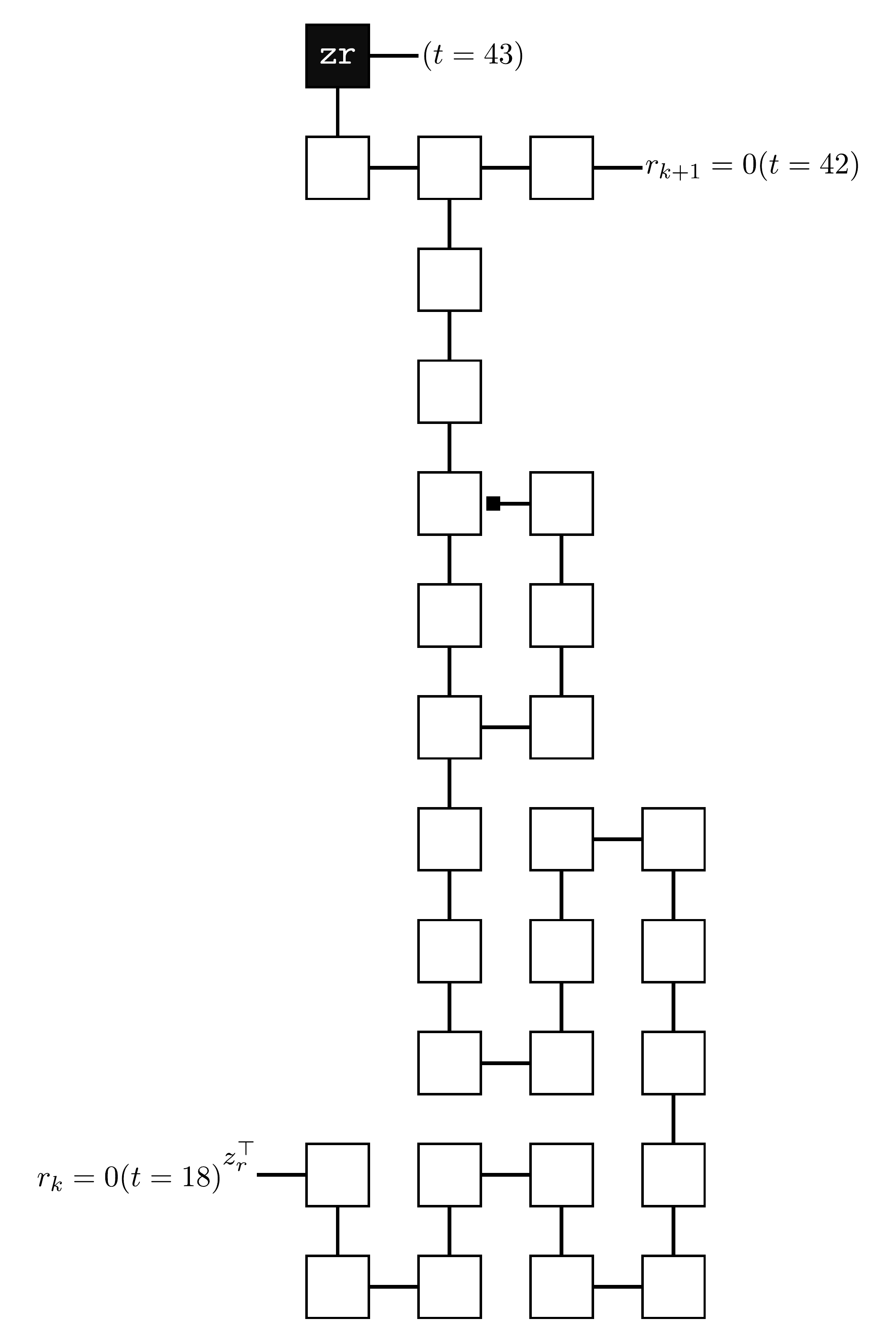}
        \caption{}
    \end{subfigure}%
    ~ 
    \begin{subfigure}[t]{0.5\textwidth}
        \centering
        \includegraphics[width=\linewidth]{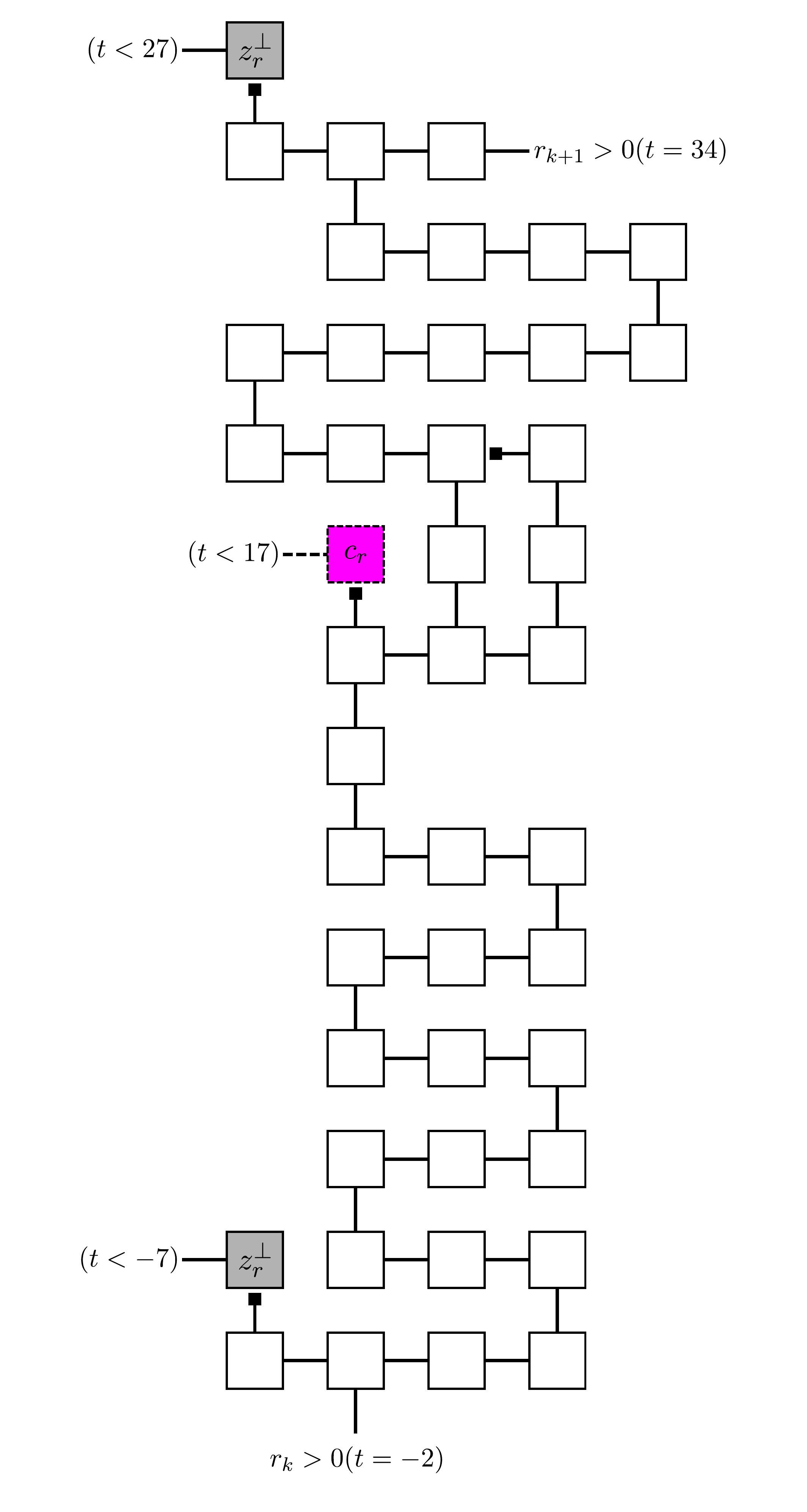}
        \caption{}
    \end{subfigure}
    \caption{Non-cooperative Right \emph{Keep} gadgets. The left ones are symmetrical. (a) The counter value is zero. The ternary bit read is 0. (b) The counter value is non-zero. The ternary bit read is 1.}
    \label{fig:band-keep-temp1}
\end{figure}

\begin{figure}
    \centering
    \begin{subfigure}[t]{0.5\textwidth}
        \centering
        \includegraphics[width=\linewidth]{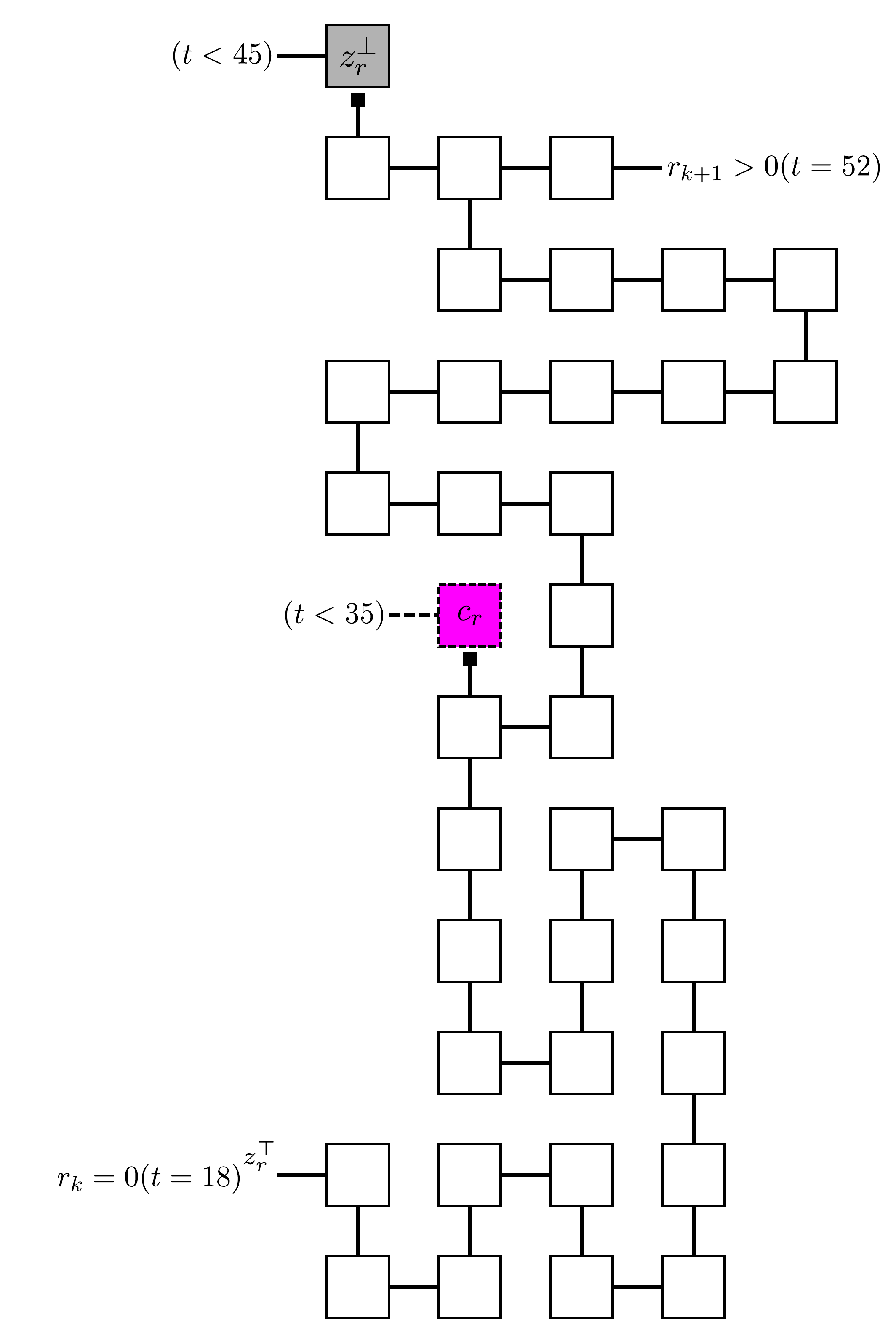}
        \caption{}
    \end{subfigure}%
    ~ 
    \begin{subfigure}[t]{0.5\textwidth}
        \centering
        \includegraphics[width=\linewidth]{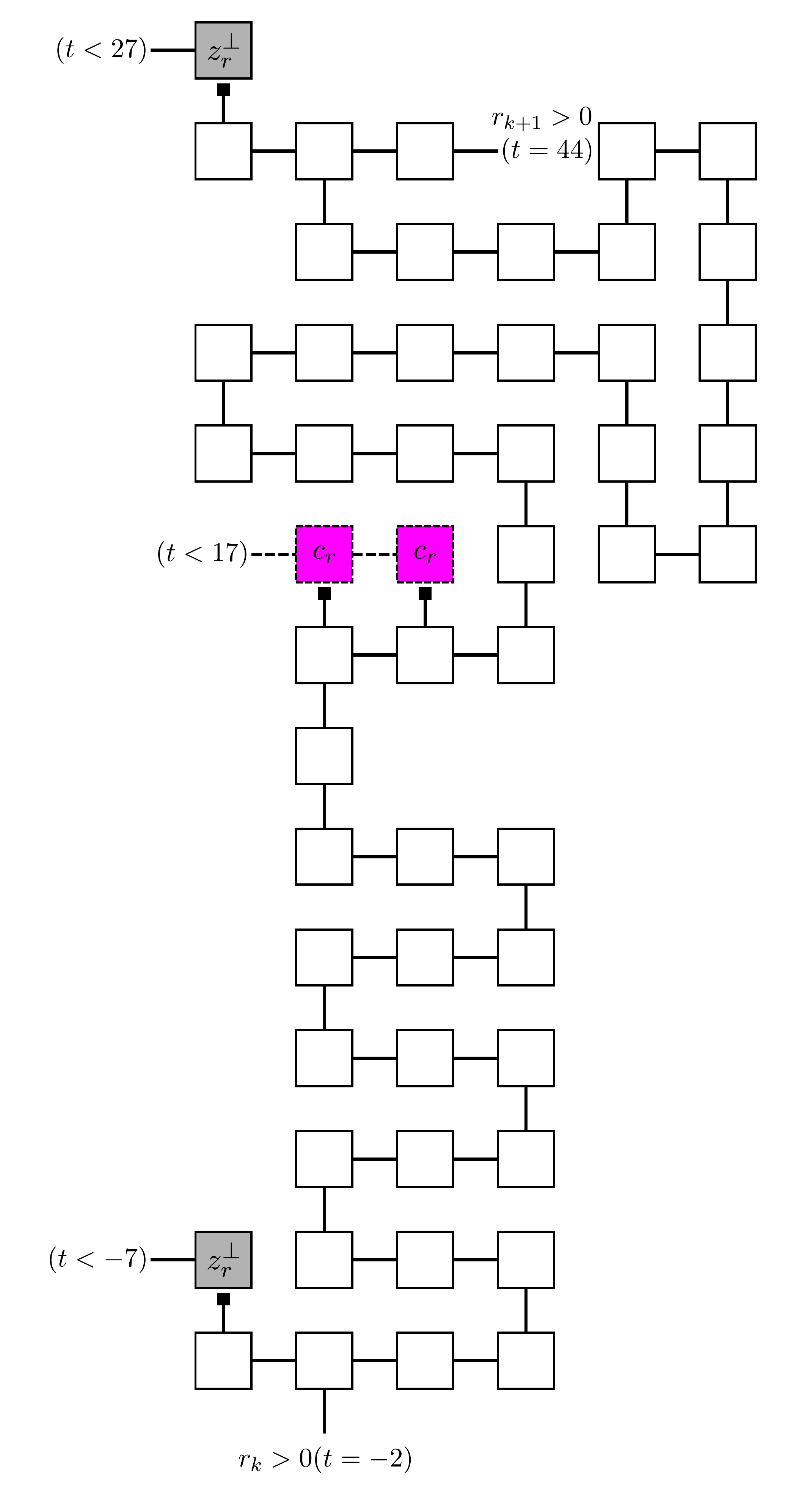}
        \caption{}
    \end{subfigure}
    \caption{Non-cooperative Right \emph{Increase} gadgets. The left ones are symmetrical. (a) The counter value is zero. The ternary bit read is 1. (b) The counter value is non-zero. The ternary bit read is 2.}
    \label{fig:band-incr-temp1}
\end{figure}

\begin{figure}
    \centering
    \input{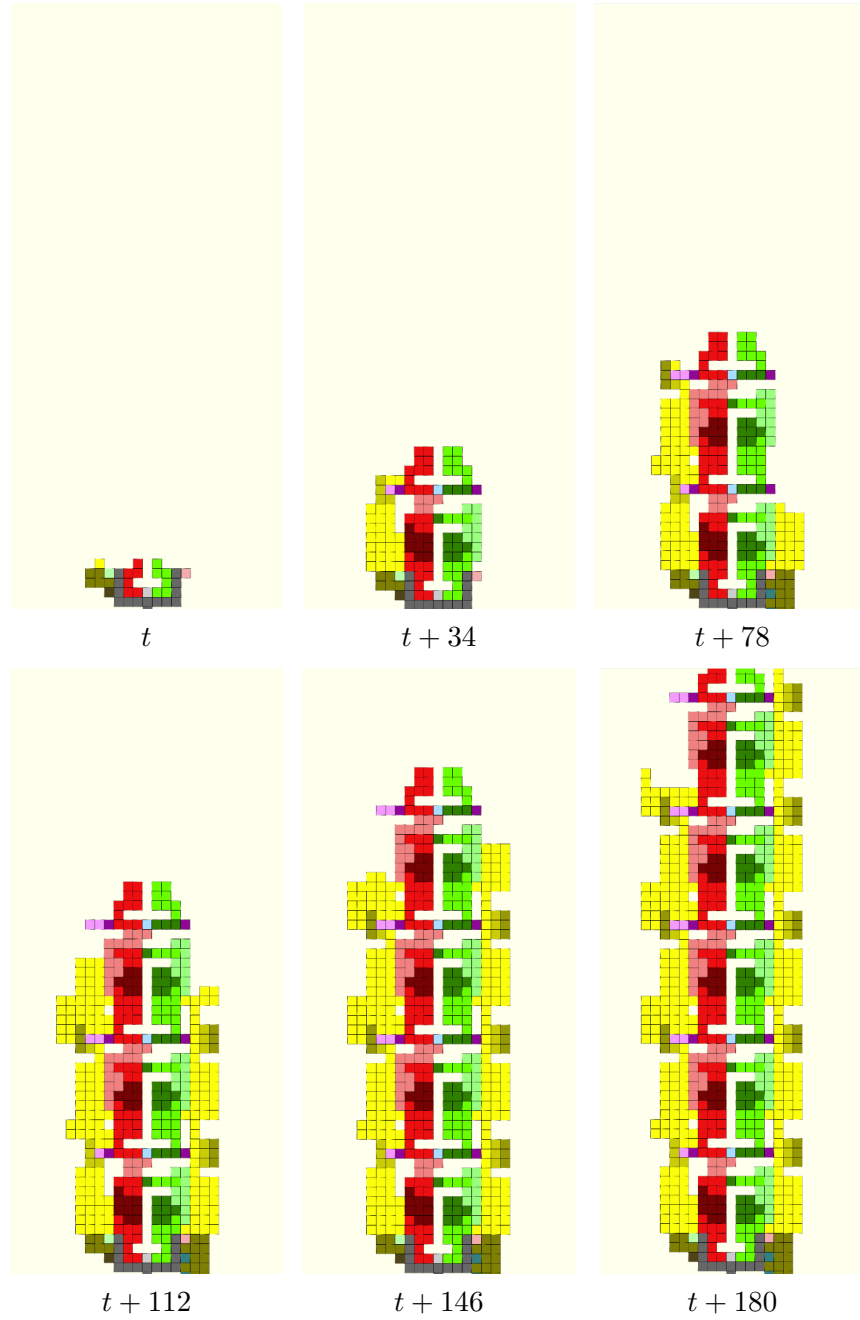}
    \caption{Non-cooperative version of the simulation shown in \cref{fig:band-example-sim}. As opposed to the cooperative version of the two-counter machine, the non-cooperative version has images taken every 34 ticks, due to needing more time and vertical space to compute the same amount of information.}
    \label{fig:band-example-sim-temp1}
\end{figure}

\paragraph*{Note on the initialization of counters at the start of the simulation} In order to prove \cref{thm:Turing_band}, it is enough to simulate
any two-counter machine $M$ on a starting configuration where both counters are at 0. It is also possible to initialize the counter at a different value $v$ by replacing the tile of the seed which starts the counter with a path of length $8v$ for the cooperative construction and $10v$ for the non-cooperative construction.

Note that we have implemented a JavaScript simulator for two-counter machines, along with functionality that can convert an arbitrary two-counter machine into a syncTAM system that uses bounded-width and can be run in an embedded syncTAM simulator, and have made them (along with examples) available online here: \url{http://self-assembly.net/software/two_counter_machines/two_counter_machines-latest/} \cite{TwoCounterSimulator}. It uses the seed trick discussed above. When the simulated machine halts, it goes into a hard-coded loop decrementing the right counter and attaching special blue tiles at the position of the left counter gadget. This way, the final value of the right counter can be read as the number of ticks at the top left of the terminal assembly.

\section{Simulating arbitrary syncTAM systems using systems with simplified dynamics}\label{sec:adj-proof}

\adjonly*

\begin{figure}
    \centering
    \input{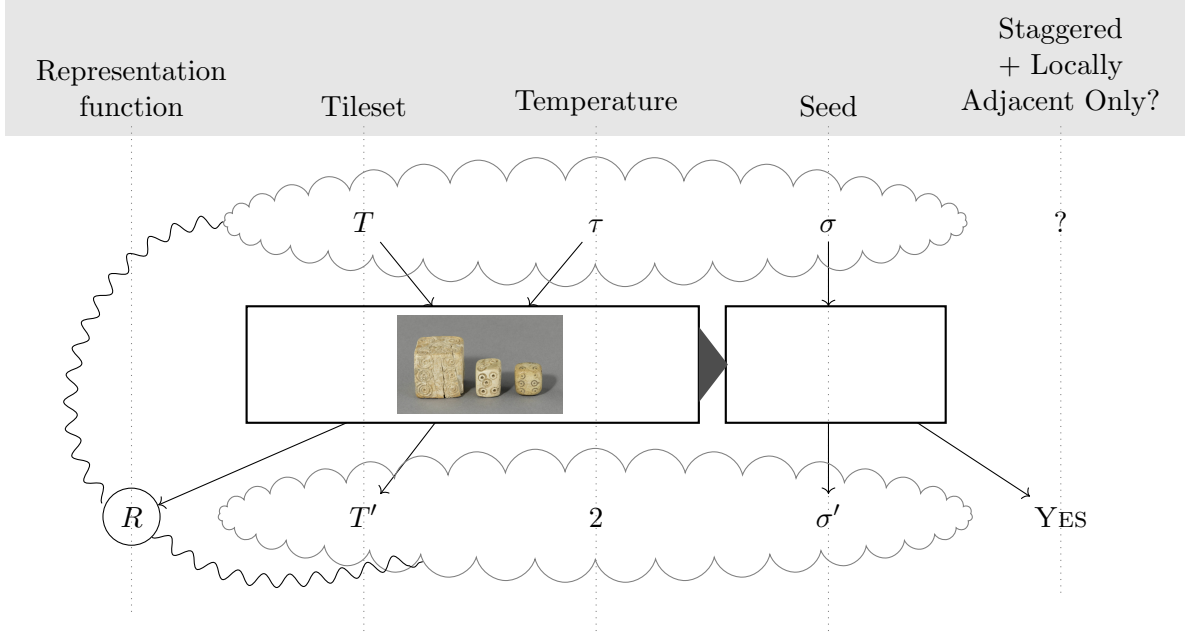}
    \caption{A visual overview of \cref{lem:adj-only}. The squiggly line between the two TAS shows that $(T', 2, \sigma')$ simulates $(T, \tau, \sigma)$ under the representation function $R$. The box representing the operation of the lemma is cleft to the left of the seed output to show that the simulation scheme is seed-independent: the outputs $R, \calT'$ do not depend on the input $\sigma$. \\
    \emph{The dice illustration is from the collections of Musée du Louvre: \url{https://collections.louvre.fr/ark:/53355/cl010167272}}
    }
    \label{fig:dice}
\end{figure}

\begin{proof}
We prove \cref{lem:adj-only} by construction. Given arbitrary syncTAM system $S = (T, \sigma, \tau)$, we will define $\calS' = (T', \sigma', 2)$, $c,s \in \mathbb{N}$ where both $c$ and $s = O(|T|)$, and macrotile representation function $R$ such that $\calS'$ simulates $\calS$ at spatial scale factor $c$ and temporal scale factor $s$ under $R$. The simulation will be seed-independent, and $\calS'$ will be locally adjacent-only and staggered. A visual overview of \cref{lem:adj-only} is provided in \cref{fig:dice}.

Using spatial scale factor $c$, $\calS'$ will simulate $\calS$ using macrotiles of size $c \times c$.
Using temporal scale factor $s$, each step of $\calS$ will be represented by $s$ assembly steps in $\calS'$, and for every $s \in \mathbb{N}$ such that $s \mod s = 0$, the assembly of $\calS'$ will map to a valid assembly of $\calS$. For each assembly step $i \in \mathbb{Z}^+$ of $\calS$, we will refer to the series of $s$ steps in $\calS'$ where $(i-1)s < s \le is$ as the $i$th \emph{macrostep}.%\todo{These definitions should be cleaned up and improved and included in a definition of simulation.}

\begin{figure}
    \centering
    \includegraphics[width=5.0in]{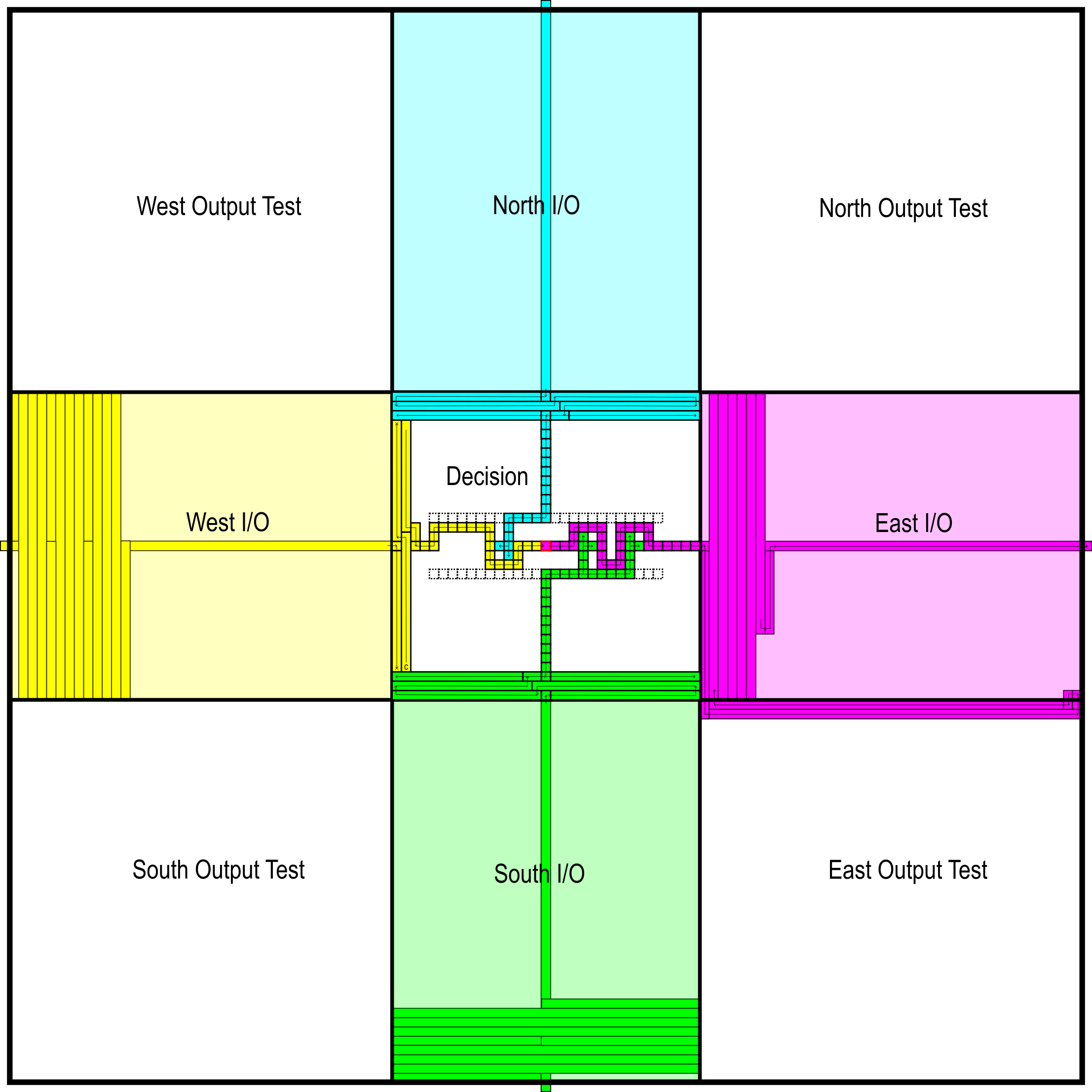}
    \caption{An example macrotile for the construction of the proof of \cref{lem:adj-only}, with the submacrotile names shown. In this example, the N, W, and S were input sides and the E an output side. The N had input delay 0 so arrived at the Decision submacrotile first, then the S arrived, and finally the W, which provided enough input glue strength to simulate a tile attachment by placing a tile in the central POC, and then cause output of the one remaining glue, on the E side.\label{fig:adj-only-macrotile}}
\end{figure}

Let $\vec{\alpha}$ be an arbitrary assembly sequence in $\calS$ and let $0 \le i < |\vec{\alpha}|$ be the index of the assembly step in $\vec{\alpha}$ which yields assembly $\alpha_i$. Let $L_i \subset \mathbb{Z}^2$ be the set of locations in $\calS$ where tiles are placed by step $i$. For each $l \in L_i$ let $M(l)$ represent the set of locations in $\calS'$ that map to the location $l$ (i.e. the macrotile representing $l$). For $l = (x,y)$, these locations will be $M((x,y)) = \{(x',y')\mid cx \le x' < c(x+1)$ and $cy \le y' < c(y+1)\}$. During macrostep $i$, $M(l)$ will complete growth (to be explained) that causes it to change from mapping to empty space under $R$ to mapping, under $R$, to a tile of the type that $\calS$ places in $l$. Additionally, the tile placements during the final step of macrostep $i$ will place a tile within, but adjacent to the boundary of, the macrotile $M(l)$ at each side for which an \emph{output glue} needs to be represented. 
(That is, in each direction in which the tile in $\calS$ in location $l$ exposed a non-null glue that wasn't used in its initial binding, as an \emph{input glue}, but which instead is adjacent to an empty location and may be used as an input glue during the binding of a neighboring tile in a future step.)

Then, during the first step of macrostep $i+1$, each of those output glues will initiate growth into those neighboring macrotile locations. That is, for each $l \in L_i$, for each direction $d$, if the tile added into location $l$ in $\calS$ during step $i$ placed an output glue in direction $d$, then let $l'$ refer to the location on the $\texttt{OPP}(d)$ side of $l$. During macrostep $i+1$, in $\calS'$ growth will enter side $\texttt{OPP}(d)$ of macrotile $M(l')$. The first tile placed into $M(l')$, during the first step of macrostep $i+1$, will be in the center of side $\texttt{OPP}(d)$. %into a region of the macrotile called an \emph{I/O submacrotile}.

An example macrotile of $\calS'$ can be seen in \cref{fig:adj-only-macrotile}. Each macrotile is logically divided into $9$ square and rectangular regions called \emph{submacrotiles}. For each direction $d \in D$, there is a square \emph{output test} submacrotile along the counterclockwisemost third of side $d$. This submacrotile only ever contains any growth if side $d$ attempts to grow an output glue representation. 
The $4$ submacrotiles adjacent to the center of each exterior macrotile side are rectangular and are the \emph{I/O submacrotiles} (i.e. \emph{input/output submacrotiles}). 
The center submacrotile is square and is the \emph{decision submacrotile}. The center location of the decision submacrotile is called the \emph{central point of competition} (POC). 
Initial growth into any macrotile enters through the central location of a side (through an \emph{input side} representing an input glue), into that side's I/O submacrotile. 
Final growth out of any macrotile leaves through the central location of a side (through an \emph{output side} representing an output glue), out of that side's I/O submacrotile.

In order to ensure that the completion of growth of macrotile $M(l)$ is synchronized, across all $l \in L_i$, various stages of the growth of each $M(l)$ will consist of paths that \emph{delay} for carefully determined amounts of time. This is done within the first submacrotile type to be explained (in \cref{sec:IO}), the I/O submacrotile. However, first we provide a high-level overview of the phases of growth of a macrotile to give context to the growth in the individual regions.

\subsection{Macrotile growth phases}

As before, let $\vec{\alpha}$ be an arbitrary assembly sequence in $\calS$. Let $l \in \mathbb{Z}^2$ be a location in $\calS$ that is initially empty (i.e. not in the seed $\sigma$). Let $\texttt{IN}(l,i) \subseteq \{N,E,S,W\}$ be the set of directions such that for each $d \in \texttt{IN}(l,i)$, at step $i$ of $\vec{\alpha}$ a tile is placed in the location adjacent to $l$, in direction $d$, which has a non-null glue on side $\texttt{OPP}(d)$. %\todo{PD: I moved the definition of \texttt{IN}() to \cref{sec:prelims}}
That is, the set of directions $\texttt{IN}(l,i)$ represents the sides of $l$ adjacent to neighboring locations that received tiles in step $i$ and have glues facing $l$. These sides are (potential) input sides, as those neighboring glues may contribute binding to a tile placed in location $l$ in some step $j > i$. Note that for the following discussion we will refer to all such glues as input glues, even though some of them may not actually contribute to binding of a tile in $l$, since it is possible that: (1) location $l$ may permanently remain empty in the case where there are not neighboring glues with enough strength to allow for a tile attachment, or (2) a subset of neighboring glues not including all (or any) of these may serve as the input glues for a later tile attachment, with some (or all) of these glues remaining unbound as mismatched glues. 
However, the growth within a macrotile that combines information about neighboring glues in order to determine which, if any, tile type attachment should be simulated is essentially the same until a possible attachment type is selected, so we'll refer to all such sides as input sides for ease of discussion.

Given location $l \not \in \dom{\sigma}$, let $M(l)$ be the macrotile in $\calS'$ that maps to location $l$, and note that $M(l)$ is initially completely empty.

\begin{enumerate}
    \item For any step $i$ such that $l$ is empty and $\texttt{IN}(l,i) \neq \emptyset$, macrostep $i+1$ proceeds as follows:
    \begin{enumerate}
        \item Simultaneously, for each $d \in \texttt{IN}(l,i)$:
        
        \begin{enumerate}
            \item A delay $D(d)$ is randomly selected (from the values in the row in \cref{tab:go-first} corresponding to direction $d$), and the I/O submacrotile of $M(l)$ on side $d$ begins growth that delays its entry into the decision submacrotile by an amount proportional to its value $D(d)$.
        \end{enumerate}

        \item Let $D_S$ be the ordered list of elements of $\texttt{IN}(l,i)$ sorted by the delay amount selected by that side, from smallest to largest. In order, for each $d' \in D_S$:

        \begin{enumerate}
            \item The delay growth for side $d'$ either: (1) is blocked by growth coming out of the decision submacrotile and terminates, or (2) completes and growth proceeds into the $d'$ side of the decision submacrotile.

            \item If growth from side $d'$ proceeds into the decision submacrotile, then $M(l)$ has not yet grown to represent a tile in $\calS$. That growth combines with growth from any other sides that may have arrived earlier.
            
            \begin{enumerate} 
                \item If the glues from the directions that have already grown into the decision submacrotile are sufficient to allow simulation of a tile placement, growth continues to place a tile in the central POC, and then out into each I/O submacrotile that hadn't previously initaited growth into the decision submacrotile. A row grows from the center of each such I/O submacrotile into the output test macrotile for that side. There, it grows a path that tests whether or not that side has initiated any input growth (to be explained). If not, growth then continues within the I/O submacrotile to perform sufficient delay (to be explained) and then grow to the central location of that side adjacent to the exterior of the macrotile to represent the output glue.

                \item Otherwise, growth halts to wait for the possible arrival of one or more additional glue representations from input sides.
            \end{enumerate}
        \end{enumerate}
    \end{enumerate}
\end{enumerate}

Note that during any step $i$ in which $\texttt{IN}(l,i) \neq \emptyset$, macrostep $i+1$ may have enough input sides present for $M(l)$ to complete growth causing it to represent a tile in $\calS$, in which case it places a tile in the central POC and grows a path to the center location adjacent to the exterior of each output side (with some output path growth possibly colliding with already-existing input growth from any side that was an input side before reaching the macrotile boundary). 
However, if the input sides present during macrostep $i+1$ are insufficient to allow for simulation of a tile binding, then growth stalls in the decision submacrotile and does not resume unless and until a future macrostep $j > i$ such that $\texttt{IN}(l,j) \neq \emptyset$, at which point the same growth process occurs for the newly arrived input side(s).

% \begin{enumerate}
%     \item Define the locations and corresponding macrotile locations of growth for each (macro)step

%     \item Define the corresponding frontier locations

%     \item Discuss the timing of growth into a frontier macrotile

%     \item Discuss the high-level function of the I/O submacrotile

%     \item Discuss the high-level function of the decision submacrotile
% \end{enumerate}

\subsection{I/O submacrotiles}\label{sec:IO}

Each I/O submacrotile serves the primary purpose of providing space for an input or output path to stall for precisely determined numbers of steps (to be explained) by growing zig-zag paths that delay the growth through the I/O submacrotile. Growth through an I/O submacrotile from an adjacent macrotile on its exterior toward the decision submacrotile is called \emph{input growth} and its purpose is to delay arrival at the decision submacrotile by an amount that ensures that input growth from multiple directions will reach the decision submacrotile at different times, sufficiently spread out in time to prevent unintended collisions. (Specifics details will be provided in \cref{sec:decision-sub}.)

On a side that has input growth, the first tile placed into its I/O submacrotile is nondeterministically selected from a set of $12$ possible \emph{input-delay} tile types that are specific, and distinct, for the direction of that side. The values represented by the $12$ possible tile types, unique for each side, can be seen in \cref{tab:go-first}. The values in that table are a slight modification (i.e. the value of $1$ has been subtracted from each) of results for 4-player ``Go-First Dice'' \cite{gofirstford2023,gofirstPurcell,gofirst4player} which guarantee that there are never ties and all permutations of orderings are fairly generated. %\todo{Could add more explanation and could also discuss possibly smaller sets of values from related results.} 
From whichever tile type of the $12$ that nondeterministically attaches, a zig-zag path whose length is proportional to the value $r$, which is the nondetermistically selected value uniquely represented by that tile type, grows first, and when that completes a row grows to the center of adjacent side of the decision submacrotile. The length of the zig-zag path is $r \cdot u_d$, where $u_d$ is the \emph{unit of delay} (to be explained). In \cref{fig:adj-only-macrotile}, example input-delay paths are seen growing from the sides N, S, and W.

% \begin{table}
% \centering
% \begin{tabular}{| c | l | l | l | l | l | l | l | l | l | l | l | l |}
% \hline
% \textbf{N}& 1 & 8 & 11 & 14 & 19 & 22 & 27 & 30 & 35 & 38 & 41 & 48 \\
% \hline
% \textbf{E}& 2 & 7 & 10 & 15 & 18 & 23 & 26 & 31 & 34 & 39 & 42 & 47 \\
% \hline
% \textbf{S}& 3 & 6 & 12 & 13 & 17 & 24 & 25 & 32 & 36 & 37 & 43 & 46 \\
% \hline
% \textbf{W}& 4 & 5 & 9 & 16 & 20 & 21 & 28 & 29 & 33 & 40 & 44 & 45 \\
% \hline
% \end{tabular}\label{tab:go-first}
% \caption{The 12 random numbers assigned to each direction.}
% \end{table}

\begin{table}
\centering
\begin{tabular}{| c | l | l | l | l | l | l | l | l | l | l | l | l |}
\hline
\textbf{N}& 0 & 7 & 10 & 13 & 18 & 21 & 26 & 29 & 34 & 37 & 40 & 47 \\
\hline
\textbf{E}& 1 & 6 & 9 & 14 & 17 & 22 & 25 & 30 & 33 & 38 & 41 & 46 \\
\hline
\textbf{S}& 2 & 5 & 11 & 12 & 16 & 23 & 24 & 31 & 35 & 36 & 42 & 45 \\
\hline
\textbf{W}& 3 & 4 & 8 & 15 & 19 & 20 & 27 & 28 & 32 & 39 & 43 & 44 \\
\hline
\end{tabular}
\caption{The 12 random numbers assigned to each direction to determine the length of input-delay.\label{tab:go-first}}
\end{table}

Growth into the decision submacrotile from one or more input sides may result in the macrotile resolving to a tile in $T$. If and when this occurs, growth (called \emph{output growth}) will be initiated outward to any I/O submacrotiles which did not provide input growth into the decision submacrotile. 
This may occur on sides where no input growth has ever been initiated, or where it was initiated but had an input-delay sufficiently long to delay its arrival until after growth completes in the decision submacrotile. 
On a side $d$ where output growth is initiated, this will occur in two phases. In the first phase, a path will enter the I/O submacrotile for direction $d$ in the center of the side adjecnt to the decision submacrotile and grow directly to the output test submacrotile of side $d$ (as shown in \cref{fig:output-test}). There, it will grow into the second row where it will grow a path to the edge of the macrotile. At this point, it initiates the growth that \emph{tests} whether or not input growth has already occurred in side $d$'s I/O submacrotile. It does this by attempting to grow up into the nearest corner of $d$'s I/O submacrotile. If input growth has already occurred in side $d$, this growth will be blocked and the output growth will terminate (since no output can, or needs to be, exposed on side $d$. Otherwise, a path will grow back to the side from which the output growth initially began. It will then begin zig-zag growth for the \emph{output-delay} (whose length will be discussed) and then a path that grows to the center of the exterior boundary of the I/O submacrotile (and thus the exterior boundary of the macrotile itself). This path carries the information about the output glue to be represented by that side of the macrotile.

Note that if input growth had previously been initiated and terminated the growth of the output, the input growth will eventually collide with the initial portion of the output growth and terminate as well, which is by design since the macrotile will already have resolved to a tile in $T$. %However, that is not a problem since the macrotile will already correctly resolve to a tile in $T$ and there is no need to complete the output growth since there must have already been an adjacent macrotile (that has also resolved to a tile in $T$) which initiated the colliding input growth, and thus there is no need for the output growth to complete.

\begin{figure}
    \centering
    \includegraphics[width=5.5in]{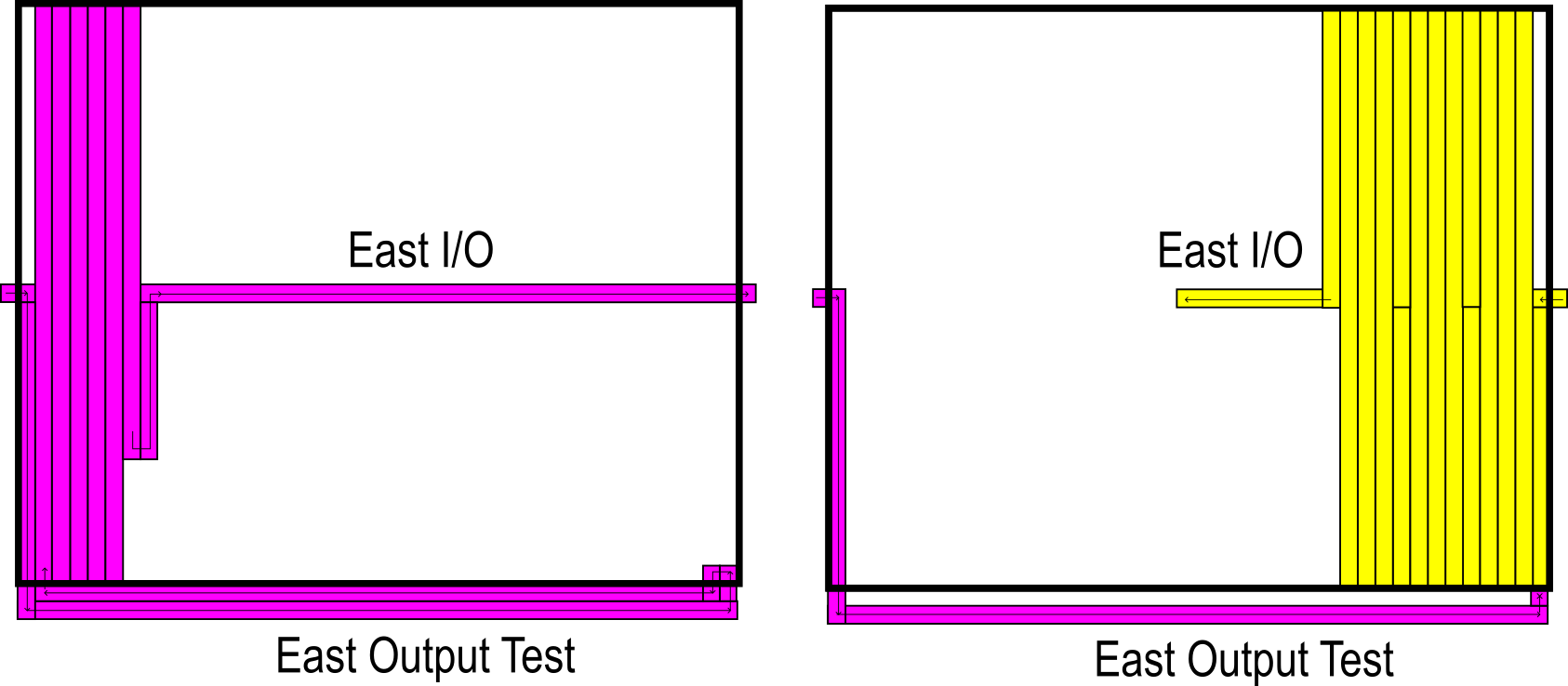}
    \caption{(Left) An example E I/O submacrotile into which output growth was initiated from the decision submacrotile. That growth  creates a path into the second row of the E output test submacrotile which proceeds to grow to the far right side then upward into the southern location of the easternmost row of the E I/O submacrotile. Since no input growth has arrived, that growth is unblocked. It then proceeds back down into the first (northernmost) row of the E output test submacrotile, where it grows back to the west. It then grows the necessary output delay zig-zag rows, then a row to the E side of the macrotile to represent the output glue. (Right) A similar situation, but here input growth has already begun from the E. This blocks the growth upward from the E output test submacrotile, which terminates its growth. Once the E input growth completes its input delay growth, it proceeds to the west where it will eventually crash into the output growth from the decision submacrotile, which will cause it to terminate.\label{fig:output-test}}
\end{figure}

The length of the output delay is dependent upon the full length of the path leading from the beginning of the last input growth arriving in the decision submacrotile that contributed to the resolution of the macrotile. Those paths will be discussed in \cref{sec:decision-sub}. For the discussion of the I/O submacrotile's output growth, we will need to define the unit of delay, $u_d$. This will be related to the longest path that can grow into the decision submacrotile, through the central POC, and then to the location where an output path grows out of the decision submacrotile. We will call the length of such a longest path $l_{max}$. The growth of the input delay path is a zig-zag path that grows the width of the I/O submacrotile, which is the same as the dimensions of the decision submacrotile and which we refer to as $w$. (This value is defined in Section \ref{sec:probe-region} and shown to be $O(|T|)$.) Let $l_{div} = \lceil (l_{max}+1) / w\rceil$ and $l_{mod} = (l_{max} + l_{div} + 1) \mod w$. The unit of delay, $u_d = l_{max} + l_{mod} + 1$, which is set so that even if sides $d$ and $d'$ select random numbers $r$ and $r+1$, respectively, making their input delays $r\cdot u_d$ and $(r+1) \cdot u_d$, respectively, then the growth into the decision submacrotile from side $d$ will have time to complete growth through the submacrotile and place at least one tile as output growth into the I/O submacrotile for side $d'$ before the input growth of side $d'$ can reach that point. (That is in the case that the input growth from side $d$ is sufficient to provide enough, possibly additional, input glue strength for the macrotile to resolve. Otherwise, its growth in the decision submacrotile will stall and the input growth from $d'$ will have the opportunity to arrive and possibly complete the growth.) Note that $u_d = O(|T|)$ and that $u_d$ is a multiple of dimension $w$. Therefore, since input delay growth begins from the center of an I/O submacrotile's side, input delay growth will always complete one position away from the center. From that point it will begin growth down the center of the I/O submacrotile directly to the decision submacrotile. If the macrotile has not yet resolved to a tile, it will then enter the decision submacrotile. If it has, then growth will already have come out of the decision submacrotile placing at least one tile in the I/O submacrotile, directly in line with that path, causing it to terminate there (shown on the right of \cref{fig:output-test} where the yellow path will crash into the magenta).

If a side $d$ attempts to grow output and no input growth has begun on that side (so that its output test growth is not blocked), once the output test path returns to the side closest to the decision submacrotile (the left side in \cref{fig:output-test}), the length of the output delay path that then grows, via zig-zagging rows, is determined as follows. Let $d'$ be the side whose input growth was the last to enter the decision submacrotile, causing it to grow into the central POC, resolve into a tile of $T$, then initiate the output growth to direction $d$. Let $r'$ be the random number selected by side $d'$ to determine its input delay, and let $l'$ be the length of the path (i.e. the number of assembly steps) that began from the first tile placed by side $d'$ in the decision submacrotile until the last placed in the decision submacrotile before exiting into the I/O submacrotile on side $d$. The length of the output delay is $u_d(47-r') + (l_{max} - l')$. Effectively, this makes the path length and number of assembly steps that begins with the first tile placed in the input side that causes the macrotile to resolve, and ends with the final tile placed adjacent to the macrotile's exterior on the output side, to be padded out to be the maximum possible. That is, given the maximum possible input delay multiple of 47, the difference between this path and one that grew that maximal input delay is first grown as output delay. Then, additional delay is added to account for how ever much shorter the path through the decision submacrotile was than the path of maximal length. Thus, all output paths are padded to the same length and growth of the macrotile will complete at exactly the final step of a macrostep in $\calS'$.

The dimension of the I/O submacrotiles that is not the same as that of the decision submacrotile is set to $48(u_d / w)+2$. 
This provides sufficient space for the maximum output delay growth (since $l_{max} - l' \le u_d$ and therefore requires $\le u_d/w$ zig-zag rows) and the maximum size of the other portion of the output delay requires $47(u_d / w)$ rows. 
The extra two rows allow for the growth of the rows for the output test. Note that  $48(u_d / w)+2 = O(|T|)$. A final point to mention is that, since this must be an even number and the other dimension is an odd number, the overall dimensions of a macrotile are odd numbers (i.e. the sum of two even numbers and one odd number). 
Since paths enter and exit all macrotiles and submacrotiles through the center locations of sides with odd lengths, all paths must be of odd length (due to being paths through the grid graph). This also means that the differences in lengths between any two paths must be even values, and therefore for any necessary length of an output delay, a valid zig-zag path can be created for that delay (noting that one zig and zag for each will likely not cover the full width of the I/O submacrotile.

%it is sufficient to note that such a path will have started with some input delay value $r$ (as it grew through its initial I/O submacrotile) and that its path through the decision submacrotile, that eventually reaches the I/O submacrotile for output, will have some additional length $l$. Furthermore, for any system $\calS'$ the maximum input delay value possible is $47$ (see \cref{tab:go-first}) and the maximum length path through the decision submacrotile is some $l_{max}$ (which is dependent upon the system $\calS$ being simulated and further discussed in \cref{sec:decision-sub}). Let $w$ be the width of the I/O submacrotiles (in the dimension in which the I/O paths zig and zag). The length of the output delay path will be $(l_{max} - l) + (47-r)4w$. Since $l_{max} + (47)4r$ is the longest that any such path can be, 

%this ensures that all output paths are delayed so that, with the delay, they are of that maximum length and thus all the same length. Therefore, all output paths that complete during a macrostep do so during the exact same assembly step in $\calS'$.

\subsection{Decision submacrotiles}\label{sec:decision-sub}

\begin{figure}
    \centering
    \includegraphics[width=3.0in]{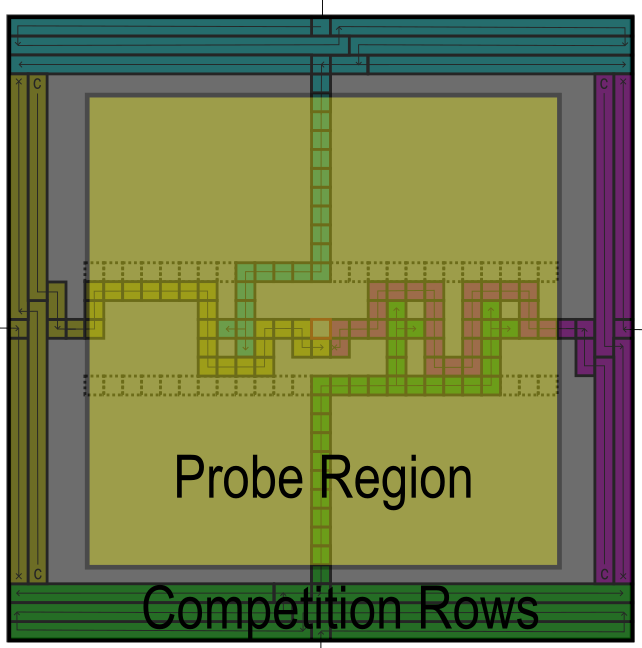}
    \caption{The layout of a decision submacrotile. The outer 4 rows are the competition rows, and the region inside of that is the probe region.\label{fig:decision-submacrotile}}
\end{figure}

Each decision submacrotile consists of 4 rows immediately inside its outer boundary called the \emph{competition rows}, and the remainder is called the \emph{probe region}. (An example can be seen in \cref{fig:decision-submacrotile}.)

\subsubsection{Probe region structure}\label{sec:probe-region}

The main function of the probe region is to provide (when necessary) subassemblies, called \emph{probes}, that grow from opposite directions (N and S, or E and W) that can cooperate (using adjacent cooperation) if the glues represented by those opposite directions could cooperate to place a tile in $\calT$. However, those probes must also be carefully designed so that, if that pair of opposite-side glues aren't sufficient to allow simulation of a placed tile, they don't prevent growth initiated by the other pair of sides that may need to pass between the probes. Note that any individual macrotile will grow either (1) no probes, (2) N and/or S probes for N-S glues, or (3) E and/or W probes for E-W glues. There will never be an attempt to grow probes from more than one pair of opposite sides (and the mechanism for determining this will be described in Section \ref{sec:decision-protocol}).

%We will describe the probes for the N-S pair of sides and note that those for the E-W pair of sides are created using the same procedure with all directions rotated by 90 degrees. 
Let $T_{N,S} \subseteq T$ be the subset of tile types from $T$ such that, for each $t \in T_{N,S}$, neither its N nor S glue is null or has strength $\tau$, but the sum of those two glues is $\ge \tau$. Essentially, $T_{N,S}$ is the set of tile types that could potentially bind across-the-gap using exactly their N and S pair of glues. Let $T_{E,W} \subseteq T$ be defined similarly but for the directions E and W.

\begin{figure}
    \centering
    \includegraphics[width=5.5in]{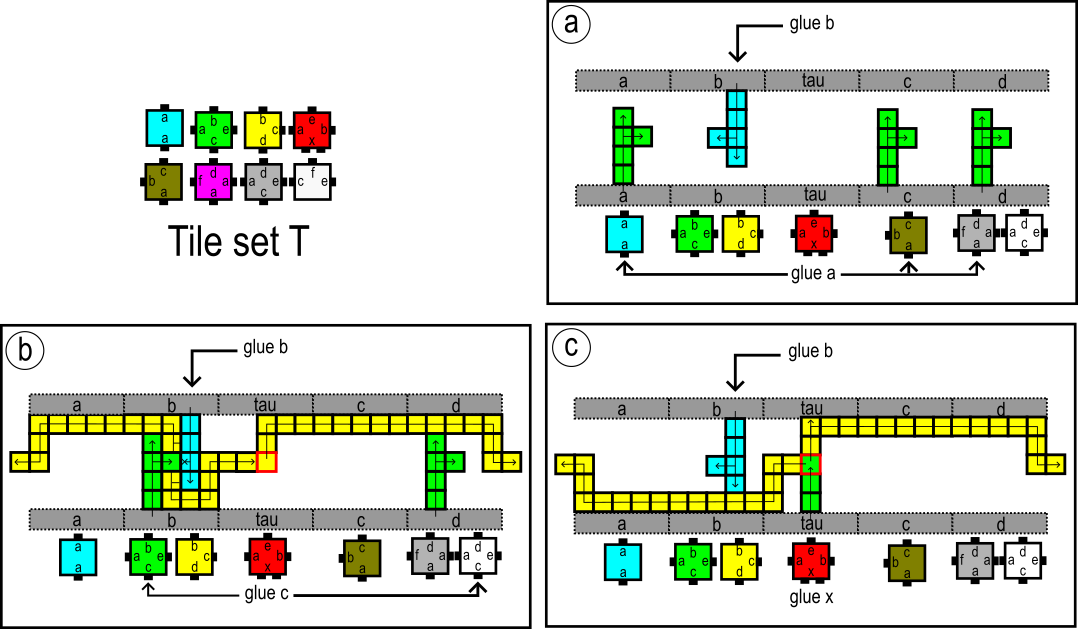}
    \caption{Examples of the probe regions for a pair of N and S input glues given the tile set $T$. The set $G_N$ for $T$ is $\{a,b,c,d\}$, since those are the glues such that each is capable of cooperating with a glue from the S to place a tile of $T$ (assuming $\tau=2$ for this example). There are probe locations for each, as well as a central region for $\tau$-strength glues. (a) The N input glue is $b$ and the S input glue is $a$. Since there is no tile type with $b$ on the N side and $a$ on the S side, no pair of probes share the same probe location. (b) The N input glue is $b$ and the S input glue is $c$. Since there are two tiles with $c$ on the S, one with $b$ on its N and one with $d$ on its N, S probes are grown in the $b$ and $d$ locations. This causes probes to meet in the probe location for $b$, resulting in a tile of the green tile type being simulated. Paths then grow to the E and W (with the E path also claiming the center POC) to attempt to grow the corresponding output glues for the green tile type. (c)The N input glue is $b$ and the S input glue is $x$, which is $\tau$-strength. Therefore, a probe grows to the central POC in the $\tau$ location. Since it wins the POC, it attempts to grow output to W, N, and E sides.\label{fig:probes}}
\end{figure}

Given a tile set $T$ and a side $d \in D$, let $G(T,d)$ be a function that returns the set of all glues that appear on the $d$ side of any tile type in $T$. Let $G_N = G(T_{N,S},N)$ and $G_E = G(T_{E,W},E)$ be the sets of southern and western glues on the tiles in the tile sets $T_{N,S}$ and $T_{E,W}$, respectively. Now let $g_{max} = \texttt{max}(|G_N|,|G_E|)$ and let $g_{max}' = g_{max} + (g_{max} \mod 2)$ (i.e. $g_{max}'$ is the smallest even number $\ge g_{max}$).  The width and height of the probe region are $p = 5(g_{max}'+1)$. The width and height of the decision submacrotile  are $w = p + 8$, which is guaranteed be an odd number. Note that since $w$ is dependent upon the maximum of the sizes of the sets of northern and eastern glues, $w = O(g'_{max}) = O(|T|)$.

The layout of the probe region grown for N-S probes can be seen in \cref{fig:probes} and at a high level works similarly to the probe region of \cite{IUSA}. However, since the tile set of system $\calS'$ is not a universal simulator but is instead specifically generated to simulate tiles of $\calS$, the probe region is much more compact that that of \cite{IUSA}. The probes in the probe region either grow from N and S, or they grow from E and W. The directions depend upon which input side grows into the decision submacrotile first. However, the general structure of N-S probes is the same as E-W probes, with the E-W probes being similar to the N-S probes but rotated $90$ degrees clockwise. Therefore, we will describe the structure of N-S probes and then mention the differences for the E-W probes afterward.

Given the subset of tiles from $T$ whose N and S glues are sufficient to cooperatively allow attachment, namely $T_{N,S}$, there will be a $5 \times 5$ \emph{probe location} for each glue in $G_N = G(T_{N,S},N)$. If the number of these glues is odd, an empty probe location is added to the E end. These regions are ordered by sorting the glue labels lexicographically, and a special probe region is inserted into the center to contain the central POC. Whenever input from the N represents some glue $g \in G_N$, a probe will grow from the N into the probe location for $g$. Examples can be seen in the top and bottom depictions of probe regions in \cref{fig:probes}.

Input from the S is handled a bit differently. If the S input represents some glue $g \in G(T_{N,S},S)$, then let $T' \subseteq T_{N_S}$ be the set of tile types from $T_{N,S}$ that have glue $g$ on their S and let $G' = G(T',N)$ be the set of N glues from all of those tiles. From the S, a probe is grown into the location for each glue in $G'$. In this way, a southern probe is grown into the location of every glue that could be represented by the N that the glue from the south could cooperate with to cause the macrotile to resolve to a tile in $T$. 
Furthermore, if the glues represented by N and S probes are not the N and S glues of a tile type in $T$, or those glues do not sum to strength $\ge \tau$, any probes that grow will be in different probe locations, leaving room for a path of tiles to pass through the entire probe region later, if needed.

If probes from the N and S grow into the same probe location, they each grow a vertical column of $4$ tiles, but the first to arrive also places one tile horizontally adjacent to the third tile of its column. The second to arrive is blocked from placing a tile there by the first, but then paths are cooperatively initiated above and below that middle tile. There is (adjacent-side) cooperation between tiles of both probes to initiate those two paths so that both input glues are known to both, and the tile type in $T$ that the macrotile resolves to, and thus also the output glues for the E and W directions, is then fixed and encoded in each corresponding path. The path above carries the information about the output glue to the W, and the path below carries the information for the output glue to the E. Note that if there is nondeterminism in $T$, meaning that the N and S input glues could cooperate to place tiles of multiple types, then the type of the S probe will be nondeterministically selected from a set representing each so that the random ``choice'' is made by the placement of the first tile of the S probe. In this way, all options can be simulated and it is guaranteed that the random selection is consistent across both the upper and lower paths. Further note that the ordering of arrival of the probes is clearly determined by the placement of the center tile between probes and thus which tiles cooperate to form the paths, so that timing information can be encoded into the paths. The utility of this information will be explained later.

\subsubsection{Decision submacrotile growth protocol}\label{sec:decision-protocol}

The growth of the portions of the decision submacrotile are carefully designed to follow a protocol that ensures that, as soon as a set of input paths arrive that represent glues sufficient to allow attachment of a tile in system $\calS$ (and thus the resolution of the macrotile in system $\calS'$), they are able to cooperate (if necessary), grow to the central POC, and grow paths to any I/O submacrotiles that may need to grow output paths.

We now present the protocol that governs the growth of the decision submacrotile.

\begin{enumerate}
    \item Any path that enters or exits a decision submacrotile does so through the middle location of one of the 4 sides. (Recall that the side lengths are guaranteed to be an odd number, so there is always a unique center location.)

    \item Let $d \in D$ be the side of a decision submacrotile into which an input path grows. Once the first tile has been placed inside of side $d$, it initiates a path which first grows to the corner on the counterclockwise side, and then across to the corner on the clockwise side.

    \item At each corner, if the row is able to grow all the way into that corner and place a tile in the corner location we way that side $d$ \emph{Wins} (W) that corner. If, instead, some other side has already grown a path that occupies the corner, the growth to the corner is blocked and we say that side $d$ \emph{Loses} (L) that corner.

\begin{figure}
    \centering
    \includegraphics[width=3.5in]{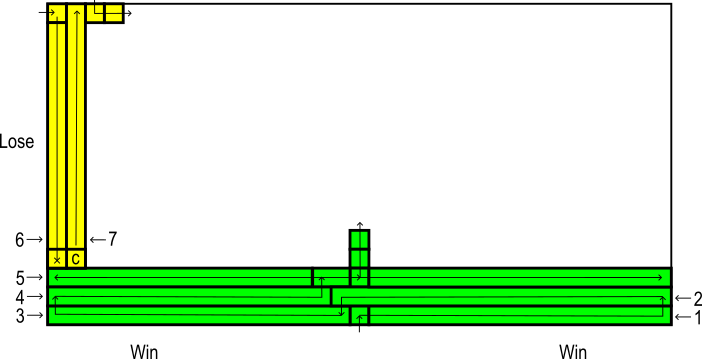}
    \caption{An example of the growth of competition rows, showing the full rows grown by the S side (but not its probes) and the bottom half of those grown for the west side. In this example, the S side grew first and Won both corners. When the west side arrived later, its competition rows were able to cooperate with a tile of the S side's competition rows to allow completion of their second row while also propagating information about the S glue and the fact that the S side Won both corners. (The tile placed by cooperation is labeled with 'C'.)\label{fig:competition-rows}}
\end{figure}

    \item Regardless of whether the paths of side $d$ Win or Lose either corner, they grow back to the center of side $d$.
    \begin{enumerate}
        \item W: For a Winning corner, the path is simply a hard-coded path where each tile attaches to the previous in the path with strength $\tau$.

        \item L: For a Losing corner, there are guaranteed to be $3$ rows of tiles that have grown inside of that corner by the adjacent side. (The reason for that guarantee is the difference in input delay path lengths that different sides must have.) The final tile placed of the first (blocked) row grown by side $d$ will have a glue on its side facing toward the center of the decision submacrotile, and that will be able to cooperatively place a tile with the diagonally adjacent tile placed by the other side. This tile will initiate the second row on side $d$, which grows back toward the center of side $d$. \cref{fig:competition-rows} shows an example of the completed competition rows for the S side (which, in this case, grew into the decision submacrotile before any other side). The labels `1` through `5' denote that those completed in that order. When the competition rows for the west side begin growth, the first one (labeled `6' in \cref{fig:competition-rows} because it is the 6th overall to grow) has its final 3 tiles of growth blocked by the 3rd row of the S side. The tile in the location labeled `C' attaches via cooperation between the adjacent green and yellow tiles. The tile that attaches encodes the information about the glue on the S (as well as the fact that the S side Won both corners), and initiates growth of the second west competition row (labeled `7') back to the center of the west side.
    \end{enumerate}

    \item The second competition rows from each side grow back to the center, with the growth from the second corner then growing up into the third competition row and initiating growth toward the corner of any side that was Won, plus growth toward the center of the decision submacrotile.

    \item For any corner that is Lost, on that half of the decision submacrotile side $d$ only two competition rows are grown.
    
    \item For any corner that is Won, a third competition row is grown. (See the competition rows in \cref{fig:decision1} for examples.) That third row grows from the center to that corner and, since both corners have now been visited, it encodes information about:
    \begin{enumerate}
        \item the input glue on side $d$,

        \item the Win and Lost status of both corners, and

        \item if the other corner Lost, the glue and Win/Loss information of the side $d'$ that won that corner instead of $d$. This will be available because the Losing competition row is guaranteed to grow after the competition rows of $d'$ have completed, and since the corner is Won by $d'$, then the third competition row of $d'$ will include the glue information and Win/Loss status of its corners. (For example, in \cref{fig:decision1}, the third row of the S side will encode information about the S input glue and the fact that the S side Won both corners. That information will be acquired by the second W row that grows back up to the center of the W side, and then transferred to the top corner of the W side by its third row, which will also contain information about the S input glue and the fact that the S side won both of its corners. The same process occurs if growth for an input glue occurs on the E side. Then, in the event that a N input glue appears later, when the N grows its competition rows they acquire information about all 4 input glues, which allows a decision to be made if there is any set of input glues that could allow the macrotile to resolve to a tile of $T$.)
    \end{enumerate}

    \item Once rows have grown toward both corners of side $d$ then returned to the middle, they can initiate further growth toward the center of the decision submacrotile and encode information that determines how that growth will proceed. This information is dependent upon the Won/Lost status of the two corners. The following list contains all possible options. Keep in mind that once a side places its first tile into the decision submacrotiles, all growth inside of the decision submacrotile initiated by that side will complete before any new growth is initiated by another side. (Again, the reason for this is the differences in input delays that is guaranteed between any set of input sides.) As soon as growth from a side contains information about adjacent glue(s) that is sufficient to allow for the simulation of tile attachment of some type $t \in T$, growth from that side will claim the central POC then grow to the necessary I/O submacrotiles (for the directions in which output glues need to be represented), all of which would complete before growth from any new side would begin in the decision submacrotile. In the following case analysis, we consider the ``worst case'' where inputs from all four sides would need to arrive to determine any simulated tile placement. This allows us to demonstrate how information can be correctly gathered from all sides without blocking any necessary output paths up until the point where a fourth input side is required (at which point no output growth is needed).
    
    The cases for the two corners of side $d$ are:
    \begin{enumerate}
        \item Win-Win: if both corners were Won, the side grows probes following the protocol specific to side $d$, previously described. (See the S side of \cref{fig:decision1} for an example.)

        \item Win-Lose: let $d'$ be the side that $d$ Lost a corner to. If $d'$ had competition rows that were:

        \begin{figure}
            \centering
            \includegraphics[width=3.0in]{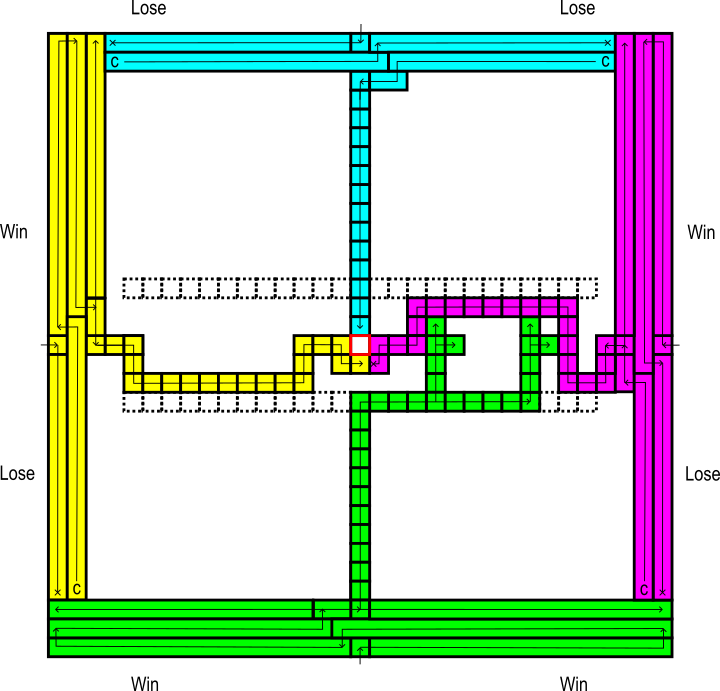}
            \caption{An example of growth within a decision submacrotile showing a particular ordering of inputs from all 4 sides. The S side is first and Win-Win and thus grows probes, then E and W sides arrive and grow toward the center, dodging the probes of the S, but there is still not enough glue cooperation to allow simulation of tile attachment. Finally, the N side arrives. If tile attachment could then be simulated, a final tile would be placed in the central POC.\label{fig:decision1}}
        \end{figure}
            
        \begin{enumerate}
            \item Win-Win: then side $d'$ will have grown probes. Since side $d$ gains the information (via the cooperation that allows it to finish its second competition row) about the glue on side $d'$ which is encoded by those probes, it grows a path toward the POC that dodges those probes. (Examples can be seen in \cref{fig:decision1}, where the E and W sides ``know'' that the S side Won both corners and will therefore grow probes, but they also know what the S input glue is so they grow paths that are designed to dodge those probes. With respect to this case, in \cref{fig:decision1}, side $d$ is E and side $d'$ is S.) 
            For each of the E and W input glues, if they are $\tau$-strength, or their combination with the S input glue is sufficient to simulate the attachment of some tile type $t \in T$ then the path grows into the central POC and the macrotile then resolves to $t$ under $R$, and then output paths will be grown to the other sides. Otherwise, the path grows immediately adjacent to the central POC, then grows S one tile, then grows below the POC and stops. (For example, the yellow path in \cref{fig:decision1}.) If it is the second such path to attempt to do this (e.g. the magenta path in \cref{fig:decision1}) it will be blocked from placing its final tile by the first path that arrived. 
            However, those two paths will now have information about at least the E, S, and W side input glues and if any set of them could cooperate to simulate attachment of a tile type $t \in T$, then a tile can attach using adjacent cooperativity with tiles of both paths into the central POC. This makes the macrotile resolve to $t$ under $R$ and then an output path is grown to the N. Note, very importantly (for reasons to be discussed later), there are different tile types that attach into the POC based on which of the E and W paths arrived first and ``won'' the location below the POC (since the sides cooperating to place the tile into the POC will be different in each case). Because of this, the output path ``knows'' which input path arrived last.

            \begin{figure}
                \centering
                \includegraphics[width=3.0in]{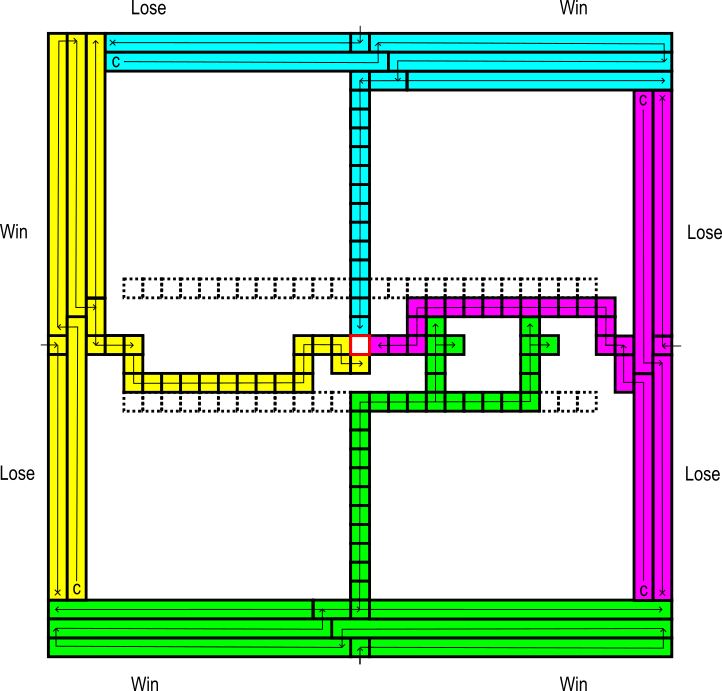}
                \caption{An example of growth within a decision submacrotile showing a particular ordering of inputs from all 4 sides. The ordering of side arrivals is S, W, N, E. The S is the only Win-Win and thus the only to grow probes.\label{fig:decision2}}
            \end{figure}
            
            \item Win-Lose:\label{case:dec2} then side $d'$ will have passed along info about the side that beat it at one corner (which we'll refer to as $d''$), so side $d$ has information about the input glues of 3 sides. (With respect to this case, in \cref{fig:decision2}, side $d$ is N, $d'$ is W, and $d''$ is S.) They are its own glue, that of side $d'$ which is adjacent to it, and that of side $d''$ which is opposite it. Since $d'$ Lost a corner, it will not grow probes. Also, since $d''$ beat $d'$, which beat $d$, and $d$ Won the remaining side, $d''$ must have also Won both of its sides and will therefore have grown probes. However, side $d$ now already has the information about the glue on side $d''$, and if any combination of the glues of sides $d$, $d'$, and $d''$ could cooperate to allow simulation of a tile attachment, side $d$ will grow a path directly to the central POC to ``claim'' the macrotile, and then a path to the remaining output side. If not, that path will stop one position before, and adjacent to, the central POC. (See the aqua path from the N side of \cref{fig:decision2} for an example.)

        \end{enumerate}

        \item Lose-Lose: Here there are 3 possible scenarios for the sides, $d'$ and $d''$, to which side $d$ Lost.
            \begin{enumerate}
                \item Both were Win-Win. (With respect to this case, in \cref{fig:decision3}, side $d$ is E and sides $d'$ and $d''$ are N and S.) In this case, both $d'$ and $d''$ grew probes, but side $d$ knows what types they are and can grow around them. As discussed previously, the path grown from side $d$ will then either claim the central POC (if the glues represented by sides $d$, $d'$, and $d''$ suffice to allow for the simulation of attachment of some tile type $t \in T$) or grow to a location adjacent to that, if not.

                    \begin{figure}
                        \centering
                        \includegraphics[width=3.0in]{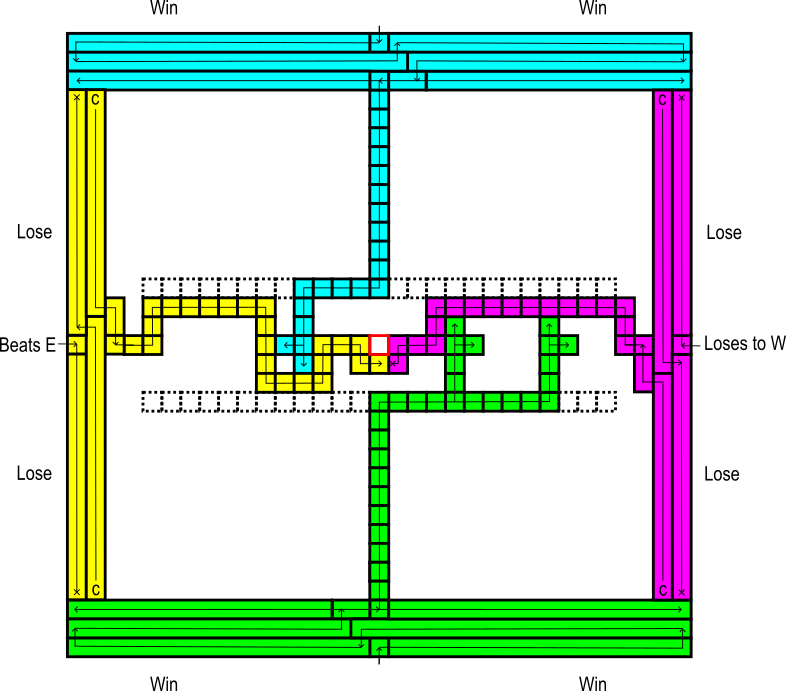}
                        \caption{An example of growth within a decision submacrotile, showing a particular ordering of inputs from all 4 sides. In this case, N and S are Win-Win and thus each grow probes. The E and W sides which arrive later have information about both sides so are able to grow paths that dodge those probes.\label{fig:decision3}}
                    \end{figure}

                \item Side $d'$ was Win-Win and side $d''$ was Win-Lose. (With respect to this case, in \cref{fig:decision2}, side $d$ is E, $d'$ is S, and $d''$ is N.) This case was discussed above in \cref{case:dec2}, and essentially side $d$ grows a path into, or adjacent to, the POC (depending on whether the glues represented by sides $d$, $d'$, and $d''$ are sufficient to simulate a tile attachment) that dodges the probes from side $d'$.

                \item Both were Win-Lose. (With respect to this case, in \cref{fig:decision1}, side $d$ is N, and side $d'$ and $d''$ are E and W.) In this case, side $d$ will have gathered information about all $4$ input glues because they all arrived before it. It therefore can terminate growth if no subset of those glues is sufficient to allow simulation of a tile attachment, or otherwise grow a path directly to the central POC (since it knows that each of its neighbors Lost a corner and therefore neither grew probes) and claim it. At this point, the tile type $t \in T$ that is simulated is fixed and no further growth needs to occur because no representations of output glues need to be grown to the exterior of the macrotile (because all sides already had incident input glue representations).
            \end{enumerate}
    \end{enumerate}
\end{enumerate}

The correctness of synchronous simulation requires that the growth of macrotiles is synchronized and thus that the output paths all complete at the same time. The ability of each path to be unique for the necessary output delay requires information about the path from the side, which we'll refer to as $d^*$, that was the last to provide input into the decision submacrotile to allow the macrotile to resolve. It needs to be unique for the random delay value selected for side $d^*$, as well as for the length of the path that grew through the decision macrotile from $d^*$. Note that it is important that input sides could have grown into the decision submacrotile during different macrosteps. Because of this, it isn't necessarily the input side with the largest input delay that arrives last. For this reason, it is important that in any location in which a tile attaches cooperatively to tiles that grew from different input sides, there are different subsets of sides used for that attachment depending on which arrived first, and thus a unique tile type attaches for each. This allows the cooperatively attaching tile to correctly encode the information about the path that arrived last (even if it had a smaller input delay but arrived in a later macrostep). Examples of such locations are those on the two sides of cooperating probes (see \cref{fig:probes}), as well as the central POC in some of the cases with 3 input sides.

\subsection{Seed structure}\label{sec:adj-seed}

We design a special set of tile types $T'_\sigma \subset T'$ to make up the seed $\sigma'$ so that the simulation done by $\calS'$ can be seed-independent. Therefore, we explain the full process for creating the tile types in $T'_\sigma$ and the function that outputs the seed $\sigma'$ such that for any arbitrary $\sigma$ of $S = (T,\sigma,\tau)$, $R^*(\sigma') = \sigma$ and $\calS' = (T',\sigma',2)$ correctly simulates $\calS$ using fixed representation function $R$, spatial scale factor $c$, and temporal scale factor $s$ independent of $\sigma$.

However, in order to ensure that the tile types in $T'_\sigma$ are constrained to the seed structure and cannot also attach later in the assembly process, we use a technique called \emph{I/O marking} which we first describe.

\subsubsection{I/O marking}\label{para:adj-io-marking}

The only tiles that form bonds across macrotile boundaries are those in the center locations of the exterior boundaries of I/O submacrotiles, with those glues used to propagate input glue information into a macrotile, or output glue information out of a macrotile. These glues are always strength-2 glues (regardless of the strength of the glue from a tile in $T$ that they are representing). Therefore, it is important that when a glue exposed on the exterior of a macrotile is propagating information that is being output from the macrotile, that the tile that attaches to it (in the adjacent location of the neighboring macrotile) initiates input growth. If, for instance, some glue $g$ is exposed on the E side of a tile type in $T$ and the W side of the same or another type in $T$, then the glues in $T'$ that simulate those must be carefully designed. If the final tiles on the output paths that grow both E and W expose the same glue to represent $g$, then a macrotile could grow that has that glue on the easternmost edge of its easternmost (final) output tile. The desired behavior would be for the tile that attaches to it in the neighboring macrotile to begin growth of an input path. However, if the same glue was used for E and W output paths, the last tile of an output path that grows W could attach. This would lead, incorrectly, to reverse growth of the output path into the macrotile, rather than forward growth of the input path. 

Another problem that must be avoided is that the special tiles in $T'_\sigma$ that are used to encode seed tiles must not be able to attach in locations outside of the pre-formed macrotiles representing the seed. This is because they do not perform any of the necessary logic of simulating inputs, computing valid tile attachments, or synchronizing timing. Therefore, it must not be possible for the tiles that finish the output paths from macrotiles to allow tiles in $T'_\sigma$ to attach to them.

Both problems can be solved by \emph{I/O marking}. This is the process of appending an extra, ``directional,'' symbol to the end of glue labels. The symbols used for input in the directions N,E,S,W are `v', `$<$', `\textasciicircum', `$>$', respectively. The symbols used for output in the directions N,E,S,W, are `\textasciicircum', `$>$', `v', `$<$', respectively. By appending these symbols to the labels of glues on the corresponding sides (and noting that a glue on side $d$ of a tile can only ever be adjacent to a glue on the side opposite $d$ on another tile), the only glues that can ever match are those for opposite directions with opposite I/O markings. (For example, output glue $g$ on the E side of a tile would be I/O marked `$g>$', which would match with I/O marked input glue `$g>$' on the W side of a tile.) We define the function $\texttt{MARK\_IN}(g,d)$ which takes as input an unmarked glue label $g$ and direction $d$, and returns the I/O marked version of $g$ for input, and $\texttt{MARK\_OUT}(g,d)$ similarly but for output marking. Note that all tiles at the edge of macrotiles that begin input paths will have exterior facing glues that were marked using $\texttt{MARK\_IN}$, and those that finish output paths (including those of the seed macrotiles) will have exterior facing glues that were marked using $\texttt{MARK\_OUT}$. In this way, output paths that terminate in one macrotile only ever initiate input paths into another, and since seed macrotiles only have output paths, their tiles can never attach outside of the seed.

\subsubsection{Seed tile types}\label{para:adj-seed-tile-types}

For each tile type $t \in T$, a set of tiles is created to be used to form $T'_\sigma$ as follows. (Examples can be seen in \cref{fig:seed-tiles}.) Note that all glues except those marked using $\texttt{MARK\_OUT}$ are only used by tile types in $T'_\sigma$ and are found on no other tiles in $T'$. They are also all strength-2, even when representing glues that are strength-1 in $T$.

\begin{enumerate}
    \item A \emph{central} tile type $c_t$ will be created that is recognized by $R$ and causes a macrotile with $c_t$ at its center to be mapped to $t$. For each side $d$ of $t$ which has a non-null glue $g$, the $d$ side of $c_t$ is given a glue $g_\sigma$.

    \item For each side $d \in \{N,S\}$ of $t$ which has a non-null glue $g$, create the following tile types if they haven't already been created (which could have happened if glue $g$ was already seen on a side $d \in \{N,S\}$ of another tile type):
    \begin{enumerate}
        \item \emph{vertical} tile type $v_g$, which has glue $g_\sigma$ on its N and S sides

        \item tile type $g_N$, which has glue $g_\sigma$ on its S and glue $\texttt{MARK\_OUT}(g,N)$ on its N

        \item tile type $g_S$, which has glue $g_\sigma$ on its N and glue $\texttt{MARK\_OUT}(g,S)$ on its S

        \item N \emph{connector} tile type $g_{N\_}$, which has glue $g_\sigma$ on its S side and glue $\_$ on its N side

        \item S \emph{connector} tile type $g_{S\_}$, which has glue $g_\sigma$ on its N side and glue $\_$ on its S side
    \end{enumerate}

    \item For each side $d \in \{E,W\}$ of $t$ which has a non-null glue $g$, create the following tile types if they haven't already been created (which could have happened if glue $g$ was already seen on a side $d \in \{E,W\}$ of another tile type):
    \begin{enumerate}
        \item \emph{horizontal} tile type $h_g$, which has glue $g_\sigma$ on its E and W sides

        \item tile type $g_E$, which has glue $g_\sigma$ on its W and glue $\texttt{MARK\_OUT}(g,E)$ on its E

        \item tile type $g_W$, which has glue $g_\sigma$ on its E and glue $\texttt{MARK\_OUT}(g,W)$ on its W

        \item E \emph{connector} tile type $g_{E\_}$, which has glue $g_\sigma$ on its W side and glue $\_$ on its E side

        \item W \emph{connector} tile type $g_{W\_}$, which has glue $g_\sigma$ on its E side and glue $\_$ on its W side
    \end{enumerate}
\end{enumerate}

Finally, the following two tile types are created and added to $T'_\sigma$. The \emph{vertical connector} tile type $v_\_$, which has glue $\_$ on its N and S sides, and the \emph{horizontal connector} tile type $h_\_$, which has glue $\_$ on its E and W sides. (Note that the connector glue $\_$ is unique to tiles in $T'_\sigma$ and is not used in any other tile types in $T'$.)

\begin{figure}
    \centering
    \includegraphics[width=5.5in]{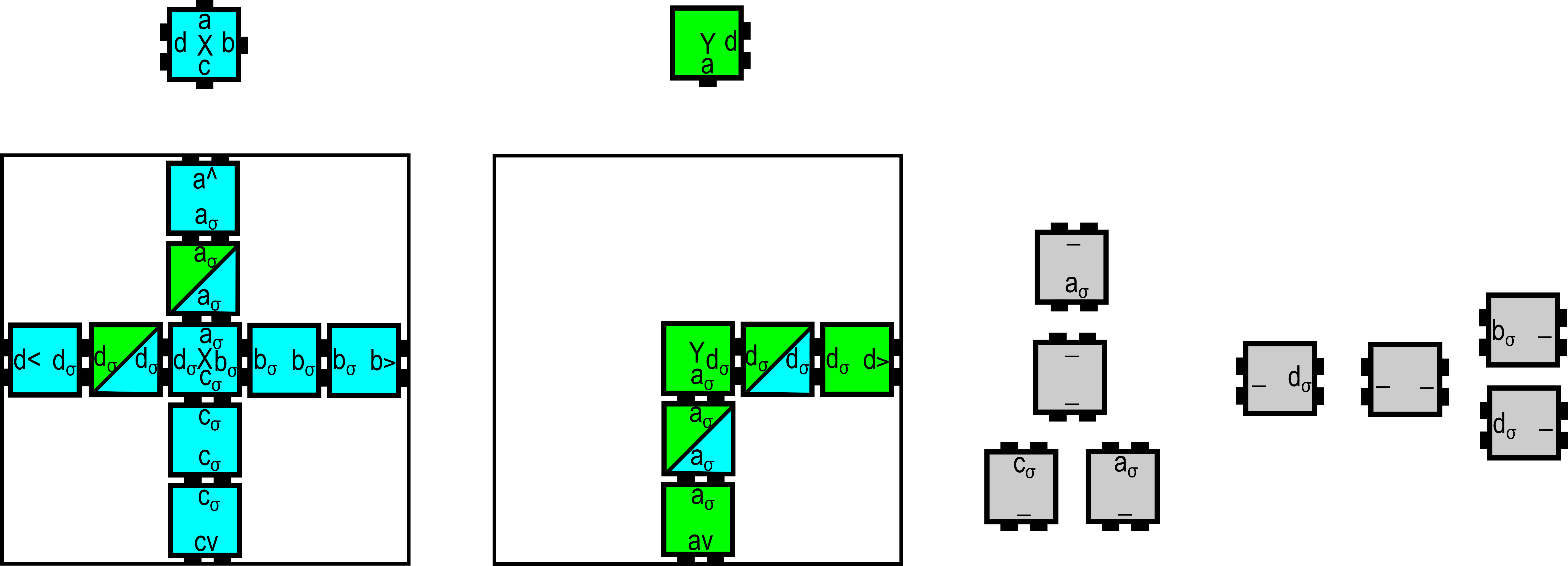}
    \caption{(Top) An example of two tile types $t_1, t_2 \in T$, labeled `X' (aqua) and `Y' (green), respectively. (Bottom) The tile types of $T'_\sigma \subset T'$ created to allow any seed over tiles of $T$ to be created. Below each of $t_1$ and $t_2$ are the tile types created to allow macrotiles representing each to be formed. (Multi-colored tile types are shared by both.) The grey tiles to the right are the tile types created to form connected paths between macrotiles when multi-tile seeds are formed.\label{fig:seed-tiles}}
\end{figure}

\subsubsection{Seed creation} \label{para:seed-creation}

The algorithm which takes as input the seed $\sigma$ in $\calS$ and returns the seed $\sigma'$ for $\calS'$ is the following:

\begin{enumerate}
    \item For each tile, given its tile type $t \in T$ and location $(x,y) \in \mathbb{Z}^2$, such that $(t,(x,y)) \in \sigma$:
    \begin{enumerate}
        \item The locations of the macrotile representing the tile will be $M((x,y)) = \{(x',y')\mid cx \le x' < c(x+1)$ and $cy \le y' < c(y+1)\}$.

        \item Let $m = (c-1)/2$ be the distance from a location on the boundary of a macrotile to the center location. The macrotile's center location $l_c = (cx+m, cy+m)$. A tile of type $c_t$ is placed in location $l_c$.

        \item For each side $d$ of $t$ that has a non-null glue $g$, and that glue does not form a bond with a neighboring tile in $\sigma$:

        \begin{enumerate}
            \item If $d \in \{N,S\}$ then the vertical tile $v_g$ is placed on the $d$ side of the center location, with additional copies of it placed in direction $d$ until stopping before the final location in the macrotile. Else, if $d \in \{E,W\}$ then the horizontal tile $h_g$ is placed similarly for direction $d$.

            \item A tile of type $g_d$ is placed in the final location of direction $d$, at the edge of the macrotile.
        \end{enumerate}

        \item For each side $d$ of $t$ that has a non-null glue, and that glue forms a bond with a neighboring tile in $\sigma$:

        \begin{enumerate}
            \item A tile of the connector type for glue $g$, $g_d\_$ is placed on side $d$ of the center location.

            \item If $d \in \{N,S\}$, copies of the vertical connector tile type $v_\_$ are placed in all remaining locations in the row in direction $d$ to the edge of the macrotile. Else if $d \in \{E,W\}$, copies of the horizontal connector tile type $h_\_$ and placed similarly for direction $d$.
        \end{enumerate}
    \end{enumerate}
\end{enumerate}

Thus, each tile of $\sigma$ is represented by a macrotile of the form shown in \cref{fig:adj-seed}, which consists of a tile in the central POC and a path of tiles to each of the macrotile boundaries on sides where glues are exposed or form bonds within the seed, ending in the center of the I/O submacrotiles. For each direction in which a tile in $\sigma$ has an exposed non-null glue, that side has a strength-2 glue exposed toward the neighboring macrotile that will initiate growth of a corresponding input path into that macrotile. For each direction in which a tile of $\sigma$ forms a bond with a neighboring tile in $\sigma$, there is a path connecting the corresponding macrotiles in $\sigma'$, so $\sigma'$ is a stable assembly that correctly represents $\sigma$ under $R$, using scale factor $c$, and will initiate valid simulation.

Additionally, given any arbitrary $\tau$-stable seed $\sigma$ over $T$, the tile types in $T'_\sigma \subset T'$ are sufficient to allow for the creation of $\sigma'$ such that $R^*(\sigma') = \sigma$. The algorithm given above is clearly computable, so $(T',2)$ seed-independently simulates $(T,\tau)$ under representation function $R$ and scale factor $c$.

\begin{figure}
    \centering
    \includegraphics[width=4.2in]{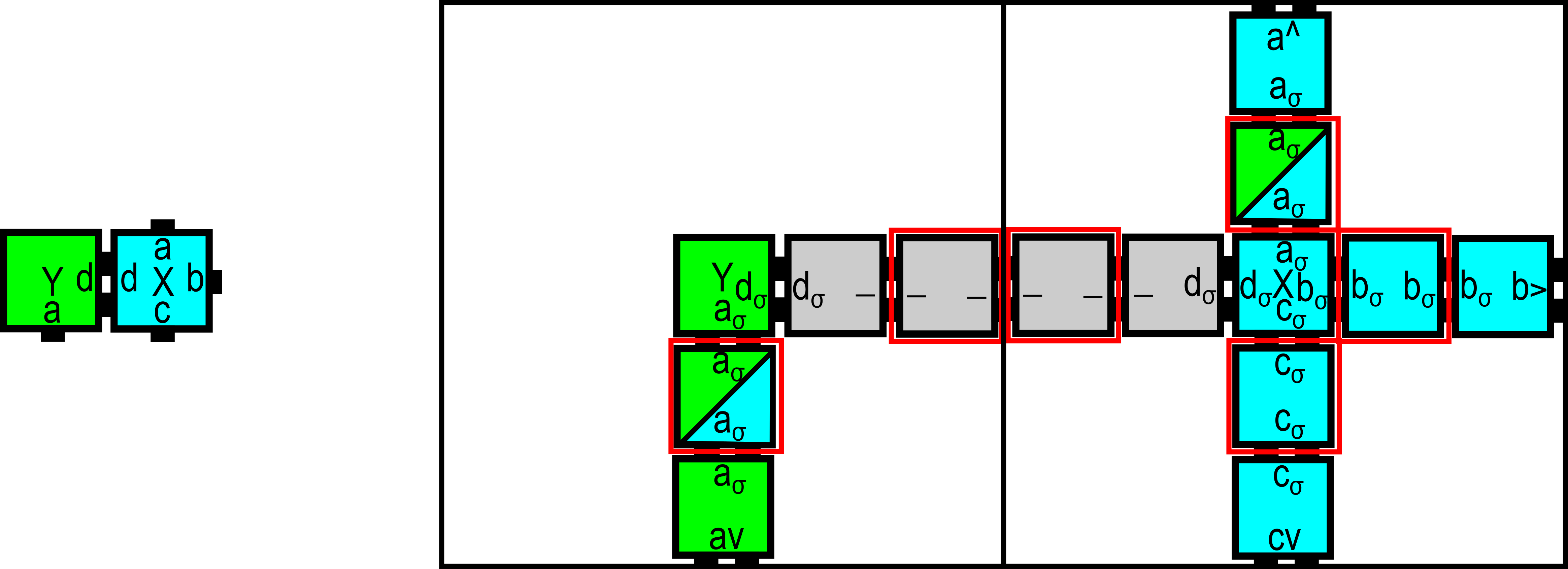}
    \caption{(Left) An example seed $\sigma$ of $\calS$ containing two tiles. (Right) The macrotiles in $\calS'$ that form $\sigma'$, where $R(\sigma') = \sigma$ at spatial scale factor $5$ (which is unrealistically small but sufficient to display the design for any scale factor, since the only change would be the repetition of tiles of the types outlined in red).\label{fig:adj-seed}}
\end{figure}

\subsection{Correctness of the construction for \cref{lem:adj-only}}

Here we prove how the construction for \cref{lem:adj-only} proves the lemma. Recall that we begin with an arbitrary syncTAM system $\calS = (T, \sigma, \tau)$. Throughout \cref{sec:adj-proof} we have discussed the construction of system $\calS' = (T', \sigma', 2)$ and how it self-assembles. We now review why $\calS'$:

\begin{enumerate}
    \item Synchronously simulates $\calS$

    \item Has spatial and temporal scale factors that are $O(|T|)$

    \item Is locally adjacent-only

    \item Is staggered

    \item Is seed-independent
\end{enumerate}

We will prove that under $R$ and with spatial scale factor $c$ and temporal scale factor $s$, $\calS'$ synchronously simulates $\calS$ by showing that (1) $\calS$ synchronously follows $\calS'$, and (2) $\calS'$ models $\calS$, and that these imply that $\calS$ and $\calS'$ have equivalent productions. %We start from the seed of $\calS$, $\sigma$, and the seed of $\calS'$, $\sigma'$. As shown in \cref{sec:adj-seed}, $R(\sigma') = \sigma$.

%showing that for any arbitrary assembly sequence $\vec{\alpha}$ of $\calS$, there is a valid assembly sequence $\vec{\alpha'}$ of $\calS'$ that synchronously simulates it at spatial scale factor $c$ and temporal scale factor $s$ under representation function $R$, and no valid assembly sequence in $\calS'$ creates an invalid assembly. Let $i$ refer to an index of $\vec{\alpha}$ and $j$ an index of $\vec{\alpha'}$, with $\alpha_i$ being the $i$th assembly in $\vec{\alpha}$ and $\alpha'_j$ the $j$th assembly in $\vec{\alpha'}$. The induction hypothesis is that, for $j \mod s == 0$ and $i = j/s$, if $R(\alpha'_j) = \alpha_i$, then there exists $\vec{\alpha'}$ such that at step $j+s$ of $\vec{\alpha'}$, $R(\alpha'_{j+s}) = \alpha_{i+1}$. Furthermore, there exists no valid assembly sequence $\vec{\alpha''}$ in $\calS$ such that 

%The base case is for $j = i = 0$, i.e. $\alpha_i = \sigma$ and $\alpha'_j = \sigma'$. As shown in \cref{sec:adj-seed}, $R(\alpha'_j) = \alpha$.

First, we show that $\calS$ synchronously follows $\calS'$ under $R$ with temporal scale factor $s$. Let $\vec{\beta} = \beta_0, \beta_1, \beta_2, \ldots$ be an arbitrary assembly sequence in $\calS'$. We will prove the existence of a valid assembly sequence $\vec{\alpha} = \alpha_0, \alpha_1, \alpha_2, \ldots$ in $\calS$, where for each $i'$ such that $i' \mod s == 0$, $R^*(i') = \alpha_i$ using induction. Our inductive hypothesis is that for arbitrary $i'$ such that $i' \mod s == 0$, $\alpha_i = R^*(\beta_{i'})$. For the base case, by construction (as shown in Section \ref{sec:adj-seed}), for $\beta_0 = \sigma'$, $R^*(\beta_0) = \sigma = \alpha_0$. 
For the inductive step, we begin with arbitrary $i'$ such that $i' \mod s == 0$ and $\beta_{i'}$ such that $R^*(\beta_{i'}) = \alpha_i$. As explained previously, the subsequent sequence of $s$ assembly steps in $\calS'$ will be the $(i+1)$th macrostep in $\calS'$ (where $i = i' / s$) and simulate the $(i+1)$th step in $\calS$. This is because at the beginning of each macrostep in $\calS'$, for any macrotile location $M(l)$ that had an output glue that appeared adjacent to $M(l)$ during the previous macrostep, all such output glues will arrive as input glues to the decision submacrotile of $M(l)$, one at a time, until the first such set that is sufficient to allow for simulation of the attachment of a tile of some type $t \in T$. At this point the macrotile will grow so that it resolves to a tile of type $t$ (by placing a corresponding tile in the central POC) and the output glue representations will grow to the outside edges of all output sides of $M(l)$ using output delays so that they all complete during the same assembly step. 
This will complete at exactly the final assembly step of macrostep $i+1$ (as shown in Section \ref{sec:IO}). Thus for $i'' = (i+1) \cdot s$, $\alpha_{i+1} = R^*(\beta_{i''})$, which must be a valid assembly such that $\alpha_i \rightarrow^\calS_1 \alpha_{i+1}$ because if and only if a macrotile in $\calS'$ changed from mapping to empty space to mapping to a tile in $T$, a tile of the same type could have attached in the corresponding location in $\calS$. 
Therefore, $\vec{\alpha}$ is a valid assembly sequence in $\calS$ and $\calS$ synchronously follows $\calS'$ under $R$ with temporal scale factor $s$. 

We now show that $\calS'$ models $\calS$ under $R$. We do this by showing how, for every producible assembly $\alpha \in \prodasm{S}$, there exists a set of producible assemblies $\Pi \subset \prodasm{S'}$ such that each maps under $R$ to $\alpha$ and for every assembly $\beta$ producible from $\alpha$ in $\calS$ in one step, that (1) every assembly in $\Pi$ can grow into an assembly in $\calS'$ representing any such $\beta$, and (2) every assembly sequence in $\calS'$ from some $\alpha'$ that represents $\alpha$ to some $\beta'$ that represents $\beta$ includes an assembly in $\Pi$.

For this we also use induction. 
Our induction hypothesis is that, given arbitrary producible assembly $\alpha_i \in \prodasm{S}$ there exists a set of producible assemblies $\Pi_i \subset \prodasm{S'}$ that satisfies the requirements stated above. The base case is again the seed assemblies $\sigma = \alpha_0$ in $\calS$ and $\sigma'$ in $\calS'$. In this case, the only element of $\Pi_0$ is $\sigma'$. By construction (as shown in Section \ref{sec:adj-seed}), $R^*(\sigma') = \sigma = \alpha_0$. 
Trivially, since $\sigma'$ is the only element of $\Pi_0$, the second condition holds (an assembly from $\Pi_0$ is included in every assembly sequence from $\sigma'$). 
For the first condition (i.e. for each assembly that $\alpha_0$ can grow into in one step, the assembly in $\Pi_0$ can grow into at least one assembly that represents it) we can simply recall how the construction works. For every empty location $l$ in $\calS$ toward which a non-null glue of a seed tile is facing, input growth will enter the corresponding macrotile location $M(l)$ in $\calS'$. If there are multiple such glues adjacent to a location, for every possible ordering of those glues there is an assembly sequence in $\calS'$ that selects random numbers for the input delays so that the glues arrive at the decision submacrotile in that order. As soon as the first set arrives that are sufficient to simulate placement of a tile of type $t \in T$, the macrotile will resolve to such a tile and then begin growth of outputs that are delayed so that they will all complete growth to the output sides of the macrotile at the exact same step, which in this case will be the $s$th step (where $s$ is the temporal scale factor of the simulation). Because of this, for every location $l$, for every tile type $t$ that could be placed there at this step in $\calS$, there is an assembly sequence beginning from the sole assembly $\sigma'$ contained in $\Pi_0$ that causes the corresponding macrotile to resolve to that tile type, and therefore the assembly $\sigma' \in \Pi_0$ can grow into a representation of anything that $\sigma$ can grow into and the base case is shown.

For the induction step, given arbitrary producible assembly $\alpha_i \in \prodasm{S}$ and its corresponding set $\Pi_i$, we will show that for any assembly $\alpha_{i+1}$ such that $\alpha_i \rightarrow^1_\calS \alpha_{i+1}$, there exists set $\Pi_{i+1}$ that satisfies the necessary conditions. Recall that the construction is designed so that every step $i$ of $\calS$ is simulated by $s$ steps in $\calS'$, with each such macrostep in $\calS'$ completing at exactly step $i \cdot s$, with every macrotile receiving its final output tiles at that step. 
We will define each $\Pi_i \subset \prodasm{S'}$ to include the assemblies producible during the final assembly step of every macrostep.
Let $\beta_i \in \Pi_i$, and thus $R^*(\beta_i) = \alpha_i$. 
Let $0 < j \le s$ be the assembly step in $\calS'$ within macrostep $i$. 
Let $k$ be the number of tiles added to $\alpha_i$ during step $i+1$ to create this $\alpha_{i+1}$. 
As previously shown, the corresponding $k$ macrotile locations in $\calS'$ will grow to resolve to tiles of $\calS$ during macrostep $i$, with there being at least one valid assembly sequence that leads to macrotiles resolving in $\calS'$ to the same tiles placed in $\calS$ to form $\alpha_{i+1}$.
The very first step of macrostep $i$ will consist of the first input tiles attaching into each input side of macrotiles that received neighbors with output glues facing them during the previous macrostep. Those input tiles are the ones that nondeterministically select one of the $12$ input delay values that are specific for their sides and thus determine the ordering of growth for the inputs into each macrotile location during that macrostep. 
Also as previously shown, for any tile type $t \in T$ that could appear in a location in $\calS$ during step $i$, the corresponding macrotile location in $\calS'$ could grow into a macrotile representing a tile of type $t$. 
As the $s$ steps of macrostep $i$ occur, until $k$ new macrotiles resolve, the assembly $\beta_{i+j}$ (for $0 \le j < s$) will be equivalent to $\alpha_i$, i.e. $R^*(\beta_{i+j}) = \alpha_i$. 
If the $k$th new macrotile resolves during the $j'$th step of the macrostep, then for each $j''$ where $j' \le j'' \le s$, $R^*(\beta_{i+j''}) = \alpha_{i+1}$. 
Any such assembly will now deterministically grow into an assembly that is the final for macrostep $i$, and these and that final assembly of macrostep $i$ will be in $\Pi_{i+1}$. 
This will satisfy the requirements for the set $\Pi_{i+1}$ because (1) the final assembly created by macrostep $i$ and representing $\alpha_{i+1}$ is guaranteed, for any $\alpha_{i+2}$ such that $\alpha_{i+1} \rightarrow^1_\calS \alpha_{i+2}$, to be able to grow into some assembly $\beta_{i+2} \in \prodasm{S}$ such that $R^*(\beta_{i+2}) = \alpha_{i+2}$, and (2) as soon as any assembly resolves under $R$ to $\alpha_{i+1}$ it is in $\Pi_{i+1}$, which means that all assembly sequences that resolve to assemblies $\alpha_i$ and $\alpha_{i+1}$ include an assembly from $\Pi_{i+1}$. Thus, the induction step is shown and $\calS'$ is shown to model $\calS$.

$\calS$ and $\calS'$ have equivalent productions because $\calS$ follows $\calS'$ and $\calS'$ models $\calS$, which means that, starting from the seed assemblies $\sigma$ in $\calS$, and $\sigma'$ in $\calS'$, as assembly $\alpha$ is producible (resp. terminal) in $\calS$ if and only if an assembly $\beta$ is producible (resp. terminal) in $\calS'$ such that $\beta$ maps under $R$ to $\alpha$ or an equivalent assembly. Thus, $\calS'$ is shown to synchronously simulate $\calS$ under $R$ and using temporal scale factor $s$.

In \cref{sec:IO} and \cref{sec:decision-sub}, it was shown that the scale factors for the IO submacrotiles and decision submacrotiles are $O(|T|)$, and since the dimensions of each macrotile are the sums of those of two IO submacrotiles and a decision submacrotile, the overall size of a macrotile, and thus the spatial scale factor of simulation, is $O(|T|)$. For the temporal scale factor $s$, we note that it was previously shown that all macrosteps are of the exact same number of steps which is the common length of all paths that enter a macrotile as input and ultimately terminate as output at another macrotile side, since they are all padded to be the same length. (For the following, recall the $w$ is the width and height of the decision submacrotile and that $w = O(|T|)$ as shown in Section \ref{sec:probe-region}, $l_{max}$ is the length of the longest path through the decision submacrotile (whose length will be explained below), and $u_d$ is the unit of delay defined in Section \ref{sec:IO} where it was shown that $u_d$ must be less than $2l_{max}$.) The length of such a path is bounded by $47u_d + 4(48u_d/w + 2) + (w + 4) + l_{max}$, where:
\begin{enumerate}
    \item $47u_d$ accounts for the sum of the input and output delays that don't account for the portion $l$ through the decision submacrotile and its difference from $l_{max}$.

    \item $48u_d/w + 2$ is the full length of each IO submacrotile, and it is crossed once by the input path and then $3$ times by the output path (with $2$ of those due to the output test path).

    \item $w + 4$ bounds the additional growth needed for the output test path.

    \item $l_{max}$ accounts for the length of the path through the decision submacrotile plus the remaining output delay growth.
\end{enumerate}

Note that it was shown that $w = O(|T|)$ and that $u_d \le 2l_{max}$. Thus it remains to show that $l_{max} = O(|T|)$. This holds because the longest path through the decision submacrotile would consist of: (1) the completed first two competition rows (because the third is not on the path to the output) which each grow at the same time from the center to the sides and back to the center, (2) the growth to the center and then of the probes, (3) cooperation between a pair of probes at one of the ends of the probe region, and (4) growth of a probe dodging path through the entire probe region (in a worst case scenario as shown in \cref{fig:probe-dodging}) and then to the output location at the opposite side of the decision submacrotile. Each of these quantities for (1), (2), and (3) are clearly bounded by the dimension $w$ of the decision submacrotile and therefore combined remain $O(|T|)$, and \cref{fig:probe-dodging} shows that (4) is also $O(|T|)$. Therefore, $l_{max} = O(|T|)$, and since $u_d$ must be less than $2l_{max}$, then $u_d = O(|T|)$. This means that the total for the length of a complete path through a macrotile, and thus the number of assembly steps in a macrostep is $O(|T|)$, and the temporal scale factor $s = O(|T|)$.

We show that $\calS'$ is a locally adjacent-only system by inspecting the growth of its macrotiles and first showing that there are no attachment through opposite sides of a frontier location, then that it is locally adjacent-only. Opposite sides of a frontier position can be involved in an attachment because (1) there could be a tile attachment in which a tile could bind on more than two sides when it attaches, or (2) if there could be a tile attachment in which a tile could bind cooperatively using two opposite sides (N and S, or E and W). There are very few locations in a macrotile in which tiles are placed by cooperation (i.e. using more than one side), and simple inspection shows them to all use exactly two sides which are adjacent to each other, never opposite. For example, cooperation occurs when a side's competition rows are blocked in the decision submacrotile, when two probes are aligned to allow for simulation of cooperation across two opposite sides, or when multiple sides grow paths to the central POC and need to cooperate to place a tile there. In every case, the paths of the macrotiles are designed so that the cooperation is between two adjacent sides. Thus, whenever two sides of a frontier position are involved in an attachment of $\calS'$, they are adjacent. Finally, in all such locations and situations of cooperative binding, there is always exactly one pair of sides to which a tile can bind. The only situations where a location may, in different assembly sequences, have tiles attach cooperatively using differing subsets of sides are (1) above and below the tile in between two cooperating probes, and (2) the central POC. However, the subset of sides used for cooperation in both situations is determined by the ordering of arrival of the neighboring tiles, and for each ordering there is only one pair of sides that will expose glues to which a tile can bind, enforcing the property of being locally adjacent-only. Therefore, $\calS'$ is locally adjacent-only.

We show that $\calS'$ is a staggered system also by inspecting the growth of its macrotiles. For $\calS'$ not to be staggered, there would need to be at least one possible empty location $l \in \mathbb{Z}^2$ where, during a single assembly step, more than one neighboring location receives a tile, and each has a non-null glue facing $l$. This can only happen if there are two paths of growth that simultaneously arrive at locations neighboring $l$. Every macrotile grows by first having input delay growth in some subset of its I/O submacrotiles. This initial growth in each I/O submacrotile is completely separate from all others and is formed by a continuous single (zig-zagging) path of tiles, so none has more than one path and thus cannot violate the property of being staggered. For each I/O submacrotile which is able to complete its input delay growth and grow into the decision submacrotile, the input delays are such that the growth from each direction completes before growth from any other can arrive. This means that to violate the property of being staggered in the decision submacrotile, the only possibility would be for two paths that begin from the same side to simulataneously converge on a location $l$. However, the only times a path splits in the decision submacrotile are (1) to grow the final competition row for a side as well as a path toward the center, (2) to grow a row of probes, and (3) to grow the two paths away from a point of cooperation between two probes. In all such cases, the paths strictly grow away from each other and cannot converge on any single location $l$, so growth within the decision submacrotile cannot violate the property of being staggered. Finally, growth out of the decision submacrotile and into an I/O submacrotile only occurs if that side hadn't previously grown into the decision submacrotile. If growth proceeds out of the decision submacrotile into the I/O submacrotile in some direction $d$, there are only two cases to consider: (1) no input growth has begun in side $d$, so the output test growth, then output delay, then final output path all complete growth as a single unbranching path that cannot violate the property of being staggered, or (2) input growth has begun in side $d$. In this case, the input delay is guaranteed to be sufficient for the output test path out of the decision submacrotile to begin growth before any collisions occur. The output test path will grow as a single path that crashes into the first row of the input delay growth and terminate. The input delay growth will eventually complete and result in an input path that grows toward the decision submacrotile but crashes into the initial portion of the output test path. In both cases, no branching paths occur and paths that crash do so in locations where one of them completed many assembly steps prior to the arrival of the second, so there is also no possibility for the property of being staggered to be violated. 
This concludes the analysis of all cases within a macrotile, and since the growth of separate macrotiles is all disjoint, $\calS'$ is a staggered system.

Finally, we show that the simulation of $\calS$ by $\calS'$ is seed-independent. In \cref{sec:adj-seed} it was shown how the tile set $T'_\sigma \subset T'$ is created from the tile set $T$ so that for any seed $\sigma$ in $\calS$, tiles from $T'_\sigma$ can be used to form seed $\sigma'$ such that $R^*(\sigma') = \sigma$ and $\calS' = (T',\sigma',2)$ synchronously simulates $\calS = (T,\sigma,\tau)$ under representation function $R$ with spatial scale factor $c$ and temporal scale factor $s$, where $T'$, $R$, $c$, $s$, and the temperature $2$ for $\calS'$ are all fixed, independent of the seed $\sigma$ in $\calS$. The function that generates the seed $\sigma'$ for $\calS'$ is clearly computable (with an example of its output shown in \cref{fig:adj-seed}). Therefore $(T',2)$ seed-independently simulates $(T,\tau)$.

Thus, \cref{lem:adj-only} is proven.

\begin{figure}
    \centering
    \includegraphics[width=4.0in]{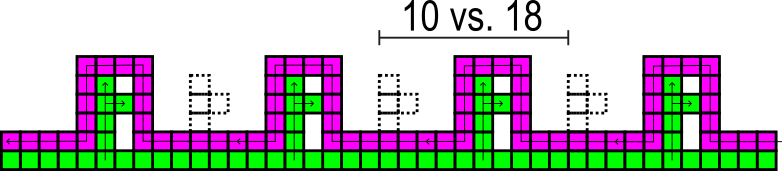}
    \caption{An example of a maximal length path dodging probes. Assuming that the path growing from the E had information about the glue on the S, including that it would grow probes, and it needed to grow to the W across the entire probe region, it would grow along the bottom as much as possible to dodge a possible probe from the N. Each time it must grow over a probe from the S, the path becomes longer. However, if two adjacent probe locations have probes, it doesn't need to grow back down then up between them (since it wouldn't be able to dodge a probe from the N that had already grown between them), a longest possible path occurs when each alternating position grows a probe from the S. In such a case, the path length between the $i$th probe location and the $(i+2)$th, when there is a probe at the $(i+1)$th location, is $18$ (rather than the $10$ it would be if there was no probe in between). Thus, the average length between probe locations is $9$, with $g'_{max}$ probe locations (see Section \ref{sec:probe-region}), for a maximal path length of $9g'_{max} = O(|T|)$.\label{fig:probe-dodging}}
\end{figure}

\end{proof}

\section{Simulating locally adjacent-only systems at temperature 1}

\tempOneSimulationAdj*

\begin{figure}
    \centering
    \input{tikz/shine_on_you/temp1simulation-adj}
    \caption{A visual overview of \cref{lem:temp1simulation-adj}. The squiggle between the $\calS = (T, 2, \sigma)$ TAS and the ``Staggered + locally adjacent-only'' column indicates that $\calS$ needs to be staggered for the simulation to work.\\
    The ``pixel-art allegory of winter'' is \emph{Figure représentant l'hiver} by Federico \textsc{de Vinciolo} from the collections of Musée du Louvre \url{https://collections.louvre.fr/ark:/53355/cl020582721}
    }
    \label{fig:temp1simulation-adj}
\end{figure}
We prove \cref{lem:temp1simulation-adj} by construction, thus, starting with an arbitrary temperature 2 SyncTAM system $\calS = (T, \sigma, 2)$ which is adjacent-only and staggered, we will define $\calS' = (T', \sigma', 1)$ and macrotile representation function $R$ such that $\calS'$ simulates $\calS$ under $R$ at spatial scale factor $c$ and temporal scale factor $s$, where $c, s = O(\log |T|)$. The simulation by $\calS'$ will be seed independent, and $\calS'$ will be staggered. We provide a visual overview of the statement of \cref{lem:temp1simulation-adj} in \cref{fig:temp1simulation-adj}.

For each glue $g \in G_T$ on any tile in $T$, we assign a unique binary string of length $\lceil\log |G_T|\rceil + 1$, with the null glue assigned to the string consisting of all $0$s. Let the binary string which represents $g$ be referred to as $b_g$. We use the functions $\mathtt{num\mathunderscore 1s}(b_g)$ and $\mathtt{num\mathunderscore 0s}(b_g)$ to refer to the number of $1$s, and $0$s, respectively, in $b_g$.

Let $\vec{\alpha}$ be an arbitrary assembly sequence in $\calS$ and let $0 \le i < |\vec{\alpha}|$ be the index of the assembly step in $\vec{\alpha}$ which yields assembly $\alpha_i$. Let $L_i \subset \mathbb{Z}^2$ be the set of locations in $\calS$ where tiles are placed in step $i$. For each $l \in L_i$ let $M(l)$ represent the set of locations in $\calS'$ that map to the location $l$ (i.e. the macrotile representing $l$). During macrostep $i$, $M(l)$ will complete growth that causes it to change from mapping to empty space under $R$ to mapping, under $R$, to a tile of the type that $\calS$ places in $l$. Additionally, the tile placements during the final step of macrostep $i$ will place a tile within, but adjacent to the boundary of the macrotile $M(l)$ at each side for which an \textit{output glue} needs to be represented (That is, in each direction in which the tile in $\calS$ in location $l$ exposed a non-null glue that wasn't used in its initial binding as an \textit{input glue},  but which is adjacent to an empty location and may be used as an input glue during the binding of a neighboring tile in a future step.). As with \cref{sec:adj-proof}, we use the functions $\mathtt{MARK\mathunderscore IN}(g,d)$ and $\mathtt{MARK\mathunderscore OUT}(g,d)$ for $g \in G_T, d \in \{N,E,S,W\}$ to refer the the input, and output marked versions of the glue $g$.

Then, during the first step of macrostep $i + 1$, each of those output glues will initiate growth into those neighboring macrotile locations. That is, for each $l \in L_i$, for each direction $d$, if the tile added into location $l$ in $\calS$ during step $i$ placed an output glue in direction $d$, then let $l'$ refer to the location on the \texttt{OPP}$(d)$ side of $l$. During macrostep $i+1$, in $\calS'$ growth will enter side \texttt{OPP}$(d)$ of macrotile $M(l')$. The first tile placed in $M(l')$ during the first step of macrostep $i+1$  will be in the center of the side \texttt{OPP}$(d)$.
\begin{figure}
    \centering
    \includegraphics[width=0.75\linewidth]{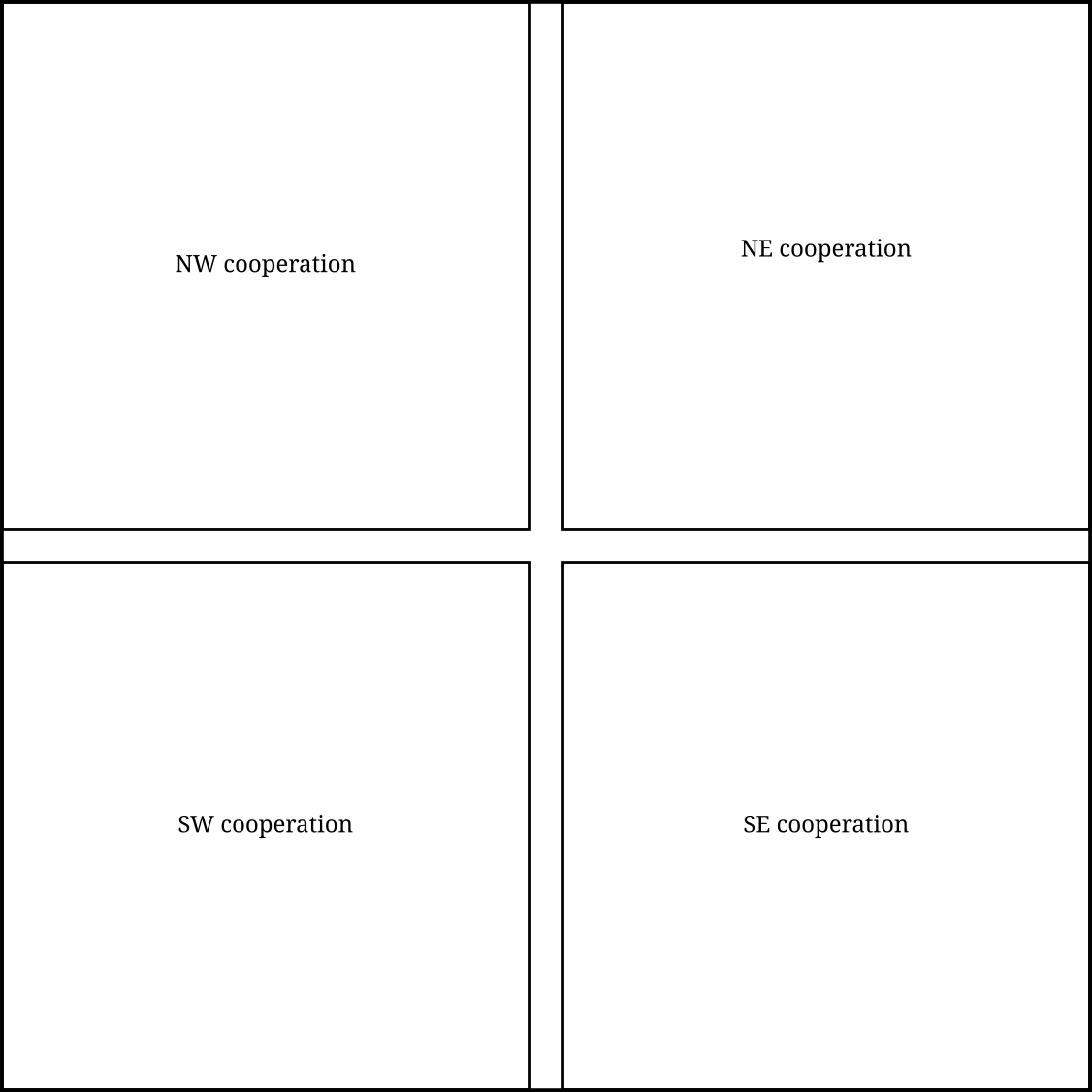}
    \caption{The regions of a macrotile for the construction of the proof of \cref{lem:temp1simulation-adj}. The four corner regions are labeled, and each facilitate cooperation between a set of two adjacent sides. The center, cross-shaped region is called the \textit{decision region}, and facilitates the resolution of the macrotile to a specific tile type, as well as the output and input of glue values to/from adjacent macrotile locations.}
    \label{fig:temp1simulator-macrotile-layout}
\end{figure}

\begin{figure}
    \centering
    \includegraphics[width=0.75\linewidth]{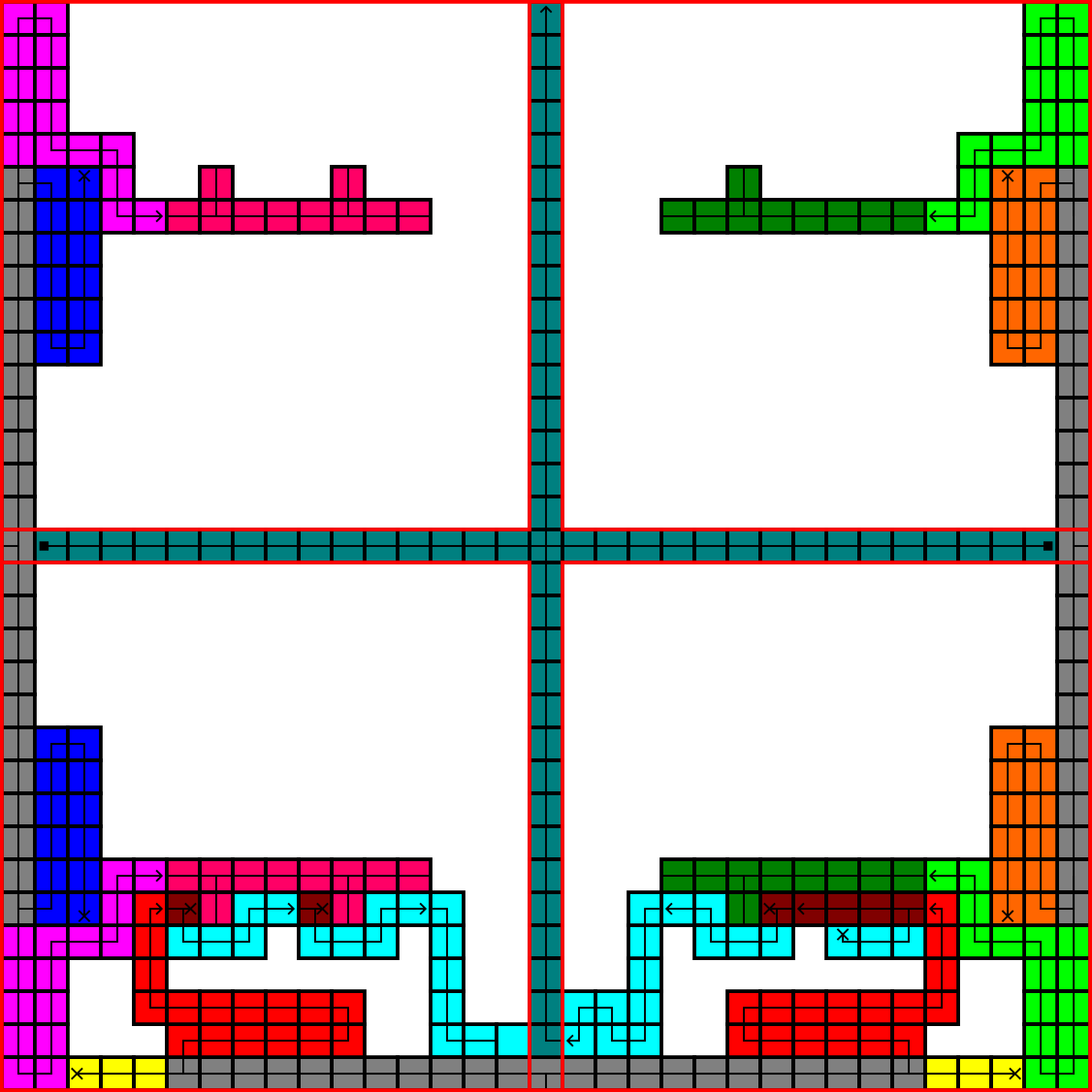}
    \caption{An example macrotile for the construction of the proof of \cref{lem:temp1simulation-adj} which receives strength-1 input from the adjacent macrotiles on its $E$, $S$ and $W$ sides. In this instance, the $E$ and $S$ glues cooperate, and initiate the growth of output to adjacent macrotile locations. The glues presented to the $S$ and $W$ sides of the macrotile  fail to cooperate for lack of a suitable tile type in $T$, and thus do not present a glue to the central decision region. The $E$ and $W$ inputs arrived before the $S$ input and each claimed the two corners next to them (in fuchsia and green respectively), which allows the corner-claiming gadgets on $S$ side of the macrotile (yellow) to start simulating cooperation (red tiles). Each macrotile region is enclosed with a red border.}
    \label{fig:temp1simulator-example-macrotile}
\end{figure}

Each macrotile is logically divided into five unique regions, which are shown in Figure \ref{fig:temp1simulator-macrotile-layout}.  The regions in each corner of the macrotile are called \textit{cooperation regions}, as they contain the locations in which paths growing from adjacent resolved macrotiles may interact and read each other's respective glue values. Conversely, the cross-shaped region in the center of the macrotile is called the \textit{decision region}, and may informally be seen as a ``highway'' throughout the macrotile, as its center tile is the location which denotes the tile type $t \in T$ that $M(l)$ resolves to under $R$. From the aforementioned center tile grows output paths which provide input to adjacent macrotiles. An example macrotile is shown in \cref{fig:temp1simulator-example-macrotile}, with each region outlined in red.

In order to ensure that the completion of growth of the macrotile $M(l)$ is synchronized across all $l \in L_i$, each macrotile $M(l)$ will contain space to allow for the growth of paths that \textit{delay} for carefully determined amounts of time. This behavior is described in \cref{sec:temp1simulation-synchronization}, however we first describe some reusable gadgets, as well as the growth procedure of the macrotile as to provide context for how/when this delay is determined.

\subsection{Commonly used gadgets} \label{sec:temp1simulator-common-gadgets}

\begin{figure*}[t!]
    \centering
    \begin{subfigure}[t]{0.25\textwidth}
        \centering
        \includegraphics[width=0.5\linewidth]{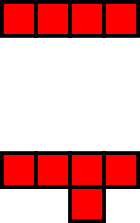}
        \caption{}
        \label{fig:temp1simulator-bitreader}
    \end{subfigure}%
    ~ 
    \begin{subfigure}[t]{0.25\textwidth}
        \centering
        \includegraphics[width=0.5\linewidth]{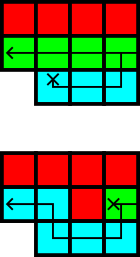}
        \caption{}
        \label{fig:temp1simulator-bitwriter}
    \end{subfigure}
    \caption{(a) A bit-writer gadget expressing a $0$, and a $1$, indicated by the presence or lack thereof of a single-tile bump. (b) A bit-reader gadget reading a 0, and a 1, respectively.}
    \label{fig:temp1simulator-bitreaderwriter}
\end{figure*}

\begin{figure*}[t!]
    \centering
    \begin{subfigure}[t]{0.5\textwidth}
        \centering
        \includegraphics[width=0.5\linewidth]{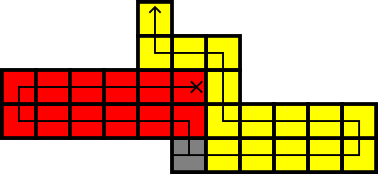}
        \caption{}
    \end{subfigure}%
    ~ 
    \begin{subfigure}[t]{0.5\textwidth}
        \centering
        \includegraphics[width=0.5\linewidth]{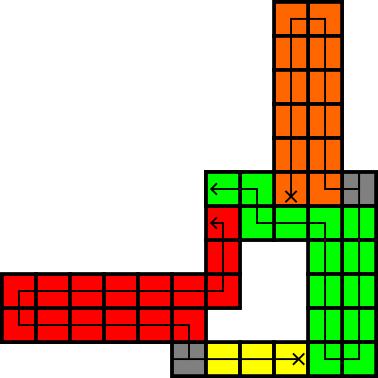}
        \caption{}
    \end{subfigure}
    \caption{A corner-claimer gadget which attempts to claim the $SE$ corner of a macrotile location. (a)   The shorter path, colored yellow, successfully claims it, blocks the red path from continuing on, and completes growth. (b) The shorter path is blocked by the presence of a corner-claimer gadget which grew from the $E$ side of the macrotile. The yellow path crashes into the already present corner-claimer, and the red path is allowed to complete its growth.}
    \label{fig:temp1simulation-cornerclaimer}
\end{figure*}

In order to simplify the description of $\calS'$, we will first describe a set of gadgets which will be used throughout our construction. Namely, these gadgets are as follows: \textit{bit-reader}, \textit{bit-writer}, and \textit{corner-claimer}.

We will first describe the operation of the complementary bit-reader and bit-writer gadgets, whose designs are borrowed from \cite{syncTAM}, and are shown in \cref{fig:temp1simulator-bitreaderwriter}. $\calS'$ is a non-cooperative system, and thus information between two paths of tiles must be communicated through geometry. A bit-writer gadget is able to encode a single bit of information through the presence (or lack thereof) of a single-tile bump. When connected to each other in series, a set of bit-writers may represent an arbitrarily long bit string. In concert, a bit-reader gadget detects the presence of the aforementioned bump by growing two parallel paths of connected tiles. One of these paths is exactly $2$ tiles longer than the other, however the shorter path is blocked when the aforementioned bump is present. As a result of the synchronous nature of the syncTAM, when there is nothing blocking the shorter path of tiles (i.e. when there is no bump placed by a previous bit-writer) the longer path is blocked, and thus the victorious shorter path encodes within its west-most tile the fact that it has read a $0$. Conversely, if the shorter path is blocked by the presence of a single-tile bump placed by a previous bit-reader, the longer path is able to complete growth, and the west-most tile of the bit-reader will encode the fact that it has read a $1$.

Bit-reader gadgets may be chained together in series to read arbitrarily long binary strings. However, it must be noted that after the first bit of the binary string, there must be tile types incorporated into the tile set which correspond to two bit-reader gadgets in order to read the second position of the binary string (for a total of $3$). For a tile set to have the capability of reading, and encoding into the last tile of the chain of bit-readers a binary string of length $n \in \mathbb{N}$, it must incorporate tile types corresponding to $2^{n}-1$ unique bit-reader gadgets.

Using the same notion of computation through geometry, corner-claimer gadgets, shown in \cref{fig:temp1simulation-cornerclaimer} appear on the corners of a macrotile location, and are used to detect the presence (or lack thereof) of previous input growth on adjacent sides of the macrotile locations. Corner-claimer gadgets consist of two paths of tiles which grow synchronously, one of length 11, and one of length 13, both of which race to reach a point of competition. The longer path only wins if the shorter one is blocked by the presence of the shorter path of another rotated/reflected corner-claimer.

We note that bit-writer, and corner-claimer gadgets are clearly staggered, as during their growth there is never a single location $l \in \mathbb{Z}^2$ such that the two paths which constitute the gadget ever face $l$, and present a glue during the same assembly step.

\subsection{Macrotile growth procedure} \label{sec:temp1simulation-growth-procedure}

As before, let $\vec{\alpha}$ be an arbitrary assembly sequence in $\calS$. Let $l \in \mathbb{Z}^2$ be a location in $\calS$ that is initially empty (i.e. not in the seed $\sigma$). If $\texttt{IN}(l,i) \neq \emptyset$, then each side $d \in \texttt{IN}(l,i)$ is a (potential) input side, as that neighboring glue may contribute binding to a tile placed in location $l$ in some step $j > i$. Note that for the following discussion we will refer to all such glues as input glues, even though some of them may not actually contribute to binding of a tile in $l$, since it is possible that: (1) location $l$ may permanently remain empty in the case where there are not neighboring glues with enough strength to allow for a tile attachment, or (2) a subset of neighboring glues not including all (or any) of these may serve as the input glues for a later tile attachment, with some (or all) of these glues remaining unbound as mismatched glues. However, the growth within a macrotile that combines information about neighboring glues in order to determine which, if any, tile type attachment should be simulated is essentially the same until a possible attachment type is selected, so we'll refer to all such sides as input sides for ease of discussion. For the context of this construction, we are able to guarantee that $|\texttt{IN}(l,i)|=1$ as a result of the fact that the system that we are simulating, $\calS$, is staggered, and thus at macrostep $i$, for each empty location $l \in \mathbb{Z}^2$, at most one neighboring location $l'$ resolves to a tile and contains a non-null glue abutting $l$. Additionally, for each $d \in \{N,E,S,W\}$, let $\texttt{ADJ}(d) \subset \{N,E,S,W\}$ refer to the sides of the macrotile which are adjacent to side $d$. (For example, the sides adjacent to $N$ are $E$ and $W$.)

Given location $l \not \in \dom{\sigma}$, let $M(l)$ be the macrotile in $\calS'$ that maps to location $l$, and $0 \le i < |\vec{\alpha}|$ be some assembly step of $\vec{\alpha}$. 
For all macrosteps such that $\texttt{IN}(l,i-1) = \emptyset$, no tiles will be placed in $M(l)$. We will now describe the growth procedure which occurs in $M(l)$ in $\calS'$ each macrostep such that $\texttt{IN}(l,i-1) \neq \emptyset$, wherein a tile with the glue $g \in G_T$ abuts the decision region of $M(l)$ on its $d$ side. For the context of this section, we will assume that $g$ is a strength-1 glue and defer the discussion of simulating strength-2 glues to \cref{sec:temp1simulation-strength2}. Note that in this description we will reference the $d d_{adj}$ cooperation region (where $d_{adj} \in \texttt{ADJ}(d)$), whose name corresponds to a region shown in \cref{fig:temp1simulator-macrotile-layout}.

Growth within $M(l)$ begins with the attachment of a tile to glue $\mathtt{MARK\mathunderscore IN}(g, d)$ in the center-most location of side $d$. This tile contains glues on sides $\texttt{ADJ}(d)$. From each side $d_{adj} \in \texttt{ADJ}(d)$, the following series of tile attachments is initiated: 
\begin{enumerate}
        \item A single-tile wide path of tiles of length $w_{sync} + 4(\lceil \log |G_T| \rceil + 1)+1 $ in direction $d_{adj}$, where $w_{sync} \in \mathbb{N}$ is the length necessary to perform the \textit{delay} phase of the macrotile growth procedure, which is discussed in \cref{sec:temp1simulation-synchronization}. The last tile of this path contains a glue on side $d_{adj}$. \label{step:temp1simulation-1}
        
        \item To the aforementioned glue attaches a corner-claimer gadged, rotated / reflected such that its shorter path grows towards the $dd_{adj}$ corner of $M(l)$. \label{step:temp1simulation-2}
        
        \item If the shorter path of the corner-claimer gadget ``wins'', then a corner-claimer gadget must \textit{not} have grown into the $dd_{adj}$ corner of $M(l)$ during a previous macrostep. Therefore, it is necessary to write the value of $g$ such that if $d_{adj} \in \texttt{IN}(l,j)$ (for some $j > i$), then the information regarding glue $g$ on side $d$ of the macrotile may be read: \label{step:temp1simulation-3}
        
        \begin{enumerate}
            \item Onto the $\texttt{OPP}(d)$ side of the last tile placed by the shorter path of the corner-claimer gadget attaches a set of $\lceil \log |G_T| \rceil + 1$ bit-writer gadgets connected in series, each of which encode a single bit in the string $b_g$.\label{step:temp1simulation-2a}
            \item After the last bit-writer has completed growth, no more tiles attach in the $d d_{adj}$ cooperation region during macrostep $i$. \label{step:temp1simulation-2b}
        \end{enumerate}
        
        \item If the longer path of the corner-claimer gadget ``wins'', a corner-claimer gadget must have been placed during a previous macrostep in the $dd_{adj}$ corner of $M(l)$, whose shorter path won. As such, there is a set of bit-writer gadgets present which encode $b_{g'}$ (where $g'$ is the glue which abuts side $d_{adj}$ of $l$ in $\alpha_i$).
        \label{step:temp1simulation-4}
        
        \item To the $\texttt{OPP}(d_{adj})$ side of last tile of the longer path of the corner-claimer gadget attaches $\lceil \log |G_T| \rceil + 1$ bit-reader gadgets (rotated / reflected as is necessary) connected in series such that they read $b_{g'}$. 
        \label{step:temp1simulation-5}

        \item To the $\texttt{OPP}(d_{adj})$ side of the last tile to be placed by the $(\lceil \log |G_T| \rceil + 1)$th bit-reader attaches a single-tile wide path of 4 tiles in the $\texttt{OPP}(d)$ direction, to the $\texttt{OPP}(d_{adj})$ side of which attaches a single-tile wide path which grows in the $\texttt{OPP}(d_{adj})$ direction for $w_{sync}$ tiles, abutting the decision region. During the growth of this path of tiles, the delay growth phase is performed, which is described in \cref{sec:temp1simulation-synchronization}. 

        The tile which abuts the decision region presents a glue on its $d_{adj}$ side iff there exists a tile type $t \in T$ which may attach through the cooperation of $g$ and $g'$, otherwise that side will contain the null glue. Note that even though both paths of tiles which grow through the cooperation regions associated with each adjacent direction in $\texttt{ADJ}(d)$ will synchronously abut the same location in the decision region, only one of them will have read a glue for which a tile type may make an attachment through cooperation with it and $g$. This is because $\calS$ is locally-adjacent-only. 
        \label{step:temp1simulation-6}
        
        \item To the aforementioned glue facing the decision region attaches a path of tiles of length $w_{sync}+4(\lceil\log|G_T|+1)+5$, whose last tile abuts the center of the decision region (which is the center of $M(l)$), and presents the glue $\mathtt{MARK\mathunderscore OUT}(g_{g,g'}, d)$.  \label{step:temp1simulation-7}

        \item To said glue, in the center of the decision region attaches the \textit{decision tile} $t'_{g,g'}$ that corresponds to the tile $t \in T$ which, in $\calS$, may attach to $\alpha_{i}$ through the cooperation of glues $g$ and $g'$ on its $d$ and $d_{adj}$ sides. If there are multiple such tile types (i.e. if $\calS$ is nondeterministic), then there are correspondingly multiple decision tiles that may attach in the central square. 
        
        Each side $d_{res} \in \{N,E,S,W\}-\{d\}$ of the decision tile contains a glue which encodes $g_{t,d_{res}}$. And the notion of a decision tile is discussed further in \cref{sec:temp1simulation-representation-function}
        \label{step:temp1simulation-8}.

        \item From each side $d_{res} \in \{N,E,S,W\}-\{d\}$ of the decision tile grows a single-tile wide path of length $w_{sync} + 4(\lceil \log |G_T|\rceil + 1) + 6$ in the $d_{res}$ direction whose last tile abuts the $d_{res}$ side of $M(l)$. If $g_{t,d_{res}}$ is not the null glue, then the $d_{res}$ side of the last tile of this path presents the glue $\mathtt{MARK\mathunderscore OUT}(g_{t,d_{res}}, d_{res})$.\label{step:temp1simulation-9}

    \end{enumerate}

    During future discussion of the macrotile growth procedure, we will refer to the path(s) of tiles that constitute items \ref{step:temp1simulation-1}-$\ref{step:temp1simulation-7}$ as \textit{input paths}, as they attach to an input glue from an adjacent macrotile which resolved to a tile type in $T$ during the previous macrostep, and we will refer to the paths of tiles that constitute item \ref{step:temp1simulation-8} as \textit{output paths}, as they grow to the edge of the macrotile location, and to them may attach input paths in adjacent macrotile locations (if previously attached geometry allows it).

\subsection{Synchronization and scale Factor}\label{sec:temp1simulation-synchronization}

In this section, we will prove that the temporal scale factor $s$ of $\calS'$ is $8\lfloor \frac{\lceil \log |G_T|\rceil}{4}\rfloor + 20(\lceil \log |G_T| \rceil + 1) + 
2(\lceil \log |G_T| \rceil) + 30$, and that the spatial scale factor $c$ of $\calS'$ is $16 \lfloor \frac{\lceil \log |G_T|\rceil}{4}\rfloor + 8\lceil \log |G_T|\rceil+13$. In doing so, we will find an appropriate value of $w_{sync}$, which is selected to ensure the paths grown during the macrotile growth procedure provide the space to perform the necessary attachments such that for each macrostep $i$ of $\vec{\alpha}$, each location that resolves to a tile of type $T$ under $R$ presents output to adjacent macrotiles at step $is + s$.

Before being able to discuss our method of ensuring synchronization, we must first carefully document the number of assembly steps which occur during each portion of the macrotile growth procedure. In order to do so, we have split the macrotile growth procedure into a set of arbitrarily chosen \textit{growth phases}. These growth phases are cataloged in \cref{tab:temp1simulator-synchronization}, to which \cref{fig:temp1simulator-synchronization-coloring} is intended to be used as a reference. The exact number of assembly steps necessary to complete each growth phase is provided, however for the sake of brevity we leave the demonstration of the provided values as exercises for the reader.

Let us now discuss the synchronization method used to ensure $\calS'$ has temporal scale factor $s$. Let $l \in \mathbb{Z}^2$ such that, during assembly step $i-1$, $\texttt{IN}(l,i-1) \neq \emptyset$, and $d \in \texttt{IN}(l,i-1)$ is the side of $M(l)$ which receives an input glue as a result of the resolution of an adjacent macrotile location in macrostep $i$ (recall that $|\texttt{IN}(l,i-1)| \le 1$). As with \cref{sec:temp1simulation-synchronization}, we defer the discussion of synchronization when $g$ is strength-2 to \cref{sec:temp1simulation-strength2}, and note that when $\texttt{IN}(l,i-1) = \emptyset$, no growth occurs in $M(l)$ during macrostep $i$.. 

We will start by determining the total number of assembly steps necessary to complete the macrotile growth procedure. We have marked each growth phase that is completed by the connected path of tiles between the input glue $g$ (which is presented on side $d$ of $M(l)$ at assembly step $is$), and each output glue on sides $\{N,E,S,W\}-\{d\}$ of $M(l)$ (assuming that input growth from adjacent sides does not block the growth of output paths to the edge of the macrotile) without concerning ourselves with the \textit{delay} growth phase. That is to say, let us find the total number of assembly steps necessary for $M(l)$ to resolve, and express output without ensuring that the number of assembly steps within each macrotile which received input glue during macrostep $i$ is equal. We have labeled all applicable growth phases with a ``$\star$'' in \cref{tab:temp1simulator-synchronization}, and the total number of assembly steps necessary to complete them is $s_{unsync} = 4w_{sync} + 20(\lceil \log |G_T|\rceil + 1) + 2 \mathtt{num\mathunderscore1s}(b_{g'})+30$. Notice that for all macrotile locations that receive an input glue $g$ during macrostep $i$, the only way $s_{unsync}$ differs between each of them is a result of the number of $1$s in $b_{g'}$ (due to the fact that a bit-reader requires $2$ more assembly steps to read a $1$ than a $0$). We will refer to the necessary amount of delay to ensure that growth within a macrotile which reads the adjacent input glue $g'$  as $s_{delay, g'} = 2(\mathtt{num\mathunderscore0s}(b_{g'}) - \texttt{min}_{g'' \in G_T}(\mathtt{num\mathunderscore0s}(b_{g''}) )$, which is $2$ times the difference of number of $0$s between $b_{g'}$ and $b_{g''}$, which is the glue $g'' \in G_T$ such that $b_{g''}$ has the \textit{least} $0$s.

In order to incorporate $s_{delay}$ ``delay'' assembly steps into the macrotile growth procedure, we simply make the following modification to item \ref{step:temp1simulation-6} of the growth procedure: Modify the single-tile wide path that grows in the direction $\texttt{OPP}(d_{adj})$, which we will call the \textit{delay path} for the context of this discussion, such that there are sufficient ``bumps'' of height $\le 4$ which each delay the path from reaching the decision region by 2, 4, 6, or 8 assembly steps (for bumps heights 1, 2, 3, and 4, respectively) whose summed lengths are equal to $s_{delay, g'}$. The maximum height that bumps can reach is chosen arbitrarily such that the uppermost tile of the delay path is parallel to the chain of bit-writer gadgets in the $dd_{adj}$ cooperation region. Additionally, the height of each individual bump is not important as it is only necessary that the summation of every bump's height is $\frac{s_{delay, g'}}{2}$. 

Using this technique, the largest delay which may be incurred is dependent on the width of the delay path, $w_{sync}$.
The maximum amount of delay that may be handled by a delay path of length $w_{sync}$ is $s_{max} = 4(2\lfloor\frac{w_{sync}}{2}\rfloor)$. In order to ensure that $w_{sync}$ is large enough to provide sufficient delay when reading an arbitrary bit-string corresponding to a glue in $G_T$, we must ensure that $8 \lfloor \frac{w_{sync}}{2}\rfloor) \ge \texttt{max}_{g'' \in G_T}(s_{delay, g''})$, which is equivalent to: 

\[8\lfloor \frac{w_{sync}}{2}\rfloor \ge 2(\texttt{max}_{g'' \in G_T}(\mathtt{num\mathunderscore0s}(b_{g''}) - \texttt{min}_{g'''\in G_T}(\mathtt{num\mathunderscore 0s}(b_{g'''})))\]

Note that the most $0$s that any bit string corresponding to a glue $G_T$ may have is $\lceil \log |G_T|\rceil + 1$ (which corresponds to the null glue), and all bit strings corresponding to glues in $G_T$ must have at least 1 $0$s, as each is padded with a single $0$, which gives us : 

\[8\lfloor \frac{w_{sync}}{2}\rfloor \ge 2(\lceil \log |G_T|\rceil+1 - 1)\]

Solving for $w_{sync}$, we get:

\[w_{sync} \in \{2\lfloor \frac{\lceil \log |G_T| \rceil}{4}\rfloor, 2 \lfloor \frac{\lceil \log |G_T|\rceil}{4}\rfloor + 1\}\]

Either value of $w_{sync}$ presented above is sufficient to allow for the necessary amount of delay bumps to grow, and thus we set $w_{sync} = 2\lfloor \frac{\lceil \log |G_T|\rceil}{4}\rfloor$. With the incorporation of the delay growth phase, the number of assembly steps which occurs within a macrotile $M(l)$ that resolves to a tile type $t \in T$ during macrostep $i$ through cooperation between two strength-1 glues, $g$ and $g'$, is $s_{unsync} + s_{delay, g'} = 4w_{sync}+20(\lceil \log |G_T|\rceil + 1) + 2\lceil \log |G_T|\rceil + 30$. Therefore, the temporal scale factor of  $\calS'$ when simulating a cooperative attachment is $s = s_{sync}+s_{unsync, g'} = 8\lfloor \frac{\lceil \log |G_T| \rceil}{4}\rfloor + 20(\lceil \log |G_T|\rceil + 1) + 2 \lceil \log |G_T| \rceil + 30 = O(\log |T|)$. 

Additionally, we note that (as per \cref{sec:temp1simulation-growth-procedure}), the length of each output path which grows from the center of a macrotile location to its edge is $w_{sync} + 4(\lceil \log |G_T|\rceil + 1) + 6$, and thus the spatial scale factor of each macrotile location is $c = 16 \lfloor \frac{\lceil \log |G_T|\rceil}{4}\rfloor + 8\lceil \log |G_T|\rceil+13 = O(\log |T|)$ (since output paths must be able to grow out of two opposite sides of the decision tile).

\begin{figure}
    \centering
    \includegraphics[width=0.75\linewidth]{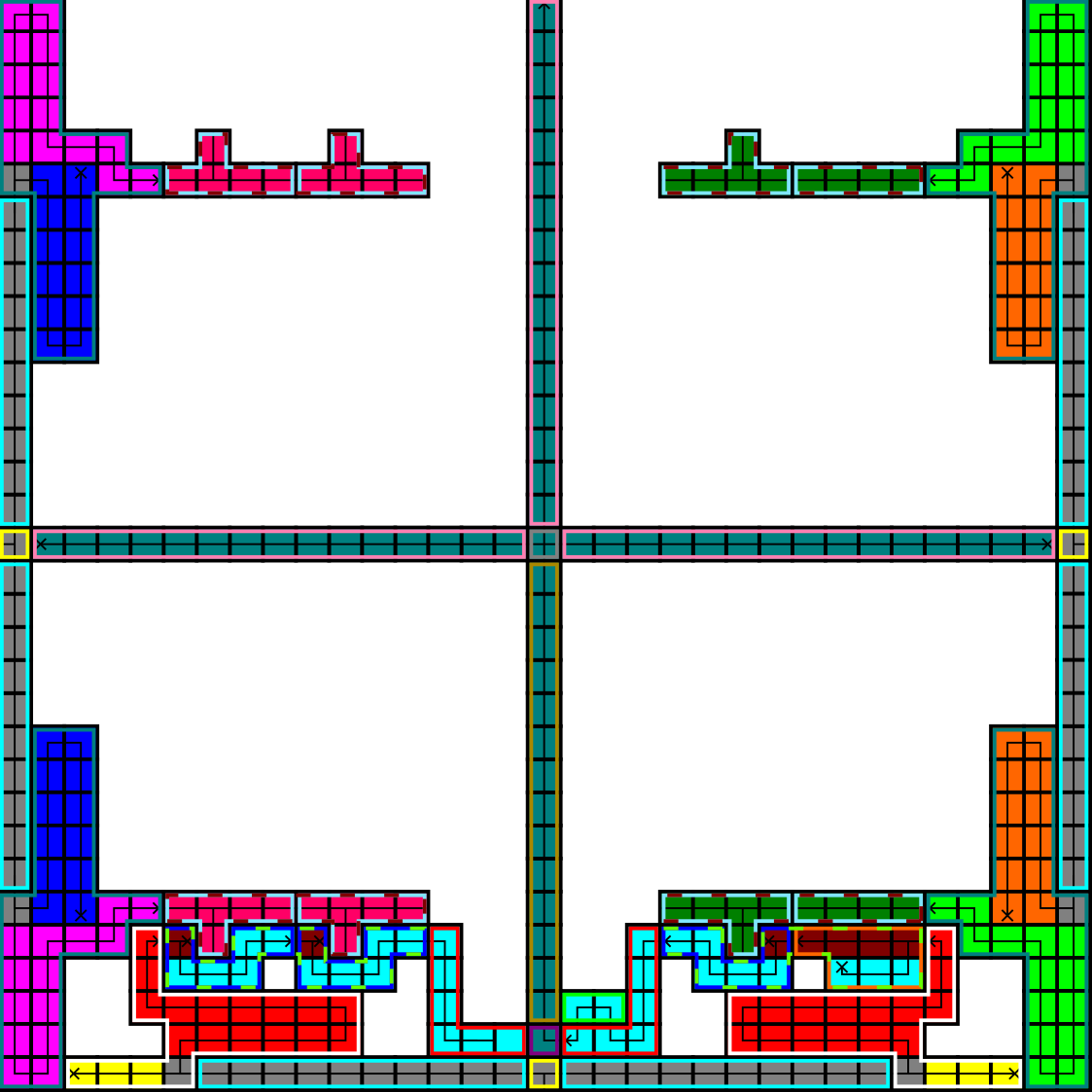}
    \caption{The example macrotile shown in \cref{fig:temp1simulator-example-macrotile} with additional outlines added to indicate the growth phase in which they attach into the assembly. This figure is intended to serve as a reference for \cref{tab:temp1simulator-synchronization}. Portions of the macrotile which map to multiple growth phases presented in \cref{tab:temp1simulator-synchronization} are represented using dashed lines.}
    \label{fig:temp1simulator-synchronization-coloring}
\end{figure}

\begin{table}[!htp]\centering\small
\resizebox{\textwidth}{!}{
\begin{tabular}{lrrrr}\toprule
Growth Phase &Number of assembly steps &Steps in \cref{sec:temp1simulation-growth-procedure} &Color in \cref{fig:temp1simulator-synchronization-coloring} \\\midrule
Initial tile attachment in $M(l)$ $\star$ &$1$ &Before step \ref{step:temp1simulation-1} &Yellow \\
Path of tiles from the decision region, to the corner-claimer gadget $\star$ &$w_{sync} + 4(\lceil \log |G_T| \rceil + 1)$ &Step \ref{step:temp1simulation-1} &Cyan \\
Corner-claimer gadget (shortest path wins and str$(g) == 1)$ &$15$ &Steps \ref{step:temp1simulation-2} and \ref{step:temp1simulation-3} &Navy \\
Corner-claimer gadget (shortest path wins and str$(g) == 2)$ &$14$ &\cref{sec:temp1simulation-strength2} &Shown in \cref{fig:temp1simulation-cornerclaimer-strength2} \\
Corner-claimer gadget (longest path wins and str$(g) == 1)$ $\star$ &$16$ &Steps \ref{step:temp1simulation-2} and \ref{step:temp1simulation-4} &White \\
Bit writer gadget &$4$ &\cref{sec:temp1simulator-common-gadgets} &Shown in \cref{fig:temp1simulator-bitreader} \\
Chain of bit-writer gadgets which write $b_g$ $\star$ &$4(\lceil \log |G_T| \rceil + 1)$ &Steps \ref{step:temp1simulation-2a} and \ref{step:temp1simulation-2b} &Dark Red \\
Bit-reader gadget (longest path wins) &$6$ &\cref{sec:temp1simulator-common-gadgets} &Shown in \cref{fig:temp1simulator-bitreader} \\
Bit-reader gadget (shortest path wins) &$4$ &\cref{sec:temp1simulator-common-gadgets} &Shown in \cref{fig:temp1simulator-bitreader} \\
Chain of bit-reader gadgets which read $b_{g'}$ $\star$ &$4(\lceil \log |G_T| \rceil + 1) + 2(\mathtt{num\mathunderscore1s}(b_g))$ &Step \ref{step:temp1simulation-5} &Light Green (Dashed) \\
Path from last bit-reader gadget towards decision region $\star$ &$w_{sync} + 4$ &Step \ref{step:temp1simulation-6} &Red \\
Delay growth &$2(\mathtt{num\mathunderscore0s}(b_g))$ &Step \ref{step:temp1simulation-6} &Light Green (Solid) \\
Placement of first tile in decision region $\star$ &$1$ &Step \ref{step:temp1simulation-7} &Purple \\
Growth towards center of $M(l)$ $\star$ &$w_{sync}+ 4(\lceil \log |G_T| \rceil + 1) + 3$ &Step \ref{step:temp1simulation-7} &Gold \\
Placement of the decision tile &$1$ &Step \ref{step:temp1simulation-8} &Gray \\
Growth to edges of $M(l)$ $\star$ &$w_{sync} + 4(\lceil \log |G_T| \rceil + 1) + 5$ &Step \ref{step:temp1simulation-9} &Pink \\
\bottomrule
\end{tabular}
}
\caption{Table which depicts the number of steps required to complete each possible growth phase which may occur within $M(l)$ during some macrostep such that $M(l)$ receives input growth. For each growth phase, a reference to the related steps as is described in \cref{sec:temp1simulation-growth-procedure} (where applicable) is provided, along with a key (where applicable) to the reference provided in \cref{fig:temp1simulator-synchronization-coloring}. Some growth phases are marked with ``$\star$'' to facilitate easier reference in \cref{sec:temp1simulation-growth-procedure}}\label{tab:temp1simulator-synchronization}
\end{table}

\subsection{Strength 2 glues}\label{sec:temp1simulation-strength2}
\begin{figure}
    \centering
    \includegraphics[width=0.5\linewidth]{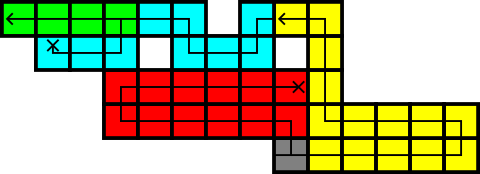}
    \caption{The modified behavior of a corner-claimer gadget, and subsequently attaching bit-readers that assembles when a strength-2 input path successfully claims a corner.}
    \label{fig:temp1simulation-cornerclaimer-strength2}
\end{figure}
In this section, we discuss the growth procedure of a macrotile which simulates a tile in $\calS$ that attaches using a strength-2 glue. As with \cref{sec:temp1simulation-growth-procedure}, for any assembly step $i$ of $\vec{\alpha}$, let $l \in \mathbb{Z}^2$ be a location in $\calS$ such that $\texttt{IN}(l,i-1) \neq \emptyset$, $d \in \texttt{IN}(l,i-1)$ is the direction of macrotile $M(l)$ which receives an input glue in step $i$ of $\vec{\alpha}$, and we note that $|\texttt{IN}(l,i-1)| = 1$ as a result of $\calS$ being staggered.

The growth procedure within $M(l)$ is identical to what is described in \cref{sec:temp1simulation-growth-procedure} up until the attachment of a corner-claimer gadget. Unlike when $g$ is a strength-1 glue, it is not necessary that the corner of the macrotile be claimed, as there is no need to read a glue from the $d_{adj}$ (for $d_{adj} \in \texttt{ADJ}(d)$) side of the macrotile in order to determine if $l$ resolves to a tile type in $T$. However, for the sake of ensuring a consistent temporal scale factor $s$, we still attach a corner-claimer gadget, which behaves as follows, depending on the presence of another corner claimer gadget on side $d_{adj}$ of $M(l)$.

If the longer path wins (i.e. there is already a corner-claimer gadget present on side $d_{adj}$), then to the last tile of the corner-claimer gadget attaches $\lceil \log |G_T| \rceil+1$ bit-readers, and growth towards the decision region occurs as normal. However, regardless of the value of $g'$ which was observed, there will only be a glue pointing towards the decision region if there is both a tile $t \in T$ which may attach using glue $g$ on its $d$ side, and the cooperation region that the tile pointing towards the decision region is in is to the ``right'' of the side of $M(l)$ which was presented a glue (when $d = S$, this is the $SE$ decision region, when $d=E$, this is the $NE$ decision region, etc.).

Otherwise, if the shorter path wins (i.e. there is no corner-claimer gadget present on side $d_{adj}$, and the corner is successfully claimed), then unlike when $g$ is a strength-1 glue, to the last tile of the corner-claimer gadget \textit{still} attaches $\lceil \log |G_T| \rceil + 1$ bit-readers connected in series (which all read $0$s), and growth towards the decision region occurs similarly to the growth procedure in the case that the corner is already claimed. Note that the shorter path of the corner-claimer gadget will win prior to the attachment of the aforementioned bit-readers, and entrance into the portion of the growth procedure which grows bit-readers will happen two assembly steps ``early''. However, if each bit-reader were to be unmodified, then growth through the bit-reading portion of the assembly will be unexpectedly fast (since reading a $0$ requires exactly two fewer assembly steps than reading a $1$), and there will be exactly 2 more steps of delay which will be necessary to incur during growth back to the decision region that cannot always be made up for during said growth. As such, the first bit-reader in the chain of $\lceil \log |G_T| \rceil + 1$ bit readers is modified such that it do not contain a `$0$' path, which incurs the two steps of delay that were made necessary by the shorter assembly path taken during growth of the corner-claimer gadget. At this point, during the delay portion of the assembly (which is explained in \cref{sec:temp1simulation-synchronization}), the amount of delay that is imparted is equivalent to that which occurs when the null glue (i.e. the glue whose binary representation is all $0s$) is read. 
Additionally, note that if a future corner-claimer gadget grows into the $dd_{adj}$ corner, which only occurs in the extraneous case wherein there is a tile in $\calS$ which contains a strength-2 glue to which no tiles may bind, then there will be no bit-writers in the $d d_{adj}$ cooperation region which correspond to glue $g$ on side $d$, and the null glue will consequentially be read.

The only other change to the macrotile growth procedure which is necessary to correctly simulate strength-2 glues is the modification of the last glue of the path described in item \ref{step:temp1simulation-7} to instead be the glue $\mathtt{MARK\mathunderscore OUT}(g_g,d)$.

\subsection{Seed structure}\label{sec:temp1simulation-seed}

The seed $\sigma'$ of $\calS'$ is composed of a special set of tile types $T'_{\sigma} \subset T'$. Here we explain the process of creating the tile types in $T'_{\sigma}$ and the function that outputs the seed $\sigma'$ for any arbitrary $\sigma$ of $\calS = (T, \sigma, \tau)$, so that $R^*(\sigma') = \sigma$, and $\calS' = (T', \sigma', 1)$ correctly simulates $\calS$ using fixed representation function $R$, spatial scale factor $c$, and temporal scale factor $s$ independent of $\sigma$. This section is structured in a similar manner to that of \cref{sec:adj-seed}, with details deferred when relevant .

\subsubsection{I/O marking}

As is described in \cref{para:adj-io-marking}, it is necessary to ensure the following two behaviors of $\calS'$: 

\begin{enumerate}
\item When a glue is exposed on the exterior of a macrotile to propagate output, the tile that attaches to it initiates input growth into the neighboring macrotile location.
\item The tiles of $T'_{\sigma}$ must not attach in any locations outside of the pre-formed macrotiles representing the seeds.
\end{enumerate}

We ensure these behaviors through the usage of I/O marking. That is to say, all tiles which face the edge of a macrotile that finish output paths have exterior facing glues that are marked using $\mathtt{MARK\mathunderscore OUT}$, and all tiles which attach to the edge of said output paths to begin input paths do so with glues marked using $\mathtt{MARK\mathunderscore IN}$. We defer further detail regarding I/O marking to \cref{para:adj-io-marking}.

\subsubsection{Creation of seed tile types and assembly}

For each tile type $t \in T$ a set of tiles is created to form $T'_{\sigma}$, which are generated using a nearly identical process as to that which is shown in \cref{para:adj-seed-tile-types}. The only difference between the seed tile types generation function described in \cref{para:adj-seed-tile-types}, and that which is used in this section is that each glue is of strength-1, as opposed to strength-2.

As for the creation of the seed $\sigma'$, the algorithm that takes as input the seed $\sigma$ in $\calS$ and returns the seed $\sigma'$ for $\calS'$ is identical to that which is described in \cref{para:seed-creation}. Thus, each tile of $\sigma$ is represented by a macrotile which consists of a tile in its center-most square of the decision region, and a path of tiles to each of the macrotile boundaries on sides where glues are exposed or form bonds within the seed. For each direction in which a tile in $\sigma$ has an exposed non-null glue, that side has a glue exposed towards the neighboring macrotile that will initiate growth of a corresponding input path into that macrotile. For each direction in which a tile of $\sigma$ forms a bond with a neighboring tile on $\sigma$, there is a path connecting the corresponding macrotiles in $\sigma'$, so $\sigma'$ is a stable assembly that correctly represents $\sigma$ under $R$, using scale factor $c$, and will initiate valid simulation.

\subsection{Representation function}\label{sec:temp1simulation-representation-function}

In this section, we define the representation function $R$ which maps macrotiles of $\calS'$ to tile types of $\calS$. 

For each tile type $t \in T$, there is a corresponding set of decision tile types $T'_{R,t} \subset T'$ which is used by $R$ to determine if a macrotile location resolves to $t$. $T'_{R,t}$ consists of (1) the central tile type $c_t \in T'_\sigma$, and (2) a set of \textit{decision} tiles.

The decision tiles corresponding to tile type $t$ are created using the following algorithm for each side $d \in \{N,E,S,W\}$ : 

\begin{enumerate}
\item If str$(g_{t,d})==1$, then for each $d_{adj} \in \texttt{ADJ}(d)$ : 
\begin{enumerate}
\item If str$(g_{t, d_{adj}}) == 1$, then add the tile type $t'_{g_{t,d},g_{t, d_{adj}}}$ to $T'_{R,t}$ which has the glue $\mathtt{MARK\mathunderscore IN}(g_{g,g'}, d)$ on its $d$ side, and glues corresponding to $g_{t,d'}$ on each side $d' \in \{N,E,S,W\}-\{d\}$ (including if $g_{t,d'}$ is a null glue).
\end{enumerate}

\item Otherwise, if str$(g_{t,d})==2$, then add the tile type $t'_{g_{t,d}}$ to $T'$ which has the glue $\mathtt{MARK\mathunderscore IN}(g_{g}, d)$on its $d$ side, and glues corresponding to $g_{t,d'}$ on each side $d' \in \{N,E,S,W\}-\{d\}$ (including if $g_{t,d'}$ is a null glue).
\end{enumerate}

For all producible assemblies of $\calS$, the tile types in $T'_{R,t}$ only appear in the center-most square of the decision region of a macrotile. As such $R$ operates by taking as input a macrotile, and if there is a tile type in its center-most location which is in $T'_{R,t}$, then it maps said macrotile to be the corresponding tile type $T$.

Note that under the previously discussed macrotile growth procedures, all macrotiles which resolve to a tile type in $T$ during a given macrostep do so simultaneously. That is to say, for any macrostep $i$, and any location $l \in \mathbb{Z}^2$ such that $l$ resolves to a tile of type $t \in T$ during macrostep $i$, $l$ will map under $R$ to $t$ at \textit{exactly} assembly step $i \cdot s + 6 \lfloor \frac{\lceil \log |G_T|\rceil}{4}\rfloor + 16(\lceil \log|G_T|\rceil + 1)+2\lceil\log|G_T|\rceil+25$ (by \cref{tab:temp1simulator-synchronization}).

\subsection{Correctness of the construction for \cref{lem:temp1simulation-adj}}

Here we present our argument for the correctness of our construction for \cref{lem:temp1simulation-adj}. Recall that we begin with a system $\calS = (T, \sigma, 2)$ that is staggered and locally adjacent-only, and throughout Sections \ref{sec:temp1simulation-growth-procedure}, \ref{sec:temp1simulation-strength2}, and \ref{sec:temp1simulation-synchronization} \ref{sec:temp1simulation-seed} we've discussed the construction of the system $\calS' = (T', \sigma', 1)$ and how it self-assembles. In our argument for the correctness of our construction, we will show why $\calS'$ : 

\begin{enumerate}
\item Synchronously simulates $\calS$
\item Has spatial and temporal scale factors that are $O(\log |T|)$
\item Is staggered
\item Is seed-independent
\end{enumerate}

Let us start by showing that under $R$ and with spatial scale factor $c$ and temporal scale factor $s$, $\calS'$ synchronously simulates $\calS$ by showing that (1) $\calS$ synchronously follows $\calS'$, (2) $\calS'$ models $\calS$,  and (3) these imply that $\calS$ and $\calS'$ have equivalent productions. 

First, we show that $\calS$ synchronously follows $\calS'$ under $R$ with temporal scale factor $s$. Let $\vec{\beta} = \beta_0, \beta_1, \beta_2,\ldots$ be an arbitrary assembly sequence in $\calS'$. We will prove the existence of a valid assembly sequence $\vec{\alpha} = \alpha_0, \alpha_1, \alpha_2,\ldots$ in $\calS$, where for each $i'$ such that $i' \mod s == 0$, $R^*(\beta_{i'}) = \alpha_i$ using induction. Our inductive hypothesis is that for arbitrary $i'$ such that $i' \mod s == 0$, $\alpha_i = R^*(\beta_{i'})$. For the base case, by construction (as shown in \cref{sec:temp1simulation-seed}), for $\beta_0 = \sigma'$, $R^*(\beta_0) = \sigma = \alpha_0$. For the inductive step, we begin with arbitrary $i'$ such that $i' \mod s ==0$ and $\beta_{i'}$ such that $R^*(\beta_{i'}) = \alpha_i$. 
As explained previously, the subsequent sequence of $s$ assembly steps in $\calS'$ will be the $(i+1)$th macrostep in $\calS'$ (where $i = i'/s$) and simulate the $(i+1)$th step in $\calS$. This is because at the beginning of each macrostep in $\calS'$, for any macrotile $M(l)$ which had an input glue appear adjacent to $M(l)$ during the previous macrostep, that input glue will have arrived adjacent to $M(l)$ at time step $is$. If a set of input glues have been presented to the macrotile $M(l)$ which are sufficient to allow for the simulation of the attachment of a tile type $t\in T$ (the most recent of which arriving at time step $is$), then within the macrotile will occur growth sufficient to resolve $M(l)$ to a tile of type $t$ (by placing a corresponding tile in the center location of the decision region), and output glue representations will grow to the outside edges of all output sides of $M(l)$ at exactly time step $(i+1)s$ (As shown in \cref{sec:temp1simulation-synchronization}). Thus for $i'' = (i+1) \cdot s, \alpha_{i+1}=R^*(\beta_{i''})$, which must be a valid assembly such that $\alpha_i \rightarrow^\calS_1 \alpha_{i+1}$ because if and only if a macrotile in $\calS'$ changed from mapping to empty space to mapping to a tile in $T$, a tile of the same type could have attached in the corresponding location in $\calS$. Therefore, $\vec{\alpha}$ is a valid assembly sequence in $\calS$ and $\calS$ synchronously follows $\calS'$ under $R$ with temporal scale factor $s$.

We now show that $\calS'$ models $\calS$ under $R$. We do this by showing how, for every producible assembly $\alpha \in \prodasm{S}$, there exists a set of producible assemblies $\Pi \subset \prodasm{S'}$ such that each maps under $R$ to $\alpha$ and for every assembly $\beta$ producible from $\alpha$ in $\calS$ in one step, that (1) every assembly in $\Pi$ can grow into an assembly in $\calS'$ representing any such $\beta$, and (2) every assembly sequence in $\calS'$ from some $\alpha'$ that represents $\alpha$ to some $\beta'$ that represents $\beta$ includes an assembly in $\Pi$.

For this we also use induction. Our induction hypothesis is that, given an arbitrary producible assembly $\alpha_i \in \prodasm{S}$ there exists a set of producible assemblies $\Pi_i \subset \prodasm{S'}$ that satisfies the requirements stated above. The base case is again the seed assemblies $\sigma = \alpha_0$ in $\calS$ and $\sigma'$ in $\calS'$. In this case, the only element in $\Pi_0$ is $\sigma'$. By construction (as shown in \cref{sec:temp1simulation-seed}), $R^*(\sigma') = \sigma = \alpha_0$. Trivially, since $\sigma'$ is the only element in $\Pi_0$, the second condition holds (an assembly from $\Pi_0$ is included in every assembly sequence from $\sigma'$). For the first condition (i.e. for each assembly that $\alpha_0$ can grow into in one step, the assembly in $\Pi_0$ can grow into at least one assembly that represents it) we can simply recall how the construction works. For every empty location $l$ in $\calS$ toward which a non-null glue of a seed tile is facing, input growth will enter the corresponding macrotile location $M(l)$ in $\calS'$. $\calS$ is staggered, and thus there is never more than one of such glues facing a location. If the glue facing a location is sufficient to simulate placement of a tile of type $t \in T$, the macrotile will resolve to such a tile and then begin growth of outputs, each of which abut an output glue to the edges of the macrotile on the $s$th assembly step (where $s$ is the temporal scale factor of the simulation). Because of this, for every location $l$, for every tile type $t$ that could be placed there at this step in $\calS$, there is an assembly sequence beginning from the sole assembly $\sigma'$ contained in $\Pi_0$ that causes the corresponding macrotile to resolve to that tile type, and therefore the assembly $\sigma' \in \Pi_0$ can grow into a representation of anything that $\sigma$ can grow into and the base case is shown.

For the induction step, given arbitrary producible assembly $\alpha_i \in \prodasm{S}$ and its corresponding set $\Pi_i$, we will show that for any assembly $\alpha_{i+1}$ such that $\alpha_i \rightarrow^1_{\calS}\alpha_{i+1}$, there exists a set $\Pi_{i+1} $ that satisfies the necessary conditions. Recall that the construction is designed so that every step $i$ of $\calS$ is simulated by $s$ steps in $\calS'$, with each macrostep in $\calS'$ completing at exactly step $i\cdot s$, with every macrotile receiving its final output tiles at that step. We will define each $\Pi_s \subset \prodasm{S'}$ to include the assemblies producible during the final assembly step of each macrostep. Let $\beta \in \Pi_i$, and thus $R^*(\beta_i) = \alpha_i$. Let $0 < j \le s$ be the assembly step in $\calS'$ within macrostep $i$. Let $k$ be the number of tiles added in $\alpha_i$ during step $i+1$ to create $\alpha_{i+1}$. As previously shown, the corresponding $k$ macrotile locations in $\calS'$ will grow to resolve to tiles of $\calS$ during macrostep $i$, with there being at least one valid assembly sequence that leads to macrotiles resolving in $\calS'$ to the same tiles placed in $\calS$ to form $\alpha_{i+1}$. Additionally, as previously shown, for any tile of type $t \in T$ that could appear in a location in $\calS$ during step $i$, the corresponding macrotile location in $\calS'$ could grow into a macrotile representing a tile of type $t$. Lastly, as was previously shown, for the first $r = 6 \lfloor \frac{\lceil \log |G_T|\rceil}{4}\rfloor + 16(\lceil \log|G_T|\rceil + 1)+2\lceil\log|G_T|\rceil+25$ steps of macrostep $i$, the assembly $\beta_{i+j}$ (for $0 \le j < r)$ will be equivalent to $\alpha_i$, i.e. $R^*(\beta_{i+j})=\alpha_i$. On the $r$th step of macrostep $i$, all $k$ macrotiles resolve to a tile type in $T$, and for all subsequent steps $r \le r' \le s$, $R^*(\beta_{i+r'})=\alpha_{i+1}$. Any such assembly will now deterministically grow into an assembly that is final for macrostep $i$, and these and that final assembly of macrostep $i$ will be in $\Pi_{i+1}$. This will satisfy the requirements for the set $\Pi_{i+1}$ because (1) the final assembly created by macrostep $i$ and representing $\alpha_{i+1}$ is guaranteed, for any $\alpha_{i+2}$ such that $\alpha_{i+1}\rightarrow^1_\calS \alpha_{i+2}$ to be able to grow into some assembly $\beta_{i+2} \in \prodasm{S}$ such that $R^*(\beta_{i+2}) = \alpha_{i+2}$, and (2) as soon as any assembly resolves under $R$ to $\alpha_{i+1}$ it is in $\Pi_{i+1}$. Thus, the induction step is shown and $\calS'$ is shown to model $\calS$.

$\calS$ and $\calS'$ have equivalent productions because $\calS$ follows $\calS'$ and $\calS'$ models $\calS$, which means that, starting from the seed assemblies $\sigma$ in $\calS$, and $\sigma'$ in $\calS'$, as assembly $\alpha$ is producible (resp. terminal) in $\calS$ if and only if an assembly $\beta$ is producible (resp. terminal) in $\calS'$ such that $\beta$ maps under $R$ to $\alpha$ or an equivalent assembly. Thus $\calS'$ is shown to synchronously simulate $\calS$ under $R$ and using temporal scale factor $s$. 

In \cref{sec:temp1simulation-synchronization}, it was shown that both the spatial scale factor, and the temporal scale factor of $\calS'$ is $O(\log |T|)$.

We show that $\calS'$ is a staggered system by inspecting the growth of its macrotiles. For $\calS'$ not to be staggered, there would need to be one possible empty location $l \in \mathbb{Z}^2$ where, during a single assembly step, more than one neighboring location receives a tile, and each has a non-null glue facing $l$. This can only happen if there are two paths of growth that simultaneously arrive at locations neighboring $l$. As was described in \cref{sec:temp1simulation-growth-procedure}, every macrotile grows by first creating two separate paths of growth, which travel to separate corners of the macrotile in their own respective cooperation regions, and return back to center decision region. During the period of these paths' growth which is localized to their own cooperation regions, they are entirely separate from each other. The bit-reader and corner-claimer gadgets which are assembled by each path are staggered, as is described in \cref{sec:temp1simulation-growth-procedure}. Both paths abut the same location of the decision region when returning to it, and will have a glue facing said location iff there is a tile-type that can attach by the cooperation of glues which have been presented to both sides of their respective corner (or for the case of temperature-2 input paths, if they are the path corresponding to the ``right'' of the location where they entered the macrotile, as is described in \cref{sec:temp1simulation-strength2}). However, since $\calS'$ synchronously simulates $\calS$, and $\calS$ is a locally adjacent-only system, only (at most) one of the two paths may have read a value corresponding to a glue that they could make a simulated attachment to, and it is never the case that they both do so. Additionally, we note that at all assembly steps, in all macrotile locations representing frontier locations, all paths which are contained in it are each connected to a single input glue from an adjacent macrotile. This is a result of $\calS$ being a staggered system, $\calS'$ synchronously simulating $\calS$, and the fact that all growth within a macrotile location from a specific input glue either halts, or expresses output in $s$ assembly steps. Similarly, it is never the case that two output paths from adjacent macrotile locations present glues to the same location, as each place tiles along the edge of their respective macrotile on the $s$th time step, synchronously crashing into each other (i.e. the distance between them goes from $2$ to $0$ in a single assembly step). This concludes the analysis of all cases within a macrotile, and since the growth of separate macrotiles is all disjoint, $\calS'$ is a staggered system. %\todo{PD: I'm pretty sure this paragraph is sufficient, but I'm open to being told otherwise. MJP: I agree}
 
Finally, we show that the simulation of $\calS$ by $\calS'$ is seed-independent. In \cref{sec:temp1simulation-seed} it was shown how a tile set $T'_{\sigma} \subset T'$ is created from the tile set $T$ so that for any seed $\sigma$ in $\calS$, tiles from $T'_{\sigma}$ can be used to form $\sigma'$ such that $R^*(\sigma') = \sigma$ and $\calS' = (T', \sigma'', 1)$ synchronously simulates $\calS = (T, \sigma, 2)$ under representation function $R$ with spatial scale factor $c$ and temporal scale factor $s$, where $T', r, c, s$, and the temperature $1$ of $\calS'$ are all fixed, independent of the seed $\sigma$ in $\calS$. The function that generates $\sigma'$ of $\calS'$ is clearly computable as previously shown. Therefore, ($T', 1$) seed-independently simulates $(T,2)$.

Thus, \cref{lem:temp1simulation-adj} is proven.

\begin{corollary}
        \label{thm:temp1simulation}

    For any TAS $\calS$, there is a temperature 1 staggered TAS $\calS_1$ which simulates $\calS$ in the syncTAM.

    Moreover, this simulation scheme is seed-independent, and $\calS_1$ can be obtained from $\calS$ computably.
\end{corollary}    

\begin{figure}
    \centering
    \input{tikz/shine_on_you/temp1simulation}
    \caption{A visual proof of \cref{thm:temp1simulation}}
    \label{fig:temp1simulation}
\end{figure}
    
\begin{proof}
By \cref{lem:adj-only} and \cref{lem:temp1simulation-adj}.   
\end{proof}

Notably, this extends the construction \cref{sec:bounded} to temperature 1 syncTAS. A visual depiction of this proof is provided in \cref{fig:temp1simulation}.

\syncTAMToneBounded*

\begin{proof}
    Apply \cref{thm:temp1simulation} to the system of \cref{thm:Turing_band}.
\end{proof}

\section{A Temperature 2 Universal Simulator for the temperature 1 syncTAM}\label{sec:temp2iufortemp1}

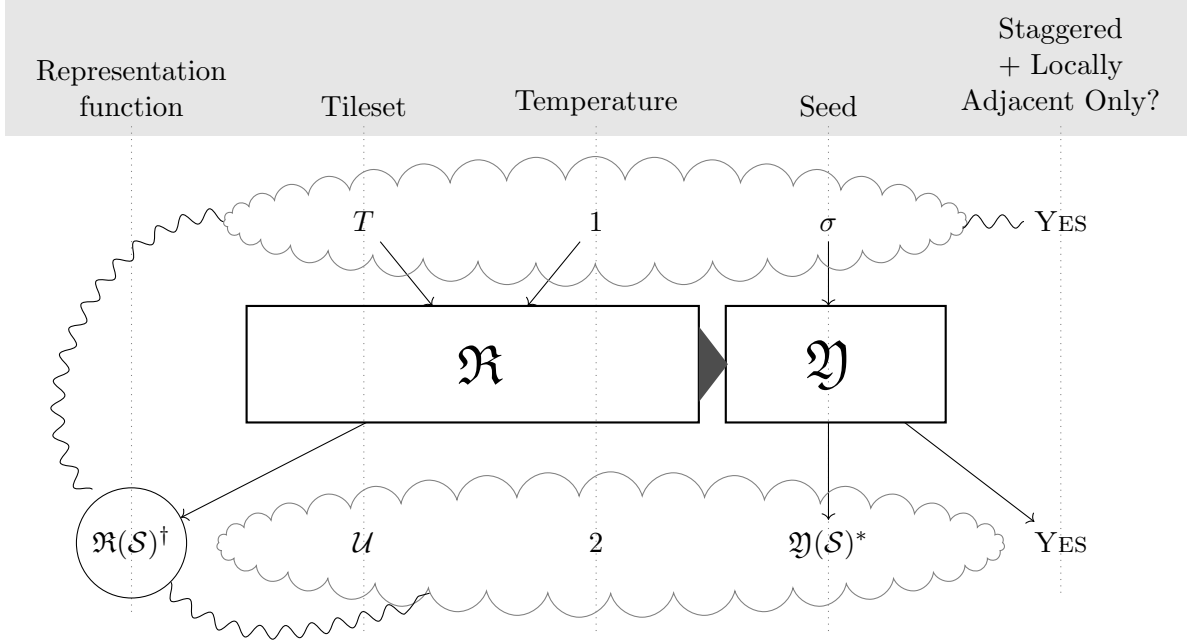
\begin{figure}
    \centering
    \begin{tikzpicture}[
    header/.style={align=center, text width=17ex, anchor=south},
    tas/.style={draw, gray, cloud, cloud ignores aspect, cloud puffs=39}
    ]
    \pgfdeclarelayer{background}
    \pgfdeclarelayer{foreground}
    \pgfsetlayers{background,main,foreground}
    \matrix (M) [
        matrix of nodes,
        nodes in empty cells, row sep={1cm},
        row 1/.style={nodes={header}},
        row 3/.style={text height=1cm, text width=17ex}
        %every node/.style={draw}
    ]{
    Representation function & Tileset & Temperature & Seed & Staggered + Locally Adjacent Only? \\
                            & $T$ & $1$      & $\sigma$ & \textsc{Yes} \\
                            &         &             &          &   \\
        |[draw, circle]|  $\mathfrak{R}(\mathcal{\mathcal{S}})^\dagger$   & $\mathcal{U}$ & $2$      & $\mathfrak{Y}(\mathcal{S})^*$ & \textsc{Yes} \\    
    };

    \begin{pgfonlayer}{background}
        \node[fill=gray!20, fit=(M-1-1) (M-1-5)] {}; 
        \foreach \col in {1,...,5} {
            \draw[gray, dotted] (M-1-\col) -- +(0,-7);
        }
    \end{pgfonlayer}

    \node[fit=(M-3-2) (M-3-3), inner xsep=0] (machine-cell-left) {};
    \node at (machine-cell-left.center) {\huge $\mathfrak{R}$};
    \draw[thick] (machine-cell-left.north east) ++(-.5em, 0) coordinate (machine-left-tr) rectangle (machine-cell-left.south west);    
    \node[fit=(M-3-4), inner xsep=0] (machine-cell-right) {};
    \node at (machine-cell-right.center) {\huge $\mathfrak{Y}$};
    \draw[thick] (machine-cell-right.north west) ++(.5em, 0) coordinate (machine-right-tl) rectangle (machine-cell-right.south east);
    \fill[black!70] (machine-cell-left.east) ++(-.5em, .5cm) -- ++(1em, -.5cm) -- ++(-1em, -.5cm) -- cycle;

    \draw[->] (M-2-2) -- (machine-cell-left);
    \draw[->] (M-2-3) -- (machine-cell-left);
    \draw[->] (M-2-4) -- (machine-cell-right);
    
    \draw[->] (machine-cell-left) -- (M-4-1);
    
    \draw[->] (machine-cell-right) -- (M-4-4);
    \draw[->] (machine-cell-right) -- (M-4-5);

    \node[tas, fit=(M-4-2) (M-4-4)] (tas-out) {};
    \node[tas, fit=(M-2-2) (M-2-4)] (tas-in) {};

    \draw[decorate, decoration={snake}] (M-2-5) to (tas-in) to[out=180, in=135] (M-4-1) to[out=-45, in=-165] (tas-out); 

\end{tikzpicture}
    \caption{A visual overview of \cref{thm:T2-IU-for-T1}. In an intrinsically universal simulation, the simulator tile set $\mathcal{U}$ is the same for any input tile set. Note that $\mathcal{U}$ works at temperature $2$, whereas $\mathcal{S} = (T, 1, \sigma)$ must be a temperature 1 syncTAS. The parts of the output system which depend on $\mathcal{T}$ are marked with a dagger, and with a star for a dependency on both $\calT$ and $\sigma$.}
    \label{fig:T2-IU-for-T1}
\end{figure}

\TtwoIUforTone*

%\begin{figure}
%    \centering
%    \input{tikz/shine_on_you/T2-IU-for-T1}
%    \caption{A visual overview of \cref{thm:T2-IU-for-T1}. In an intrinsically universal simulation, the simulator tileset $\mathcal{U}$ is the same for any input tileset. Note that $\mathcal{U}$ works at temperature $2$, while $\mathcal{T}$ needs to be a temperature 1 system. The parts of the output system which depend on $\mathcal{T}$ are marked with a dagger, and with a star for a dependency on both $\calT$ and $\sigma$.}
    %\label{fig:T2-IU-for-T1}
%\end{figure}

\begin{figure}
    \centering
    \includegraphics[width=0.5\linewidth]{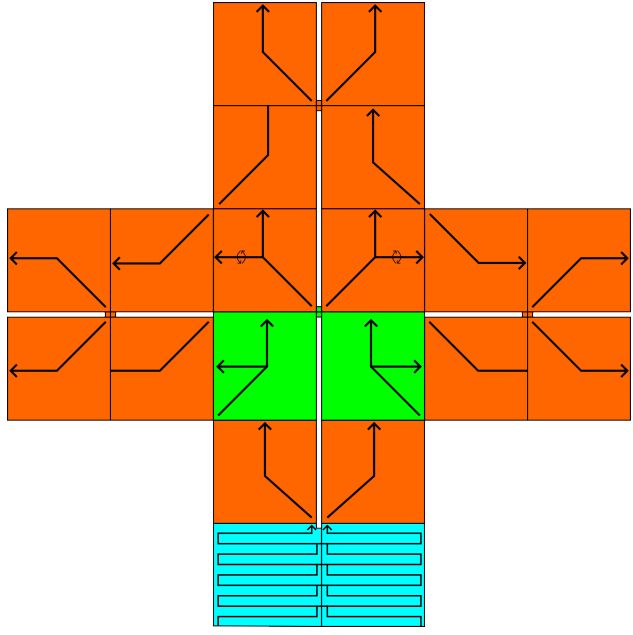}
    \caption{A (very) high level depiction of a macrotile of $\calS'$. Squares colored in light-blue represent zig-zag assemblies which co-simulate $\mathcal{M}_{d}$, squares colored orange represent propagation gadgets, and squares colored light-green represent modified propagation gadgets.} 
    %\label{fig:temp2IUforTemp1-macrotile-layout}
\end{figure}

%The rest of this section is the proof of \cref{thm:T2-IU-for-T1}.

The proof is by construction, thus we present a temperature 2 tile set $\calU$, representation function generator function $\mathfrak{R}$, and seed generator function $\mathfrak{Y}$ such that, provided a temperature 1 syncTAM system $\calT = (T, \sigma, 1)$ that is staggered, $\calS' = (\calU, S(\calT), 2)$ simulates $\calT$ under $\mathfrak{R}(\calT)$ at spatial scale factor $c$ and temporal scale factor $s$, where $c = O(|T|\log^2|T|)$, and $s = O(|T|^2\log^4|T|)$. A visual overview of the statement of \cref{thm:T2-IU-for-T1} is provided in \cref{fig:T2-IU-for-T1}.

\subsection{Commonly used gadgets} \label{sec:temp2iufortemp1-reusable-gadgets}

In this section we introduce a class of \textit{propagation} gadgets, which appear frequently throughout our construction of $\calS'$. 

The purpose of a propagation gadget is to transmit information contained within a set of glues from one side, $d_{in} \in \{N,E,S,W\}$ (from which growth of the gadget is initiated), to a set of its sides $D_{out} \subseteq \{N,E,S,W\}-\{d_{in}\}$. A propagation gadget which seeks to propagate a string of glues of length $n \in \mathbb{N}$ will reach a height of $n$, and completes its growth after $n^2$ assembly steps. We call the first tile to attach during the growth of a propagation gadget the \textit{start} tile, and the last tile to attach during the growth of a propagation gadget the \textit{end} tile. Three example square propagation gadgets, along with their abstract depictions, are shown in Figures \ref{fig:temp2iufortemp1-prop-gadget} and \ref{fig:temp2iufortemp1-prop-gadget-2}. 

\begin{figure*}[t!]
    \centering
    \begin{subfigure}[t]{0.48\textwidth}
        \centering
        \includegraphics[width=\linewidth]{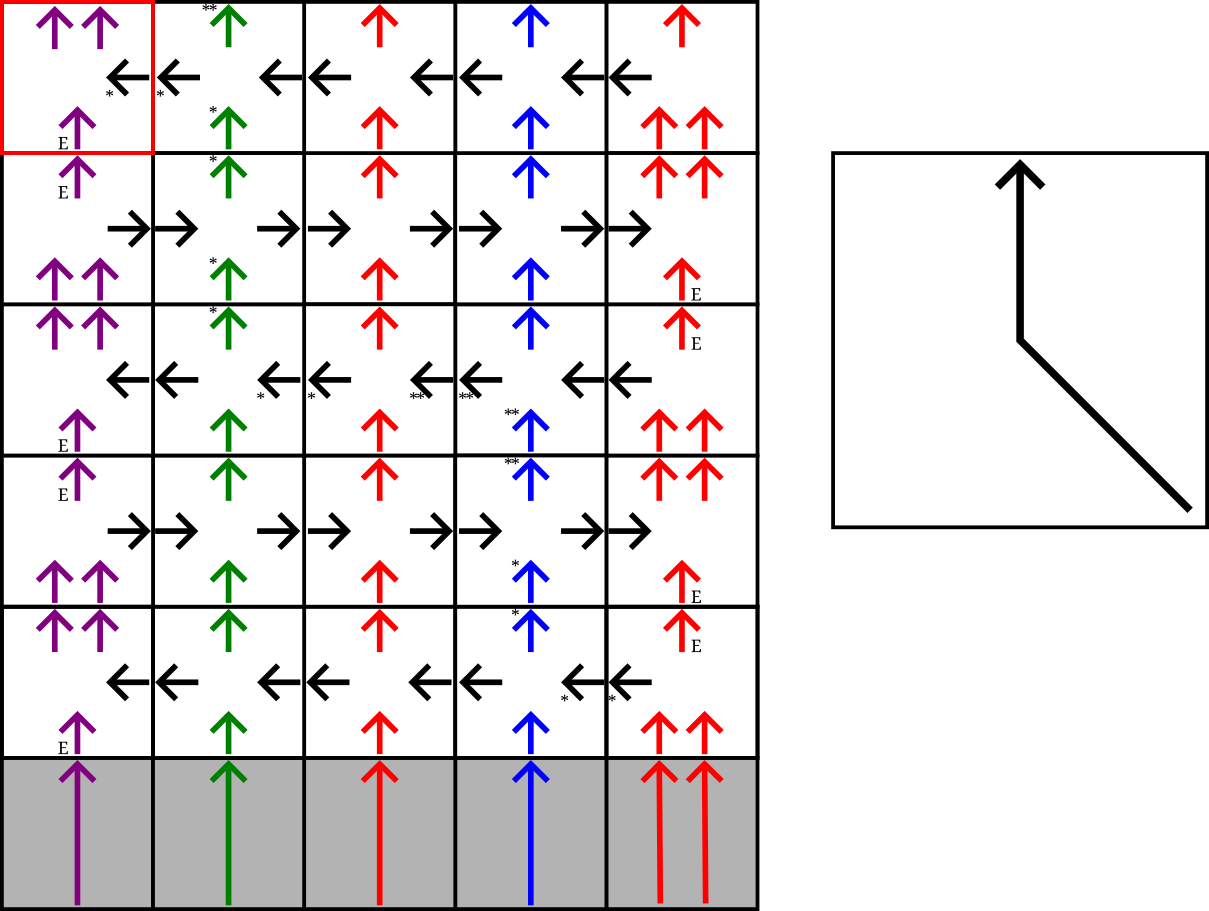}
        \caption{}
    \end{subfigure}%
    ~ 
    \begin{subfigure}[t]{0.48\textwidth}
        \centering
        \includegraphics[width=\linewidth]{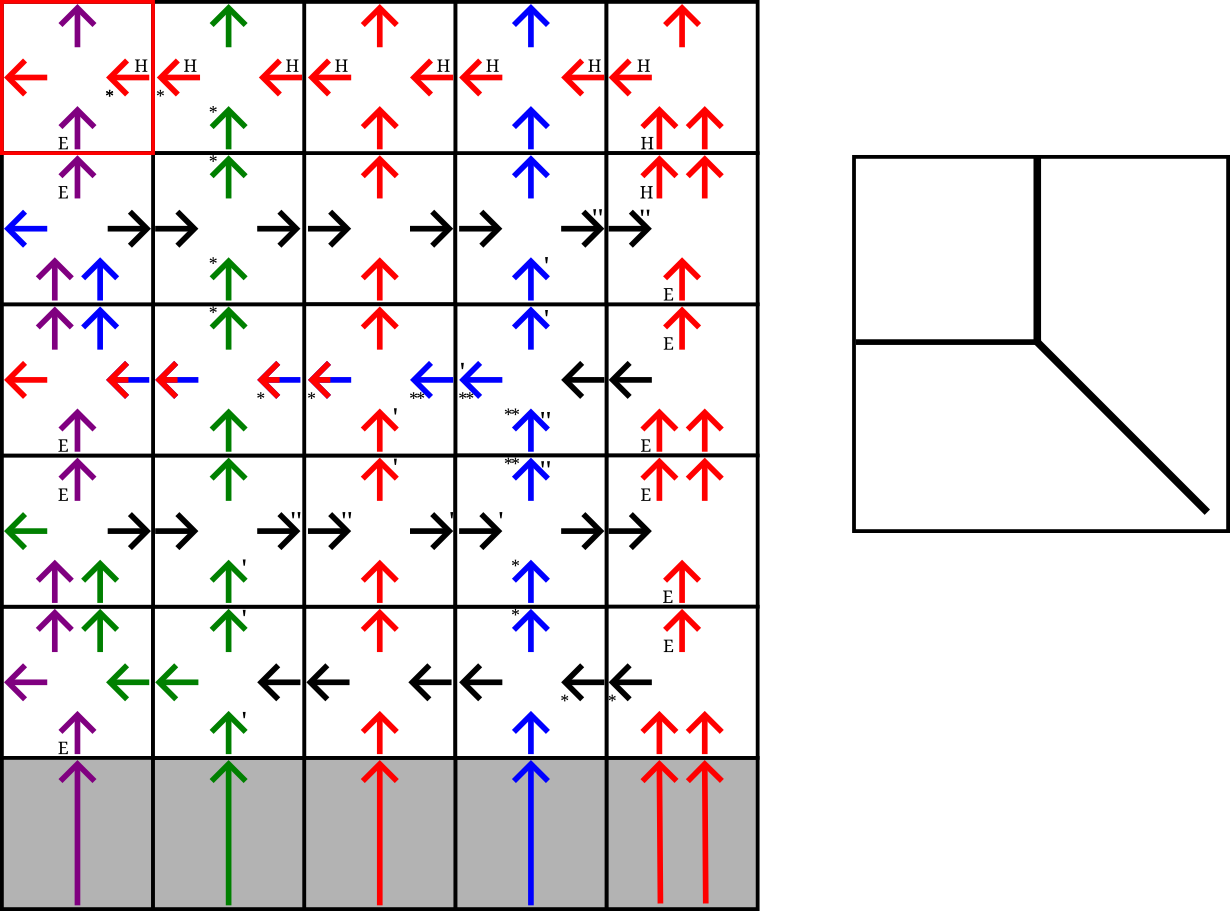}
        \caption{}
    \end{subfigure}
    \caption{A set of two example propagation gadgets.  Colors are used to indicate the symbol encoded by each glue, with the color black corresponding to ``no symbol''. The end tile of each gadget is outlined in red, and to each assembly's right is an abstract depiction of the gadget. Markings are used to indicate any extra information that must be encoded into each glue to direct growth. (a) A propagation gadget which propagates the string of symbols on its $S$ side to its $N$ side, which begins growth in the $SE$ position. The `$*$' marking travels to the left by two locations on each ``zag'', and once the `$*$' marking reaches the tile on the opposite edge to that which it began, the gadget ceases to grow any taller. The end tile of the gadget has a strength-2 glue that allows for a subsequent, chaining propagation gadget to attach. (b) A propagation gadget which propagates the string of symbols on its $S$ side to its $N$ and $W$ sides. Along with the aforementioned `$*$' marking, a `$'$' marking is used to indicate which symbols propagate to the west on each ``zig''. Sides of a tile whose glue encodes multiple symbols are represented by the presence of multiple colors.}
    \label{fig:temp2iufortemp1-prop-gadget}
\end{figure*}

\begin{figure}
    \centering
    \includegraphics[width=0.6\linewidth]{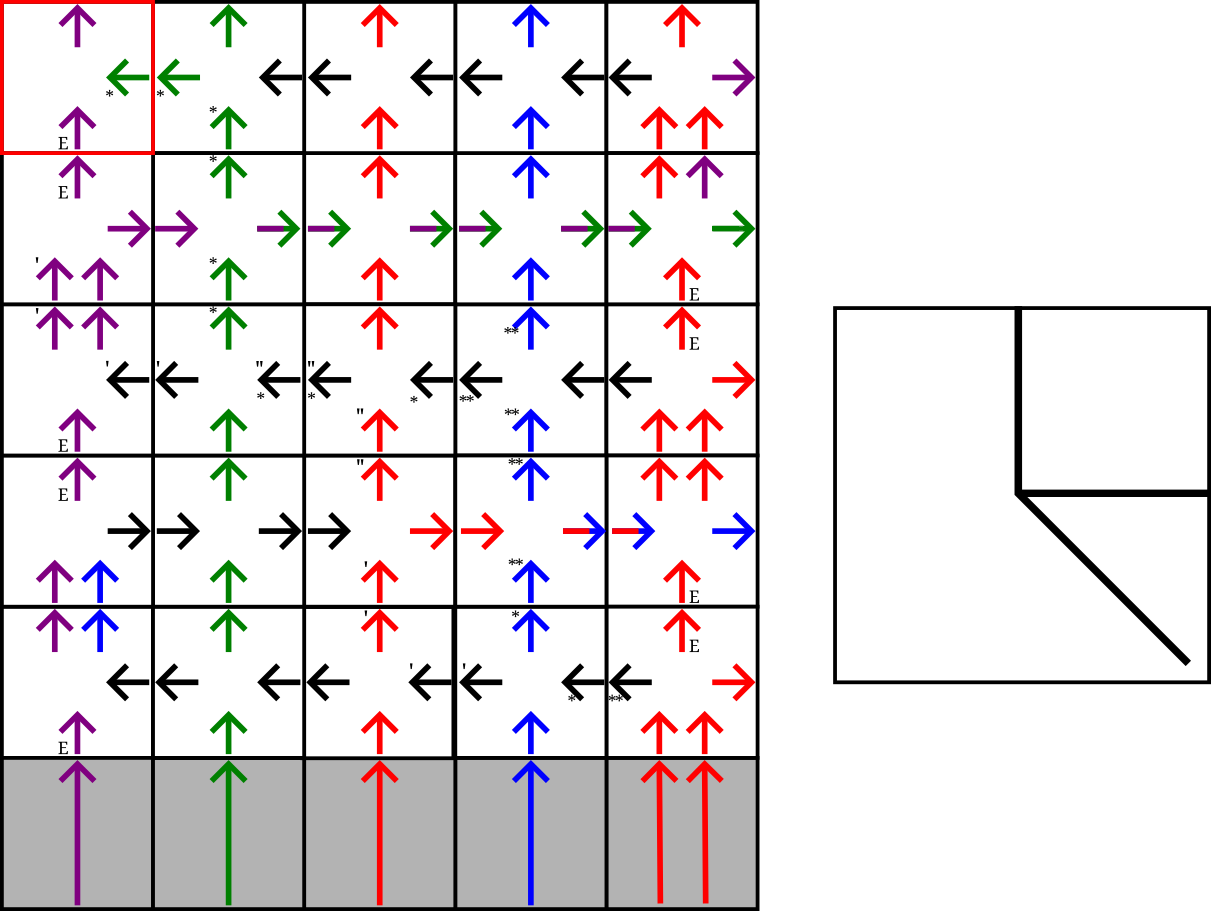}
    \caption{A propagation gadget which propagates the string of symbols on its $S$ side to its $N$ and $E$ side, which begins growth in the $SE$ position. The `$*$' and `$'$' markings travels to the left by two locations on each ``zig''. The `$'$' marking is used to determine which symbols are propagated to the east on each ``zag'', whereas the `$*$' symbol is used to denote when the gadget ceases to grow any taller.}
    \label{fig:temp2iufortemp1-prop-gadget-2}
\end{figure}

Growth of a propagation gadget is initiated by the attachment of a start tile with a strength-2 in exactly one of its $d_{in} d_{adj}$ corners, where $d_{adj} \in \texttt{ADJ}(d_{in})$. A propagation gadget which initiates growth in its $d_{in}d_{adj}$ corner will place its last tile (i.e. its end tile) in its $\texttt{OPP}(d_{in})\texttt{OPP}(d_{adj})$ corner if the length of its input string of symbols is \textit{odd}, and in its $\texttt{OPP}(d_{in})d_{adj}$ corner if the length of its input string of symbols is \textit{even}. For the purposes of our construction, we will only use propagation gadgets to propagate odd-length inputs, a decision whose justification we defer to \cref{sec:temp2iufortemp1-tile-lookup-procedure}. 

We also use propagation gadgets to reverse the string of symbols which it transmits from its $d_{in}$ side. Particularly, we use propagation gadgets whose start tile attaches to its $d_{in}d_{adj}$ corner to output the reversed version of their supplied input string \textit{exclusively} on their $\texttt{OPP}(d_{adj})$ side. An example of this behavior can be seen in \cref{fig:temp2iufortemp1-prop-reversing}. 

\begin{figure}
    \centering
    \includegraphics[width=0.6\linewidth]{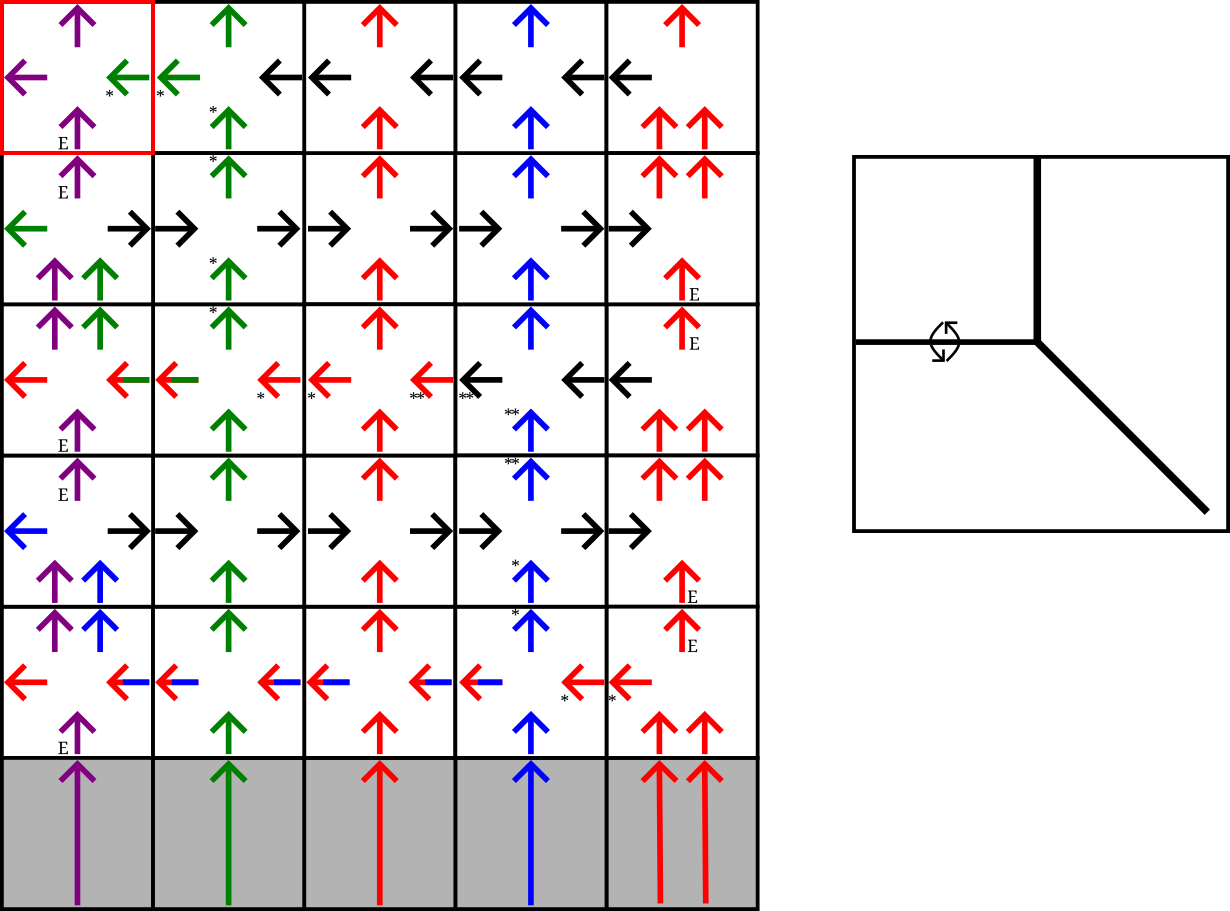}
    \caption{An example propagation gadget which transmits the string of symbols on its $S$ side to its $N$ and $W$ sides, with the version of the aforementioned string of symbols transmitted to its $W$ side being reversed. The `$*$' marking, in addition to denoting the maximum height of the gadget, also denotes which symbols to propagate to the left on each ``zig''. The end tile of the provided gadget is outlined in red, and sides of a tile whose glue encodes multiple symbols are represented by the presence of multiple colors.} 
    \label{fig:temp2iufortemp1-prop-reversing}
\end{figure}

Finally, we note that propagation gadgets may easily be ``chained'' onto each other (i.e. a propagation gadget may be used to transmit the set of symbols presented on an output side of another propagation gadget) by the inclusion of strength-2 glues to the tiles which attach in one of the gadget's corners. When a propagation gadget chains off of the end tile of another, the chaining gadget starts growth only after the gadget which initiated its growth has completed construction, meaning that there are no assembly steps where it must ``wait'' for the gadget which it grew off of to complete its assembly. Let $d_{in}d_{adj}$ be the corner of a propagation gadget where a start tile is placed, our construction exclusively uses chaining propagation gadgets whose growth is initiated by the tile in the $\texttt{OPP}(d_{in})\texttt{OPP}(d_{adj})$ corner, taking the output on the $\texttt{OPP}(d_{in})$ or $\texttt{OPP}(d_{adj})$ side of the gadget it propagates off as input, and the $d_{in}\texttt{OPP}(d_{adj})$ corner, taking the output on the $\texttt{OPP}(d_{adj})$ side of the gadget it propagates off of as input. 

\subsection{Macrotile structure} \label{sec:temp2iufortemp1-macrotile-structure}

Here we discuss the layout of a $c \times c$ macrotile of $\calS'$. A macrotile which has completed growth contain tiles along its output sides which express the following two pieces of information to an adjacent macrotile location : (1) the tile type $t$ that the macrotile resolved to, and (2) an encoding of the tile set $T$, in the form of a data structure called a \textit{Tile Lookup Table} (TLT). Both of these pieces of information are expressed using a finite set of glues in $G_{\calU}$, which correspond to the symbols of $\Sigma_{Macrotile}$. For the sake of clarity, we will denote the length of the encodings necessary to express this information as $l_{T}$, and defer further discussion of this encoding to \cref{sec:temp2iufortemp1-Side-Encoding}. 

In order to maintain synchronicity, the aforementioned information is presented twice on each output side of a completed macrotile, in two sections which we logically describe as the $L$, and $R$ outputs on a side of the macrotile. The $L$ and $R$ outputs on a particular side of the macrotile each encode identical information (i.e. the tile type $t$, and the TLT), however they are separated from each other by exactly one empty location between them, with the information provided in the $R$ output being the reversed version of that which is given in the $L$ output.

The $L$ and $R$ outputs on the sides of all completed macrotiles in $\calS'$ have length $l_T$. We logically divide each $c \times c$ macrotile into a $6 \times 6$ grid of 36 $l_T \times l_T$ \textit{sections}, with the center-most two sections along each side $d\in \{N,E,S,W\}$ separated by a single-tile wide space. On side $d \in \{N,E,S,W\}$, we denote the two center-most sections nearest to the sides of the macrotile as $L_d$ and $R_d$. All other sections which receive tiles during the growth of a macrotile in $\calS'$ are denoted as $fill_{(x,y)}$ where $(x,y) \in \{0,1,\ldots,5\} \times \{0,1,\ldots,5\}$, and all sections which do not receive tiles during the growth of a macrotile as $empty_{(x,y)}$. For any two regions $Region_1$ and $Region_2$ of a macrotile, if the sides of $Region_1$ and $Region_2$ are separated by a single-tile wide space between them, we say that $Region_1$ and $Region_2$ are \textit{neighbors} (for example, $L_d$ neighbors $R_d$, for all $d \in \{N,E,S,W\}$). Note that the four center-most regions of a macrotile each have two neighbors, since two of their sides are bordered by a single-tile wide gap. A schematic of this layout is shown in \cref{fig:temp2iufortemp1-macrotile-layout-grid}.

\begin{figure}
    \centering
    \includegraphics[width=0.75\linewidth]{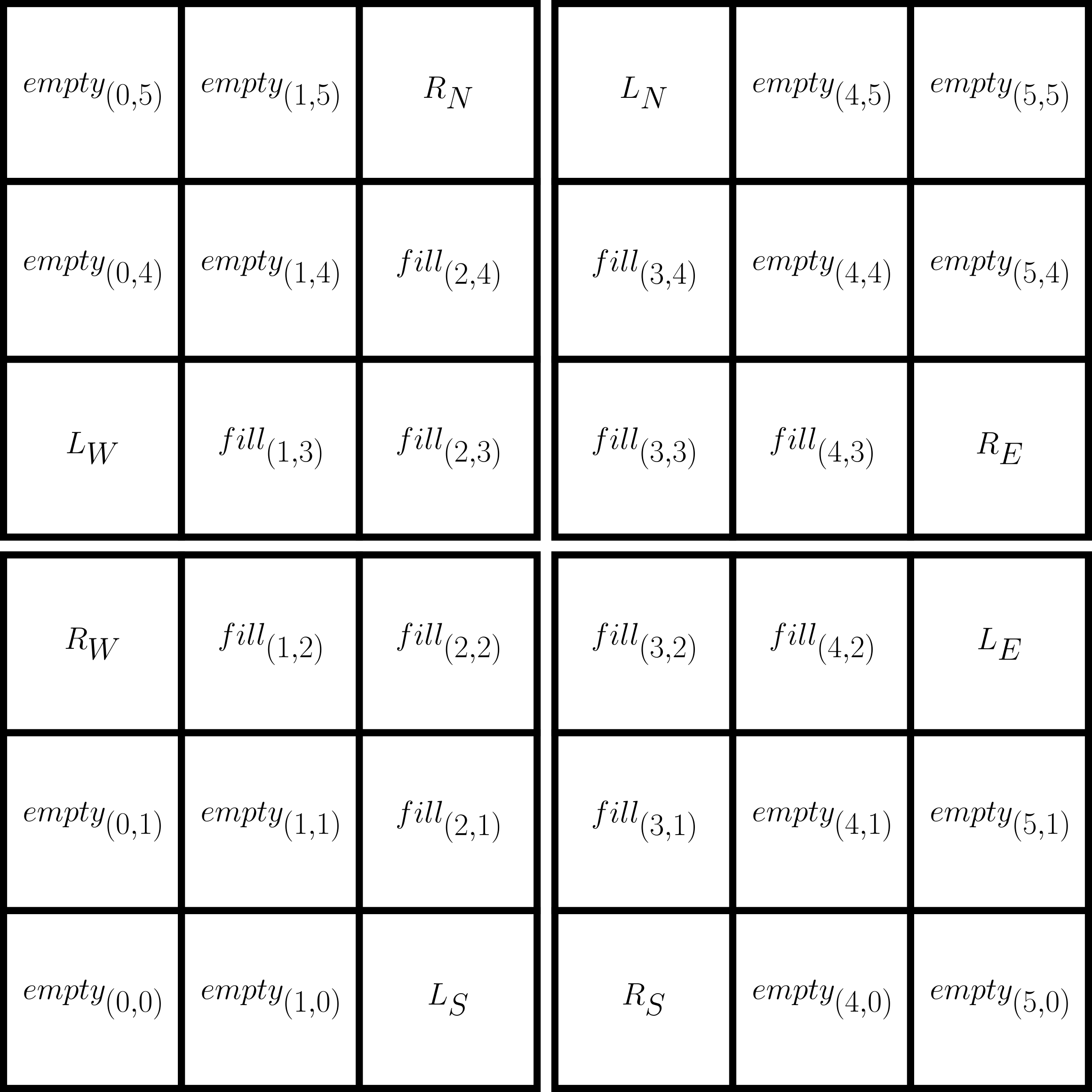}
    \caption{A schematic depicting the layout of a macrotile constructed by $\calS'$. Each $l_{out} \times l_{out}$ cell is labeled as per. the description provided in \cref{sec:temp2iufortemp1-macrotile-structure}.}
    \label{fig:temp2iufortemp1-macrotile-layout-grid}
\end{figure}

\subsection{Tile lookup table encoding} \label{sec:temp2iufortemp1-TLT-encoding}

The Tile lookup table is a data structure which lists out representations of each tile type $t \in T$, more specifically, the TLT is a string of symbols in the alphabet $\Sigma_{TLT}=\{0,1,\mid,N,E,S,W\}$. In this section, we showcase our procedure for creating the tile lookup table corresponding to tile set $T$. The maximum number of bits necessary to uniquely express each glue label in $G_T$ is $\lceil \log |G_T|\rceil$. Similarly, the maximum number of bits necessary to uniquely express an identifier for each tile type $t \in T$ is $\lceil \log |T|\rceil$. 

Given a set of glue labels $G_T$, let $\texttt{LEX}_G(G_T)$ be a lexicographic ordering of the glue labels in $G_T$, and for each glue $g \in G$, let $\texttt{INDEX}_G(g)$ be a function which returns the index of $g$ in $\texttt{LEX}(G)$. Similarly, given a set of tile types $T$, let $\texttt{LEX}_T(T)$ be a lexicographic ordering of the tile types in $T$, and $\texttt{INDEX}_T(t)$ be a function which returns the index of $t$ in $\texttt{LEX}_T(T)$. To create a binary encoding for each tile type, define the function $\mathtt{PAD\mathunderscore{}BIN} :\mathbb{N}\times\mathbb{N} \rightarrow \{0,1\}^* $ such that $\mathtt{PAD\mathunderscore{}BIN}(n,l)$ is the string $b \in \{0,1\}^*$ consisting of the binary representation of $n$ padded on the left with the number of $0$s needed to make $|s|=l$, with the requirement that $|\mathtt{BIN}(n)| \le l$, where $\mathtt{BIN}(n)$ is the standard binary representation of $n$.

For each tile type $t \in T$, and glue $g \in G_T$, let $b_t$ and $g_t$ be the strings $\mathtt{PAD\mathunderscore{}BIN}(\texttt{INDEX}_T(t), \lceil \log |T|\rceil)$ and $\mathtt{PAD\mathunderscore{}BIN}(\texttt{INDEX}_G(g), \lceil \log |G_T|\rceil)$, respectively. For each $t \in T$, the TLT lookup entry corresponding to $t$ is represented by the function $\texttt{TLT}(t)$, which returns the following string : 

\[\text{``}\mid b_t\text{ N } b_{g_{t,N}}\text{ E }b_{g_{t,E}} \text{ S } b_{g_{t,S}} \text{ W } b_{g_{t,W}} \text{''}\]

The length of $\texttt{TLT}(t)$ is $\lceil \log |T|\rceil + 4\lceil \log |G_T|\rceil + 5$. The string corresponding to the Tile Lookup Table for tile set $T$, which we will call $\Gamma_{T}$, is the concatenation of $\texttt{TLT}(t)$ for each tile $t \in T$, ordered by $\texttt{INDEX}(t)$, and the length of $\Gamma_T$ is $|T|(\lceil \log |T|\rceil + 4 \lceil \log |G_T|\rceil + 5)$.

\subsection{Macrotile side encoding}\label{sec:temp2iufortemp1-Side-Encoding}

Here we discuss the information which is presented on the output sides of a macrotile location that resolves to a tile type $t \in T$ in $\mathcal{S}'$. Let $M(l)$ be a macrotile location in $\calS'$ that resolves to a tile type $t \in T$ during macrostep $i$. The output sides of $M(l)$ are those which never received input growth from an adjacent macrotile during a previous macrostep. In order to direct future assembly (i.e. propagate glue information such that future macrotiles may resolve), our construction insists that during the last step of macrostep $i$, information is presented on each side $d \in D_{out}$ (where $D_{out} \subset \{N,E,S,W\}$ is the output sides of $M(l)$) which encode the tile type $t$, and the TLT $\Gamma_T$.

Along the output sides of a macrotile is a string of symbols in the alphabet $\Sigma_{Macrotile} =\{0,1,\mid,\mid',N,E,S,W,\_\}$ (note that $\Sigma_{TLT} \subset \Sigma_{Macrotile}$). We use the function $\texttt{OUTPUT}(t)$ to represent the encoding of tile type $t$ that will be presented on each output side of $M(l)$. $\texttt{OUTPUT}(t)$ returns $\texttt{TLT}(t)$, with the symbol `$\mid$' replaced with `$\mid'$'. $\texttt{blank}(T)$ is a function which, when provided with a tile set $T$, returns a string consisting of a large number of `$\_$' symbols. The length of $\texttt{blank}(T)$ has significance in \cref{sec:temp2iufortemp1-tile-lookup-procedure}, to which we defer its discussion. The function $\mathtt{SIDE\mathunderscore ENCODING}(t)$ returns the macrotile side encoding which is presented on each output side of a $M(l)$, and is as follows :

\[\text{``}\texttt{OUTPUT}(t)\text{ }\Gamma_T\text{ }\texttt{blank}(T)\text{''}\]

After the last step of macrostep $i$, for each output side $d \in D_{out}$ of $M(l)$ there are 2 groups of $|\mathtt{SIDE\mathunderscore ENCODING}(t)|$ tiles abutting side $d$ of $M(l)$, within the $L_d$, and $R_d$ regions, respectively. The tiles on the edge of the $L_d$ region present glues to the adjacent macrotile which encode $\mathtt{SIDE\mathunderscore ENCODING}(t)$, whereas the tiles on the edge of the $R_d$ region present glues to the adjacent macrotile which encode $\mathtt{SIDE \mathunderscore ENCODING}(t)^R$ (i.e. the reversed version of the string). Every glue abutting the adjacent macrotile apart from those on the corners of $L_d$ and $R_d$ and the adjacent $empty$ regions is strength-1, and are in the glue set $G_{\Sigma_{Macrotile},d} \subset G_\calU$. $G_{\Sigma_{Macrotile},\texttt{OPP}(d)}$ contains exactly $|\Sigma_{Macrotile}|$ glues, each of which we denote as $g_{\lambda,\texttt{OPP}(d)}$ for $\lambda \in \Sigma_{Macrotile}$. Glues that are presented in the corner of the sides of $M(l)$, and the boundary with the adjacent $empty$ region (for $L_d$ and $R_d$) are strength-2, and will always encode the symbol $\lambda$ that is in the first position of $\mathtt{SIDE \mathunderscore ENCODING}(t)$ (or the last position of $\mathtt{SIDE\mathunderscore ENCODING}(t)$), as well as the output side $d$ that the information is being expressed upon. We call the aforementioned strength-2 glue $\mathtt{MARK\mathunderscore OUT}(g_{\lambda, d,L})$ (or $\mathtt{MARK\mathunderscore OUT}(g_{\lambda, d,R})$) where $\lambda \in \Sigma_{Macrotile}$ is a symbol and $d \in \{N,E,S,W\}$ is an output side that a macrotile may take.

\subsection{Tile lookup procedure} \label{sec:temp2iufortemp1-tile-lookup-procedure}

Dissimilar to many other IU constructions in the literature \cite{IUSA, SADSSF}, our system $\calS'$ determines valid tile attachments through the operation of a zig-zag, fixed-width, probabilistic Turing machine $\mathcal{M}_{d}$ where $d \in \{N,E,S,W\}$. When we refer to $\mathcal{M}_{d}$, we are actually referring to one of $4$ TMs $\{\mathcal{M}_{N}, \mathcal{M}_{E}, \mathcal{M}_{S}, \mathcal{M}_{W}\}$, which we describe for arbitrary $d \in \{N,E,S,W\}$. We will describe the tile lookup procedure of $\calS'$ by first describing the operation of $\mathcal{M}_{d}$.

At a high level, for $t \in T$, $\mathcal{M}_{d}$, on input $\mathtt{SIDE\mathunderscore ENCODING}(t)$, determines whether or not there is a tile type $t' \in T$ whose glue on its $d$ side can form a bond with the glue on side $\texttt{OPP}(d)$ of $t$. If there is such a tile, then $\mathcal{M}_d$ enters an accept state with the string $``\texttt{OUTPUT}(t')\Gamma_T\texttt{blank}(T)'' = \mathtt{SIDE\mathunderscore ENCODING}(t')$ after exactly $\mu = \frac{|\mathtt{SIDE\mathunderscore ENCODING}(t)|-1}{2}$ steps. If there is not such a tile, $\mathcal{M}_d$ will halt after $<\mu$ steps.

Formally, $\mathcal{M}_{d} = (Q_d, \Sigma_{Macrotile}, \Sigma_{\mathcal{M}}, q_0 \in Q_d, A_d \subset Q_d, \delta_{1,d}, \delta_{2,d})$ where $\Sigma_{\mathcal{M}} = \{0,1,\mid, \mid', N,E,S,W,0^*,1^*,!,\#\}$, $|A_d| = 1$ (there is exactly one accept state), and $\delta_{1,d}$ and $\delta_{2,d}$ are the transition functions of $\mathcal{M}_d$ (which has two transition functions due to being nondeterministic). $\mathcal{M}_{d}$ operates in four phases : \textit{Lookup}, \textit{Marking}, \textit{Randomized Choice}, and \textit{Copy-back}, which we will describe below. Note that in the pseudocode that we provide for $\mathcal{M}_{d}$, we use the notation $d+1$ to refer to the side which is $90^\circ$ clockwise to $d$ (if $d ==N$, then $d+1 = E$, etc.), $h$ to refer to the symbol that the head of $\mathcal{M}_{d}$ points to at the current step, and $r \in \{1,2\}$ to refer to whether or not $\mathcal{M}_{d}$'s next transition will be in $\delta_{1.d}$, or $\delta_{2,d}$. 

\subsubsection{Lookup stage}

The lookup stage of $\mathcal{M}_{d}$ is responsible for determining the set of tiles $T_{bond} \subseteq T$ that can form an attachment to the glue on side $\texttt{OPP}(d)$ of $t$. This is done through a bitwise comparison of $b_{g_{t,\texttt{OPP}(d)}}$ and $b_{g_{t',d}}$, for every $t' \in T$. Each bitwise comparison is performed by scanning the entirety of ``$\texttt{OUTPUT}(t)$ $\Gamma_T$'' twice (once to the right, and once to the left). Once all $\lceil\log|G_T|\rceil$ right and leftward passes have been performed, this stage concludes with a single rightward pass which halts at the end of $\Gamma_T$. Therefore, the lookup stage of $\mathcal{M}_d$ requires exactly $(2\lceil\log|G_T|\rceil+1)(|\texttt{OUTPUT}(t)| + |\Gamma_T|)$ computation steps to complete. Additionally, during the first $\lceil \log |G_T|\rceil$ rightward passes, exactly one symbol in $\{0,1\}$ corresponding to $b_{g_{t,\texttt{OPP}(d)}}$ in the $\texttt{OUTPUT}(t)$ portion of the input has their symbol changed to have `$*$' appended (i.e. $0$ is transformed to $0^*$, and $1$ to $1^*$). During each rightward pass, the aforementioned instance of $b_{g_{t,\texttt{OPP}(d)}}$ is observed, and if every symbol for which said string is made up of is $0^*$, then $\mathcal{M}_d$ will determine that $b_{g_{t,\texttt{OPP}(d)}}$ corresponds to the null glue, and will halt, since there is no tile type $t'$ which makes a bond with the null glue. The pseudocode for this portion of $\mathcal{M}_{d}$ is shown in \cref{alg:temp2iufortemp1-Lookup}.

\begin{algorithm}
\caption{An algorithm that performs a bitwise comparison of $b_{g_{t,\texttt{OPP}(d)}}$ and $b_{g_{t',d}}$, for each tile type $t' \in T$. After the completion of this algorithm, if $b_{g_{t,\texttt{OPP}(d)}} == b_{g_{t',d}}$, all $\lceil \log |G_T|\rceil$ bits of $b_{g_{t',d}}$ within $\texttt{TLT}(t')$ will be marked with a `$^*$' symbol, otherwise, $<\lceil \log |G_T|\rceil$ bits will be marked.} 
\label{alg:temp2iufortemp1-Lookup}
\begin{algorithmic}[1]
    \While{$h \notin \{\texttt{OPP}(d), \#, \_\}$} \label{algstep:temp2iufortemp1-Lookup-start} 
    %\algorithmiccomment Travel right to $b_{g_{t,d}}$
        \State Move right without changing the tape
    \EndWhile
    \State $null \gets T$
    \While{$h \notin \{0,1, E,S,W,\mid\}$}
    %\algorithmiccomment Travel right to the first unmarked symbol of $b_{g_{t,d}}$, or the end of $b_{g_{t,d}}$
        \If{$h \neq 0^*$}
            \State $null \gets F$
        \EndIf
        \State Move right without changing the tape
    \EndWhile
    \If{$null == T$}
    %\algorithmiccomment Halt immediately if all symbols of $b_{g_{t,d}}$ are $0^*$
        \State \textbf{halt}
    \EndIf
    \If{$h \in \{0,1\}$}
        \State $b_0 \gets h$ %\algorithmiccomment Save the first unmarked bit of $b_{g_{t,d}}$
        \State Write $b_0^*$ to the tape
        \While{$h \notin \{\mid, \_\}$} \label{algstep:temp2iufortemp1-Lookup-entry}
            \State Move right without changing the tape
        \EndWhile
        \If{$h ==$ $\mid$}
            \While{$h \neq d$}
                \State Move right without changing the tape
            \EndWhile
            \While{$h \notin \{0,1\}$}
                \State Move right without changing the tape
            \EndWhile
            \State $b_1 \gets h$
            \If{$b_1 == b_0$}
                \State Write $b_1^*$ to the tape and move right
            \EndIf
            \State \textbf{goto} \ref{algstep:temp2iufortemp1-Lookup-entry}

        \Else
            \While{$h \neq$ $\mid'$}
                \State Move left without changing the tape
            \EndWhile
            \State \textbf{goto} \ref{algstep:temp2iufortemp1-Lookup-start}
        \EndIf
    \Else 
        \While{$h \neq \_$}
            \State Move right without changing the tape
        \EndWhile
        \State \textbf{goto} \cref{alg:temp2iufortemp1-Marking}
    \EndIf    
\end{algorithmic}
\end{algorithm}

\subsubsection{Marking stage}

The marking stage of $\mathcal{M}_{d}$ is responsible for marking the first character of each TLT entry corresponding to a tile $t' \in T_{bond}$ with the symbol $\mid^*$, with the exception TLT entry for $t'$ which appears \textit{last} on the tape, which is marked with the symbol $\mid^{**}$. This operation is performed using only a single leftward pass of the $\texttt{OUTPUT}(t)$ and $\Gamma_T$ portions of the tape. Thus, the marking stage requires exactly $|\texttt{OUTPUT}(t)| +|\Gamma_T|$ steps to complete, and the pseudocode for this portion of $\mathcal{M}_{d}$ is shown in \cref{alg:temp2iufortemp1-Marking}.

\begin{algorithm}
\caption{An algorithm that marks the first character of each TLT entry whose glue on side $d$ matches the glue on side $\texttt{OPP}(d)$ of the provided tile-type definition. The \textit{last} matching $TLT$ entry (ordered along the tape) is marked with a `$^{**}$' symbol, and all others are marked with `$^*$' symbol.}
\label{alg:temp2iufortemp1-Marking}
\begin{algorithmic}[1]
    \State $a \gets T$ 
    
    \If{$d \neq W$} \label{algstep:temp2iufortemp1-Marking-start}
        \While{$h \neq d+1$}
            \State Move left without changing the tape
        \EndWhile
    \EndIf
    \State Move left without changing the tape

    \If{$h \in \{0^*, 1^*\}$}
        \While{$h \notin \{\mid, \mid'\}$}
            \State Move left without changing the tape
        \EndWhile
        
        \If{$h ==$ $\mid$}
            \If{$a == T$}
                \State $a \gets F$
                \State Write $h^{**}$ to the tape and move left
            \Else
                \State Write $h^*$ to the tape and move left
            \EndIf
            \State \textbf{goto} \ref{algstep:temp2iufortemp1-Marking-start}
        \Else
            \State \textbf{goto} \cref{alg:temp2iufortemp1-Randomized-Choice}        
        \EndIf
    \EndIf
\end{algorithmic}
\end{algorithm}

\subsubsection{Randomized choice stage}

The randomized choice stage of $\mathcal{M}_{d}$ is responsible for performing the non-deterministic selection of a single tile type in $T_{bond}$. 
This is accomplished by iterating a single time to the right through each entry corresponding to a tile type $t'\in T_{bond}$ and, if said entry was marked with a `$\mid^*$' symbol, the corresponding transition in $\delta_{1,d}$ ``accepts'' $t'$, and marks its entry with a `$!$' symbol. However, if the corresponding transition in $\delta_{2,d}$ ``rejects'' $t'$, and does not mark it with a `$!$'. 
In order to ensure that a tile type $t' \in T$ is always selected (if there is one available), the entry marked with `$\mid^{**}$' is \textit{always} ``accepted'' iff no previous entry marked with `$\mid^*$' was ``accepted''. 
If there is no such entry corresponding to a tile type $t'$ marked with `$\mid^*$' or `$\mid^{**}$', then $\mathcal{M}_d$ will determine that there is no tile type in $T$ which binds to side $\texttt{OPP}(d)$ of $t$, and halt at the end of the aforementioned rightward pass. The pseudo-code for this portion of $\mathcal{M}_{d}$ is shown in \cref{alg:temp2iufortemp1-Randomized-Choice}. 

Let $t'_i \in T_{bond}$ be the $i$th TLT entry (for $0 < i \le |T_{bond}|$) marked with a $\mid^*$, or $\mid^{**}$ symbol by \cref{alg:temp2iufortemp1-Marking}. We note that for each $t'_i \in T_{bond}$ , there is a sequence of transitions by $\mathcal{M}_d$ that marks the entry related to $t'_i$ with a `$!$' symbol, and that the likelihood of $t'_i$ receiving the aforementioned marking is $\frac{1}{2^i}$ (as long as $|T_{bond}| > 1$, otherwise the likelihood of $t'_i$ receiving a marking is $1$).

\begin{algorithm}
\caption{An algorithm that chooses a single random matching $TLT$ entry (by \cref{alg:temp2iufortemp1-Lookup} and \cref{alg:temp2iufortemp1-Marking}), and marks it with a $!$.}
\label{alg:temp2iufortemp1-Randomized-Choice}
\begin{algorithmic}[1]
    \While{$h \notin \{\mid^*, \mid^{**}, \_\}$}
        
        \State Move right without changing the tape
    \EndWhile

    \If{$h \in \{\mid^*, \mid^{**}\}$}
        \If{($h == \mid^*$ and $r == 1$) or $h ==$ $\mid^{**}$}    
            \State Write $!$ to the tape and move right
            \While{$h \neq \_$}
                \State Move right without changing the tape
            \EndWhile
            \While{$h \neq$ $\mid'$}
                \If{$h == 0^*$}
                    \State Write $0$ to the tape and move left
                \ElsIf{$h == 1^*$}
                    \State Write $1$ to the tape and move left
                \Else
                    \State Move left without changing the tape
                \EndIf
            \EndWhile
            \State \textbf{goto} \cref{alg:temp2iufortemp1-Copyback}
        \EndIf
    \Else
        \State \textbf{halt}
    \EndIf
\end{algorithmic}
\end{algorithm}

\subsubsection{Copy-back stage}

The copy-back stage of $\mathcal{M}_{d}$ is responsible for ensuring that (1) the tape only consists of members of the input alphabet $\Sigma_{Macrotile}$, and (2) the first $|\texttt{OUTPUT}(t)|$ locations of the tape encode $\texttt{OUTPUT}(t')$ (where $t'$ is the tile type $t' \in T$ that was marked with a `$!$' by \cref{alg:temp2iufortemp1-Randomized-Choice}). This is done through the completion of $|\texttt{OUTPUT}(t)-5|$ left and right sweeps of the first $|\texttt{OUTPUT}(t)|+|\Gamma_T|$ locations of the tape. Each sweep to the right copies a single bit of $\texttt{TLT}(t')$, which is then written to its respective location in $\texttt{OUTPUT}(t')$ at the start of the tape during the subsequent rightward sweep. After $\texttt{TLT}(t')$ has been copied to the start of the tape, $\mathcal{M}_d$ enters the \textit{accept} state with its head exactly on the first location of $\Gamma_T$. As such, this portion of $\mathcal{M}_{d}$ requires exactly $2(|\texttt{OUTPUT}(t)|-5)(|\texttt{OUTPUT}(t)|+|\Gamma_T|)+|\texttt{OUTPUT}(t)|$ steps to complete. Pseudocode for this portion of $\mathcal{M}_{d}$ is shown in \cref{alg:temp2iufortemp1-Copyback}. We note that for the context of $\calS'$, placing a tile corresponding to the entrance of the accept state by $\mathcal{M}_{d}$ is sufficient for the macrotile to map to $t'$ under the representation function.

\begin{algorithm}
\caption{An algorithm for copying the information of the single matching TLT entry marked with a $!$ into first $\lceil \log |T| \rceil + 4\lceil \log |G_T| \rceil + 5$ locations of the tape, and replacing all symbols on the tape with symbols in $\Sigma_{Macrotile}$}
\label{alg:temp2iufortemp1-Copyback}
\begin{algorithmic}[1]
    
    \While{$h \neq$ $!$} \label{algstep:temp2iufortemp1-Copyback-start}
        \State Move right without changing the tape
    \EndWhile

    \While{$h \notin \{0,1, \mid, \_\}$}
        \State Move right without changing the tape
    \EndWhile
    \If{$h \in \{0,1\}$}
        \State $b \gets h$
        \If{$h == 0$}
            \State Write $0^*$ to the tape and move right
        \ElsIf{$h==1$}
            \State Write $1^*$ to the tape and move right
        \EndIf
        \While{$h \neq \_$}
            \State Move right without changing the tape
        \EndWhile
        \While{$h \neq$ $\mid'$}
            \State Move left without changing the tape
        \EndWhile
        \While{$h \notin \{0,1\}$}
            \State Move right without changing the tape
        \EndWhile
        \State Write $b^*$ to the tape and move right
        \State \textbf{goto} \ref{algstep:temp2iufortemp1-Copyback-start}
    \Else 
        \While{$h \neq \_$}
            \State Move right without changing the tape
        \EndWhile
        \While{$h \neq$ $\mid'$}
            \State Move left without changing the tape
        \EndWhile
        \While{$h \neq$ $\mid$}
            \If{$h == 0^*$}
                \State Write $0$ to the tape and move right
            \ElsIf{$h==1^*$}
                \State Write $1$ to the tape and move right
            \Else
                \State Move right without changing the tape
            \EndIf
        \EndWhile
        \State \textbf{accept}
    \EndIf
\end{algorithmic}
\end{algorithm}

\subsubsection{Runtime analysis}

$\mathcal{M}_d$ always reaches its accept state in exactly $\mu$ steps, or halts after $<\mu$ steps. The value of $\mu$ is as follows : 

\begin{align*}
\mu 
&= (2\lceil\log|G_T|\rceil+1)(|\texttt{OUTPUT}(t)| + |\Gamma_T|)\notag\\
&\quad + (|\texttt{OUTPUT}(t)|+|\Gamma_T|)\notag\\
& \quad + 2(|\texttt{OUTPUT}(t)|+|\Gamma_T|)\notag\\
&\quad + 2(|\texttt{OUTPUT}(t)|-5)(|\texttt{OUTPUT}(t)+|\Gamma_T|)+|\texttt{OUTPUT}(t)|
\end{align*}

We note that $|\texttt{OUTPUT}(t)| = \lceil \log |T|\rceil + 4\lceil \log G_T\rceil+5$ (for all $t \in T$), and that $|\Gamma_T| = |T|(\lceil \log |T|\rceil + \lceil \log |G_T|\rceil+5)$, which gives us : 

\begin{align*}
\mu 
&= (2\lceil \log|G_T|)\rceil+1)((|T|+1)\lceil \log |T|\rceil + 4\lceil \log |G_T|\rceil + 5))\notag\\
&\quad +(|T|+1)(\lceil \log |T|\rceil + 4\lceil \log |G_T|\rceil + 5)\notag\\
&\quad + 2(|T|+1)(\lceil \log |T|\rceil + 4\lceil \log|G_T|\rceil + 5)\notag\\
&\quad + 2(\lceil \log |T|\rceil +4\lceil \log |G_T|\rceil)(|T|+1)(\lceil \log |T|\rceil+\lceil \log |G_T|\rceil + 5)
\end{align*}

Notice that the runtime, $\mu$, of $\mathcal{M}_d$, on input $\mathtt{SIDE\mathunderscore ENCODING}(t)$ is independent of both the tile-type $t$, and the length of $\texttt{blank}(T)$. We will now derive $|\texttt{blank}(T)|$ by evaluating the equality $\mu = \frac{|\mathtt{SIDE\mathunderscore ENCODING}(t)|-1}{2}$, which gives us : 

\begin{align*}
\mu 
&= \frac{|\texttt{OUTPUT}(t)|+|\Gamma_T|+|\texttt{blank}(T)|-1}{2}
\end{align*}

Therefore : 

\begin{align*}
|\texttt{blank}(T)| = 2\mu - (|T|+1)(\lceil \log |T|\rceil + 4\lceil \log |G_T|\rceil + 5)
\end{align*}

Notice that $|\mathtt{SIDE\mathunderscore ENCODING}(t)|=2\mu+1$, and therefore $|\mathtt{SIDE\mathunderscore ENCODING}(t)|$ must always be odd. Additionally, for the duration of this proof we will refer to $|\mathtt{SIDE\mathunderscore ENCODING}(t)|$ as $l_T$, and note that $l_T = O(|T|\log^2|T|)$.

\subsubsection{Simulation of $\mathcal{M}_{d}$ by $\calS'$}

Recall that when a macrotile $M(l)$ receives an input glue on side $d$, the set of glues along the edges of its $L_d$ and $R_d$ regions encode $\mathtt{SIDE\mathunderscore ENCODING}(t)$ and $\mathtt{SIDE\mathunderscore ENCODING}(t)^R$, respectively.
Within each region $L_d$ and $R_d$ grows a modified version of the standard zig-zag Turing machine design \cite{CookFuSch11}, which each compute $\mathcal{M}_{d}$. The tiles used to simulate both of these TMs are of the set $T_{\mathcal{M}_d} \subset \calU$. 

The zig-zag TM design used in the literature simulates a standard TM with a one(or two)-sided infinite tape. Each transition to the right is performed by a path of tiles, starting on the left, wherein each tile makes a cooperative attachment with the tile placed in the location to its left, and the row below (i.e. in a ``zig''), and each transition to the left is performed by a path of tiles, starting from the right, wherein each tile tile attaches via. a cooperative attachment with the tile placed in the location to its right, and the row below (i.e. in a ``zag''). A singly-infinite zig-zag TM expands to its right by growing the size of the simulated tape by $1$ on every zig, and a doubly-infinite zig-zag also expands to its left by growing the size of the simulated tape by $1$ on every zag.

Our modified version of the standard zig-zag TM construction does not expand the size of the simulated tape during any zig, nor any zag. Additionally, in order to successfully co-simulate $\mathcal{M}_{d}$, the simulation of $\mathcal{M}_{d}$ within the both the $L_d$ and $R_d$ regions must always make the \textit{same} choice as to whether to make a transition using $\delta_{1,d}$ or $\delta_{2,d}$. We call this behavior \textit{co-simulation}. In order to allow for co-simulation of $\mathcal{M}_d$ in the $L_d$ and $R_d$ regions of a macrotile, we make the following changes to the standard zig-zag TM construction which grows within $L_d$: 

\begin{enumerate}
    \item During the first zig, no transitions of $\mathcal{M}_d$ are simulated, and each simulated transition of $\mathcal{M}_{d}$ instead occurs during a (zag,zig) pair.
    \item At the end of every zig, attach a tile with a strength-2 glue towards the $R_d$ regions, into the single-tile wide space between the $L_d$ and $R_d$ regions (i.e. it grows over one additional space toward the $R_d$ region).
    \item Onto the $\texttt{OPP}(d)$ side of the aforementioned tile, one of two tile types $t_{r,d} \in T_{\mathcal{M}_d}$ for $r \in \{1,2\}$ nondeterministically attaches.
    \item On each side of $t_{r,d}$ which is in $\texttt{ADJ}(d)$ is a strength-1 glue. Cooperation with the aforementioned glue initiates the growth of the next (zag, zig) pair of the movements in the simulations in $L_d$ and $R_d$.

    The next transition to be simulated in the subsequent (zag,zig) pair will be in $\delta_{r.d}$.
\end{enumerate}

In concert, the following changes are made to the simulation of $\mathcal{M}_{d}$ in the $R_d$ region: 

\begin{enumerate}
    \item During the first zig, no transitions of $\mathcal{M}_{d}$ are simulated, and each simulated transition of $\mathcal{M}_{d}$ occurs during a (zag,zig) pair.
    \item At the end of every zig are no strength-2 glues. Growth of the next zag may only be initiated by cooperation with a tile $t_{r,d}$, for $r \in \{0,1\}$, placed in the single-tile wide space between $L_d$ and $R_d$. In this way, both $R_d$ and $L_d$ share the same nondeterministic choices.
    \item During each (zag,zig) pair of movements, only transitions in $\delta_{r,d}$ may be made.
\end{enumerate}

An illustration of our construction which co-simulates the fixed-width, probabilistic $TM$ $\mathcal{M}_S$ is shown in \cref{fig:temp2iufortemp1-probabilistic-TM}. Recall that $\mathcal{M}_d$ is operated on an input of the form $\mathtt{SIDE\mathunderscore ENCODING}(t)$, and always either reaches an accept state after exactly $\frac{l_T-1}{2}$ steps, or halts after $<\frac{l_T-1}{2}$ steps (either as a result of $b_{g_{t,\texttt{OPP}(d)}}$ being the null glue, or there being no tile type such that a simulated attachment may occur, in both cases $\mathcal{M}_d$ does not complete all $4$ computation stages and growth from that input direction halts.). Thus, in all cases where the input glue is sufficient to allow for simulation of a tile attachment, our construction of the co-simulator of $\mathcal{M}_d$ reaches a height of exactly $2(\frac{l_T-1}{2})+1=l_T$, and the maximal dimensions of our co-simulator is exactly $l_T \times 2l_T+1$ (or $2l_T+1 \times l_T$), which neatly fits into the $L_d$ and $R_d$ regions of $M(l)$ (with the single-tile wide space between them). The aforementioned tiles which encode the accept state of $\mathcal{M}_d$ are of the type $t_{accept,d,L}$ and $t_{accept,d,R}$ (in the $L_d$ and $R_d$ regions, respectively), and to these tiles occur subsequent attachments which complete the zig on which the tile was placed in both $L_d$, and $R_d$ such that the last attachment in $L_d$ and $R_d$ occurs simultaneously on the $(l_T(l_T+1)-1)$th step of the co-simulator's growth, and is in the corner of the adjacent $fill$ region, and the single-tile wide space between them. Lastly, we note the glues of each row in the co-simulator which form a cooperative attachment with a tile in the subsequent row are of the form $g_{\lambda\in\Sigma_{\mathcal{M}}, d}$, every strength-1 glue which is placed on the top-most row of the co-simulator is of the form $g_{\lambda\in\Sigma_{Macrotile}, d}$, and the last tiles placed in $L_d$ and $R_d$ each have a strength-2 glue which abuts the adjacent $fill$ region, and initiates the next step of the growth procedure, which we detail in \cref{sec:temp2iufortemp1-macrotile-growth-procedure}.

\begin{figure}
    \centering
    \includegraphics[width=1\linewidth]{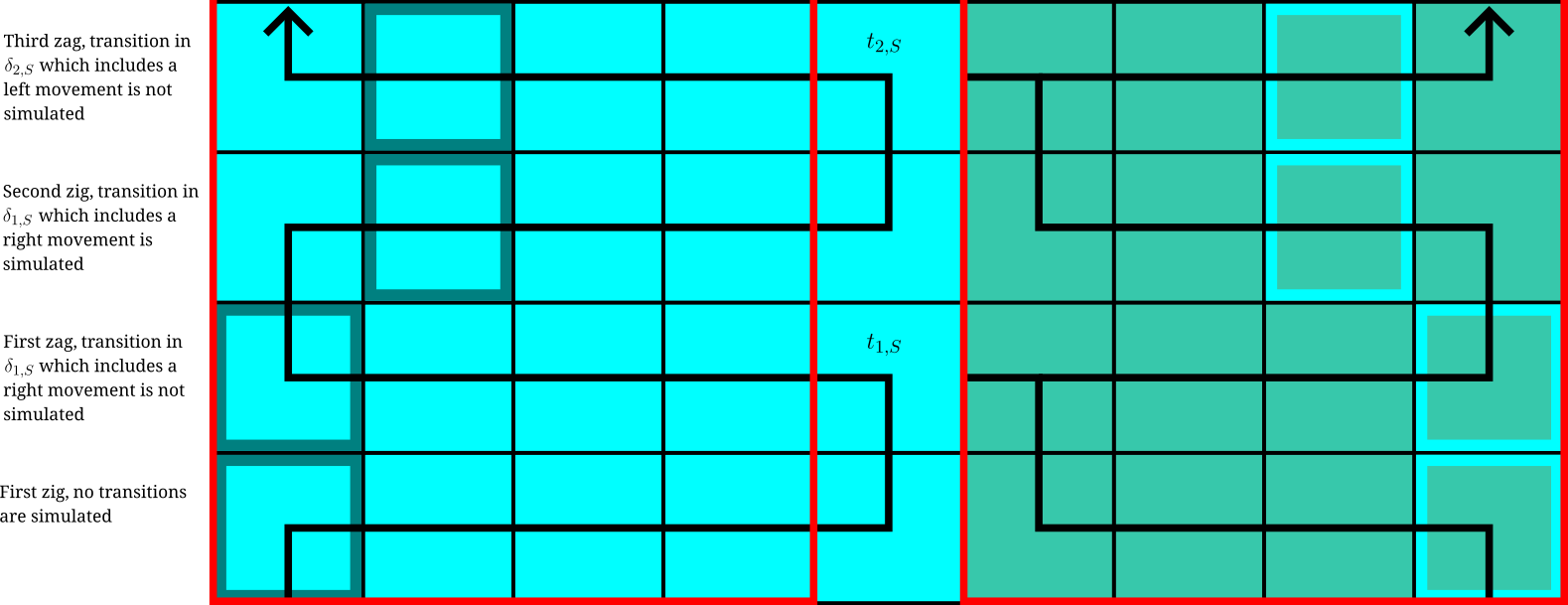}
    \caption{High-level illustration of two zig-zag TMs co-simulating $\mathcal{M}_S$ in the $L_S$ and $R_S$ regions of a macrotile location, with their widths/heights truncated. $L_S$ and $R_S$ are each outlined with red borders, and the tiles whose type corresponds to the head of $\mathcal{M}_{S}$ are outlined.} 
    \label{fig:temp2iufortemp1-probabilistic-TM}
\end{figure}

\subsection{Macrotile growth procedure} \label{sec:temp2iufortemp1-macrotile-growth-procedure}

Let $\vec{\alpha}$ be an arbitrary assembly sequence in $\calT$. Let $l \in \mathbb{Z}^2$ be a location in $\calT$ that is initially empty (i.e. not in the seed $\sigma$), and $\texttt{IN}(l,i) \in \{N,E,S,W,\emptyset\}$ be the direction such that at step $i$ in $\vec{\alpha}$ a tile is placed in a location adjacent to $l$, in direction $d$, which has a non-null glue on side $\texttt{OPP}(d)$ (or $\emptyset$, if either no such tile is placed, or if $\texttt{OPP}(d)$ contains a null glue). That is, $\texttt{IN}(l,i)$ represents the side of $l$ adjacent to a neighboring location that received a tile in step $i$ and has a glue facing $l$. This side is a (potential) input side, as that neighboring glue may allow binding to a tile placed in location $l$. Note that for the context of this construction, we are guaranteed that $|\texttt{IN}(l,i)| \le 1$ as a result of $\calT$ being staggered.

Given location $l \notin \dom \sigma$, let $M(l)$ be the macrotile in $\calS'$ that maps to location $l$, and $0 \le i < |\vec{\alpha}|$ be some assembly step in $\vec{\alpha}$. We will now discuss the growth procedure which occurs in $M(l)$ in $\calU$ each macrostep $i+1$ such that $\texttt{IN}(l,i) \neq \emptyset$.

Let $d \in \texttt{IN}(l,i)$ be the side of $M(l)$ which experiences input growth during macrostep $i+1$ (recall that $|\texttt{IN}(l,i)| \le 1$ by the fact that the system $\calT$ which is being simulated is staggered), and let $t \in T$ be the tile whose simulated attachment was performed in an adjacent macrotile during the previous macrostep. Input growth within $M(l)$ starts in $L_d$, and $R_d$, whose side $d$ contains a sequence of glues of the form $g_{\lambda\in\Sigma_{Macrotile}, d}$ corresponding to $\mathtt{SIDE\mathunderscore ENCODING}(t)$ and $\mathtt{SIDE\mathunderscore ENCODING}(t)^R$, respectively. The first tile attachments in $M(l)$ during macrostep $i+1$ occur synchronously within the $L_d$ and $R_d$ regions of $M(l)$, and are to the strength-2 glues $\mathtt{MARK\mathunderscore IN}(g_{M_d,L},d)$ and $\mathtt{MARK\mathunderscore IN}(g_{M_d,R},d)$, respectively. Input growth within the $L_d$ and $R_d$ regions of $M(l)$ performs the co-simulation of $\mathcal{M}_d$, described in \cref{sec:temp2iufortemp1-tile-lookup-procedure}. 

Once growth within the $L_d$ and $R_d$ regions is complete, if there is a tile type $t' \in T$ that can attach using the glue presented to side $d$ of $M(l)$,  there will be a sequence of glues of the form $g_{\lambda \in \Sigma_{Macrotile},d}$ which encode $\mathtt{SIDE\mathunderscore ENCODING}(t')$ and  $\mathtt{SIDE\mathunderscore ENCODING}(t')^R$,  along the $\texttt{OPP}(d)$ sides of $L_d$ and $R_d$, respectively. All subsequent attachments within $M(l)$ during macrostep $i+1$ are dedicated to propagating these glues to each side $d'\in \{N,E,S,W\} -\{d\}$ of $M(l)$. These attachments occur in regions of $M(l)$ labeled $fill$, $L$, and $R$, whose constituent growth is of instances of propagation gadgets (which we presented in \cref{sec:temp2iufortemp1-reusable-gadgets}), and the orientation of each of these gadgets is depicted in \cref{fig:temp2iufortemp1-macrotile-layout-grid}. 

Notice that two square-propagation gadgets colored light-green must cross the single-tile wide space which extends across the center of the macrotile. In order to accomplish this, we modify the square-propagation gadgets which grow within these regions such that the tiles placed on their $\texttt{OPP}(d_{in})$ side each have a strength-2 glue onto which a tile may attach (therefore raising the height of both of said propagation gadgets by $1$ tile). Propagation gadgets which chain onto an adjacent side to $d_{in}$ of these modified gadgets do so as normal. However, propagation gadgets which chain onto the $\texttt{OPP}(d_{in})$ must do so through a cooperative attachment to a ``synchronization path'', which is simply a two-tile long ``nub'' which grows into the single-tile wide region between two neighboring $fill$ regions from exactly one propagation gadget. In this case, the gadget which grows the synchronization path is the right-most one. A depiction of this behavior is shown in \cref{fig:temp2iufortemp1-crossing-the-gap}

\begin{figure}
    \centering
    \includegraphics[width=\linewidth]{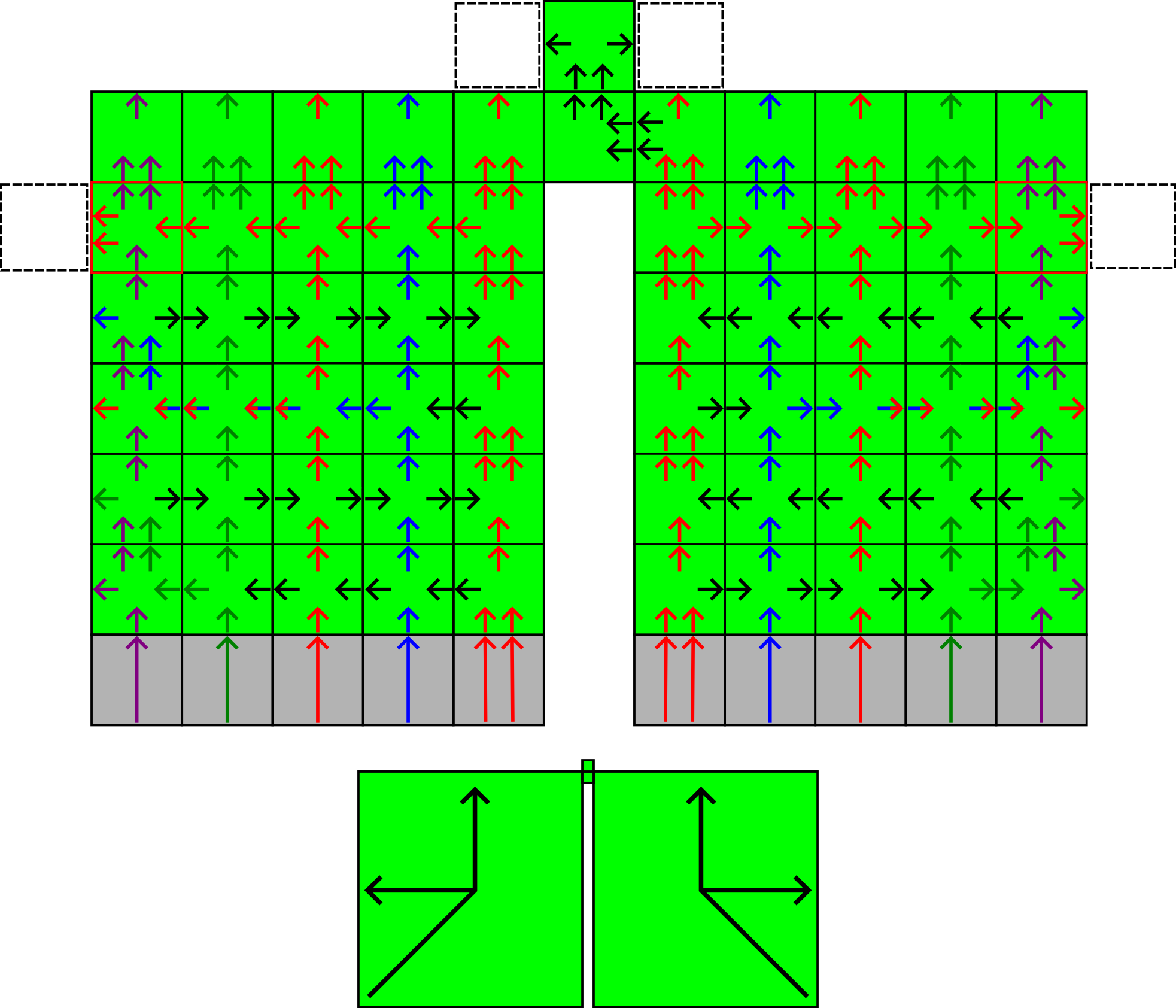}
    \caption{A depiction of two modified propagation gadgets which cross the one-tile wide region in the center of the macrotile. Colors are used to indicate the symbol encoded by each glue. The ``end'' tile of each gadget is outlined in red, and locations where subsequent chained propagation gadgets may attached are outlined with a dotted line.}
    \label{fig:temp2iufortemp1-crossing-the-gap}
\end{figure}

Additionally notice that fill regions which chain off of the side adjacent to those propagation gadgets colored light-green complete assembly at an earlier assembly step than their neighbors. In order to ensure that subsequent neighboring propagation gadgets grow synchronously, this is another location wherein the attachment of a start tile to any subsequent propagation gadgets may only occur through a cooperative bind to a synchronization path. In this case, the synchronization path exclusively attaches to the propagation gadget which arrives last. 

During the last assembly step of macrostep $i+1$, a tile is placed within the $L_{d'}$ and $R_{d'}$ region of $M(l)$ (for all $d' \in \{N,E,S,W\}-\{d\}$) which initiates growth into the macrotile location adjacent to side $d'$ of $M(l)$. Additionally, on the side of $L_{d'}$ which abuts an adjacent macrotile location is a series of strength-1 glues (except for one) of the form $g_{\lambda \in \Sigma_{Macrotile}, d'}$ which encode $\mathtt{SIDE\mathunderscore ENCODING}(t')^R$, and the parallel side of $R_{d'}$ is a series of strength-1 glues of the same which encode $\mathtt{SIDE\mathunderscore ENCODING}(t')$. The aforementioned exceptions are a strength-2 glue in the $L_{d'}$ region of the form $\mathtt{MARK\mathunderscore OUT}(g_{\mathcal{M}_{\texttt{OPP}(d')}, R},d')$, and a strength-2 glue in the $R_{d'}$ region of the form $\mathtt{MARK\mathunderscore OUT}(g_{\mathcal{M}_{\texttt{OPP}(d')},L},d')$. Each of the aforementioned glues initiate growth of the co-simulator of $\mathcal{M}_{\texttt{OPP}(d')}$ in the adjacent $L_{\texttt{OPP}(d')}$, and $R_{\texttt{OPP}(d')}$ regions of the adjacent macrotile location, and appear on the last tiles to be placed in $L_{d'}$, and $R_{d'}$ regions of $M(l)$.

We note that even though for all locations $l \in \mathbb{Z}^2$, and macrosteps $i$ of $\calS$, $|\texttt{IN}(l,i)| \le 1$, but it is not necessarily true that input growth within $M(l)$ occurs on only a single side $d \in \texttt{IN}(l,i)$. This is a result of the fact that output information is presented on each side of a macrotile which performs a simulated attachment of the tile $t'\in T$, even when the glue on the respective side is the null glue. However, this does not impact the growth procedure, as all sides of $M(l)$ for which input growth occurs corresponding to the null glue of an adjacent macrotile that growth halts within the $L$ and $R$ regions of that side by the tile lookup procedure provided in \cref{sec:temp2iufortemp1-macrotile-growth-procedure}. As such, output propagation within $M(l)$ that simulates the attachment of $t' \in T$ during macrostep $i+1$ is able to occur unimpeded for all sides of $M(l)$ which are adjacent to empty macrotile locations.

\subsection{Synchronization and scale factor}\label{sec:temp2iufortemp1-synchronization}
In this section, we will prove that the temporal scale factor $s$ of $\calS'$ is $6l_T^2+l_T+4$, and that the spatial scale factor $c$ of $\calS'$ is $6l_T+1$. 

Let us start by documenting the number of assembly steps that are necessary to fully assemble the propagation gadgets shown in \cref{sec:temp2iufortemp1-reusable-gadgets}, along with number of assembly steps necessary to complete the co-simulation of $\mathcal{M}_{d}$ shown in \cref{sec:temp2iufortemp1-tile-lookup-procedure}. For both of those gadgets, we assume that their input is of length $l_T$, which is equal to $|\mathtt{SIDE\mathunderscore ENCODING}(t)|$ for all $t\in T$. This information is presented in \cref{tab:temp2iufortemp1-gadget-steps}.

\begin{table}[!htp]\centering\small

%\resizebox{\textwidth}{!}{ % use this if the table is too large
\begin{tabular}{lrr}\toprule
Gadget Type &Number of Assembly Steps \\\midrule
Propagation &$l_T^2$ \\
Co-simulation of $\mathcal{M}_{d}$ &$l_T(l_T+1)-1$ \\
\bottomrule
\end{tabular}
\caption{The number of assembly steps required to complete the growth of two gadget types which appear within $M(l)$ during some macrostep such that $M(l)$ receives input growth.}\label{tab:temp2iufortemp1-gadget-steps}
\end{table}

We will now discuss the assembly step during which each region in a macrotile completes growth. Let $\vec{\alpha}$ be an arbitrary assembly sequence of $\calT$, location $l \notin \dom \sigma$, and $0 \le i < \vec{\alpha}$ be some assembly step in $\vec{\alpha}$ such that $\texttt{IN}(l,i) \neq \emptyset$. During macrostep $i+1$, $d \in \texttt{IN}(l,i)$ is a side of $M(l)$ which experiences input growth, and may or may not resolve to a tile type $t \in T$ during macrostep $i+1$. Without loss of generality, we will restrict ourselves to the discussion of macrotiles which experience input growth from a single side during macrostep $i+1$, and for whom all but the $d$ side never experienced input growth during a prior macrostep.

We say that a region of $M(l)$ initiates growth, and call it an \textit{initiating} region iff after the placement of the last tile within that region, a tile is placed in at least one adjacent region $\le 4$ assembly steps later. We will describe growth within $M(l)$ with respect to \textit{phases}, and we denote phase $1 \le j \le 5, j \in \mathbb{N}$ as $h_j$. Each phase is associated with a set of initiating regions of $M(l)$ which we denote as $\mathtt{REGIONS}(h_j,d)$. $\mathtt{REGIONS}(h_1,d)$ is the regions $L_d$ and $R_d$, and $h_1$ is the set the assembly steps between (and including) $i\cdot s+1$ and the last tile placement within $L_d$ and $R_d$. $\mathtt{REGIONS}(h_j,d)$ (for all $j > 1$) is all initiating regions adjacent to at least one region in $\mathtt{REGIONS}(h_{j-1}, d)$. Consequentially, a growth phase $h_j$ is the set of assembly steps wherein tiles within their respective growth regions are placed. For all $j \le 4$, the growth phase $h_j$ starts at the first assembly step where a tile attaches in at least one region in $\mathtt{REGIONS}(h_j,d)$, and ends at the assembly step prior to the first assembly step where a tile attaches in at least one region in $\mathtt{REGIONS}(h_{j+1},d)$. For $j==6$, $h_j$ starts at assembly step of the first tile attachment in $\mathtt{REGIONS}(h_j,d)$, and ends at the last tile attachment within any region in $\mathtt{REGIONS}(h_j,d)$. A guide to the phase of each initiating region of $M(l)$ where $\texttt{IN}(l,i) == S$ is provided in \cref{fig:temp2IUfortemp1-macrotile-growth-phases} , and a table depicting the number of assembly steps, along with the first / last assembly step of each phase is provided in \cref{tab:temp2iufortemp1-phase-steps}.
\begin{figure}
    \centering
    \includegraphics[width=0.5\linewidth]{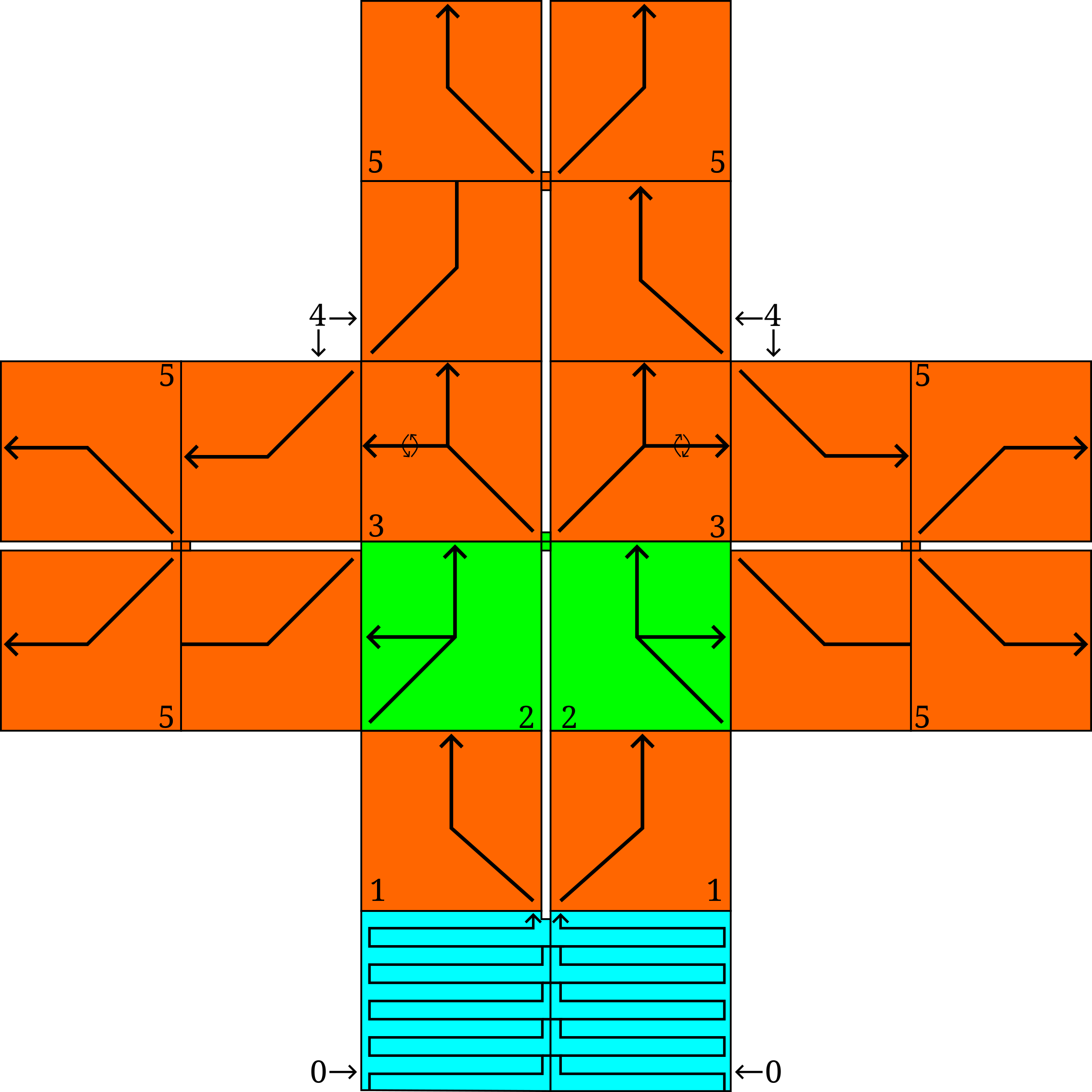}
    \caption{A (very) high level depiction of a macrotile of $\calS'$. Squares colored in light-blue represent zig-zag assemblies which co-simulate $\mathcal{M}_{d}$, squares colored orange represent propagation gadgets, and squares colored light-green represent modified propagation gadgets. Each region which is an initiating region has its phase denoted by a number within its borders where permitted, or to its side.}
    \label{fig:temp2IUfortemp1-macrotile-growth-phases}
\end{figure}

\begin{table}[!htp]\centering\small

%\resizebox{\textwidth}{!}{ % use this if the table is too large
\begin{tabular}{lrrrr}\toprule
Phase &First assembly step &Last assembly step &Number of assembly steps \\\midrule
1 &$i \cdot s+1$ &$i \cdot s + l_T(l_T+1)-1$ &$l_T(l_T+1)-1$ \\
2 &$i \cdot s + l_T^2+l_T$ &$i \cdot s + 2l_T^2+l_T-1$ &$l_T^2$ \\
3 &$i \cdot s + 2l_T^2+l_T$ &$i \cdot s + 3l_T^2+l_T+2$ &$l_T^2+3$ \\
4 &$i \cdot s + 3l_T^2+l_T+3$ &$i \cdot s + 4l_T^2+l_T+2$ &$l_T^2$ \\
5 &$i \cdot s +4l_T^2+l_T+3$ &$i \cdot s + 5l_T^2+l_T+4$ &$l_T^2+2$ \\
6 &$i \cdot s + 5l_T^2+l_T+5$ &$i \cdot s + 6l_T^2+l_T+4$ &$l_T^2$ \\
\bottomrule
\end{tabular}
%}
\caption{Table which depicts the first and last assembly step, along with the total number of assembly steps during which tiles are placed corresponding to each phase of growth within $M(l)$ during some macrostep such that $M(l)$ receives input growth, and has at least one adjacent macrotile location to which it may express output. The total number of assembly steps necessary to complete each growth phase is with respect to $l_T = |\mathtt{SIDE\mathunderscore ENCODING}(t)|$ (for all $t \in T$).}\label{tab:temp2iufortemp1-phase-steps}

\end{table}

Let us now describe why each phase contains the number of assembly steps assigned in \cref{tab:temp2iufortemp1-phase-steps}: 

\begin{enumerate}
\item The first phase of growth within $M(l)$ consists of the co-simulation of $\mathcal{M}_d$ within the $L_d$ and $R_d$ regions of $M(l)$. The number of steps this takes is provided in \cref{sec:temp2iufortemp1-macrotile-growth-procedure}, and is equal to $l_T(l_T+1)-1$.

\item The second phase of growth within $M(l)$ consists of the synchronous growth of propagation gadgets within the $fill$ regions adjacent to $\texttt{REGIONS}(h_0)$. Attachment of tiles within subsequent initiating regions is to the ``end'' tile of each of those propagation gadgets, exactly $l_T^2$ assembly steps after the start of this growth phase.

\item The third phase of growth within $M(l)$ consists of the synchronous growth of propagation gadgets within the $fill$ regions adjacent to $\texttt{REGIONS}(h_1)$, whose ``end'' tiles also attach after exactly $l_T^2$ steps of growth. However, the attachment of tiles within subsequent initiating regions may only occur after the attachment of a $3$-tile long synchronization path to one of the aforementioned synchronization gadgets, which results in this phase of growth requiring exactly $l_{T}^2+3$ assembly steps to complete.

\item The fourth phase of growth within $M(l)$ consists of the synchronous growth of two propagation gadgets. Subsequent initiating regions may only attach tiles attach to the ``end'' tile of each of these propagation gadgets, exactly $l_T^2$ assembly steps after the start of this growth phase.

\item The fifth phase of growth within $M(l)$ consists of the synchronous growth of four propagation gadgets. Each of these gadgets place their ``end'' tile after $l_T^2$ assembly steps, and future initiating regions may only attach after the attachment of a $2$-tile long synchronization path to all but one of these gadgets, which results in this phase of growth requiring exactly $l_T^2+2$ assembly steps to complete.

\item The sixth and final growth phase within $M(l)$ consists of the synchronous growth of four propagation gadgets, each of which require $l_T^2$ assembly steps to place their last tile. 
\end{enumerate}

$|h_j|$ ($1 \le j \le6$) is the number of assembly steps it takes to complete the $j$th growth phase within a macrotile. The number of assembly steps necessary for a macrotile $M(l)$ to resolve to a tile type $t \in T$ during macrostep $i$, and express glues to adjacent macrotile locations sufficient to initiate input growth is $\Sigma_{j=1}^6 |h_j|=6l_T^2+l_T+4$. Therefore, the temporal scale factor of $\calS'$ is $s = \Sigma_{j=0}^5 |h_j|=6l_T^2+l_T+4=O(|T|^2\log|T|^4)$. 

Additionally, we note that (as per \cref{sec:temp2iufortemp1-macrotile-structure}), a macrotile location consists of a $6 \times 6$ grid of $l_T \times l_T$ regions, with a single-tile wide cross through the center, and thus the spatial scale factor of each macrotile location is $c = 6l_T+1=O(|T|\log |T|^2)$.

\subsection{Seed structure} \label{sec:temp2iufortemp1-seed-structure}

The seed $\mathfrak{Y}(\calT)$ of $\calS' = (\calU, \mathfrak{Y}(\calT), 2)$ is composed of a special set of tile types $T_{\mathfrak{Y}} \subset \calU$. Here we explain the tile types in $T_{\mathfrak{y}}$, and the function $\mathfrak{Y}$ that takes as input a staggered temperature 1 syncTAS $\calT$, and outputs the seed of $\calS'$ such that $\calS'$ simulates simulates $\calT$ under $\mathfrak{R}(\calT)$ (given in \cref{sec:temp2iufortemp1-representation-function}) with spatial scale factor $c$, and temporal scale factor $s$. $\mathfrak{Y}$ differs from the seed generation functions provided in \cref{sec:adj-seed} and \cref{sec:temp1simulation-seed} in that the tile set $\calU$ of $\calS'$ is fixed, and thus the same set of seed tile types must be reusable for all simulations. Additionally, because each macrotile in $\mathfrak{Y}(\calT)$ must map to a tile type under $\mathfrak{R}(\calT)$, an \textit{indicative bit-string} made up of tiles types in $T_{\mathfrak{Y}}$ must be present in each macrotile in $\mathfrak{Y}(\calT)$. However, the structure of this section is still largely similar to that of \cref{sec:adj-seed}. 

\subsubsection{I/O marking}

As is described in \cref{sec:adj-seed}, it is necessary to ensure the following two behaviors of $\calS'$: 

\begin{enumerate}
\item When a glue is exposed on the exterior of a macrotile to propagate output, the tile that attaches to it initiates input growth into the neighboring macrotile location.
\item The tiles of $T_{\mathfrak{Y}}$ must not attach in any locations outside of the pre-formed macrotiles representing the seeds.
\end{enumerate}

We ensure these behaviors through the usage of $I/O$ marking. That is to say, all tiles which face the edge of a macrotile that finish output growth have their strength-2 exterior facing glue marked using $\mathtt{MARK\mathunderscore OUT}$, and all tiles which attach to the edge of said output growth to begin the growth of input do so with glues marked using $\mathtt{MARK\mathunderscore IN}$. As with \cref{sec:temp1simulation-seed}, we defer further details regarding I/O marking to \cref{sec:adj-seed}.

\subsubsection{Seed tile types}

The tile types which constitute $T_{\mathfrak{Y}}$ are as follows, note that all glues except those marked using $\mathtt{MARK \mathunderscore OUT}$ are either only used by tile types in $T_{\mathfrak{Y}}$, or are strength-1. 

\begin{enumerate}
\item A central tile type $c_{\mathfrak{Y}}$ which has a strength-2 glue $g_{\mathfrak{Y}}$ on each side.

\item A \textit{vertical} tile type $v_{\mathfrak{Y}}$ that has $g_{\mathfrak{Y}}$ on its $N$ and $S$ sides.

\item For each side $d \in \{N,E,S,W)$, a $d$ \textit{connector} tile type $d_{\mathfrak{Y}}$ which has the strength two glue $g_{\mathfrak{Y}}$ on its $d$, and $d_{adj} \in \texttt{ADJ}(d)$ sides.

\item For each side $d \in \{N,S\}$, the following tile types : 

\begin{enumerate}

\item For every symbol $\lambda \in \Sigma_{Macrotile}$, a \textit{symbol printing} tile type $d_{\mathfrak{Y}, \lambda}$ which has the strength-2 glue $g_{\mathfrak{Y}}$ on its $E$ and $W$ sides and the strength-1 glue $g_{\lambda}$ on its $d$ side.

\item If $d==N$ : 

\begin{enumerate}

\item $E, R$ \textit{string cap} tile type $d_{\mathfrak{Y}, E,R}$ which has the strength-2 glue $g_{\mathfrak{Y}}$ on its $W$ side and, strength-2 glue $\mathtt{MARK\mathunderscore OUT}(g_{\mathcal{M}_{\texttt{OPP}(d')},R},d)$ on its $d$ side.

\item $W, L$ \textit{string cap} tile type $d_{\mathfrak{Y}, W,L}$ which has the strength-2 glue $g_{\mathfrak{Y}}$ on its $E$ side, and strength-2 glue $\mathtt{MARK\mathunderscore OUT}(g_{\mathcal{M}_{\texttt{OPP}(d')},L},d)$ on its $d$ side.
\end{enumerate}

\item If $d==S$ : 

\begin{enumerate}

\item $E, L$ \textit{string cap} tile type $d_{\mathfrak{Y}, E,L}$ which has the strength-2 glue $g_{\mathfrak{Y}}$ on its $W$ side and, strength-2 glue $\mathtt{MARK\mathunderscore OUT}(g_{\mathcal{M}_{N},L},S)$ on its $S$ side.

\item $W, R$ \textit{string cap} tile type $d_{\mathfrak{Y}, W,R}$ which has the strength-2 glue $g_{\mathfrak{Y}}$ on its $E$ side, and strength-2 glue $\mathtt{MARK\mathunderscore OUT}(g_{\mathcal{M}_{N},R},S)$ on its $S$ side.

\end{enumerate}

\end{enumerate}

\item A \textit{horizontal} tile type $h_{\mathfrak{Y}}$ that has $g_{\mathfrak{Y}}$ on its $E$ and $W$ sides.

\item For each side $d \in \{E,W\}$, the following tile types : 

\begin{enumerate}

\item For every symbol $\lambda \in \Sigma_{Macrotile}$, a \textit{symbol printing} tile type $d_{\mathfrak{Y}, \lambda}$ which has the strength-2 glue $g_{\mathfrak{Y}}$ on its $N$ and $S$ sides and the strength-1 glue $g_{\lambda}$ on its $d$ side.

\item If $d==E$ : 

\begin{enumerate}
\item $N, L$ \textit{string cap} tile type $d_{\mathfrak{Y}, N,L}$ which has the strength-2 glue $g_{\mathfrak{Y}}$ on its $S$ side and, strength-2 glue $\mathtt{MARK\mathunderscore OUT}(g_{\mathcal{M}_{W},L},E)$ on its $E$ side.

\item $S, R$ \textit{string cap} tile type $d_{\mathfrak{Y}, S,R}$ which has the strength-2 glue $g_{\mathfrak{Y}}$ on its $N$ side, and strength-2 glue $\mathtt{MARK\mathunderscore OUT}(g_{\mathcal{M}_{W},R},E)$ on its $E$ side.
\end{enumerate}

\item If $d==W$ : 

\begin{enumerate}

\item $N, R$ \textit{string cap} tile type $d_{\mathfrak{Y}, N,R}$ which has the strength-2 glue $g_{\mathfrak{Y}}$ on its $S$ side and, strength-2 glue $\mathtt{MARK\mathunderscore OUT}(g_{\mathcal{M}_{E},R},W)$ on its $W$ side.

\item $S, L$ \textit{string cap} tile type $d_{\mathfrak{Y}, S,L}$ which has the strength-2 glue $g_{\mathfrak{Y}}$ on its $N$ side, and strength-2 glue $\mathtt{MARK\mathunderscore OUT}(g_{\mathcal{M}_{E},L},W)$ on its $W$ side.
\end{enumerate}

\end{enumerate}

\item An \textit{indicative bit-string start} tile type $B_{\mathfrak{Y}}$ which has the strength-2 glue $g_{\mathfrak{Y}}$ on its $N$, $E$, and $S$ sides.

\item An \textit{indicative bit-string 0} tile type $B_{\mathfrak{Y},0}$ which has the strength-2 glue $g_{\mathfrak{Y}}$ on its $E$ and $W$ sides.

\item An \textit{indicative bit-string} $1$ tile type $B_{\mathfrak{Y},1}$ which has the strength-2 glue $g_{\mathfrak{Y}}$ on its $E$ and $W$ sides.

\item An \textit{indicative bit-string cap} tile type $B_{\mathfrak{Y}, c}$ which has the strength-2 glue $g_{\mathfrak{Y}}$ on its $W$ side.

\end{enumerate}

\subsubsection{Seed creation} The seed function generator $\mathfrak{Y}$, which takes as input the staggered temperature 1 syncTAS $\calS$ and outputs the seed $S(\calT)$ is the following :

\begin{enumerate}
    \item For each tile, given its tile type $t \in T$ and location $(x,y) \in \mathbb{Z}^2$, such that $(t,(x,y)) \in \sigma$:
    \begin{enumerate}
        \item The locations of the macrotile representing the tile will be $M((x,y)) = \{(x',y')\mid cx \le x' < c(x+1)$ and $cy \le y' < c(y+1)\}$.

        \item Let $m = (c-1)/2$ be the distance from a location on the boundary of a macrotile to the center location. The macrotile's center location $l_c = (cx+m, cy+m)$. A tile of type $c_\mathfrak{Y}$ is placed in location $l_c$.

        \item For each side $d$ of $t$ that has a glue $g$ (including the null glue), and that glue does not form a bond with a neighboring tile in $\sigma$:

        \begin{enumerate}
            \item If $d == N$ :

            \begin{enumerate}
                \item The vertical tile $v_\mathfrak{Y}$ is placed on the $N$ side of the center location, with additional copies of it placed to its $N$ direction, until stopping one tile away from the edge of the macrotile.\label{item:green1}

                \item On the $N$ side of the final $v_{\mathfrak{Y}}$ tile is a tile of the type $N_{\mathfrak{Y}}$.
                \item On the $E$ and $W$ sides of $N_{\mathfrak{Y}}$, $l_T-1$ tiles are placed such that, for $0 \le i < l_T-1$, the $i$th tile to be placed to the $E$ of $N_{\mathfrak{Y}}$ is $N_{\mathfrak{Y}, \lambda_{E,i}}$, and the $i$th tile to be placed its $W$ is $N_{\mathfrak{Y}, \lambda_{W,i}}$, where $\lambda_{E,i}$ and $\lambda_{W,i}$ are the $i$th symbol of $\mathtt{SIDE\mathunderscore ENCODING}(t)^R$ and the $i$th symbol of $\mathtt{SIDE\mathunderscore ENCODING}(t)$, respectively. \label{item:gold}\label{item:yellow}
                
                \item To the $E$ and $W$ sides of the $l_T-1$th tile to the $E$ and $W$ of $N_{\mathfrak{Y}}$, respectively, tiles of the type $N_{\mathfrak{Y},E,R}$, and $N_{\mathfrak{Y},W,L}$ tiles are placed, respectively.
            \end{enumerate}

            \item If $d == E$ : 

            \begin{enumerate}
                \item the horizontal tile $h_\mathfrak{Y}$ is placed on the $E$ side of the center location, with additional copies of it placed to its $E$ direction, until stopping one tile away from the edge of the macrotile.\label{item:green2}

                \item On the $E$ side of $h_{\mathfrak{Y}}$ is a tile of the type $E_{\mathfrak{Y}}$.
                \item On the $N$ and $S$ sides of $E_{\mathfrak{Y}}$, $l_T-1$ tiles are placed such that, for $0 \le i < l_T-1$, the $i$th tile to be placed to the $N$ of $E_{\mathfrak{Y}}$ is $E_{\mathfrak{Y}, \lambda_{N,i}}$, and the $i$th tile to be placed its $S$ is $E_{\mathfrak{Y}, \lambda_{S,i}}$, where $\lambda_{N,i}$ and $\lambda_{S,i}$ are the $i$th symbol of $\mathtt{SIDE\mathunderscore ENCODING}(t)$ and the $i$th symbol of $\mathtt{SIDE\mathunderscore ENCODING}(t)^R$, respectively.\label{item:pink}\label{item:red}

                \item To the $N$ and $S$ sides of the $(l_T-1)$th tile placed to the $N$ and $S$ of $E_{\mathfrak{Y}}$, respectively, tiles of the type $E_{\mathfrak{Y},N,L}$ and $E_{\mathfrak{Y},S,R}$ are placed, respectively.
            \end{enumerate}

            \item If $d == S$

            \begin{enumerate}
                \item the vertical tile $v_\mathfrak{Y}$ is placed on the $S$ side of the center location.

                \item On the south side of $v_{\mathfrak{Y}}$ is an indicative bit-string start tile of the type $B_{\mathfrak{Y}}$.

                \item On the $E$ side of $B_{\mathfrak{Y}}$ is a set of $\lceil \log |T|\rceil$ tiles such that, for $0 \le i < \lceil \log |T|\rceil$, the $i$th tile to the $E$ of $B_{\mathfrak{Y}}$ is an indicative bit-string tile $B_{\mathfrak{Y}, k}$, where $k$ is the $i$th bit of $b_t$.\label{item:blue}

                \item On the $E$ side of the $\lceil \log |T|\rceil$th tile to the $E$ of $B_{\mathfrak{Y}}$ is an indicative bit-string cap tile of type $B_{\mathfrak{Y},c}$. 

                \item On the $S$ side of $B_{\mathfrak{Y}}$, copies of $v_{\mathfrak{Y}}$ are placed, until stopping one tile away from the $S$ edge of the macrotile.\label{item:dark-green}
                \item On the $S$ side of $v_{\mathfrak{Y}}$ is a tile of the type $S_{\mathfrak{Y}}$.
                \item On the $E$ and $W$ sides of $S_{\mathfrak{Y}}$, $l_T-1$ tiles are placed such that, for $0 \le i < l_T-1$, the $i$th tile to be placed to the $E$ of $S_{\mathfrak{Y}}$ is $S_{\mathfrak{Y}, \lambda_{E,i}}$, and the $i$th tile to be placed its $W$ is $S_{\mathfrak{Y}, \lambda_{W,i}}$, where $\lambda_{E,i}$ and $\lambda_{W,i}$ are the $i$th symbol of $\mathtt{SIDE\mathunderscore ENCODING}(t)$ and the $i$th symbol of $\mathtt{SIDE\mathunderscore ENCODING}(t)^R$, respectively.\label{item:purple}\label{item:cyan}

                \item To the $E$ and $W$ sides of the $l_T-1$th tile to the $E$ and $W$ of $S_\mathfrak{Y}$, respectively, tiles of the type $S_{\mathfrak{Y},E,L}$ and $S_{\mathfrak{Y},W,R}$ are placed, respectively.
                
            \end{enumerate}

            \item If $d == W$
            \begin{enumerate}
                \item the horizontal tile $h_\mathfrak{Y}$ is placed on the $W$ side of the center location, with additional copies of it placed to its $W$ direction, until stopping one tile away from the edge of the macrotile.\label{item:green3}

                \item On the $W$ side of $h_{\mathfrak{Y}}$ is a tile of the type $W_{\mathfrak{Y}}$.
                \item On the $N$ and $S$ sides of $W_{\mathfrak{Y}}$, $l_T-1$ tiles are placed such that, for $0 \le i < l_T-1$, the $i$th tile to be placed to the $N$ of $S_{\mathfrak{Y}}$ is $W_{\mathfrak{Y}, \lambda_{N,i}}$, and the $i$th tile to be placed its $S$ is $W_{\mathfrak{Y}, \lambda_{S,i}}$, where $\lambda_{N,i}$ and $\lambda_{S,i}$ are the $i$th symbol of $\mathtt{SIDE\mathunderscore ENCODING}(t)^R$ and the $i$th symbol of $\mathtt{SIDE\mathunderscore ENCODING}(t)$, respectively.\label{item:grey}\label{item:black}
                \item To the $N$ and $S$ side of the $(l_T-1)$th tile place to the $N$ and $S$ of $W_{\mathfrak{Y}}$, respectively, tiles of the type $W_{\mathfrak{Y},N,R}$ and $W_{\mathfrak{Y}, S,R}$ are placed, respectively.
            \end{enumerate}
        \end{enumerate}
        
        \item For each side $d$ of $t$ that has a non-null glue, and that glue forms a bond with a neighboring tile in $\sigma$:

        \begin{enumerate}
            \item If $d == N$:
            \begin{enumerate}
                \item The vertical tile $v_\mathfrak{Y}$ is placed on the $N$ side of the center location, with additional copies of it placed to its $N$ direction, until stopping at the edge of the macrotile.
            \end{enumerate}

            \item if $d == S$

            \begin{enumerate}
                \item The vertical tile $v_\mathfrak{Y}$ is placed on the $S$ side of the center location.

                \item On the south side of $v_{\mathfrak{Y}}$ is an indicative bit-string start tile of the type $B_{\mathfrak{Y}}$.

                \item On the $E$ side of $B_{\mathfrak{Y}}$ is a set of $\lceil \log |T|\rceil$ tiles such that, for $0 \le i < \lceil \log |T|\rceil$, the $i$th tile to the $E$ of $B_{\mathfrak{Y}}$ is an indicative bit-string tile $B_{\mathfrak{Y}, k}$, where $k$ is the $i$th bit of $b_t$.

                \item On the $E$ side of the $\lceil \log |T|\rceil$th tile to the $E$ of $B_{\mathfrak{Y}}$ is an indicative bit-string cap tile of type $B_{\mathfrak{Y},c}$. 

                \item On the $S$ side of $B_{\mathfrak{Y}}$, copies of $v_{\mathfrak{Y}}$ are placed, until stopping at the edge of the macrotile.
            \end{enumerate}
            \item If $d\in\{N,S\}$
            
            \begin{enumerate}
                \item The vertical tile $v_\mathfrak{Y}$ is placed on the $d$ side of the center location, with additional copies of it placed to its $d$ direction, until stopping at the edge of the macrotile.
            \end{enumerate}

            \item If $d\in\{E,W\}$

            \begin{enumerate}
                \item The horizontal tile $h_{\mathfrak{Y}}$ is placed on the $d$ side of the center location, with additional copies of it placed to its $d$ direction, until stopping at the edge of the macrotile.
            \end{enumerate}
            
        \end{enumerate}
    \end{enumerate}
\end{enumerate}

\begin{figure}
    \centering
    \includegraphics[width=1\linewidth]{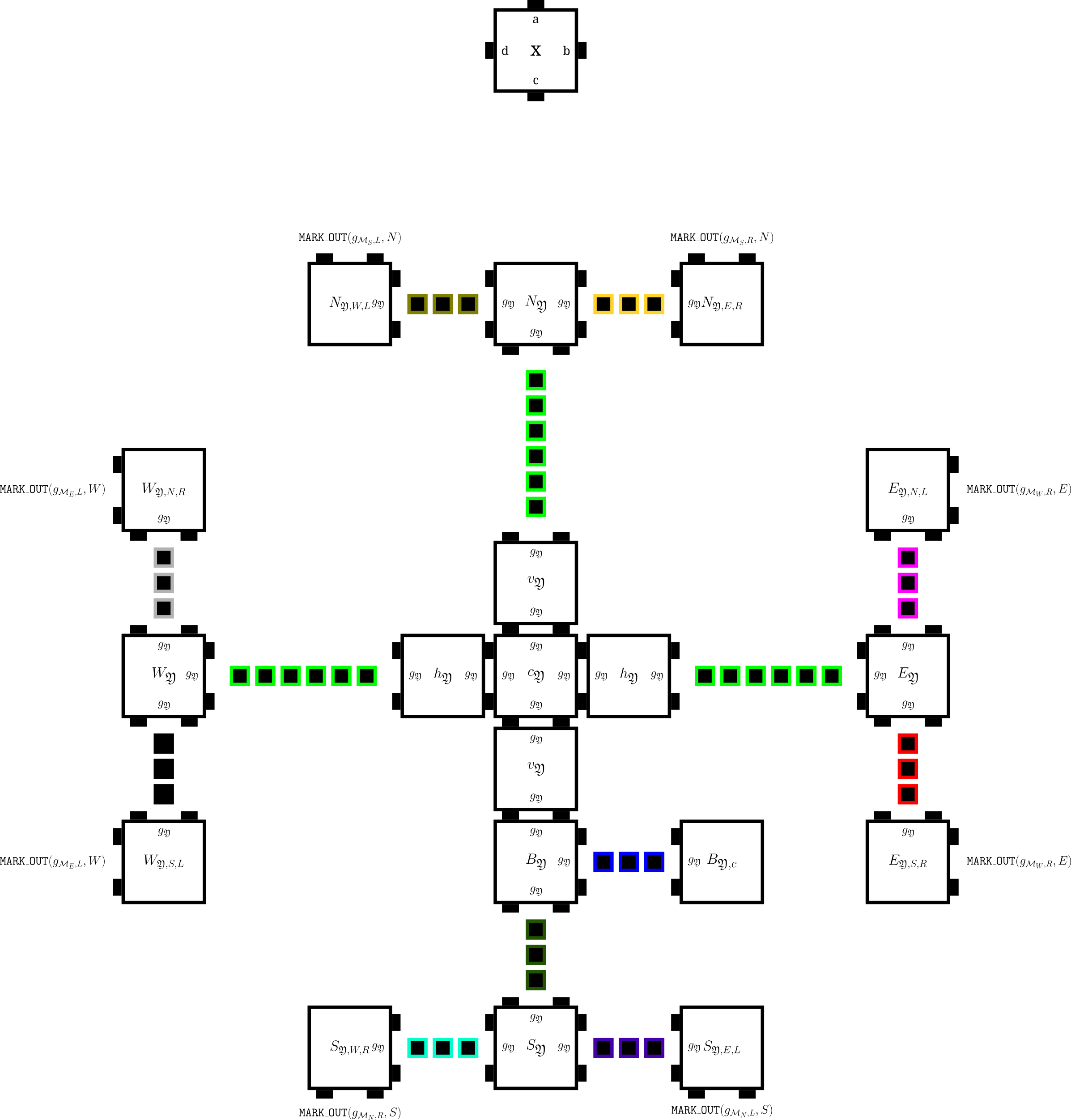}
    \caption{(Top) An example seed $\sigma$ of a temperature 1 staggered syncTAS $\calT$. (Bottom) An example seed assembly $\mathfrak{Y}(\calT)$ such that $\calS '=(\calU, \mathfrak{Y}(\calT),2)$ synchronously simulates $\calT$. Portions of the seed assembly which could not be shown due to space constraints are outlined with colors to indicate their contents, a guide to which is provided in \cref{tab:temp2iufortemp1-seed-guide}.}
    \label{fig:temp2iufortemp1-seed-example}
\end{figure}

\begin{table}[!htp]\centering\small

%\resizebox{\textwidth}{!}{ % use this if the table is too large
\begin{tabular}{lrr}\toprule
Color &Portion of Seed Assembly \\\midrule
Light Green &Items \ref{item:green1}, \ref{item:green2}, \ref{item:green3} \\
Dark Green &\cref{item:dark-green} \\
Blue &\cref{item:blue} \\
Gold &\cref{item:gold} \\
Yellow &\cref{item:yellow} \\
Pink &\cref{item:pink} \\
Red &\cref{item:red} \\
Purple &\cref{item:purple} \\
Cyan &\cref{item:cyan} \\
Black &\cref{item:black} \\
Gray &\cref{item:grey} \\
\bottomrule
\end{tabular}
\caption{A guide to truncated portions of the seed assembly shown in \cref{fig:temp2iufortemp1-seed-example}}\label{tab:temp2iufortemp1-seed-guide}
\end{table}

Thus, each tile of $\sigma$ is represented by a macrotile of the form shown in \cref{fig:temp2iufortemp1-seed-example}, which consists of a tile in the central square and a path of tiles to each of the macrotile boundaries on sides where glues are exposed or form bonds within the seed, along with a binary encoding of the tile type $t\in T$ that the macrotile maps to. On each macrotile boundary in which a tile in $\sigma$ has an exposed glue (including the null glue), glues are presented along the side of the $L$ and $R$ regions which encode $\mathtt{SIDE\mathunderscore ENCODING}(t)$, or $\mathtt{SIDE\mathunderscore ENCODING}(t)^R$, depending on whether the glues face an $L$, or $R$ region of the adjacent macrotile, respectively, which are sufficient to initiate the macrotile growth procedure (shown in \cref{sec:temp2iufortemp1-macrotile-growth-procedure}) in the adjacent macrotile location. For each direction in which a tile of $\sigma$ forms a bond with a neighboring tile in $\sigma$, there is a path connecting the corresponding macrotiles in $\mathfrak{Y}(\calT)$, so $\mathfrak{Y}(\calT)$ is a stable assembly that correctly represents $\sigma$ under $\mathfrak{R}(\calT)$ using scale factor $c$, and will initiate valid simulation.

\subsection{Representation function} \label{sec:temp2iufortemp1-representation-function}

In this section we present a representation function generator $\mathfrak{R}$ that, when provided with a staggered temperature $1$ syncTAS $\mathcal{T} = (T, \sigma, 1)$, outputs a representation function $\mathfrak{R}(\calT)$ that maps macrotiles of $\mathcal{S'}=(\calU, \mathfrak{Y}(\calT),2)$. 

Let us first define the function $\mathtt{PAD\mathunderscore BIN}^{-1} : \{0,1\}^* \rightarrow \mathbb{N} \times \mathbb{N}$ such that $\mathtt{PAD\mathunderscore BIN}(b)$ (where $b$ is a binary string) returns $b$ in integer form, and $|b|$. Additionally, for $n \in \mathbb{N}, n < |T|$, let $\texttt{INDEX}_T^{-1}(n)$ be a function which returns the $n$th tile type in $\texttt{LEX}_T(T)$.

$\mathfrak{R}(\calT)$ has two procedures, \textit{seed resolution}, and \textit{frontier resolution}, which are used to resolve macrotiles $M(l)$ where $l \in \dom \sigma$, and $l \notin \dom \sigma$, respectively. We will first describe the frontier-resolution process of $\mathfrak{R}(\calT)$.

Let location $l \in \mathbb{Z}^2$ such that $l \notin \dom \sigma$. If $\mathcal{S}'$ simulated the attachment of a tile $t \in T$ in the macrotile $M(l)$, then the tile lookup procedure shown in \cref{sec:temp2iufortemp1-tile-lookup-procedure} must have been performed on some side $d \in \{N,E,S,W\}$ of $M(l)$, and the co-simulation of $\mathcal{M}_d$ must have entered the accept state. By \cref{sec:temp2iufortemp1-tile-lookup-procedure}, there is tile type $t_{accept,d,L}$ (for all $d \in \{N,E,S,W\}$) whose placement in $M(l)$ indicates that the co-simulation of $\mathcal{M}_d$ entered an accept state. Recall that when $\mathcal{M}_d$ enters an accept state, the $\lceil \log|T|\rceil+4\lceil \log |G_T|\rceil+5$ tape locations to its left contain $\texttt{OUTPUT}(t)$. The frontier resolution process of $\mathfrak{R}(\calT)$ is as follows: 

\begin{enumerate}
\item For each $d \in \{N,E,S,W\}$
\begin{enumerate}
\item Observe each location $l'$ in the $L_d$ region of $M(l)$
\item If $l'$ contains a tile type $t_{accept,d} \in T_{accept,d}$, do the following :
\item Observe the tiles which encode $\texttt{OUTPUT}(t)$ that are to placed in the $\lceil \log |T|\rceil +4\lceil \log |G_T|\rceil + 5$  locations to the $d+1$ of $t_{accept,d}$.
\item Let $b_t$ be the substring of $\texttt{OUTPUT}(t)$ that is the length $\lceil \log |T|\rceil$ binary string corresponding to the tile to which $M(l)$ will be mapped.
\item Return $\mathtt{INDEX}_T^{-1}(\mathtt{PAD\mathunderscore BIN}^{-1}(b_t))$.
\end{enumerate}
\end{enumerate}

For each location $l \in \mathbb{Z}^2$ such that $l \in \dom \sigma$, the macrotile $M(l)$ was generated by the seed-generation function, shown in \cref{sec:temp2iufortemp1-seed-structure}. Recall that for each macrotile corresponding to a seed location, there is an indicative bit-string to the south of the center of the macrotile. The seed resolution process of $\mathfrak{R}(\calT)$ is as follows : 

\begin{enumerate}
\item Observe each location $l'$ in $M(l)$
\item If the tile $B_{\mathfrak{Y}}$ is placed in location $l'$ :

\begin{enumerate}
\item Let $b_t$ be the bit-string whose value is encoded in the $\lceil \log |T|\rceil$ tiles to the $E$ of $B_{\mathfrak{Y}}$.
\item Return $\mathtt{INDEX}_T^{-1}(\mathtt{PAD\mathunderscore BIN}^{-1}(b_t))$
\end{enumerate}
\end{enumerate}

Note that under the previously described tile resolution process, all macrotiles which resolve to a tile type in $T$ do so simultaneously. That is to say, for any macrostep $i$, and for any location $l \in \mathbb{Z}^2$ such that $l$ resolves to a tile type $t \in T$, $l$ will map under $R$ to $T$ at \textit{exactly} assembly step $i \cdot s + l_T^2+\lceil \log |T|\rceil + 4\lceil \log |G_T|\rceil +5$, which is the time step wherein the accept tile is placed, and is discussed in \cref{sec:temp2iufortemp1-tile-lookup-procedure}.

\subsection{Correctness of the construction for \cref{thm:T2-IU-for-T1}}

Here we present our argument for the correctness of our construction for \cref{thm:T2-IU-for-T1}. Recall that throughout Sections \ref{sec:temp2iufortemp1-reusable-gadgets}-\ref{sec:temp2iufortemp1-representation-function} we have discussed the construction of a tile set $\calU$, seed generation function $\mathfrak{Y}$, and representation function generator $\mathfrak{R}$ such that, beginning with a system $\calT = (T,\sigma, 1)$ that is staggered, $\calS'= (\calU, \mathfrak{Y}(\calT), 2)$ synchronously simulates $\calT$ under $\mathfrak{R}(\calT)$. For the duration of this section, we will use $R$ to refer to $\mathfrak{R}(\calT)$, and $\sigma'$ to refer to $\mathfrak{Y}(\calT)$. In our argument for the correctness of our construction, we will show that $\calU$ is a fixed tile set, as well as why $\calS'$ : 

\begin{enumerate}
\item Synchronously simulates $\calT$
\item Has spatial scale factor $O(|T|\log^2|T|)$
\item Has temporal scale factor $O(|T|^2\log^4|T|)$
\item Is staggered
\item Is locally adjacent-only
\end{enumerate}

Let us start by showing that under $R$ and with spatial scale factor $c$ and temporal scale factor $s$, $\calS'$ synchronously simulates $\calT$ by showing that (1) $\calT$ synchronously follows $\calS'$, (2) $\calS'$ models $\calT$, and (3) these imply that $\calT$ and $\calS'$ have equivalent productions.

First, we show that $\calT$ synchronously follows $\calS'$ under $R$ with temporal scale factor $s$. Let $\vec{\beta} = \beta_0, \beta_1, \beta_2, \ldots$ be an arbitrary assembly sequence in $\calS'$. We will prove the existence of a valid assembly sequence $\vec{\alpha} = \alpha_0, \alpha_1, \alpha_2, \ldots$ in $\calS$, where for each $i'$ such that $i' \mod s == 0$, $R^*(\beta_{i'})=\alpha_i$, using induction. Our inductive hypothesis is that for arbitrary $i'$ such that $i' \mod s == 0$, $\alpha_i = R^*(\beta_{i'})$. For the base case, by construction (as shown in \cref{sec:temp2iufortemp1-seed-structure}), for $\beta_0 = \sigma'$, $R^*(\beta_0)=\sigma=\alpha_0$. For the inductive step, we begin with arbitrary $i'$ such that $i' \mod s == 0$ and $\beta_{i'}$ such that $R^*(\beta_{i'})=\alpha_i$. As explained previously, the subsequent sequence of $s$ assembly steps in $\calS'$ will be the $(i+1)$th macrostep in $\calS'$ (where $i = i' / s)$ and simulate the $(i+1)$th step in $\calS$. This is because at the beginning of each macrostep in $\calS'$, for any macrotile $M(l)$ which had input glues appear adjacent to $M(l)$ during the previous macrostep, those input glues will have synchronously arrived at time step $is$. If a set of input glues have been presented to the macrotile $M(l)$ which are sufficient to allow for the simulation of the attachment of a tile type $t \in T$ (both of which will have arrived at time stepped $is$), then within the macrotile will occur growth sufficient to resolve $M(l)$ to a tile of type $t$ (by the entrance of an accept state during the co-simulation of $\mathcal{M}_d$, where $d \in \{N,E,S,W\}$ is the side of $M(l)$ which received the aforementioned input glues), and output representations of the information necessary to allow input growth in adjacent macrotiles will grow to the outside edges of all output sides of $M(l)$ at exactly time step $(i+1)s$ (As shown in \cref{sec:temp2iufortemp1-synchronization}). Note that the null glue will also be represented on output edges. Thus, for $i'' = (i+1)\cdot s$, $\alpha_{i+1}=R^*(\beta_{i''})$, which must be a valid assembly such that $\alpha_i \rightarrow^\calT_1 \alpha_{i+1}$ because if and only if a macrotile in $\calS'$ changed from mapping to empty space to mapping to a tile in $T$, a tile of the same type could have attached in the corresponding location in $\calT$. Therefore, $\vec{\alpha}$ is a valid assembly sequence in $\calT$ and $\calT$ synchronously follows $\calS'$ under $R$ with temporal scale factor $s$. 

We now show that $\calS'$ models $\calT$ under $R$. we do this by showing how, for every producible assembly $\alpha \in \prodasm{\calT}$, there exists a set of producible assemblies $\Pi \subset \prodasm{\calS'}$ such that each maps under $R$ to $\alpha$ and for every assembly $\beta$ producible from $\alpha$ in $\calT$ in one step, that (1) every assembly in $\Pi$ can grow into an assembly in $\calS'$ representing any such $\beta$, and (2) every assembly sequence in $\calS'$ from some $\alpha'$ that represents $\alpha$ in some $\beta'$ that represents $\beta$ includes an assembly in $\Pi$. 

For this we use induction. Our induction hypothesis is that, given an arbitrary producible assembly $\alpha_i \in \prodasm{\calT}$ there exists a set of producible assemblies $\Pi \subset \prodasm{\calS'}$ that satisfies the requirements stated above. The base case is again the seed assemblies $\sigma = \alpha_0$ in $\calT$ and $\sigma'$ in $\calS'$. In this case, the only element in $\Pi_0$ is $\sigma'$. By construction (as shown in \cref{sec:temp2iufortemp1-seed-structure}), $R^*(\sigma') = \sigma =\alpha_0$. Trivially, since $\sigma'$ is the only element in $\Pi_0$, the second condition holds (an assembly from $\Pi_0$ is included in every assembly sequence from $\sigma'$). For the first condition (i.e. for each assembly that $\alpha_0$ can grow into in one step, the assembly in $\Pi_0$ can grow into at least one assembly that represents it) we can recall how our construction works. For every empty location $l$ in $\calT$ toward which a non-null glue of a seed tile is facing, input growth will enter the corresponding macrotile location $M(l)$ in $\mathcal{S'}$. $\calT$ is staggered, and thus there is never more than one of such glues facing a location. For any location which is faced by a null glue, input growth will enter the corresponding macrotile location $M(l)$, however, by our tile lookup procedure it is impossible for the aforementioned input growth to either simulate the placement of a tile $t \in T$, or block input growth on another side of $M(l)$ which may simulate a tile placement. If the glue facing a location is sufficient to simulate the placement of a tile of type $t \in T$, the macrotile will resolve to such a tile and then begin growth of outputs, each of which will abut output glues to the edges of the macrotile on the $s$th assembly step (where $s$ is the temporal scale factor of the simulation). Because of this, for every location $l$, for every tile type $t$ that could be placed there at this step in $\calT$, there is an assembly sequence beginning from the sole assembly $\sigma'$ contained in $\Pi_0$ that causes the corresponding macrotile to resolve to that tile type, and therefore the assembly $\sigma' \in \Pi_0$ can grow into a representation of anything that $\sigma$ can grow into and the base case is shown.

For the induction step, given arbitrary producible assembly $\alpha_i \in \prodasm{\calT}$ and its corresponding set $\Pi_i$, we will show that for any assembly $\alpha_{i+1}$ such that $\alpha_i \rightarrow^1_{\calT} \alpha_{i+1}$ there exists a set $\Pi_{i+1}$ that satisfies the necessary conditions. Recall that the construction is designed so that every step $i$ of $\calT$ is simulated by $s$ steps in $\calS'$, with each macrostep in $\calS'$ completing at exactly step $i\cdot s$, with every macrotile receiving its final output tiles at that step. We will define $\Pi_{s} \subset \prodasm{\calS'}$ to include the assemblies producible during the final assembly step of each macrostep. Let $\beta \in \Pi_i$, and thus $R^*(\beta_i)=\alpha_i$. Let $0 < j \le s$ be the assembly step in $\calS'$ within macrostep $i$. Let $k$ be the number of tiles added in $\alpha_i$ during step $i+1$ to create $\alpha_{i+1}$. As previously shown, the corresponding $k$ macrotile locations in $\calS'$ will grow to resolve to tiles in $\calT$ during macrostep $i$, with there being at least one valid assembly sequence that leads to macrotiles resolving to $\calS'$ to the same tiles placed in $\calT$ to form $\alpha_{i+1}$. Additionally, as previously shown, for any tile type $t\in T$ that could appear in a location in $\calT$ during step $i$, the corresponding macrotile location in $\calS'$ could grow into a macrotile representing a tile of type $t$. Lastly, as was previously shown, for the first $r = l_T^2+\lceil \log |T|\rceil + 4\lceil \log |G_T|\rceil + 5$ steps of macrostep $i$, the assembly $\beta_{i+j}$ (for $0 \le j < r)$ will be equivalent to $\alpha_i$, i.e. $R^*(\beta_{i+j})=\alpha_i$. On the $r$th step of macrostep $i$, all $k$ macrotiles resolve to a tile type in $T$, and for all subsequent steps $r \le r' \le s$, $R^*(\beta_{i+r'})=\alpha_{i+1}$. Any such assembly will now deterministically grow into an assembly that is final for macrostep $i$, and that final assembly of macrostep $i$ will be in $\Pi_{i+1}$. This will satisfy the requirements for the set $\Pi_{i+1}$ because (1) the final assembly created by macrostep $i$ and representing $\alpha_{i+1}$ is guaranteed, for any $\alpha_{i+2}$ such that $\alpha_{i+1}\rightarrow^1_{\calT}\alpha_{i+2}$ to be able to grow into some assembly $\beta_{i+2}\in\prodasm{\calT}$ such that $R^*(\beta_{i+2})=\alpha_{i+2}$, and (2) as soon as any assembly resolves under $R$ to $\alpha_{i+1}$ it is in $\Pi_{i+1}$. Thus, the induction step is shown, and $\calS'$ is shown to model $\calT$. 

$\calT$ and $\calS'$ have equivalent productions because $\calT$ follows $\calS'$ and $\calS'$ models $\calT$, which means that, starting from the seed assemblies $\sigma$ in $\calT$, and $\sigma'$ in $\calS'$, an assembly $\alpha$ is producible (resp. terminal) in $\calT$ if and only if an assembly $\beta$ is producible (resp. terminal) in $\calS'$ such that $\beta$ maps under $R$ to $\alpha$ or an equivalent assembly. Thus, $\calS'$ is shown to synchronously simulate $\calT$ under $R$ and using temporal scale factor $s$.

In \cref{sec:temp2iufortemp1-synchronization}, it was shown that the spatial scale factor of $\calS'$ is $O(|T|\log^2|T|)$, and the temporal scale factor of $\calS'$ is $O(|T|^2\log^4|T|)$.

We show that $\calS'$ is a staggered system by inspecting the growth of its macrotiles. For $\calS'$ not to be staggered, there would need to be one possible empty location $l \in \mathbb{Z}^2$ where, during a single assembly step, more than one neighboring location receives a tile, and each has a non-null glue facing $l$. This can only happen if there are two paths of growth that simultaneously arrive at locations neighboring $l$. As was described in \cref{sec:temp2iufortemp1-macrotile-growth-procedure}, growth within a macrotile is initiated by two separated paths of growth in the $L_d$ and $R_d$ (where $d \in \{N,E,S,W\}$ is the side of macrotile where growth is initiated) regions of the macrotile which co-simulate $\mathcal{M}_d$. These two paths of growth regularly place tiles which are adjacent to the same location in the single-tile wide space between the $L_d$ and $R_d$ regions, however only one of said tiles (that being the one in the $L_d$ region) may ever have a glue which abuts this location. Each propagation gadget which is created is staggered, as it only consists of a single path of growth. In all locations where a tile of a propagation gadget presents a glue to an empty location adjacent to a tile placed by another propagation gadget during the same assembly step (i.e. when synchronization paths attach), the empty location is always in the single-tile wide space between two neighboring propagation gadgets, and no more than one such tile will have a glue which faces the aforementioned empty location. Lastly, in the case of growth of modified propagation gadgets, the only location wherein two glues from separate paths of tiles face the same location is those which are adjacent to the end of the decision path that may begin growth of subsequent propagation gadgets through a cooperative attachment. However, the last tile placed by the decision path is necessarily placed exactly $2$ assembly steps after the tiles whose glues future tiles use to attach cooperatively (as is seen in \cref{fig:temp2iufortemp1-crossing-the-gap}). Additionally, we note that, for all assembly steps, in all macrotile locations representing frontier locations, all paths which are contained in it are either connected to a single input glue from an adjacent macrotile, or are connected to a single null glue from an adjacent macrotile. This is because $\calT$ is a staggered system, $\calS'$ synchronously simulates $\calT$, and because all growth within a macrotile location must either halt, or express output in $s$ assembly steps. All paths which are connected to a null glue from an adjacent macrotile must necessarily never grow past the $L$ and $R$ regions on the side which they entered, and growth of these paths must cease after $<l_T^2+l_T-1$ assembly steps. All paths within a macrotile which are connected to a single input glue from an adjacent macrotile only reach those aforementioned $L$ and $R$ regions after $5l_T^2+4$ assembly steps, and thus cannot place a tile with a glue facing the same location as a tile from said path connected to a null glue during the same assembly step as it. Similarly, it is never the case that two outputs paths from adjacent macrotile locations present glues to the same location, as each place tiles along their edge in the same order, and during the same assembly steps, therefore, each tile which contains a glue facing an adjacent macrotile location will have said glue face a tile placed in the same assembly step, and not an empty location.

We similarly show that $\calS'$ is locally adjacent-only by inspecting the growth of its macrotiles. We first state that the glues which are present within each propagation gadget are disjoint, except for where one propagation gadget chains off of another. For example, the tiles on the side of a propagation gadget which transmits the values along its $S$ side to its $N$ side, and grows through the $fill(2,1)$ region contains glues which only may be used for its own construction, aside from those which are presented to its $N$ side. Those glues to its $N$ side may only be attached to by glues on the $S$ side of a subsequent propagation gadget which grows through $fill(2,2)$. The same is true for gadgets which perform the co-simulation of $\mathcal{M}_d$ (described in \cref{sec:temp2iufortemp1-synchronization}). By our construction, it is only the case that two paths of tiles may crash into each other if (1) Two adjacent macrotiles resolve to a tile type $t \in T$ during the same assembly step or (2) Input growth on one side of a macrotile failed to simulate the attachment of a tile, whereas input growth on another side of the same macrotile successfully simulated the attachment of a tile. In the first case, the glue used by the propagation gadgets which grow towards the edge of two adjacent macrotiles from each direction are completely disjoint, and it is impossible for a single tile to form a bond with both of them. The same is true for the second case, as the glues that are presented along the edge of input growth which fails to place a tile (i.e. when the co-simulation of $\mathcal{M}_d$ halts prematurely) are completely disjoint from those of a propagation gadget which attempts to grow through the affected $L$ and $R$ regions, meaning that no tile may form a bond with both of them. We then note that the growth of propagation gadgets, including when propagation gadgets chain off of each other is locally adjacent-only (by Sections \ref{sec:temp2iufortemp1-reusable-gadgets} and \ref{sec:temp2iufortemp1-macrotile-growth-procedure}), and that the simulation of $\mathcal{M}_d$ for all $d \in \{N,E,S,W\}$ is locally adjacent-only (by \cref{sec:temp2iufortemp1-tile-lookup-procedure}), therefore $\calS'$ is locally adjacent-only.

Lastly, we will show that $\calU$ is a fixed size tile set. Let the tile types from which each propagation gadget attaches be the set $T_{prop, Region,d}$ where $Region$ is the region of a macrotile that the propagation gadget attaches, and $d \in \{N,E,S,W\}$ is the side of a macrotile that experience input growth which may result in the placement of said propagation gadget. There are a total of $18$ regions of a macrotile wherein a propagation gadget grows, and there are $4$ directions from which input growth may enter a macrotile. Therefore, $\calU$ must contain the tile types necessary to allow the self-assembly of $48$ unique propagation gadgets (note that each propagation gadget must only be able to transmit information corresponding to the symbols in $\Sigma_{Macrotile}$), we will call this set of tiles $T_{propagation}$. $\calU$ must also contain the tile types necessary to allow the growth of the co-simulation of $\mathcal{M}_d$,for all $d \in \{N,E,S,W\}$, the tile types necessary for this are each in the set $T_{\mathcal{M}_d}$. Including the tile types necessary to allow growth of the seed, $\calU = T_{propagation}\cup T_{\mathcal{M}_N}\cup T_{\mathcal{M}_E}\cup T_{\mathcal{M}_{S}}\cup T_{\mathcal{M}_W}\cup T_{\mathfrak{Y}}$. %\todo{PD: I don't know if this paragraph is necessary... MJP: It is fine and good additional clarification.}

Thus, there is a temperature $2$ tile set $\calU$, a representation function $\mathfrak{R}$, and a seed function generator $\mathfrak{Y}$ such that for any staggered temperature $1$ syncTAS $\calT = (T,\sigma, 1)$, $\calS' = (\calU,\mathfrak{Y}(\calT), 2)$ is a staggered, locally adjacent-only syncTAS that synchronously simulates $\calT$ under $\mathfrak{R}(\calT)$, and \cref{thm:T2-IU-for-T1} is proven.

\begin{corollary}
    \label{thm:T2-actually-IU}
    The tile set $\calU$ of \cref{thm:T2-IU-for-T1} is intrinsically universal for the syncTAM.
\end{corollary}

\begin{proof}
    Let $\mathfrak{R}$ and $\mathfrak{Y}$ be the functions witnessing that $\calU$ is IU for the class of temperature 1 systems.

    For $\mathcal{S}$ an arbitrary TAS, by \cref{thm:temp1simulation}, there is $\mathcal{S}_1$ with temperature 1 which simulates $\mathcal{S}$. Pose $\mathfrak{R}'(\mathcal{S}) = \mathfrak{R}(\mathcal{S}_1)$ and $\mathfrak{Y}'(\mathcal{S}) = \mathfrak{Y}(\mathcal{S}_1)$. 

    Hence, $\calU$ is also IU for the syncTAM, through $\mathfrak{R'}$ and $\mathfrak{Y}'$.

    A visual depiction of the statement of this proof, along with the proof itself are provided in \cref{fig:T2-actually-IU} and \cref{fig:T2-actually-IU-proof}, respectively.
\end{proof}

\begin{figure}
    \centering
    \begin{tikzpicture}[
    header/.style={align=center, text width=17ex, anchor=south},
    tas/.style={draw, gray, cloud, cloud ignores aspect, cloud puffs=39}
    ]
    \pgfdeclarelayer{background}
    \pgfdeclarelayer{foreground}
    \pgfsetlayers{background,main,foreground}
    \matrix (M) [
        matrix of nodes,
        nodes in empty cells, row sep={1cm},
        row 1/.style={nodes={header}},
        row 3/.style={text height=1cm, text width=17ex}
        %every node/.style={draw}
    ]{
    Representation function & Tileset & Temperature & Seed & Staggered + Locally Adjacent Only? \\
                            & $T$ & $\tau$      & $\sigma$ & \textsc{?} \\
                            &         &             &          &   \\
        |[draw, circle]|  $\mathfrak{R}'(\mathcal{S})^\dagger$   & $\mathcal{U}$ & $2$      & $\mathfrak{Y}'(\mathcal{S})^*$ & \textsc{Yes} \\    
    };

    \begin{pgfonlayer}{background}
        \node[fill=gray!20, fit=(M-1-1) (M-1-5)] {}; 
        \foreach \col in {1,...,5} {
            \draw[gray, dotted] (M-1-\col) -- +(0,-7);
        }
    \end{pgfonlayer}

    \node[fit=(M-3-2) (M-3-3), inner xsep=0] (machine-cell-left) {};
    \node at (machine-cell-left.center) {\huge $\mathfrak{R}'$};
    \draw[thick] (machine-cell-left.north east) ++(-.5em, 0) coordinate (machine-left-tr) rectangle (machine-cell-left.south west);    
    \node[fit=(M-3-4), inner xsep=0] (machine-cell-right) {};
    \node at (machine-cell-right.center) {\huge $\mathfrak{Y}'$};
    \draw[thick] (machine-cell-right.north west) ++(.5em, 0) coordinate (machine-right-tl) rectangle (machine-cell-right.south east);
    \fill[black!70] (machine-cell-left.east) ++(-.5em, .5cm) -- ++(1em, -.5cm) -- ++(-1em, -.5cm) -- cycle;

    \draw[->] (M-2-2) -- (machine-cell-left);
    \draw[->] (M-2-3) -- (machine-cell-left);
    \draw[->] (M-2-4) -- (machine-cell-right);
    
    \draw[->] (machine-cell-left) -- (M-4-1);
    
    \draw[->] (machine-cell-right) -- (M-4-4);
    \draw[->] (machine-cell-right) -- (M-4-5);

    \node[tas, fit=(M-4-2) (M-4-4)] (tas-out) {};
    \node[tas, fit=(M-2-2) (M-2-4)] (tas-in) {};

    \draw[decorate, decoration={snake}] (M-2-5) to (tas-in) to[out=180, in=135] (M-4-1) to[out=-45, in=-165] (tas-out); 

\end{tikzpicture}
    \caption{A graphical rendition of the statement of \cref{thm:T2-actually-IU}. The parts of the output system which depend on $\mathcal{T}$ are marked with a dagger, and with a star for a dependency on both $\calT$ and $\sigma$. Note that $\mathfrak{R}$ does not depend on $\sigma$, and that $\mathcal{S} = (T, \tau, \sigma)$ is an arbitrary syncTAS.}
    \label{fig:T2-actually-IU}
\end{figure}
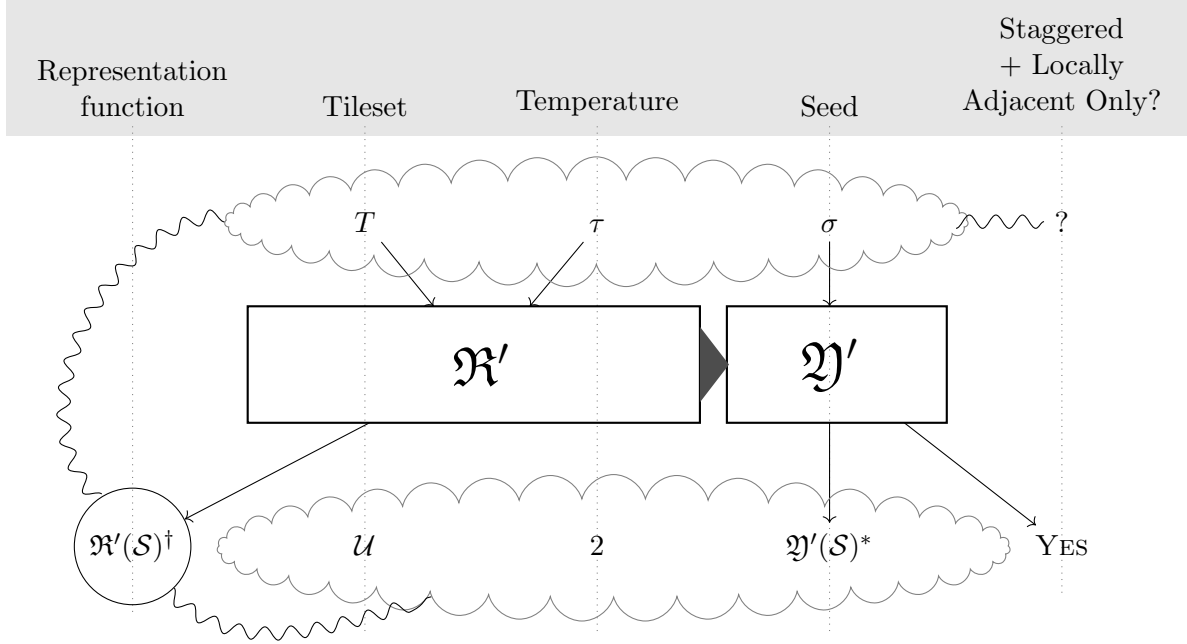

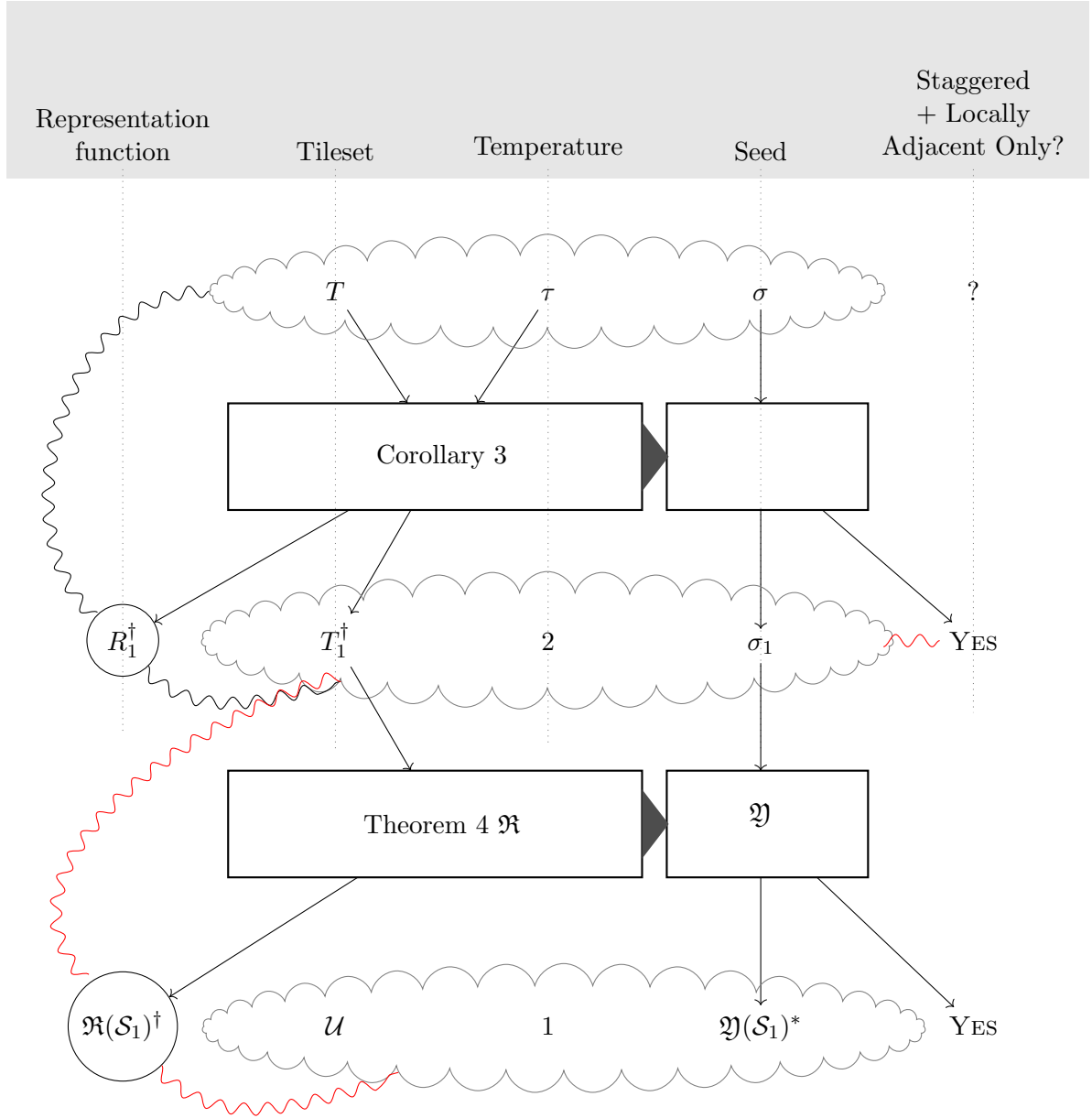
\begin{figure}
    \centering
    \begin{tikzpicture}[
    header/.style={align=center, text width=17ex, anchor=south},
    tas/.style={draw, gray, cloud, cloud ignores aspect, cloud puffs=39, inner sep=3pt}
    ]
    \pgfdeclarelayer{background}
    \pgfdeclarelayer{foreground}
    \pgfsetlayers{background,main,foreground}
    \matrix (M) [
        matrix of nodes,
        nodes in empty cells, row sep={1.5cm},
        row 1/.style={text height={}, nodes={header}},
        every odd row/.style={text height=1cm, text width=17ex}
        %every node/.style={draw}
    ]{
    Representation function & Tileset & Temperature & Seed & Staggered + Locally Adjacent Only? \\
                            & $T$ & $\tau$      & $\sigma$ & ? \\
                            &         &             &          &   \\
    |[draw, circle]|  $R_1^\dagger$   & $T_1^\dagger$ & $2$      & $\sigma_1$ & \textsc{Yes} \\
            &         &             &          &   \\
    |[draw, circle]|  $\mathfrak{R}(\mathcal{S}_1)^\dagger$   & $\mathcal{U}$ & $1$      & $\mathfrak{Y}(\mathcal{S}_1)^*$ & \textsc{Yes} \\
    };

    \begin{pgfonlayer}{background}
        \node[fill=gray!20, fit=(M-1-1) (M-1-5)] {}; 
        \foreach \col in {1,...,5} {
            \draw[gray, dotted] (M-1-\col) -- +(0,-9);
        }
    \end{pgfonlayer}

    \node[fit=(M-3-2) (M-3-3), inner xsep=0] (machine-cell-left) {};
    \node at (machine-cell-left.center) { \Cref{thm:temp1simulation}};
    \draw[thick] (machine-cell-left.north east) ++(-.5em, 0) coordinate (machine-left-tr) rectangle (machine-cell-left.south west);    
    \node[fit=(M-3-4), inner xsep=0] (machine-cell-right) {};
    \draw[thick] (machine-cell-right.north west) ++(.5em, 0) coordinate (machine-right-tl) rectangle (machine-cell-right.south east);
    \fill[black!70] (machine-cell-left.east) ++(-.5em, .5cm) -- ++(1em, -.5cm) -- ++(-1em, -.5cm) -- cycle;

    \draw[->] (M-2-2) -- (machine-cell-left);
    \draw[->] (M-2-3) -- (machine-cell-left);
    \draw[->] (M-2-4) -- (machine-cell-right);
    
    \draw[->] (machine-cell-left) -- (M-4-1);
    \draw[->] (machine-cell-left) -- (M-4-2);

    \draw[->] (machine-cell-right) -- (M-4-4);
    \draw[->] (machine-cell-right) -- (M-4-5);

    \node[tas, fit=(M-4-2) (M-4-4)] (tas-mid) {};
    \node[tas, fit=(M-2-2) (M-2-4)] (tas-in) {};

    \draw[decorate, decoration={snake}] (tas-in.west) to[in=135, out=180] (M-4-1) to[out=-45, in=-165] (tas-mid.south west); 

    \node[fit=(M-5-2) (M-5-3), inner xsep=0] (machine2-cell-left) {};
    \node at (machine2-cell-left.center) {\cref{thm:T2-IU-for-T1} $\mathfrak{R}$};
    \draw[thick] (machine2-cell-left.north east) ++(-.5em, 0) coordinate (machine2-left-tr) rectangle (machine2-cell-left.south west);    
    \node[fit=(M-5-4), inner xsep=0] (machine2-cell-right) {$\mathfrak{Y}$};
    \draw[thick] (machine2-cell-right.north west) ++(.5em, 0) coordinate (machine2-right-tl) rectangle (machine2-cell-right.south east);
    \fill[black!70] (machine2-cell-left.east) ++(-.5em, .5cm) -- ++(1em, -.5cm) -- ++(-1em, -.5cm) -- cycle;

    \draw[->] (M-4-2) -- (machine2-cell-left);
    \draw[->] (M-4-4) -- (machine2-cell-right);
    
    \draw[->] (machine2-cell-left) -- (M-6-1);
    %\draw[->] (machine2-cell-left) -- (M-6-2);

    \draw[->] (machine2-cell-right) -- (M-6-4);
    \draw[->] (machine2-cell-right) -- (M-6-5);

    \node[tas, fit=(M-6-2) (M-6-4)] (tas-out) {};

    \draw[decorate, decoration={snake}, red] (M-4-5) to (tas-mid) (tas-mid.south west) to[out=-165, in=135] (M-6-1) to[out=-45, in=-165] (tas-out); 

\end{tikzpicture}\qed
    \caption{Diagrammatic proof of \cref{thm:T2-actually-IU}. The parts of the output system which depend on $T$ are marked with a dagger, and with a star for a dependency on both $T$ and $\sigma$. Note that $\mathcal{S}_1 = (T_1, 1, \sigma_1)$ is a temperature 1 syncTAS.}
    \label{fig:T2-actually-IU-proof}
\end{figure}

\section{A Universal Temperature 1 syncTAM system}\label{sec:univ-temp1}

In this section, we chain together the series of previous constructions to yield the final, universal simulator for the syncTAM that operates at temperature 1.

\ToneIU*

\begin{proof}
The final link in the chain of constructions is to create a tile set that can be used at temperature 1 to simulate any system created for the universal simulator from the construction of \cref{thm:T2-IU-for-T1}. By careful design, those systems are temperature 2, staggered, and locally adjacent-only, and we already have a construction that takes any such system as input and outputs a temperature 1 tile set to simulate it: the construction of \cref{lem:temp1simulation-adj}. Furthermore, that construction yields a simulation scheme that is seed-independent. This means that, for the entire class of systems using the universal simulator for temperature 2, since they all use the same tile set and temperature and thus vary only by seed structure, the construction of \cref{lem:temp1simulation-adj} yields a single universal tile set that works at temperature 1 to simulate any system that can be simulated by the systems of the construction for \cref{thm:T2-IU-for-T1}. We will call this tile set $\mathcal{U}_1$.

Any arbitrary syncTAM system $\calS$ can be simulated by a system $\calS'$ that is temperature 2, locally adjacent-only, and staggered (\cref{lem:adj-only}). That system $\calS'$ can then be simulated by a temperature 1 system $\calS_1$ that is staggered (\cref{lem:temp1simulation-adj}). A seed over the tile set $\mathcal{U}$, which is universal for temperature 1 systems that are staggered, can then be designed so that the resulting system, $\mathcal{S}_{\mathcal{U}}$, using tile set $\mathcal{U}$ at temperature 2 simulates $\calS_1$. Finally, a seed over the tiles of $\mathcal{U}_1$ can be created so that the resulting system, $\mathcal{S}_{\mathcal{U}_1}$, using tile set $\mathcal{U}_1$ at temperature 1 simulates $\mathcal{S}_{\mathcal{U}}$. Thus, $\mathcal{S}_{\mathcal{U}_1}$ simulates arbitrary syncTAM system $\calS$ using tile set $\mathcal{U}_1$ at temperature 1, proving \cref{thm:T1-IU}.
\end{proof}

\section{Preserving directedness}\label{sec:directed}

In \cite{IUSA} it was shown that the aTAM is intrinsically universal. However, in \cite{DirectedNotIU} it was proven that the class of systems that are directed in the aTAM is not, namely that there exists no universal simulator that is able to remain directed while simulating all directed systems. In contrast, we now show that this is another point of differentiation between the aTAM and syncTAM. 

\directedIU*

\begin{proof}
    We prove \cref{thm:directed} by construction, and our construction requires a very few, simple modifications of the previous constructions. Let $\calS \in \mathfrak{D}$ be an arbitrary syncTAM system which is directed. We first create locally adjacent-only, staggered, temperature 2 syncTAM system $\calS'$ using the same construction as in the proof of \cref{lem:adj-only} with one slight modification. 
    Namely, rather than having a set of $12$ tile types for each direction to nondeterministically simulate ``Go-First Dice,'' $11$ tile types are discarded and only $1$ is kept for each direction (i.e. $N=0, E=1, S=2, W=3$), thus deterministically hard-coding the input delay for each direction. 
    Since $\calS$ is directed, regardless of the ordering of growth of the input sides (within a macrostep) of a macrotile representing a tile placement, only one tile type could legally be placed in $\calS$ so only that tile type can be simulated in $\calS'$. 
    The only other points of possible nondeterminism in the original construction of $\calS'$ occur in situations where a set of input glue(s) could yield placement of tiles of multiple types (e.g. the possible nondeterministic selection of a probe type to grow in a given location, or placement of one of multiple valid tiles in the central POC). Again, since $\calS$ is directed there will be no possible nondeterminism in these locations, and thus this modified $\calS'$ will retain all of the properties of the original simulator while also remaining directed.

    The conversion of $\calS'$ to a temperature 1 system via the construction for \cref{lem:temp1simulation-adj} can use that construction with no modification, as there is no nondeterminism within the resulting system when simulating a directed system.

    The temperature 1, locally adjacent-only, staggered system that results can be simulated by a universal simulator at temperature 2 that is constructed using the same construction as for \cref{thm:T2-IU-for-T1} with one very small modification. The original construction simulates nondeterministic Turing machines with two transition functions, with two tiles that can nondeterministally attach at the end of each zig row of a machine in a $L_d$ region to determine which transition function to use. The modified construction instead discards one tile type of each pair that can nondeterministically attach in these locations (thus discarding one tile type for each of the four directions), thereby hard-coding the selection and creating a deterministic Turing machine. As before, if the input system being simulated is directed, there will only ever be one choice of tile type that a Turing machine simulation can make for which tile type to simulate, and thus the correctness of simulation remains. Since this is the only point of nondeterminism, the modified universal simulator remains directed while simulating.

    Finally, again we can create a temperature 1 system $\calS_{\mathcal{U}_2}$ that simulates the universal simulator at temperature 2 via the construction of \cref{lem:temp1simulation-adj} without modification. This results in the tile set of $\calS_{\mathcal{U}_1}$ being a universal simulator of directed systems that works at temperature 1 and itself remains directed. This tile set is a witness that $\mathfrak{D}$, the class of directed syncTAM systems, is intrinsically universal. 
\end{proof}

% \begin{theorem}
%     There is an unseeded TAS $\calU$, a representation function $\mathfrak{R}$ and a seed function $\mathfrak{S}$ such that:
%     \begin{itemize}
%         \item $(\calU, \mathfrak{R}, \mathfrak{S})$ is IU for syncTAM systems
%         \item for any seeded TAS $\calS$ which is directed in the syncTAM, so is the seeded TAS $(\calU, \mathfrak{S}(\calS))$. 
%     \end{itemize}
% \end{theorem}

% Using the ``go-first dice'', when simulating a directed system, the dice results can be hardcoded (always North, then West ...)

\printbibliography

%\appendix
%\input{two-counter}

%\newpage
%\input{appendix}

\end{document}